\documentclass[aps,prx,floatfix,nofootinbib,superscriptaddress,twocolumn,longbibliography,10pt,raggedbottom]{revtex4-2}
\pdfoutput=1
\usepackage[english]{babel}

\usepackage{xcolor}
\usepackage{times}
\usepackage{graphicx}
\usepackage{xcolor}
\usepackage{xr-hyper}

\let\coloneqq\relax

\usepackage[T1]{fontenc}

\usepackage{amsthm}
\usepackage{amssymb}
\usepackage{amsmath}
\usepackage{bbold}
\usepackage{bbm}
\usepackage{nonfloat}
\usepackage{xcolor}
\definecolor{myrefcolor}{rgb}{0.067,0.5,0.5}
\usepackage[
    breaklinks,
    pdftex,
    colorlinks=true,
    linkcolor=myrefcolor,
    citecolor=myrefcolor,
    urlcolor=myrefcolor
]{hyperref}
\usepackage{braket}
\usepackage{dsfont}
\usepackage{mathdots}
\usepackage{mathtools}
\usepackage{enumerate}
\usepackage[shortlabels]{enumitem}
\usepackage{csquotes}
\usepackage{stmaryrd}
\usepackage[cal=boondox]{mathalfa}
\usepackage{graphicx}
\usepackage{stackengine}
\usepackage{scalerel}
\usepackage{microtype}
\usepackage{array}
\usepackage{makecell}
\newcolumntype{x}[1]{>{\centering\arraybackslash}p{#1}}
\usepackage{tikz}
\usepackage{pgfplots}
\usetikzlibrary{shapes.geometric, shapes.misc, positioning, arrows, arrows.meta, decorations.pathreplacing, decorations.pathmorphing, patterns, angles, quotes, calc}
\usepackage{booktabs}
\usepackage{xfrac}
\usepackage{siunitx}
\usepackage{centernot}
\usepackage{comment}
\usepackage{chngcntr}
\usepackage{breakurl}
\def\clearthms#1{ \@for\tname:=#1\do{\cleartheorem\tname} }

\newtheorem*{thm*}{Theorem}
\newtheorem{thm}{Theorem}

\newtheorem*{prop*}{Proposition}
\newtheorem{lemma}[thm]{Lemma}
\newtheorem*{lemma*}{Lemma}

\newtheorem{cor}[thm]{Corollary}
\newtheorem*{cor*}{Corollary}

\newtheorem*{cj*}{Conjecture}
\newtheorem{Def}[thm]{Definition}
\newtheorem*{Def*}{Definition}

\newtheorem{problem}[thm]{Problem}

\newtheorem{remark}[thm]{Remark}

\makeatletter
\def\thmhead@plain#1#2#3{%
  \thmname{#1}\thmnumber{\@ifnotempty{#1}{ }\@upn{#2}}%
  \thmnote{ {\the\thm@notefont#3}}}
\let\thmhead\thmhead@plain
\makeatother

\theoremstyle{definition}

\newcommand{\bb}{\begin{equation}\begin{aligned}\hspace{0pt}}
\newcommand{\be}{\begin{equation}\begin{aligned}\hspace{0pt}}
\newcommand{\bbb}{\begin{equation*}\begin{aligned}}
\newcommand{\ee}{\end{aligned}\end{equation}}
\newcommand{\eee}{\end{aligned}\end{equation*}}
\newcommand*{\coloneqq}{\mathrel{\vcenter{\baselineskip0.5ex \lineskiplimit0pt \hbox{\scriptsize.}\hbox{\scriptsize.}}} =}

\newcommand\ceil[1]{\left\lceil#1\right\rceil}
\newcommand{\texteq}[1]{\stackrel{\mathclap{\scriptsize \mbox{#1}}}{=}}
\newcommand{\eqt}[1]{\stackrel{\mathclap{\scriptsize \mbox{#1}}}{=}}
\newcommand{\leqt}[1]{\stackrel{\mathclap{\scriptsize \mbox{#1}}}{\leq}}

\newcommand{\geqt}[1]{\stackrel{\mathclap{\scriptsize \mbox{#1}}}{\geq}}
\newcommand{\ketbra}[1]{\ket{#1}\!\!\bra{#1}}

\newcommand{\tcb}[1]{{\color{blue} #1}}

\newcommand{\R}{\mathds{R}}
\newcommand{\N}{\mathds{N}}

\newcommand{\E}{\mathds{E}}

\DeclareMathOperator{\Tr}{Tr}

\DeclareMathAlphabet{\pazocal}{OMS}{zplm}{m}{n}

\DeclareMathOperator{\Id}{Id}

\DeclareMathOperator{\diag}{diag}

\newcommand{\HH}{\pazocal{H}}

\newcommand{\D}{\pazocal{D}}

\newcommand{\NN}{\mathcal{N}}

\newcommand{\lsmatrix}{\left(\begin{smallmatrix}}
\newcommand{\rsmatrix}{\end{smallmatrix}\right)}

\stackMath

\stackMath

\makeatletter
\newcommand*\rel@kern[1]{\kern#1\dimexpr\macc@kerna}
\newcommand*\widebar[1]{%
  \begingroup
  \def\mathaccent##1##2{%
    \rel@kern{0.8}%
    \overline{\rel@kern{-0.8}\macc@nucleus\rel@kern{0.2}}%
    \rel@kern{-0.2}%
  }%
  \macc@depth\@ne
  \let\math@bgroup\@empty \let\math@egroup\macc@set@skewchar
  \mathsurround\z@ \frozen@everymath{\mathgroup\macc@group\relax}%
  \macc@set@skewchar\relax
  \let\mathaccentV\macc@nested@a
  \macc@nested@a\relax111{#1}%
  \endgroup
}

\counterwithin*{equation}{part}
\counterwithin*{thm}{part}
\counterwithin*{figure}{part}

\tikzset{meter/.append style={draw, inner sep=10, rectangle, font=\vphantom{A}, minimum width=30, line width=.8, path picture={\draw[black] ([shift={(.1,.3)}]path picture bounding box.south west) to[bend left=50] ([shift={(-.1,.3)}]path picture bounding box.south east);\draw[black,-latex] ([shift={(0,.1)}]path picture bounding box.south) -- ([shift={(.3,-.1)}]path picture bounding box.north);}}}
\tikzset{roundnode/.append style={circle, draw=black, fill=gray!20, thick, minimum size=10mm}}
\tikzset{squarenode/.style={rectangle, draw=black, fill=none, thick, minimum size=10mm}}

\definecolor{Blues5seq1}{RGB}{239,243,255}
\definecolor{Blues5seq2}{RGB}{189,215,231}
\definecolor{Blues5seq3}{RGB}{107,174,214}
\definecolor{Blues5seq4}{RGB}{49,130,189}
\definecolor{Blues5seq5}{RGB}{8,81,156}

\definecolor{Greens5seq1}{RGB}{237,248,233}
\definecolor{Greens5seq2}{RGB}{186,228,179}
\definecolor{Greens5seq3}{RGB}{116,196,118}
\definecolor{Greens5seq4}{RGB}{49,163,84}
\definecolor{Greens5seq5}{RGB}{0,109,44}

\definecolor{Reds5seq1}{RGB}{254,229,217}
\definecolor{Reds5seq2}{RGB}{252,174,145}
\definecolor{Reds5seq3}{RGB}{251,106,74}
\definecolor{Reds5seq4}{RGB}{222,45,38}
\definecolor{Reds5seq5}{RGB}{165,15,21}

\usepackage{pgfplots}

\newtheorem{definition}{Definition}

\newcommand*{\addFileDependency}[1]{
  \typeout{(#1)}
  \@addtofilelist{#1}
  \IfFileExists{#1}{}{\typeout{No file #1.}}
}

\makeatother

 \newcommand{\deff}{d_{\text{eff}}}
\newcommand{\tildeO}{\tilde{O}\!}
\newcommand{\tildeTheta}{\tilde{\Theta}\!}
\newcommand{\tildeOmega}{\tilde{\Omega}\!}

\definecolor{tealblue}{HTML}{00AEB3}

\newcommand{\idop}{\mathbb{\hat{1}}}

\usepackage{algorithm}
\usepackage{algpseudocode}
\usepackage{mdframed}
\definecolor{antonio}{rgb}{.2,.5,.1}

\pgfplotsset{width=10cm,compat=1.9}
\usetikzlibrary{decorations.markings}

\newcommand{\nocontentsline}[3]{}
\newcommand{\tocless}[2]{\bgroup\let\addcontentsline=\nocontentsline#1{#2}\egroup}

\PassOptionsToPackage{pdftex,hyperfootnotes=false,pdfpagelabels}{hyperref}
\hypersetup{
	colorlinks = true,
	linkcolor = [rgb]{0.70,0.13,0.13},
	citecolor = [rgb]{0.13,0.55,0.13},
	urlcolor = [rgb]{0.25, 0.41, 0.88}}

\begin{document}

\setcounter{secnumdepth}{2}
\setlength{\parskip}{0.1pt}
\title{Learning quantum states of continuous-variable systems}

\author{Francesco A. Mele}
\thanks{\{\href{mailto:francesco.mele@sns.it}{francesco.mele}, \href{mailto:vittorio.giovannetti@sns.it}{vittorio.giovannetti}, \href{mailto:salvatore.oliviero@sns.it}{salvatore.oliviero}\}@sns.it}
\affiliation{NEST, Scuola Normale Superiore and Istituto Nanoscienze, Piazza dei Cavalieri 7, IT-56126 Pisa, Italy}

\author{Antonio A. Mele}
  \thanks{\{\href{mailto:a.mele@fu-berlin.de}{a.mele}, \href{mailto:l.bittel@fu-berlin.de}{l.bittel}, \href{mailto:lorenzo.leone@fu-berlin.de}{lorenzo.leone}\}@fu-berlin.de}
  \affiliation{Dahlem Center for Complex Quantum Systems, Freie Universit\"at Berlin, 14195 Berlin, Germany}

 \author{Lennart Bittel}
  \thanks{\{\href{mailto:a.mele@fu-berlin.de}{a.mele}, \href{mailto:l.bittel@fu-berlin.de}{l.bittel}, \href{mailto:lorenzo.leone@fu-berlin.de}{lorenzo.leone}\}@fu-berlin.de}
 \affiliation{Dahlem Center for Complex Quantum Systems, Freie Universit\"at Berlin, 14195 Berlin, Germany}
 \author{Jens Eisert}
 \thanks{\{\href{mailto:jenseisert@gmail.com}{jenseisert}, \href{mailto:ludovico.lami@gmail.com}{ludovico.lami}\}@gmail.com}
  \affiliation{Dahlem Center for Complex Quantum Systems, Freie Universit\"at Berlin, 14195 Berlin, Germany}
  \affiliation{Helmholtz-Zentrum Berlin für Materialien und Energie, Berlin, Germany}
 \author{Vittorio Giovannetti}
    \thanks{\{\href{mailto:francesco.mele@sns.it}{francesco.mele}, \href{mailto:vittorio.giovannetti@sns.it}{vittorio.giovannetti}, \href{mailto:salvatore.oliviero@sns.it}{salvatore.oliviero}\}@sns.it}
 \affiliation{NEST, Scuola Normale Superiore and Istituto Nanoscienze,
Consiglio Nazionale delle Ricerche, Piazza dei Cavalieri 7, IT-56126 Pisa, Italy}
 \author{Ludovico Lami}
\thanks{\{\href{mailto:jenseisert@gmail.com}{jenseisert}, \href{mailto:ludovico.lami@gmail.com}{ludovico.lami}\}@gmail.com}
 \affiliation{QuSoft, Science Park 123, 1098 XG Amsterdam, the Netherlands}
 \affiliation{Korteweg–de Vries Institute for Mathematics, University of Amsterdam,
Science Park 105-107, 1098 XG Amsterdam, the Netherlands}
\affiliation{Institute for Theoretical Physics, University of Amsterdam,
Science Park 904, 1098 XH Amsterdam, the Netherlands}
 \author{Lorenzo Leone}
  \thanks{\{\href{mailto:a.mele@fu-berlin.de}{a.mele}, \href{mailto:l.bittel@fu-berlin.de}{l.bittel}, \href{mailto:lorenzo.leone@fu-berlin.de}{lorenzo.leone}\}@fu-berlin.de}
  \affiliation{Dahlem Center for Complex Quantum Systems, Freie Universit\"at Berlin, 14195 Berlin, Germany}
  
\author{Salvatore F.~E. Oliviero}
   \thanks{\{\href{mailto:francesco.mele@sns.it}{francesco.mele}, \href{mailto:vittorio.giovannetti@sns.it}{vittorio.giovannetti}, \href{mailto:salvatore.oliviero@sns.it}{salvatore.oliviero}\}@sns.it}
 \affiliation{NEST, Scuola Normale Superiore and Istituto Nanoscienze,
Consiglio Nazionale delle Ricerche, Piazza dei Cavalieri 7, IT-56126 Pisa, Italy}

\begin{abstract}
Quantum state tomography, aimed at deriving a classical description of an unknown state from measurement data, is a fundamental task in quantum physics. 
In this work, we analyse the ultimate achievable performance of tomography of continuous-variable systems, such as bosonic and quantum optical systems. We prove that tomography of these systems is extremely inefficient in terms of time resources, much more so than tomography of finite-dimensional systems: not only does the minimum number of state copies needed for tomography scale exponentially with the number of modes, but it also exhibits a dramatic scaling with the trace-distance error, even for low-energy states, in stark contrast with the finite-dimensional case. On a more positive note, we prove that tomography of Gaussian states is efficient. To accomplish this, we answer a fundamental question for the field of continuous-variable quantum information: if we know with a certain error the first and second moments of an unknown Gaussian state, what is the resulting trace-distance error that we make on the state? Lastly, we demonstrate that tomography of non-Gaussian states prepared through Gaussian unitaries and a few local non-Gaussian evolutions is efficient and experimentally feasible.
\end{abstract}

\maketitle

\section{Introduction}

A fundamental task in quantum physics, known as \emph{quantum state tomography}~\cite{anshu2023survey}, is to derive a classical description of a quantum system from experimental data. It serves not only as a powerful method to investigate Nature, but also as a diagnostic tool for benchmarking and verifying quantum devices~\cite{eisert_quantum_2020,cramer_efficient_2010,anshu2023survey,LogicalLukin}. The concept of tomography dates back to the 1990s~\cite{smithey_measurement_1993,lvovsky_continuousvariable_2009}, when it was first introduced for \emph{continuous-variable} (CV) systems~\cite{BUCCO} and then for finite-dimensional systems. Over the years, tomography algorithms for CV systems have been extensively developed, becoming a bread-and-butter tool for quantum opticians~\cite{smithey_measurement_1993,lvovsky_continuousvariable_2009}.

In the last decade, rapid advances in accurately preparing quantum states~\cite{lanyon_efficient_2017} have spurred the development of a new research field known as \emph{quantum learning theory}~\cite{anshu2023survey}, which investigates the question of how to learn properties of quantum systems as efficiently as possible. 
When the focus is on obtaining a full classical description of a quantum system, such investigation reduces to quantum state tomography, which is considered the cornerstone of quantum learning theory~\cite{anshu2023survey}. Despite tomography being a longstanding concept, many of its fundamental properties have only recently been unveiled due to advancements in this new field, albeit only for finite-dimensional systems and not for CV systems.

Let us introduce the concept of quantum state tomography from a quantum learning theory perspective. Formally, the task of quantum state tomography involves two parameters: the number of state copies $N$ and the error parameter $\varepsilon$. Given $N$ copies of an unknown state, the goal is to output a classical description of a state which is guaranteed to be $\varepsilon$-close to the true unknown state with high probability 
(this is stated in all rigour in the Methods section).
Among various notion of distance to measure the $\varepsilon$-closeness, the \emph{trace distance} emerges as the most meaningful one due to its operational significance~\cite{HELSTROM, Holevo1976}: given two states $\varepsilon$-close in trace distance, no quantum measurement can distinguish them by more than $\varepsilon$. Thus, the two states are in fact indistinguishable to any physical observer within a $\varepsilon$ resolution. 
The optimal performance of tomography is quantified by the \emph{sample complexity}, which is the minimum number of copies $N$ needed to achieve tomography with error $\varepsilon$. When the sample complexity and the running time of the tomography procedure scale polynomially with the number of constituents (e.g., qubits, modes), tomography is deemed \emph{efficient}.  

With vast many-body systems naturally emerging in Nature on one side, and the rapid advancement in constructing large-scale quantum devices on the other, minimising required resources has become imperative.
Consequently, determining the sample complexity of tomography of physically relevant classes of states is currently a pressing problem~\cite{anshu2023survey}. Notably, the sample complexity has been determined for $n$-qubit systems~\cite{kueng2014low,odonnell2015efficient,Haah_2017,haah_optimal_2021,anshu2023survey}: tomography of mixed states requires $
O(4^n/\varepsilon^2)$ copies, while for pure states the required number of copies reduces to $O(2^n/\varepsilon^2)$.

While quantum learning theory has been extensively developed for finite-dimensional systems, it remains almost unexplored for CV systems~\cite{aolita_reliable_2015,gandhari_precision_2023, becker_classical_2023,oh2024entanglementenabled}. 
This is partly because the usual approaches to CV tomography are based on mere approximations of phase-space functions~\cite{smithey_measurement_1993,lvovsky_continuousvariable_2009,BUCCO}, which do not account for trace-distance error and thus lack operational significance. This is a significant gap in the literature, especially considering that in recent years photonic quantum devices have been at the forefront of attempts to demonstrate quantum advantage, particularly through boson sampling~\cite{Borealis,Zhong_2020,SupremacyReview} and quantum simulation experiments~\cite{QuantumPhotoThermodynamics}. Moreover, photonic platforms play a pivotal role in various quantum technologies, including quantum computation~\cite{Mirrahimi_2014,Ofek_nature2016,error_corr_boson,Guillaud_2019,alexander2024manufacturable}, communication~\cite{Wolf2007,TGW,PLOB,Die-Hard-2-PRL,mele2023maximum, mele2024quantum}, and sensing~\cite{SGravi2,PhysRevLett.121.160502, PhysRevLett.86.5870,q_sensing_cv}. 
Therefore, determining the ultimate achievable performance of CV tomography has become a crucial, pressing problem --- which we solve in this work.

\begin{figure*}[t]
    \centering
\includegraphics[width=1\textwidth]{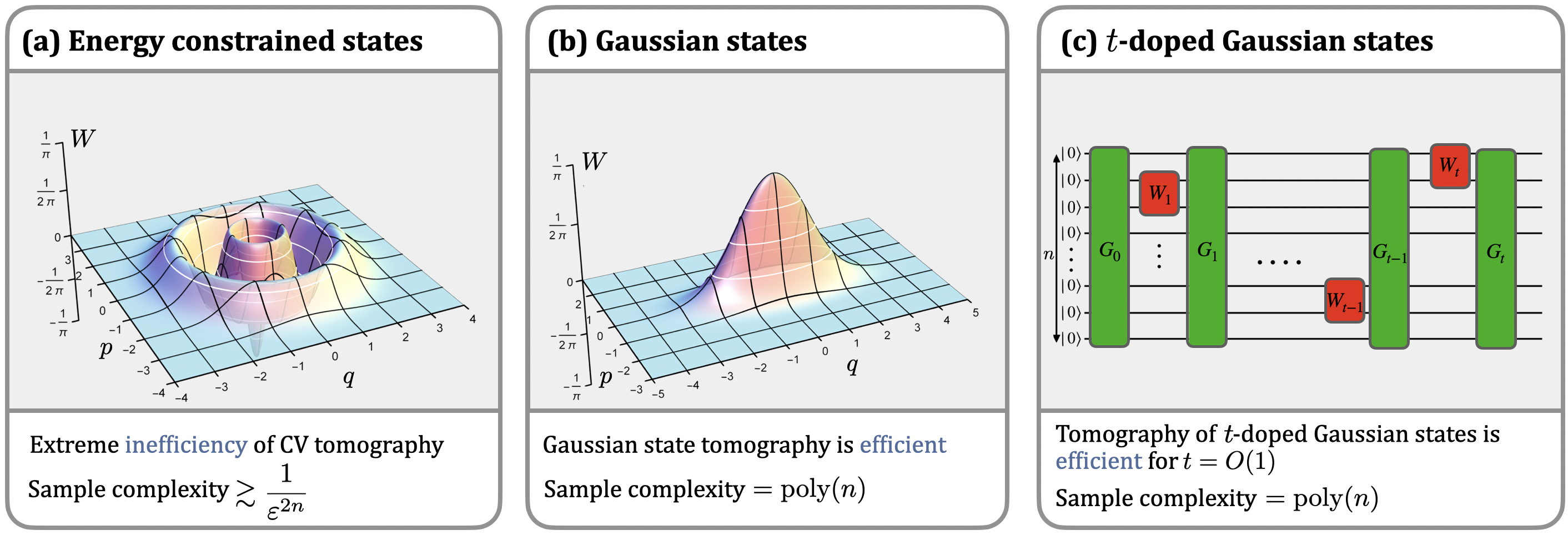}
    \caption{ We identify strong limitations against \textbf{(a)} quantum state tomography of continuous-variable systems subject to energy constraints inherent in experimental platforms. Here, $n$ is the number of modes, while $\varepsilon$ is the trace-distance error.
    Our investigation reveals a new phenomenon dubbed `extreme inefficiency' of continuous-variable quantum state tomography. Specifically, the number of copies required for tomography of $n$-mode energy-constrained states must scale at least as $\varepsilon^{-2n}$. This dramatic scaling is a unique feature of continuous-variable systems, standing in stark contrast to finite-dimensional systems where the required number of copies scales with the trace-distance error as $\varepsilon^{-2}$. Therefore, we ask whether there exist physically interesting classes of states for which tomography is efficient. We answer this in the affirmative by presenting \textbf{(b)} an efficient tomography algorithm for tomography of Gaussian states with provable guarantees in trace distance. Our analysis is based on novel technical tools of independent interest: specifically, we introduce simple bounds on the trace distance between two Gaussian states in terms of the norm distance between their first moments and covariance matrices. Finally, we demonstrate \textbf{(c)} that tomography of non-Gaussian states prepared by Gaussian unitaries and a few local non-quadratic Hamiltonian evolutions is still efficient. Remarkably, both of these efficient tomography algorithms are experimentally feasible to implement in quantum optics laboratories.}
    \label{figmain}
\end{figure*}

\subsection*{Energy-constrained states}

A CV system corresponds to $n$ modes, each associated with an infinite-dimensional Hilbert space. From the outset, the infinitely-many degrees of freedom of a single bosonic mode make tomography a challenging concept: What does it mean for a tomography algorithm to output a classical description of an infinite-dimensional state? Does this imply the output would be an infinite-dimensional matrix, which no classical computer could store? 
After all, not even Nature itself can deal with infinitely many degrees of freedom, and, as a consequence, physical states are inevitably subjected to certain constraints. Indeed, in the real world, energy is finite. The energy budget available in quantum optics laboratories is limited, as is the energy emitted by the Sun. Thus, CV quantum states only gain physical meaning under energy constraints; this turns the task of tomography from unavoidably inconceivable to a precious tool for investigating bosonic quantum systems. As explained below, we can indeed devise tomography algorithms capable of achieving arbitrarily low trace-distance error, for infinite-dimensional states subjected to some energy constraint.



By definition, an $n$-mode state $\rho$ satisfies the \emph{energy constraint} $E$ if the expectation value of the energy observable satisfies
\bb
    \Tr[\hat{E}_n\rho] \le n E\,.
    \label{def_constraint}
\ee
Here, the energy observable is defined as 
\bb
\hat{E}_n\coloneqq \frac{1}{2}\sum_{i=1}^n \left({\hat{p}_i^2}+{\hat{x}_i^2}\right)\,,
\ee
where $\{\hat{x}_i\}^{n}_{i=1}$ and $\{\hat{p}_i\}^{n}_{i=1}$ denote the position and momentum operators of the $n$ modes~\cite{BUCCO}. We normalise the right-hand-side of~\eqref{def_constraint} with $n$ because the energy is an extensive observable. Note that the energy constraint in \eqref{def_constraint} is equivalent to a constraint on the mean \emph{total photon number} of the state (see~\eqref{total_photon_number} in the Methods).


The following theorem, proven in the Supplementary Material (SM), determines the sample complexity of tomography of energy-constrained pure states.
\begin{thm}[(Tomography of energy-constrained pure states)]\label{thm_1_main}
    The sample complexity of tomography of $n$-mode pure states is $O(E^n/\varepsilon^{2n})$. Here, $\varepsilon $ is the trace-distance error and $E$ is the energy constraint.
\end{thm}
This result establishes that \emph{every} tomography algorithm --- even standard algorithms based on homodyne and heterodyne detections (see Fig.~\ref{tomfig1}) --- must use \emph{at least} $O(E^n/\varepsilon^{2n})$ copies to achieve error $\varepsilon$. Conversely, it also establishes that there exists a tomography algorithm that achieves an error $\varepsilon$ given access to the above number of copies. 

This finding reveals a striking phenomenon that we dub `extreme inefficiency' of CV tomography: not only does the number of copies required for CV tomography scale exponentially with the number of modes $n$, as for finite-dimensional systems, but it also has a dramatic scaling with respect to the error $\varepsilon$. Specifically, the scaling of $\sim  \varepsilon^{-2n}$ is a unique feature of CV tomography, being in stark contrast with the finite-dimensional setting characterised by the scaling of $\sim  \varepsilon^{-2}$. While in the finite-dimensional setting the error can be halved by increasing the number of copies by a factor $4$, which is cheap, in the CV setting one needs an exponential factor $4^n$, which is arresting.

To emphasise this remarkable behaviour, let us illustrate with an example the substantial difference between tomography of finite-dimensional systems and tomography of CV systems. Let us estimate the time required to achieve tomography with error $\varepsilon=10\%$ in both settings. Assume that each copy of the state is produced and processed every $ 1\,\mathrm{ns}$ (typical for qubits and light pulses). Then, tomography of $10$-qubit pure states only requires $\sim 0.1\,\mathrm{ms}$. However, as a consequence of Theorem~\ref{thm_1_main}, tomography of $10$-mode pure states with energy constraint $E=1$ requires $\sim 3000\,\mathrm{years}$, which shows that CV tomography becomes impractical even for a few modes. This highlights that tomography of CV systems is extremely inefficient, much more so than tomography of finite-dimensional systems.

In the forthcoming theorem, proven in the SM, we find bounds on the sample complexity of tomography of energy-constrained \emph{mixed} states.
\begin{thm}[(Tomography of energy-constrained mixed states)]\label{thm_tom_ec_states}
A number $O(E^{2n}/\varepsilon^{3n})$ of copies is sufficient to achieve tomography of $n$-mode mixed states. Conversely, tomography of such states requires at least $O(E^{2n}/\varepsilon^{2n})$ copies. Here, $\varepsilon $ is the trace-distance error and $E$ is the energy constraint.
\end{thm}


The core idea underpinning the attainable (yet arresting) performance guarantees of tomography presented in Theorem~\ref{thm_tom_ec_states} is that that energy-constrained CV states are effectively described by finite-dimensional systems. In particular, we prove that any (possibly mixed) energy-constrained state can be approximated, up to trace-distance error $\varepsilon$, by a $D$-dimensional state with rank $r$ such that 
\bb\label{approx_dim_rank_eq}
    D&=O(E^n/\varepsilon^{2n})\,,\\
    r&=O(E^n/\varepsilon^{n})\,.
\ee
The first result of Theorem~\ref{thm_tom_ec_states} simply follows from the observation that the sample complexity of tomography of $D$-dimensional states with rank $r$ is $O(Dr)$~\cite{kueng2014low,odonnell2015efficient,Haah_2017,haah_optimal_2021,anshu2023survey}. Remarkably, Eq.~\eqref{approx_dim_rank_eq} establishes that physical CV states --- those with bounded energy --- are characterised by a finite number of physically relevant parameters.


These conclusions bring us to shed light on the core question in CV tomography posed at the beginning of the section: while tomography of arbitrary CV states is indeed meaningless, considering that infinitely many parameters are physically inconceivable, with an energy constraint tomography gains real-world relevance. Indeed, in this latter case it is sufficient for the output to be a matrix of \emph{finite} dimension, i.e.~$O(E^{n}/\varepsilon^{2n})$.


The above findings can be generalised by considering constraints on any higher moments of the energy. By definition, similarly to \eqref{def_constraint}, an $n$-mode state $\rho$ satisfies the $k$-moment constraint $E$ if 
\bb
    \left(\Tr[(\hat{E}_n)^k\rho]\right)^{\!\frac{1}{k}}\le n E\,.
    \label{def_constraint_moment}
\ee
As proven in the SM, the sample complexity of tomography of pure states becomes $O(E^n/\varepsilon^{2n/k})$, while for mixed states the sample complexity is upper bounded by $O(E^{2n}/\varepsilon^{3n/k})$ and lower bounded by $O(E^{2n}/\varepsilon^{2n/k})$, thus recovering Theorems~\ref{thm_1_main}--\ref{thm_tom_ec_states} for $k=1$. Notably, the sample complexity decreases as $k$ increases, and the dramatic scaling $\sim\varepsilon^{-2n/k}$ disappears for $k\gtrsim n$. 


\begin{figure*}[t]
    \centering
    \includegraphics[width=1\textwidth]{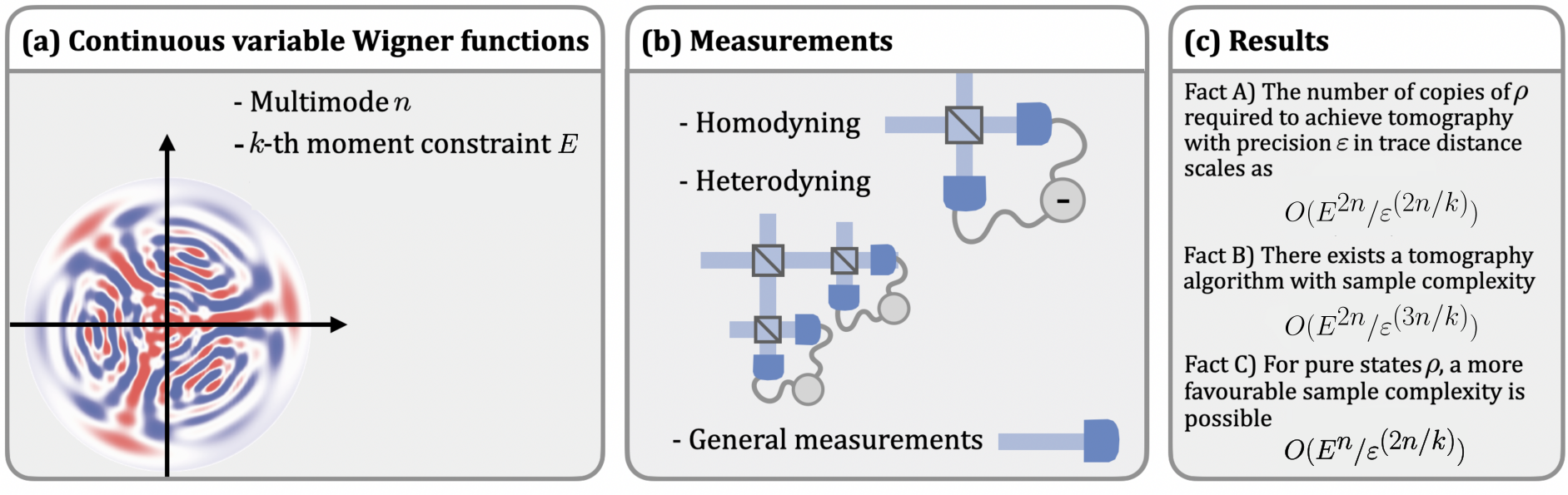}

\caption{We establish fundamental bounds on the resources required for (a) quantum state tomography of continuous-variable $k$-th moment constrained quantum states, highlighting the pronounced inefficiency of any strategy aiming to solve this task. (b) Our results encompass \emph{any} possible strategy, including those using only homodyne and heterodyne measurements, as well as other experimentally feasible operations in photonic platforms, and even general measurements. This means, independently from the techniques used, tomography of CV states is impractical. (c) We identify three key results, labelled Facts A-C. The implication is that the resources needed for tomography exhibit strong dependence on the desired accuracy, scaling as $\sim \varepsilon^{-2n/k}$.}
\label{tomfig1}
\end{figure*}

\subsection*{Gaussian states}
We have identified strong limitations on tomography of CV systems, even under stringent energy constraints. This prompts a natural question: can tomography be efficiently performed for more structured, yet physically interesting, classes of CV states? Here, we answer this question affirmatively by demonstrating that tomography of Gaussian states~\cite{BUCCO} is efficient. This finding is not only conceptually important, but also practically significant, as Gaussian states are the most commonly used states in quantum optics laboratories and have key applications in quantum sensing, communication, and computing~\cite{BUCCO}.

We start by investigating the following problem, rather fundamental for the field of CV quantum information. It is well known that Gaussian states are in one-to-one correspondence with their first moments and covariance matrices~\cite{BUCCO}. However, since in practice one has access only to a finite number of copies of an unknown Gaussian state, it is impossible to determine its first moment and covariance matrix \emph{exactly}. Instead, one can only obtain arbitrarily good approximations of them. Given the operational meaning of the trace distance~\cite{HELSTROM, Holevo1976}, it is thus a fundamental problem to answer the following question: `if we know with an error $\varepsilon$ the first moment and the covariance matrix of an unknown Gaussian state, what is the resulting trace-distance error that we make in approximating the underlying state?' The following theorem, proven in the Methods, answers this question.
\begin{thm}[(Error propagation from moments to trace distance)]\label{estimation_cov_trace_main}
    If we know with an error $\varepsilon$ the first moment and the covariance matrix of an unknown Gaussian state, the resulting trace-distance error is at most $O(\sqrt{\varepsilon})$ and at least $O({\varepsilon})$. 
\end{thm}
To prove this result, in Theorem~\ref{thm_upp_bound} and in Theorem~\ref{thm_trace_distance_lower_bound_main} of the Methods we find stringent bounds on the trace distance between two Gaussian states in terms of the norm distance between their first moments and covariance matrices, which constitute technical tools of independent interest for the field of CV quantum information. 


Notably, Theorem~\ref{estimation_cov_trace_main} is the key tool to prove the main result of this section: \emph{tomography of Gaussian states is efficient}. Specifically, the forthcoming theorem, proven in the SM, shows that the sample complexity of tomography of Gaussian states scales polynomially with the number of modes.
\begin{thm}[(Tomography of Gaussian states)]\label{tom_gaus_main}
There exists an algorithm that uses
\bb
    O\!\left( \frac{n^7E^4}{\varepsilon^4}\right)   =\mathrm{poly}(n)
\ee
copies to achieve tomography of an unknown $n$-mode Gaussian state. Here, $\varepsilon $ is the trace-distance error and $E$ is the energy constraint.
\end{thm}
This algorithm simply consists in estimating the first moment and the covariance matrix of the unknown Gaussian state, both procedures routinely performed in quantum optics laboratories using homodyne detection~\cite{BUCCO}. Additionally, the algorithm runs in $\mathrm{poly}(n)$ time, and its output is efficient to store, as it consists only of the $O(n^2)$ parameters of the first moment and covariance matrix.  


In summary, we have established the efficiency of tomography of Gaussian states. However, what if the state deviates slightly from being exactly Gaussian? This is a crucial question, especially considering imperfections during state preparation in experimental setups. In the SM, we demonstrate that our tomography algorithm is robust against little perturbations caused by non-Gaussian noise (e.g., dephasing noise). Even if the unknown state is a \emph{slightly perturbed} Gaussian state --- technically, a state with sufficiently small \emph{relative entropy of non-Gaussianity}~\cite{Marian2013} --- our algorithm remains effective and tomography remains efficient.

\begin{figure*} 
    \centering
\includegraphics[width=1\textwidth]{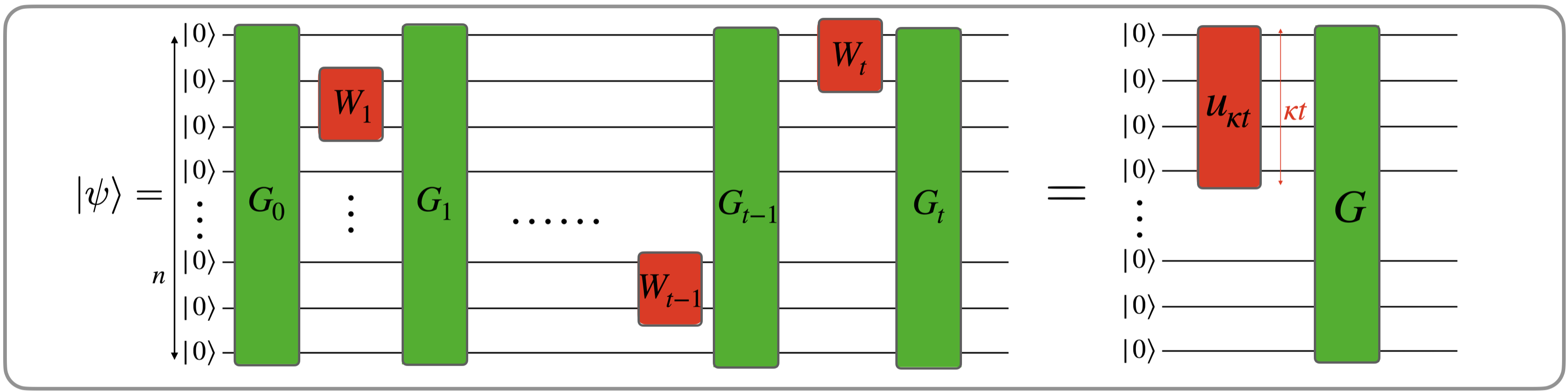}
    \caption{Pictorial representation of a $t$-doped Gaussian state. By definition, a $t$-doped Gaussian state vector $\ket{\psi}$ is a state prepared by applying Gaussian unitaries $G_0,\cdots,G_t$ (green boxes) and at most $t$ non-Gaussian $\kappa$-local unitaries $W_1,\cdots,W_t$ (red boxes) to the $n$-mode vacuum. A unitary is said to be $\kappa$-local if it is generated by a Hamiltonian which is a polynomial in at most $\kappa$ operators from the set of position operators $\{\hat{x}_i\}_{i=1}^n$ and momentum operators $\{\hat{p}_i\}_{i=1}^n$ of the $n$ modes. The figure also shows the decomposition proven in Theorem~\ref{thm_compressionMAIN}, which establishes that all the non-Gaussianity in $\ket{\psi}$ can be compressed in a localised region consisting of $\kappa t$ modes by applying a Gaussian unitary $G^\dagger$ to $\ket{\psi}$.}
    \label{figdoped}
\end{figure*}

\subsection*{$\boldsymbol{t}$-doped Gaussian states}
Having demonstrated that tomography of arbitrary (non-Gaussian) states is extremely inefficient (Theorem~\ref{thm_1_main}) and that tomography of Gaussian states is efficient (Theorem~\ref{tom_gaus_main}), we now ask whether there are relevant classes of non-Gaussian states for which tomography remains efficient. This leads us to analyse `$t$-doped Gaussian states': states prepared by applying Gaussian unitaries and at most $t$ non-Gaussian local unitaries on the vacuum state. The set of all $t$-doped Gaussian states coincides with that of all pure Gaussian states for $t=0$, it grows as $t$ increases, and it becomes the set of all pure states for $t\rightarrow\infty$~\cite{Lloyd1999}. In this sense, analysing the performance of tomography as a function of $t$ allows us to understand the trade-off between the efficiency of tomography and the degree of `non-Gaussianity'.

Let us proceed with the mathematical definitions. A unitary $U$ is a \emph{$t$-doped Gaussian} unitary if it is a composition 
\begin{equation}
    U = G_{t}W_t\cdots G_1 W_1 G_0\,,
\end{equation}
of Gaussian unitaries $G_0,G_1,\ldots,G_t$ and at most $t$ (non-Gaussian) $\kappa$-local unitaries $W_1,\ldots,W_t$, 
as depicted in Fig.~\ref{figdoped}. A unitary is $\kappa$-local if it is generated by a Hamiltonian which is a polynomial in at most $\kappa$ quadratures. An $n$-mode state vector is a \emph{$t$-doped Gaussian state vector} 
\begin{equation}
    \ket{\psi}=U\ket{0}^{\otimes n}
\end{equation}
if it can be prepared by applying a $t$-doped Gaussian unitary $U$ to the vacuum. The forthcoming theorem, proven in the SM, provides a remarkable decomposition of $t$-doped unitaries and states. 

\begin{thm}[(Compression of non-Gaussianity)]\label{thm_compressionMAIN}
If $n\ge \kappa t$, any $n$-mode $t$-doped Gaussian unitary $U$ can be decomposed as 
\bb
    U = G(u_{\kappa t}\otimes \mathbb{1}_{n-\kappa t}) G_{\text{passive}}\,,
\ee
for some suitable Gaussian unitary $G$, energy-preserving Gaussian unitary $G_{\text{passive}}$~\cite{BUCCO}, and $\kappa t$-mode (non-Gaussian) unitary $u_{\kappa t}$. In particular, any $n$-mode $t$-doped Gaussian state vector can be decomposed as 
\bb\label{dec_t_dopedMAIN}
    \ket{\psi}=G\left(\ket{\phi_{\kappa t}}\otimes\ket{0}^{\otimes(n-\kappa t)}\right)\,,
\ee
for some suitable Gaussian unitary $G$ and $\kappa t$-mode (non-Gaussian) state vector $\ket{\phi_{\kappa t}}$.
\end{thm}
Eq.~\eqref{dec_t_dopedMAIN} establishes that all the non-Gaussianity of a $t$-doped Gaussian state can be compressed into $O(t)$ modes via a suitable Gaussian unitary. This result can be seen as a bosonic counterpart of recent results within the stabiliser~\cite{leone2023learning,Oliviero_2021,Leone_2021Chaos,leone2023learning22PUBL1,leone2023learning22PUBL2,grewal2023efficient} and fermionic setting~\cite{mele2024efficient}, which reveals a fascinating parallelism among the theories of stabilisers, fermions, and bosons.

Leveraging the decomposition in~\eqref{dec_t_dopedMAIN}, we design a simple and experimentally feasible tomography algorithm for $t$-doped Gaussian states. Our algorithm involves: (i) estimating the first moment and the covariance matrix to construct an estimation of the Gaussian unitary $G$; (ii) applying its inverse to the state to compress the non-Gaussianity into the first \(\kappa t\) modes; and (iii) performing full-state tomography on the first \(\kappa t\) modes (see the SM for details). This algorithm is practical because it only requires tools commonly available in quantum optics laboratories, such as Gaussian evolutions and easily implementable Gaussian measurements like homodyne and heterodyne detection~\cite{BUCCO}. Specifically, step (iii) can be achieved using the CV classical shadow algorithm~\cite{becker_classical_2023}, which is experimentally feasible. Alternatively, in order to obtain a tighter upper bound on the sample complexity of tomography of $t$-doped Gaussian states, step (iii) can be performed using the optimal full-state tomography algorithm identified in Theorem~\ref{thm_1_main}.

The following Theorem~\ref{sample_t_doped_main}, proven in the SM, analyses the performance of our tomography algorithm. On a technical note, deriving guarantees on the estimation of the covariance matrix in step (i) requires more than just an energy constraint; a second-moment constraint is necessary. Refer to \eqref{def_constraint_moment} for the definition of second-moment constraint.
\begin{thm}[(Tomography of $t$-doped Gaussian state)]\label{sample_t_doped_main} 
There exists an algorithm that exploits
\bb\label{eq_sample_t_doped}
    \mathrm{poly}(n)+O\!\left(\left(n E/\varepsilon\right)^{2\kappa t}\right)
\ee
copies to achieve tomography of $n$-mode $t$-doped Gaussian states. Here, $E$ is the second-moment constraint, $\varepsilon$ is the trace-distance error, and $\kappa$ is the locality of the non-Gaussian unitaries.
\end{thm}
This theorem implies that the sample complexity of tomography of $t$-doped Gaussian states scales at most exponentially with $\kappa t$, indicating that tomography becomes increasingly difficult as the degree of non-Gaussianity $t$ increases. Additionally, both the time and memory complexity exhibit the same behaviour. Notably, this establishes the core finding of this section: \emph{tomography of $t$-doped Gaussian states is efficient in the regime $\kappa t = O(1)$.}


The proof of Theorem~\ref{sample_t_doped_main} hinges on the fact that a $t$-doped state is \emph{compressible}, meaning that it satisfies the decomposition in \eqref{dec_t_dopedMAIN}. Specifically, we can prove that tomography of compressible states is efficient \emph{if and only if} $\kappa t = O(1)$. This contrasts with the stabiliser~\cite{leone2023learning,grewal2023efficient} and fermionic setting~\cite{mele2024efficient}, where tomography of compressible stabiliser/fermionic states is efficient if and only if $t = O(\log(n))$. The difference in our setting ultimately arises from the infinite-dimensional nature of CV states and the presence of energy constraints.

\section{Discussion}
Our work serves as bridge between the two fields of quantum learning theory and CV quantum information. We provided the first and at the same time exhaustive investigation of tomography of CV systems with guarantees on the trace-distance error. 

First, we determined the sample complexity of tomography of energy-constrained pure states, which pinpoints the ultimate achievable performance of CV tomography. A new phenomenon, the `extreme inefficiency' of CV tomography, emerged: \emph{any} tomography algorithm for energy-constrained states must use a number of copies that dramatically scales at least as $\varepsilon^{-2n}$, where $n$ is the number of modes and $\varepsilon$ is the trace-distance error. This phenomenon, providing arresting fundamental limitations even for small $n$, is a unique feature of tomography of CV systems.

On a more positive note, we proved that tomography of Gaussian states is efficient. To establish this, we investigated a fundamental problem of CV quantum information: how the error in approximating the first moment and the covariance matrix of a Gaussian state propagates in the trace-distance error. Our solution introduces new tools of independent interest: simple, stringent bounds on the trace distance between two Gaussian states in terms of the norm distance between their first moments and covariance matrices. ~

Finally, we devised an experimentally-feasible algorithm that efficiently achieves tomography of $t$-doped Gaussian states for small $t$. This establishes that  even if a few non-Gaussian local gates are applied to a Gaussian state, tomography of the resulting non-Gaussian state remains efficient. The main tool employed here is a new decomposition of $t$-doped Gaussian states, which shows that all the non-Gaussianity in the state can be compressed into only $O(t)$ modes via a Gaussian unitary.


\medskip

\section{Acknowledgements}
We thank Matthias Caro, Nathan Walk, and Marco Fanizza for useful discussions. This work has been supported by
the BMBF (QPIC-1, PhoQuant, DAQC), 
the DFG (CRC 183), the QuantERA (HQCC), 
the MATH+ Cluster of Excellence,
the Quantum Flagship (Millenion, PasQuans2), the Einstein Foundation (Einstein Research Unit on Quantum Devices), and the Munich Quantum Valley (K-8). J.E.~is also funded by the European Research Council (ERC) within the project DebuQC. F.A.M., S.F.E.O., and V.G.\ acknowledges financial support by MUR (Ministero dell'Istruzione, dell'Universit\`a e della Ricerca) through the following projects: PNRR MUR project PE0000023-NQSTI. F.A.M.\ and S.F.E.O.\ thank the Freie Universit\"{a}t Berlin for hospitality. F.A.M.\ thanks the University of Amsterdam and QuSoft for hospitality.
\medskip


\smallskip
 
\bibliographystyle{unsrt}
\bibliography{biblio}

\clearpage


\section{Methods} 
\subsection{Trace distance}
The \emph{trace distance} between two quantum states $\rho_1$ and $\rho_2$ is defined as 
\bb
    d_{\mathrm{tr}}(\rho_1,\rho_2)\coloneqq \frac{1}{2}\|\rho_1-\rho_2\|_1\,,
\ee
where $\|A\|_1\coloneqq\Tr\sqrt{A^\dagger A}$ denotes the trace norm. The trace distance is considered the most meaningful notion of distance between two states because of its operational meaning in terms of the optimal probability of discriminating between the two states having access to a single copy of the state (Holevo--Helstrom theorem~\cite{HELSTROM, Holevo1976}). Consequently, in quantum information theory, the error in approximating a state is typically quantified using the trace distance.

\subsection{Quantum state tomography}
In this section, we precisely formulate the problem of \emph{quantum state tomography}~\cite{anshu2023survey}, which forms the basis of 
our investigation. The basic setup is depicted in Fig.~\ref{fig_tomography_main}.

\begin{problem}[(Quantum state tomography)] 
\label{def:PROBtom_methods}
    Let $\pazocal{S}$ be a set of quantum states. Consider $\varepsilon,\delta \in (0,1)$, and $N\in \N$. Let $\rho \in \pazocal{S}$ be an unknown quantum state. Given access to $N$ copies of $\rho$, the goal is to provide a classical description of a quantum state $\tilde{\rho}$ such that
    \bb
        \Pr\!\left[d_{\mathrm{tr}}(\tilde{\rho},\rho)  \le \varepsilon\right]\ge 1-\delta\,.
    \ee
    That is, with a probability $\ge 1-\delta$, the trace distance between $\tilde{\rho}$ and $\rho$ is at most $\varepsilon$.    Here, $\varepsilon$ is called the trace-distance error, while $\delta$ is called the failure probability.
\end{problem}
The \emph{sample complexity}, the \emph{time complexity}, and the \emph{memory complexity} of tomography of states in $\pazocal{S}$ are defined as the minimum number of copies $N$, the minimum amount of classical and quantum computation time, and the minimum amount of classical memory, respectively, required to solve Problem~\ref{def:PROBtom_methods} with trace distance error $\varepsilon$ and failure probability $\delta$. It is worth noting that the time complexity always upper bound the memory complexity, as well as the sample complexity.

One can think of $\pazocal{S}$ as a specific subset of the entire set of $n$-qubit states (e.g., pure states, $r$-rank states, stabiliser states) or $n$-mode states (e.g., energy-constrained states, moment-constrained states, Gaussian states, $t$-doped Gaussian states). By definition, tomography is deemed \emph{efficient} if the sample, time, and memory complexity scales polynomial in $n$; otherwise, it is deemed \emph{inefficient}. For example, in our work, we prove that tomography of energy-constrained states is (extremely) inefficient. In contrast, tomography of Gaussian states is efficient, while tomography of $t$-doped Gaussian states is efficient for small $t$ and inefficient for large $t$.

 \begin{figure}[t]
    \centering
    \includegraphics[width=0.48\textwidth]{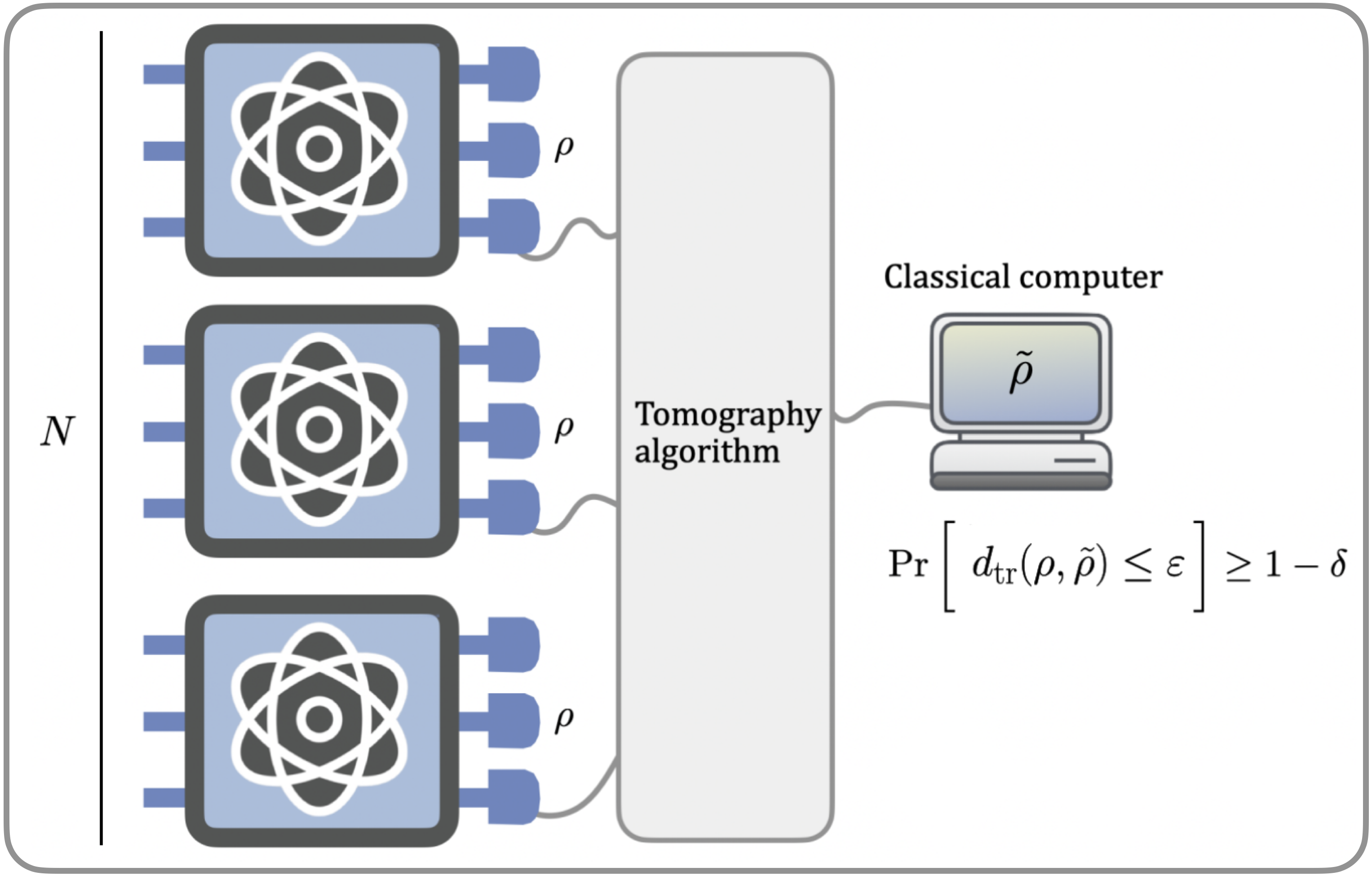}
    \caption{Pictorial representation of a quantum state tomography algorithm. Given access to $N$ copies of an unknown state $\rho$, the goal of a tomography algorithm is to output a classical description of a state $\tilde{\rho}$ that serves as a `good approximation' of the true unknown state $\rho$. Let us clarify what we mean by `good approximation'. First, the error in this approximation is measured by the trace distance $d_{\mathrm{tr}}(\rho,\tilde{\rho})$, which is most meaningful notion of distance between states~\cite{HELSTROM, Holevo1976}. Second, since quantum measurements yield probabilistic outcomes, the output $\tilde{\rho}$ is probabilistic rather than deterministic. Therefore, a `good approximation' means that the probability of getting a small trace-distance error is high. In formula, this is expressed as $\mathrm{Pr}\![d_{\mathrm{tr}}(\rho,\tilde{\rho})\le \varepsilon]\ge 1-\delta$, where $\varepsilon$ is the trace-distance error and $\delta$ is the failure probability. To measure the performance of tomography, one can define the \emph{sample complexity}. For fixed $\varepsilon$ and $\delta$, the sample complexity is the minimum number of copies $N$ required to achieve tomography with trace-distance error $\varepsilon$ and failure probability $\delta$. For example, the sample complexity of tomography of arbitrary $n$-qubit states is $O\!\left(\frac{4^n}{\varepsilon^2}\log(\frac{1}{\delta})\right)$~\cite{odonnell2015efficient,Haah_2017,haah_optimal_2021,anshu2023survey}. In general, the sample complexity of any tomography task depends at most logarithmically on the failure probability $\delta$~\cite{haah2023queryoptimal}, implying that this parameter has a minimal impact on the performance. In other words, the failure probability $\delta$ can be made very small with only a slight increase in the sample complexity. Therefore, throughout this work, we omit the dependence on $\delta$.}
    \label{fig_tomography_main}
\end{figure}

\subsection{Continuous-variable systems}
In this section, we provide a concise overview of quantum information with continuous variable (CV) systems~\cite{BUCCO}. A CV system is a quantum system associated with the Hilbert space $L^2(\mathbb R^n)$ of all square-integrable complex-valued functions over $\mathbb{R}^n$, which models $n$ modes of electromagnetic radiation with definite frequency and polarisation. A quantum state on $L^2(\mathbb{R}^n)$ is called an \emph{$n$-mode} state, and a unitary operator on $L^2(\mathbb{R}^n)$ is called an \emph{$n$-mode} unitary. The \emph{quadrature vector} is defined as
\bb
    \mathbf{\hat{R}}\coloneqq (\hat{x}_1,\hat{p}_1,\dots,\hat{x}_n,\hat{p}_n)^{\intercal}\,,
\ee
where $\hat{x}_j$ and $\hat{p}_j$ are the well-known position and momentum operators of the $j$-th mode, collectively called \emph{quadratures}. Let us proceed with the definitions of Gaussian unitary and Gaussian state.
\begin{Def}[(Gaussian unitary)]
An $n$-mode unitary is Gaussian if it is the composition of unitaries generated by quadratic Hamiltionians $\hat{H}$ in the quadrature vector:
    \bb\label{eq_quadratic_hamiltionian}
        \hat{H}\coloneqq \frac{1}{2}(\mathbf{\hat{R}}-\mathbf{m})^{\intercal}h(\mathbf{\hat{R}}-\mathbf{m})\,,
    \ee
    for some symmetric matrix $h\in \mathbb{R}^{2n,2n}$ and some vector $\mathbf{m}\in\mathbb{R}^{2n}$.
\end{Def}
\begin{Def}[(Gaussian state)]
An $n$-mode state $\rho$ is \emph{Gaussian} if it can be written as a Gibbs state of a quadratic Hamiltonian $\hat{H}$ of the form in \eqref{eq_quadratic_hamiltionian} with $h$ being positive definite. The Gibbs states associated with the Hamiltonian $\hat{H}$ are given by 
\begin{equation}
\rho= \left(\frac{e^{-\beta \hat{H}}}{\Tr[e^{-\beta \hat{H}}]}\right)_{\beta\in[0,\infty]}\,,
\end{equation}
where the parameter $\beta$ is called the `inverse temperature'.
\end{Def}
This definition includes also the pathological cases where both $\beta$ and certain terms of $H$ diverge (e.g., this is the case for tensor products between pure Gaussian states and mixed Gaussian states). An example of Gaussian state is the \emph{vacuum}, denoted as $\ket{0}^{\otimes n}$. Any pure $n$-mode Gaussian state vector
\bb
    \ket{\psi}=G\ket{0}^{\otimes n}\,
\ee
can be written as a Gaussian unitary $G$ applied to the vacuum.
A Gaussian state $\rho$ is uniquely identified by its \emph{first moment} $\textbf{m}(\rho)$ and \emph{covariance matrix} $V\!(\rho)$. By definition, the first moment and the covariance matrix of an $n$-mode state $\rho$ are given by
\bb
    \mathbf{m}(\rho)&\coloneqq \Tr\!\left[\mathbf{\hat{R}}\,\rho\right]\,,\\
    V\!(\rho)&\coloneqq\Tr\!\left[\left\{\mathbf{(\hat{R}-m(\rho)),(\hat{R}-m(\rho))}^{\intercal}\right\}\rho\right]\, ,
\ee  
where $\{\hat{A},\hat{B}\}\coloneqq \hat{A}\hat{B}+\hat{B}\hat{A}$ is the anti-commutator.

By definition, the \emph{energy} of an $n$-mode state $\rho$ is given by the expectation value $\Tr[\hat{E}_n\rho]$ of the energy observable $\hat{E}_n\coloneqq \sum_{j=1}^n \!\left(\frac{\hat{x}_j^2}{2}+\frac{\hat{p}_j^2}{2}\right)$, where it is assumed that each mode has a frequency of one~\cite{BUCCO}. It is important to note that energy is an extensive quantity, because for any single-mode state $\sigma$ the energy of $\sigma^{\otimes n}$ equals the energy of $\sigma$ multiplied by $n$. Furthermore, the energy of an $n$-mode state is always greater than or equal to $\frac{n}{2}$, with the equality achieved only by the vacuum. The \emph{total photon number} can be defined in terms of the energy observable as
\bb\label{total_photon_number}
    \hat{N}_n\coloneqq \hat{E}_n-\frac{n}{2}\mathbb{\hat{1}}\,
\ee
Given $\mathbf{k}=(k_1,\ldots,k_n) \in\mathbb{N}^n$, 
let us denote as
\bb
\ket{\mathbf{k}}=\ket{k_1}\otimes\ldots\otimes\ket{k_n}
\ee 
the $n$-mode \emph{Fock state}~\cite{BUCCO}. The total photon number is diagonal in the Fock basis:
\bb
    \hat{N}_n=\sum_{\mathbf{k}\in\mathbb{N}^n}\|\mathbf{k}\|_1\ketbra{\mathbf{k}}\,,
\ee
where $\|\mathbf{k}\|_1\coloneqq \sum_{i=1}^nk_i$.

\subsection{Effective dimension and rank of energy-constrained states}
In this section, we show that energy-constrained states can be well approximated by finite-dimensional states with low rank, as anticipated above in \eqref{approx_dim_rank_eq}. Further technical details regarding the findings presented in this section can be found in the SM.

Let $\rho$ be an $n$-mode state with total photon number satisfying the energy constraint
\bb
    \Tr[\rho\hat{N}_n]\le n E\,,
\ee
where $E\ge 0$. Given $M\in\mathbb{N}$, let $\HH_M$ be the subspace spanned by all the $n$-mode Fock states with total photon number not exceeding $M$, and let $\Pi_M$ the projector onto this space.

Let us begin by analysing the \emph{effective dimension} of the set of energy-constrained states. The trace distance between the energy-constrained state $\rho$ and its projection $\rho_M$ onto $\HH_M$, i.e.
\bb
    \rho_M\coloneqq \frac{\Pi_M\rho\Pi_M}{\Tr[\Pi_M\rho]}\,,
\ee
can be upper bounded as follows:
\bb\label{upp_bound_eff_dim_main0}
    d_{\mathrm{tr}}(\rho,\rho_M)&\leqt{(i)}\sqrt{\Tr[(\mathbb{1}-\Pi_M)\rho]}\leqt{(ii)} \sqrt{\frac{\Tr[\hat{N}_n\rho]}{M}} 
    \le \sqrt{\frac{n E}{M}}\,,
\ee
where in (i) we employed the Gentle measurement lemma~\cite{Sumeet_book} and in (ii) we used the simple operator inequality $\mathbb{1}-\Pi_M\le \frac{\hat{N}_n}{M}$. Consequently, by setting $M_1\coloneqq \lceil n E/\varepsilon^2\rceil $, it follows that the projection $\rho_{M_1}$ is $\varepsilon$-close to $\rho$ in trace distance. Moreover, the dimension of $\HH_{M_1}$ can be upper bounded as
\bb\label{upp_bound_eff_dim_main1}
    \dim\HH_{M_1}=\binom{n\!+\!M_1}{n}\le \left(\frac{e(n\!+\!M_1)}{n}\right)^{\!\!n}=O\!\left(\frac{(eE)^n}{\varepsilon^{2n}} \right) ,
\ee
where $e$ denotes the Euler's number. Hence, we conclude that any energy-constrained state $\rho$ can be approximated, up to trace-distance error $\epsilon$, by its projection $\rho_{M_1}$ onto the subspace $\HH_{M_1}$, which has a \emph{finite} dimension of $O\!\left((eE)^n/\varepsilon^{2n} \right)$.

Now, let us analyse the \emph{effective rank} of the energy-constrained state $\rho$. We say that $\rho$ has effective rank $r$ if it is $\varepsilon$-close to a state with rank $r$. Let us consider the spectral decomposition
\bb
    \rho=\sum_{i=1}^\infty p_i^\downarrow\psi_i\,,
\ee
where the eigenvalues $(p_i^\downarrow)_{i}$ are not increasing in $i$. In order to estimate the effective rank, let us choose an integer $r$ such that 
\bb
    \sum_{i=r+1}^\infty p_i^\downarrow\le \varepsilon\,,
\ee
which guarantees that the $r$-rank state 
$\rho^{(r)}\propto \sum_{i=1}^r p_i^\downarrow\psi_i$ is $O(\varepsilon)$-close to $\rho$. The \emph{infinite-dimensional Schur Horn theorem}~\cite[Proposition 6.4]{kaftal2009infinite} implies that for any $r$-rank projector $\Pi$ it holds that
\bb
    \sum_{i=r+1}^\infty p_i^\downarrow\le \Tr[(\mathbb{1}-\Pi)\rho]\,.
\ee
Moreover, by setting $M_2\coloneqq \lceil n E/\varepsilon\rceil $, the projector $\Pi_{M_2}$ is a $O\!\left((eE)^n/\varepsilon^{n} \right)$-rank projector satisfying
\bb
    \Tr[(\mathbb{1}-\Pi_{M_2})\rho]\le \frac{\Tr[\hat{N}_n\rho]}{M_2}\le \frac{nE}{M}\le\varepsilon\,,
\ee
where we employed the same inequalities used in \eqref{upp_bound_eff_dim_main0} and \eqref{upp_bound_eff_dim_main1}. Hence, by setting $\Pi=\Pi_{M_2}$, we deduce that $\rho$ is $\varepsilon$-close to a state $\rho^{(r)}$ having rank
\bb
    r=O\!\left((eE)^n/\varepsilon^{n} \right)\,.
\ee
Finally, by exploiting Gentle measurement lemma~\cite{Sumeet_book} and triangle inequality, one can easily show that the projection of $\rho^{(r)}$ onto $\HH_{M_1}$ is still $O(\varepsilon)$-close to $\rho$. Consequently, we conclude that any energy-constrained state can be approximated, up to trace-distance error $\varepsilon$, by a $D$-dimensional state with rank $r$ such that
\bb\label{approx_dim_rank_eq_methods}
    D&=O\left((eE)^n/\varepsilon^{2n}\right)\,,\\
    r&=O\left((eE)^n/\varepsilon^{n}\right)\,.
\ee

Based on these observations, we can devise a simple tomography algorithm for energy-constrained states. The first step involves performing the POVM $(\Pi_{M_1},\mathbb{1}-\Pi_{M_1})$ and discarding the post-outcome state associated with $\mathbb{1}-\Pi_{M_1}$. This step transforms the unknown state $\rho$ into the state $\rho_{M_1}$ with high probability. The state $\rho_{M_1}$ has two key properties: (i) it resides in the finite-dimensional subspace $\HH_{M_1}$ of dimension $D=O\left((eE)^n/\varepsilon^{2n}\right)$, and (ii) it is $O(\varepsilon)$-close to a state residing in $\HH_{M_1}$ with rank $r=O\left((eE)^n/\varepsilon^{n}\right)$. The second step involves performing the tomography algorithm of~\cite{wrightHowLearnQuantum} designed for $D$-dimensional state with rank $r$, which has a sample complexity of $O(Dr)$. Importantly, this algorithm remains effective even if the unknown state, which is promised to reside in a given $D$-dimensional Hilbert space, has rank strictly larger than $r$, as long as it is $O(\varepsilon)$-close to a $r$-rank state within the \emph{same} Hilbert space~\cite{wrightHowLearnQuantum}. We thus conclude that the sample complexity of tomography of energy-constrained states is upper bounded by $O(Dr)=O\left((eE)^{2n}/\varepsilon^{3n}\right)$.

Analogously, by exploiting that the sample complexity of tomography of $D$-dimensional \emph{pure} states is $O(D)$~\cite{kueng2014low,odonnell2015efficient,Haah_2017,haah_optimal_2021,anshu2023survey}, we can show that the sample complexity of tomography of energy-constrained \emph{pure} states is upper bounded by $O(D)=O\left((eE)^n/\varepsilon^{2n}\right)$.

For the sake of completeness, let us mention that the proof of the \emph{lower} bounds on the sample complexity of tomography of energy-constrained states, as presented in Theorems~\ref{thm_1_main}-\ref{thm_tom_ec_states}, primarily relies on \emph{epsilon-net} tools~\cite{vershynin_2018}.  Detailed proofs of these results can be found in the SM.

\subsection{Bounds on the trace distance between Gaussian states}
In this section, we address the question: `if we know with a certain precision the first moment and the covariance matrix of an unknown Gaussian state, what is the resulting trace-distance error that we make on the state?'

Let us formalise the problem. Let us consider a Gaussian state $\rho_1$ and assume that we have an approximation of its first moment $\textbf{m}(\rho_1)$ and an approximation of its covariance matrix $V\!(\rho_1)$. For example, these approximations may be retrieved through homodyne detection on many copies of $\rho_1$. We can then consider the Gaussian state $\rho_2$ with first moment and covariance matrix equal to such approximations: $\textbf{m}(\rho_2)$ and $V(\rho_2)$ are thus the approximations of $\textbf{m}(\rho_1)$ and $V(\rho_1)$, respectively. The errors incurred in these approximations are naturally measured by the norm distances $\|\textbf{m}(\rho_1)-\textbf{m}(\rho_2)\|$ and $\|V\!(\rho_1)-V\!(\rho_2)\|$, respectively, where $\|\cdot\|$ denotes some norm. Now, a natural question arises: `given an error $\varepsilon$ in the approximations of the first moment and covariance matrix, what is the error incurred in the approximation of $\rho_1$?' The most meaningful way to measure such an error is given by the trace distance $d_{\mathrm{tr}}(\rho_1,\rho_2)$~\cite{HELSTROM, Holevo1976}. Hence, the question becomes: `if it holds that
\bb
    \|\textbf{m}(\rho_1)-\textbf{m}(\rho_2)\|&=O(\varepsilon)\,,\\
    \|V\!(\rho_1)-V\!(\rho_2)\|&=O(\varepsilon)\,,
\ee
what can we say about the trace-distance $d_{\mathrm{tr}}(\rho_1,\rho_2)$?' Thanks to Theorems~\ref{thm_upp_bound}-\ref{thm_trace_distance_lower_bound_main} below, we can answer this question: the trace distance $d_{\mathrm{tr}}(\rho_1,\rho_2)$ is at most $O(\sqrt{\varepsilon})$ and at least $O(\varepsilon)$.

This motivates the problem of finding upper and lower bounds on the trace distance between Gaussian states in terms of the norm distance of their first moments and covariance matrices. Now, we present our bounds, which constitutes technical tools of independent interest.

The forthcoming Theorem~\ref{thm_upp_bound}, proven in the SM, shows our upper bound on the trace distance between Gaussian states. 
\begin{thm}[(Upper bound on the distance between Gaussian states)]\label{thm_upp_bound}
     Let $\rho_1$ and $\rho_2$ be $n$-mode Gaussian states satisfying the energy constraint $\Tr[\hat{N}_n\rho_1],\Tr[\hat{N}_n\rho_2]\le N$. Then, 
    \bb 
     d_{\mathrm{tr}}(\rho_1,\rho_2)\le f(N)\big(&\|\textbf{m}(\rho_1)-\textbf{m}(\rho_2)\|\\
     &\,+\sqrt{2}\sqrt{\|V\!(\rho_1)-V\!(\rho_2)\|_1}\big)\,,
    \ee
    where 
    $f(N)\coloneqq \frac{1}{\sqrt{2}} (\sqrt{N} + \sqrt{N+1})$. Here,  $\|\textbf{m}\|\coloneqq \sqrt{\textbf{m}^\intercal \textbf{m}}$ and $\|\cdot\|_1$ denote the Euclidean norm and the trace norm, respectively.
\end{thm}
The above theorem turns out to be crucial to prove the upper bound on the sample complexity of tomography of Gaussian states provided in Theorem~\ref{tom_gaus_main} in the Main text.

One might believe that proving Theorem~\ref{thm_upp_bound} would be straightforward by bounding the trace distance using the closed formula for the fidelity between Gaussian states~\cite{Banchi_2015}. However, this approach turns out to be highly non-trivial due to the complexity of such fidelity formula~\cite{Banchi_2015}, which makes it challenging to derive a bound based on the norm distance between the first moments and covariance matrices. Instead, our proof directly addresses the trace distance without relying on fidelity and involves a meticulous analysis based on the energy-constrained diamond norm~\cite{EC-diamond}.

The following theorem, proven in the SM, establishes our lower bound on the trace distance between Gaussian states.
\begin{thm}[(Lower bound on the distance between Gaussian states)]\label{thm_trace_distance_lower_bound_main}
    Let $\rho_1$ and $\rho_2$ be $n$-mode Gaussian states satisfying the energy constraint $\Tr[\hat{E}_n\rho_1],\Tr[\hat{E}_n\rho_2]\le E$. Then, 
\bb
    d_{\mathrm{tr}}(\rho_1,\rho_2)&\ge \frac{1}{200}\min\!\left\{1,\frac{\| \textbf{m}(\rho_1)-\textbf{m}(\rho_2) \|}{\sqrt{4E +1}}\right\} \,,\\
    d_{\mathrm{tr}}(\rho_1,\rho_2)&\ge \frac{1}{200}\min\!\left\{1, \frac{\|V\!(\rho_2)-V\!(\rho_1)\|_2}{4E +1}\right\}\,,
\ee
where $\|\textbf{m}\|\coloneqq \sqrt{\textbf{m}^\intercal \textbf{m}}$ and $\|V\|_2\coloneqq \sqrt{\Tr[V^\intercal V]}$ denote the Euclidean norm and the Hilbert-Schmidt norm, respectively.
\end{thm}
The proof of this theorem heavily relies on state-of-the-art bounds recently established for Gaussian probability distributions~\cite{devroye2023total}.

Theorems~\ref{thm_upp_bound}-\ref{thm_trace_distance_lower_bound_main} allow us to answer the question posed at the beginning of this section. Indeed, these theorems imply that, if we know with error $\varepsilon$ the first moment and the covariance matrix of an unknown Gaussian state, the resulting trace-distance error that we make on the state is \emph{at most} $O(\sqrt{\varepsilon})$ and \emph{at least} $O({\varepsilon})$. In particular, this proves Theorem~\ref{estimation_cov_trace_main} in the Main text.

The trace-distance bound of Theorem~\ref{thm_upp_bound} can be improved by assuming one of the Gaussian states to be pure, as we detail in the following theorem, proven in the SM. 
\begin{thm}[{(}Improved bound for pure states{)}]\label{spoiler_bounds_pure_main}
    Let $\psi$ be a pure $n$-mode Gaussian state and let $\rho$ be an $n$-mode (possibly non-Gaussian) state satisfying the energy constraints $\Tr[\psi \hat{E}_n], \Tr[\rho \hat{E}_n]\le E$. Then
    \bb\label{upper_bound_d_tr_pure}
        &d_{\mathrm{tr}}(\rho,\psi) \le \sqrt{E}\sqrt{ 2\|\textbf{m}(\rho)-\textbf{m}(\psi)\|^2+\|V\!(\rho)-V\!(\psi)\|_\infty}\,,
    \ee
    where $\|\textbf{m}\|\coloneqq \sqrt{\textbf{m}^\intercal \textbf{m}}$ and $\|\cdot\|_\infty$ denote the Euclidean norm and the operator norm, respectively.
\end{thm}
By exploiting this improved bound, in the SM we show that tomography of \emph{pure} Gaussian states can be achieved using $O\!\left( n^5E^3/\varepsilon^4\right)$ copies of the state. This represents an improvement over the mixed-state scenario analysed in Theorem~\ref{tom_gaus_main} in the Main text.

Moreover, the bound in Theorem~\ref{spoiler_bounds_pure_main} can be useful for \emph{quantum state certification}~\cite{Aolita_2015}, as we briefly detail now. Suppose one aims to prepare a pure Gaussian state $\psi$ with known first and covariance matrix. In a noisy experimental setup, however, an unknown state $\rho$ is effectively prepared. By accurately estimating the first two moments of $\rho$ (which can be done efficiently, as shown in the SM) one can estimate the right-hand side of \eqref{upper_bound_d_tr_pure}, which provides an upper bound on the trace distance between the target state $\psi$ and the noisy state $\rho$, thereby providing a measure of the quantum device's precision. Consequently, the device can be adjusted to minimise the error in the state preparation.

\newpage
\clearpage

\onecolumngrid
\begin{center}
\vspace*{\baselineskip}
{\textbf{\large Supplemental material:\\ Learning quantum states of continuous variable systems}}\\
\end{center}

\renewcommand{\theequation}{S\arabic{equation}}
\renewcommand{\thethm}{S\arabic{thm}}
\setcounter{equation}{0}
\setcounter{thm}{0}
\setcounter{figure}{1}
\setcounter{table}{0}
\setcounter{section}{0}
\setcounter{page}{1}
\makeatletter

\setcounter{secnumdepth}{2}
\section{Preliminaries}
\tableofcontents
\subsection{Notation and basics}\label{sec_notations}
Let $\mathbb{N}$ denote the set of natural numbers, and for each $n \in \mathbb{N}$, define $[n] \coloneqq \{1, 2, \ldots, n\}$. We also define $\N_+ \coloneqq \N \setminus \{0\}$ and $\mathbb{R}_+$ as the set of positive real numbers. We introduce $\mathbb{R}^{n \times n}$ and $\mathbb{C}^{n \times n}$ as the sets of $n \times n$ real and complex matrices, respectively. The notation $\lceil x \rceil$ rounds $x \in \mathbb{R}$ up to the nearest integer, while $\lfloor x \rfloor$ rounds $x$ down to the nearest integer.
Given a vector $v \in \mathbb{C}^n$ and a scalar $p \in [1, \infty]$, the $p$-norm of $v$ is denoted by $\|v\|_p$, defined as
\begin{equation}
    \|v\|_p \coloneqq \left(\sum_{i=1}^n |v_i|^p\right)^{\frac{1}{p}}.
\end{equation}
For any matrix $A \in \mathbb{C}^{n \times n}$, its Schatten $p$-norm is given by $\|A\|_p \coloneqq \text{Tr}\left((\sqrt{A^\dagger A})^p\right)^{\frac{1}{p}}$, which corresponds to the $p$-norm of the singular values of $A$. The trace norm and the Hilbert-Schmidt norm, instances of Schatten $p$-norms, are denoted $\|\cdot\|_1$ and $\|\cdot\|_2$, respectively. The infinity norm, $\|\cdot\|_\infty$, represents the maximum singular value and equals the limit of the Schatten $p$-norms as $p \to \infty$. 
The H\"older inequality, 
\begin{equation}
    |\text{Tr}(A^\dagger B)| \leq \|A\|_p \|B\|_q,
\end{equation}
applies for $1 \leq p, q \leq \infty$ with $\frac{1}{p} + \frac{1}{q} = 1$. Moreover, for all matrices $A \in \mathbb{C}^{n \times n}$ and $1 \leq p \leq q$, it holds that $\|A\|_q \leq \|A\|_p$ and $\|A\|_p \leq \text{rank}(A)^{\left(\frac{1}{p} - \frac{1}{q}\right)}\|A\|_q$.
We use the bra-ket notation, where we denote a vector $v \in \mathbb{C}^d$ using the ket notation $\ket{v}$ and its adjoint using the bra notation $\bra{v}$.
We refer to a vector $\ket{\psi}\in \mathbb{C}^d$ as a (pure) state if $\|\ket{\psi}\|_2=1$.

Given a Hilbert space $\HH$, we denote with $\pazocal{D}\!\left(\HH\right)$ the set of quantum states on $\HH$, i.e., positive semi-definite operators with unit trace. The trace distance between two quantum states $\rho$ and $\sigma$ is defined by $\frac{1}{2}\|\rho-\sigma\|_1$. The von Neumann entropy of a quantum state $\rho$ is given by $S(\rho)\coloneqq-\Tr[\rho\log_2\rho]$.

We denote with $O(n)$ the group of $n\times n$ orthogonal matrices, and with $\mathrm{U}(n)$ the group of $n \times n$ unitary matrices.
$\mathrm{Sp}(2n)$ denotes the group of symplectic matrices over the real field, defined as 
\begin{align}
    \mathrm{Sp}(2n) \coloneqq  \{S \in \mathbb{R}^{2n,2n} \,:\, S\Omega S^{T}=\Omega\},
\end{align} 
where 
\begin{align}\Omega_n  \coloneqq   \bigoplus_{i = 1}^{n} \begin{pmatrix} 0 & 1 \\ -1 & 0 \end{pmatrix}.
\end{align} 
We denote the Fourier transform \( \mathcal{F} \) of a function \( g: \mathbb{R}^{2n} \to \mathbb{R} \) as
\begin{align}
    \mathcal{F} g (\mathbf{r}) = \int_{\mathbf{r}' \in \mathbb{R}^{2n}} \mathrm{d}^{2n}\mathbf{r}' \, g(\mathbf{r}') e^{-i {\mathbf{r}}^{\intercal} \mathbf{r}'}.
\end{align}
Here, \( \mathbf{r}^{\intercal} \) denotes the transpose of \( \mathbf{r} \).
Consequently, the inverse Fourier transform \( \mathcal{F}^{-1} \) of a function \( g: \mathbb{R}^{2n} \to \mathbb{R} \) is given by
\begin{align}
    \mathcal{F}^{-1} g (\mathbf{r}) = \frac{1}{(2\pi)^{2n}} \int_{\mathbf{r}' \in \mathbb{R}^{2n}} \mathrm{d}^{2n}\mathbf{r}' \, g(\mathbf{r}') e^{i {\mathbf{r}'}^{\intercal} \mathbf{r}}.
\end{align}

\noindent
For any random variable $Z$ and any convex function $f$, Jensen Inequality states that $f(\mathbb{E}[Z]) \leq \mathbb{E}[f(Z)]$. Furthermore, if $f$ is concave, then the opposite inequality holds.

\noindent
For any concave real function $f$, any Hermitian matrix $X$ with eigenvalues pertaining to the domain of $f$, and any positive semi-definite matrix $\rho$, it holds that
\begin{align}
\Tr(\rho f(X)) \leq f(\Tr(\rho X)),
\label{eq:conc}
\end{align}
which can be shown by expanding $X$ in its eigendecomposition and using Jensen inequality.

We review some basic notions regarding the asymptotic notation:
\begin{itemize}
    \item Big-O notation: For a function \(f(n)\), if there exists a constant \(c\) and a specific input size \(n_0\) such that $f(n) \leq c \cdot g(n)$ for all \(n \geq n_0\), where \(g(n)\) is a well-defined function, then we express it as \(f(n) = O(g(n))\). This notation signifies the upper limit of how fast a function grows in relation to \(g(n)\).
    \item Big-Omega notation: For a function \(f(n)\), if there exists a constant \(c\) and a specific input size \(n_0\) such that $f(n) \geq c \cdot g(n)$ for all \(n \geq n_0\), where \(g(n)\) is a well-defined function, then we express it as \(f(n) = \Omega(g(n))\). This notation signifies the lower limit of how fast a function grows in relation to \(g(n)\).
    \item Big-Theta notation: For a function \(f(n)\), if $f(n) = O(g(n))$ and if $f(n) = \Omega(g(n))$, where \(g(n)\) is a well-defined function, then we express it as \(f(n) = \Theta(g(n))\). 
\end{itemize}
A tilde over $\pazocal{O}(\cdot)$, i.e., $\tilde{\pazocal{O}}(\cdot)$, implies that we are neglecting $n$ factors at the numerator or denominator (for us, $n$ will always represent the number of modes or qubits). For example, $f(n)=\frac{2^n}{n}$ is $\tildeTheta\,(2^n)$.  Analogously for the other asymptotic functions.

\emph{Note}: Throughout the Supplementary Material, we will use the asymptotic notation described above. However, in the main text, we opted to simplify the notation by just using $O(\cdot)$ to informally denote the "relevant asymptotic scaling." Specifically, when in the Main text we wrote that $y=O(f(\varepsilon,n))$ for some function $f(\cdot)$ of $\varepsilon$ and $n$, we informally meant that  ``$y$ is of the order of $f(\varepsilon,n)$ when $n$ is very large and $\varepsilon$ is very small.'' Importantly, to ensure clarity and avoid ambiguity, we express every sample complexity bound using both asymptotic notation and explicit exact expressions.

\subsubsection{Basics of statistical learning theory}
We present here basic results of probability and statistical learning theory; for further details we refer to Ref.~\cite{vershynin_2018}.
\begin{lemma}[(Union bound)]
    Let $B_1, B_2, \ldots, B_M$ be events in a probability space. The probability of their union is bounded by the sum of their individual probabilities, 
    \bb
        \Pr\!\left[\bigcup_{i=1}^{M} B_i\right] \le \sum_{i=1}^{M} \Pr[B_i]\,.
    \ee
\end{lemma}

\begin{lemma}[(Markov inequality)]
Let $X$ be a non-negative random variable and $a > 0$. Then, the probability that $X$ is at least $a$ is bounded by the expected value of $X$ divided by $a$, i.e.
\begin{equation}
\Pr[X \geq a] \leq \frac{\mathbb{E}[X]}{a}.
\end{equation}
\end{lemma}

\begin{lemma}[(Chernoff bound)]
\label{le:chernoff}
    Consider a set of independent and identically distributed binary random variables $\{Y_i\}_{i=1}^{N}$, taking values in $\{0,1\}$. Define $Y \coloneqq \sum_{i=1}^{N} Y_i$ and $\mu_Y \coloneqq \mathbb{E}[Y]$. For any $\alpha \in (0,1)$, the probability of $Y$ being less than $(1-\alpha)$ times its expected value is exponentially bounded as
    \bb 
    \Pr\left[Y \leq (1-\alpha) \mu_Y\right] \leq \exp\!\left(-\frac{\alpha^2 \mu_Y}{2}\right).
    \ee
\end{lemma}
We are now going to define the median-of-means estimator~\cite{vershynin_2018,JERRUM1986169,nemirovskii_problem_1983}.
Let $N', K \in \mathbb{N}$. Given $N = N'K$ samples $\{X_m\}_{m=1}^{N}$ of the random variable $X$, divide the samples into $K$ disjoint bins $\{\pazocal{B}_l\}_{l=1}^{K}$. Specifically, for each $l\in [K]$, let $\pazocal{B}_l$ be a set containing the elements $\{X_{(l-1)N'+1}, \ldots, X_{lN'}\}$, where $X_{i}$ denotes the $i$-th ordered sample. For each bin $\pazocal{B}_l$, define $\tilde{x}_l$ as the arithmetic average, i.e.
\bb
    \tilde{x}_l \coloneqq \frac{1}{N'}\sum_{X_m\in \pazocal{B}_l}X_m.
\ee
The median-of-means estimator is then given by
\bb
    \hat{\mu}(N',K)\coloneqq \operatorname{median}(\tilde{x}_1,\tilde{x}_2,\ldots,\tilde{x}_K).
\ee
We now present a Lemma which will be crucial in our subsequent analysis.
\begin{lemma}[(Medians of means \cite{JERRUM1986169,nemirovskii_problem_1983})]\label{median_of_means}
    Let $X$ be a random variable with variance $\sigma^2$. Suppose $K$ independent sample means of size $N'\ge\frac{34 \sigma^2}{\varepsilon^2}$ suffice to construct a median-of-means estimator $\hat{\mu}(N, K)$ that satisfies
    \bb 
    \operatorname{Pr}[|\hat{\mu}(N, K)-\mathbb{E}[X]| \geq \varepsilon] \leq 2 \mathrm{e}^{-K / 2}, \quad \forall \varepsilon>0\,.
    \ee
    As a consequence, $N\ge 68\log\!\left(\frac{2}{\delta}\right)\frac{\sigma^2}{\varepsilon^2}$ samples of $X$ suffice to construct a median-of-means estimator $\hat{\mu}$ which satisfies
    \bb 
    \operatorname{Pr}[|\hat{\mu}-\mathbb{E}[X]| \geq \varepsilon] \leq \delta\quad\forall \varepsilon>0\,.
    \ee
\end{lemma}
We also mention the following (standard) fact that is useful in amplifying the probability of success of an algorithm.
 
\begin{lemma}[(Enhancing the probability of success)]
\label{le:enhance-success}
Let $\pazocal{A}$ be an algorithm with a success probability of at least $ p_{\mathrm{succ}}\in(0,1]$. Let $\delta > 0$ and $N' \in \N$. If we execute $\pazocal{A}$ a total of $m$ times, with
\bb
    m \ge \left\lceil \frac{3}{2p_{\mathrm{succ}}}N' + \frac{18}{p_{\mathrm{succ}}}\log\!\left(\frac{1}{\delta}\right)\right\rceil\,,
\ee
then $\pazocal{A}$ will achieve success at least $N'$ times with a probability of at least $1 - \delta$.
\end{lemma}

\begin{proof}
Consider the binary random variables $\{X_i\}^{m}_{i=1}$ defined as
\begin{align}
    X_i \coloneqq \begin{cases}
        1 & \text{if $\pazocal{A}$ succeeds}, \\ 
        0 & \text{if $\pazocal{A}$ fails}.
    \end{cases}
\end{align}
Let $\hat{X} \coloneqq \sum^{m}_{i=1}X_i$, and note that $\mathbb{E}[\hat{X}] = m p_{\mathrm{succ}}$. Our goal is to upper bound the probability that $\pazocal{A}$ succeeds fewer than $N'$ times.
Define $\alpha \coloneqq 1 - \frac{N'}{m p_{\mathrm{succ}}}$. Using that $m\ge \frac{3}{2p_{\mathrm{succ}}}N'$, we ensure that $\alpha \in \left(\frac{1}{3},1\right)$. Applying the Chernoff bound in Lemma~\ref{le:chernoff}, we get:
\begin{align}
    \operatorname{Pr}\!\left(\hat{X} \leq N'\right) &= \operatorname{Pr}\!\left(\hat{X} \leq (1-\alpha) \mathbb{E}[\hat{X}]\right) \\
    \nonumber
    &\leq \exp\!\left(-\frac{\alpha^2 \mathbb{E}[\hat{X}]}{2}\right) \\
    \nonumber
    &= \exp\!\left(-\frac{\alpha^2}{2} p_{\mathrm{succ}}m\right) \\
    \nonumber
    &\le \exp\!\left(-\frac{p_{\mathrm{succ}}}{18}m\right) 
    \\
    \nonumber
    &\le \delta,
    \nonumber
\end{align}
where in the second inequality we have used that $\alpha\ge \frac{1}{3}$, while in the last inequality we have used 
\begin{equation}
    m \ge   \frac{18}{p_{\mathrm{succ}}}\log\!\left(\frac{1}{\delta}\right).
\end{equation}
\end{proof}

We present a lemma that enhances the probability of success of a learning algorithm. Although the proof follows standard steps similar to those in Ref.~\cite{haah2023queryoptimal} (Proposition 2.4), we provide it here with precise constants, which were not explicitly stated in Ref.~\cite{haah2023queryoptimal}.

\begin{lemma}[(Enhancing the probability of success of an algorithm for learning objects in a metric space)]
\label{le:enhance-success_states}
Let $p_{\mathrm{succ}}\in(\frac{1}{2},1]$ and $\varepsilon>0$. Let $\pazocal{A}$ be an algorithm for learning unknown objects in a metric space with distance $d(\cdot,\cdot)$. Assume that for any (unknown) input object $\rho$ the algorithm $\pazocal{A}$ outputs an object $\tilde{\rho}$ such that
\bb
    \Pr\left[d(\tilde{\rho}, \rho)\le \varepsilon\right]\ge p_{\mathrm{succ}}\,.
\ee
Then for any $\delta\in(0,1]$ there exists an algorithm $\pazocal{A}'$, which executes $\pazocal{A}$ a total of $m$ times with
\bb\label{choice_m_value}
    m \coloneqq \left\lceil \frac{2}{\left(1-\frac{1}{2p_{\mathrm{succ}}}\right)^2p_{\mathrm{succ}}}\log\!\left(\frac{1}{\delta}\right)  \right\rceil\,,
\ee
such that for any (unknown) input object $\rho$ it outputs an object $\tilde{\rho}$ such that
\bb
    \Pr\left[d(\tilde{\rho}, \rho)\le 3\varepsilon\right]\ge 1-\delta\,.
\ee
\end{lemma}

\begin{proof}[Proof]
The algorithm $\pazocal{A}'$, by executing $\pazocal{A}$ a total of $m$ times, produces $m$ random objects $\tilde{\rho}_1,\ldots,\tilde{\rho}_m$ such that each $\tilde{\rho}_i$ satisfies $\Pr\left[d(\tilde{\rho}_i, \rho)\le \varepsilon\right]\ge p_{\mathrm{succ}}$. Let us show that the probability $\bar{P}$ that there are more than $\frac{m}{2}$ objects in the set $\mathcal{S}\coloneqq\{\tilde{\rho}_1,\ldots,\tilde{\rho}_m\}$ which are close at most $\varepsilon$ to $\rho$ is not smaller than $1-\delta$. Consider the binary random variables $\{X_i\}^{m}_{i=1}$ defined as
\begin{align}
    X_i \coloneqq \begin{cases}
        1 & \text{if $d(\tilde{\rho}_i, \rho)\le \varepsilon$}, \\ 
        0 & \text{otherwise}.
    \end{cases}
\end{align}
Let $\hat{X} \coloneqq \sum^{m}_{i=1}X_i$, and note that $\mathbb{E}[\hat{X}] = m p_{\mathrm{succ}}$. 
Using that $p_{\mathrm{succ}}\in(\frac{1}{2},1)$, we ensure that the quantity $\alpha \coloneqq 1 - \frac{1}{2 p_{\mathrm{succ}}}$ satisfies $\alpha \in \left(0,1\right)$. Then, by applying the Chernoff bound in Lemma~\ref{le:chernoff}, we have that
\begin{align}
    \bar{P}&=\operatorname{Pr}\!\left(\hat{X} > \frac{m}{2}\right)\\
    \nonumber
    &= 1- \operatorname{Pr}\!\left(\hat{X} \le \frac{m}{2}\right)\\
    \nonumber
    &=1-\operatorname{Pr}\!\left(\hat{X} \leq (1-\alpha) \mathbb{E}[\hat{X}]\right) \\
    \nonumber
    &\geq 1- \exp\!\left(-\frac{\alpha^2 \mathbb{E}[\hat{X}]}{2}\right) \\
    \nonumber
    &=1- \exp\!\left(-\frac{ \left(1-\frac{1}{2  
p_{\mathrm{succ}}  }\right)^2 p_{\mathrm{succ}}}{2} m\right) \\
\nonumber
    &\geq 1- \delta\,,
    \nonumber
\end{align}
where in the last inequality we have used 
that 
\begin{equation}
m \ge   \frac{2}{\left(1-\frac{1}{2p_{\mathrm{succ}}}\right)^2p_{\mathrm{succ}}}\log\!\left(\frac{1}{\delta}\right).
\end{equation}
From now on, let us assume that there are more than $\frac{m}{2}$ objects in $\mathcal{S}$ which are close at most $\varepsilon$ to $\rho$ (which is an event that happens with probability at least $1-\delta$, as proved above). Then, by the triangle inequality, there are more than $\frac{m}{2}$ objects in $\mathcal{S}$ which are close at most $2\varepsilon$ to each other.

The algorithm $\mathcal{A}'$ is as follows: First, it computes the distance between any two objects in $\mathcal{S}$ and, second, it outputs an object $\tilde{\rho}\in\mathcal{S}$ that satisfies the property of being close at most $2\varepsilon$ to more than $\frac{m}{2}$ objects in $\mathcal{S}$. 

Let us show that the output $\tilde{\rho}$ satisfies $d(\tilde{\rho},\rho)\le3\varepsilon$. Let $\mathcal{S}_{\tilde{\rho}}$ be the set of objects in $\mathcal{S}$ that are close at most $2\varepsilon$ to $\tilde{\rho}$. Since $|\mathcal{S}_{\tilde{\rho}}|>\frac{m}{2}$, $|\mathcal{S}|=m$, and $\mathcal{S}$ contains more than $\frac{m}{2}$ objects which are close at most $\varepsilon$ to $\rho$, then there exists an object $\sigma\in\mathcal{S}_{\tilde{\rho}}$ with $d(\sigma,\rho)\le\varepsilon$. Then, by the  triangle inequality, we conclude that $d(\tilde{\rho},\rho)\le3\varepsilon$.
\end{proof}
The right-hand side of~\eqref{choice_m_value} diverges as $p_{\mathrm{succ}}\to\frac{1}{2}$. This is consistent with the simple fact that, in general, the probability of success can not be arbitrarily enhanced when $p_{\mathrm{succ}}\le\frac{1}{2}$. Indeed, in the latter case, we could consider an algorithm $\mathcal{A}$ such that for any unknown input $\rho$ it outputs: the exact true object $\rho$ with probability $p_{\mathrm{succ}}$; an object $\sigma_\rho$, which is very distant from $\rho$, with the same probability $p_{\mathrm{succ}}$; and a fixed object $\tau$ (independent of $\rho$) with probability $1-2p_{\mathrm{succ}}$. Hence, since the objects $\rho$ and $\sigma_\rho$ are statistically indistinguishable, there is no way to arbitrarily enhance the probability of success of $\mathcal{A}$, even with infinite executions of it.
\subsection{Preliminaries on continuous variable systems}
\label{sub:prelCV}
In this section, we provide a concise overview of quantum information with \emph{continuous variable} (CV) systems; 
for further details, we refer to Refs.~\cite{BUCCO,introeisert,weedbrook_gaussian_2012}. We consider $n$ modes of harmonic oscillators associated with the Hilbert space $L^2(\mathbb R^n)$, which comprises all square-integrable complex-valued functions over $\mathbb{R}^n$. Each mode represents a single-mode of electromagnetic radiation with definite frequency and polarisation. The set of $n$-mode states is denoted by $\pazocal{D}\!\left(L^2(\mathbb R^n)\right)$. For each \(j\in[n]\), the annihilation operator $a_j$ of the $j$-th mode is defined by
\bb
    a_j\coloneqq \frac{\hat{x}_j+i\hat{p}_j}{\sqrt{2}}\,,
\ee
where $\hat{x}_j$ and $\hat{p}_j$ denote the well-known position and momentum operators of the \(j\)-th mode, which are Hermitian operator satisfying the canonical commutation relations $ [\hat{x}_j,\hat{p}_k]=i\delta_{j,k}\hat{\mathbb{1}}$. 
Given a single mode with annihilation operator $a$, its $m$-th Fock state vector (corresponding to the quantum state vector with $m$ photons) is defined as
\bb
    \ket{m}\coloneqq \frac{(a^\dagger)^m}{\sqrt{m!}}\ket{0}\,,
\ee
where $\ket{0}$ is the vacuum state vector. Crucially, the Fock states  \((\ket{m})_{n\in\N}\) form an orthonormal basis for the Hilbert space \(L^2(\mathbb R)\) of a single-mode system, meaning that such a system can be effectively viewed as an infinite-dimensional qudit. The operator $a^\dagger a$ is referred to as the photon number operator and it can be diagonalised in Fock basis as
\bb
    a^\dagger a=\sum_{m=0}^\infty m \ketbra{m}\,.
\ee
By introducing the quadrature vector 
\bb
    \mathbf{\hat{R}}\coloneqq (\hat{x}_1,\hat{p}_1,\dots,\hat{x}_n,\hat{p}_n)^{\intercal}=(\hat{R}_1,\hat{R}_2,\dots,\hat{R}_{2n-1},\hat{R}_{2n})^{\intercal}\,
\ee
the canonical commutation relations can be expressed as 
\bb
    [\hat{R}_k,\hat{R}_l]=i\,(\Omega_n)_{kl}\mathbb{\hat{1}}\qquad\forall\,k,l\in[2n]\,,
    \label{comm_rel_quadrature}
\ee
where 
\bb\label{symplectic_form}
\Omega_n\coloneqq\bigoplus_{i=1}^n \left(\begin{matrix}0&1\\-1&0\end{matrix}\right)=\mathbb{1}_{n}\otimes\left(\begin{matrix}0&1\\-1&0\end{matrix}\right)
\ee
is the $n$-mode symplectic form, $\mathbb{1}_{n}$ is the $n\times n$ identity matrix, and $\mathbb{\hat{1}}$ is the identity operator over $L^2(\mathbb R^n)$. The relation in~\eqref{comm_rel_quadrature} is usually expressed in the continuous variable literature~\cite{BUCCO} in a compact form as
\bb
[\mathbf{\hat{R}},\mathbf{\hat{R}}^{\intercal}]=i\,\Omega_n\mathbb{\hat{1}}\,,
\ee
which we denote as `vectorial notation'.
The energy operator $\hat{E}_n$ is defined as 
\bb\label{energy_operator}
\hat{E}_n&\coloneqq \frac{1}{2}\mathbf{\hat{R}}^\intercal\mathbf{\hat{R}}=\sum_{j=1}^n \!\left(\frac{\hat{x}_j^2}{2}+\frac{\hat{p}_j^2}{2}\right)=\sum_{j=1}^n \!\left(a_j^\dagger a_j +\frac{\mathbb{\hat{1}}}{2}\right)=\hat{N}_n+\frac{n}{2}\mathbb{\hat{1}}\,,
\ee
where $\hat{N}_n\coloneqq \sum_{j=1}^n a_j^\dagger a_j$ is the photon number operator. The characteristic function $\chi_\rho: \mathbb{R}^{2n}\to \mathbb{C}$ of an $n$-mode state $\rho\in\pazocal{D}\!\left(L^2(\mathbb R^n)\right)$ is defined as $\chi_\rho(\mathbf r)\coloneqq\Tr[ \rho  \hat{D}_{\mathbf{r}} ]$, where for all $\mathbf{r}\in \mathbb{R}^{2n}$ the displacement operator $ \hat{D}_{\mathbf{r}}$ is given by $\hat{D}_{\mathbf{r}}\coloneqq e^{-i {\mathbf{r}}^{\intercal}\Omega_n \mathbf{\hat{R}}}$. Any state $\rho$ can be written in terms of its characteristic function as~\cite{BUCCO}
\bb\label{inverse_fourier_displacement}
\rho=\int_{\mathbb{R}^{2n}}\frac{\mathrm{d}^{2n}\mathbf{r}}{(2\pi)^n}\chi_\rho(\mathbf r) \hat{D}_{\mathbf{r}}^\dagger\,,
\ee
and hence quantum states and characteristic functions are in one-to-one correspondence. 

The Wigner function \( W_\rho: \mathbb{R}^{2n} \to \mathbb{R} \) of an \( n \)-mode state \( \rho \) is defined as the inverse Fourier transform \( \mathcal{F}^{-1} \) 
of the characteristic function \( \chi_\rho \), evaluated at \( \Omega_n \mathbf{r} \)
\begin{equation}
W_\rho(\mathbf{r}) = \mathcal{F}^{-1} \chi_\rho(\Omega_n \mathbf{r}) = \frac{1}{(2\pi)^{2n}} \int_{\mathbf{r}' \in \mathbb{R}^{2n}} \mathrm{d}^{2n}\mathbf{r}' \, \chi_\rho(\mathbf{r}') e^{i {\mathbf{r}'}^{\intercal} \Omega_n \mathbf{r}}.
\end{equation}
Consequently, the characteristic function \( \chi_\rho \) can be expressed as the Fourier transform 
\begin{equation}
\chi_\rho(\mathbf{r}) = \mathcal{F} W_\rho (\Omega_n \mathbf{r}) = \int_{\mathbf{r}' \in \mathbb{R}^{2n}} \mathrm{d}^{2n}\mathbf{r}' \, W_\rho(\mathbf{r}') e^{-i {\mathbf{r}}^{\intercal} \Omega_n \mathbf{r}'},
\end{equation}
evaluated at \( \Omega_n \mathbf{r} \), of the Wigner function \( W_\rho \).
Moreover, the Husimi function $Q_\rho(\textbf{r}): \mathbb{R}^{2n}\to \mathbb{R}$ of an $n$-mode state $\rho$ is defined as
\bb
    Q_\rho(\textbf{r})\coloneqq \frac{1}{(2\pi)^n}\bra{\textbf{r}}\rho\ket{\textbf{r}}\,,
\ee
where $\ket{\textbf{r}}\coloneqq \hat{D}_\textbf{r}\ket{0}$ is a coherent state vector. 
It turns out that the Fourier transform of the Husimi function, evaluated at \( \Omega_n \mathbf{r} \), is related to the characteristic function as~\cite{BUCCO}
\bb\label{relation_Husimi_charact}
    \int_{\textbf{r}'\in\mathbb{R}^{2n}}\mathrm{d}^{2n}\textbf{r}' \,Q_\rho(\textbf{r}') e^{-i {\mathbf{r}}^{\intercal}\Omega_n \mathbf{r}'}=e^{-\frac{1}{4}\textbf{r}^\intercal \textbf{r}}\chi_\rho(\textbf{r})\,.
\ee
The Husimi function $Q_\rho(\textbf{r})$ is a useful quantity since it is the probability distribution of the outcome $\textbf{r}\in\mathbb{R}^{2n}$ of a (experimentally feasible) measurement --- known as \emph{heterodyne measurement} --- performed on the state $\rho$~\cite{BUCCO}.
The first moment of a quantum state $\rho$ is defined as $\mathbf{m}(\rho)\coloneqq (m_1(\rho),\dots,m_n(\rho))^{\intercal}$, where
\bb
    m_k(\rho)\coloneqq\Tr\!\left[\hat{R}_k\,\rho\right]\,,
\ee
for each $k\in [n]$, or in its vectorial notation as $\mathbf{m}(\rho)=\Tr\!\left[\mathbf{\hat{R}}\,\rho\right]$.
Additionally, the covariance matrix of $\rho$ is defined by the matrix $V\!(\rho)$ with elements
\bb
	[V\!(\rho)]_{k,l}\coloneqq \Tr\!\left[\left\{{\hat{R}_k-m_k(\rho)\idop,\hat{R}_l-m_l(\rho)\idop}\right\}\rho\right]= \Tr\!\left[\left\{\hat{R}_k,\hat{R}_l\right\}\rho\right]-2m_k(\rho)m_l(\rho) \, ,
\ee
for each $k,l\in[2n]$, where $\{\hat{A},\hat{B}\}\coloneqq \hat{A}\hat{B}+\hat{B}\hat{A}$ is the anti-commutator. In its vectorial notation this reads as
\bb
	V\!(\rho)&=\Tr\!\left[\left\{\mathbf{(\hat{R}-m(\rho)),(\hat{R}-m(\rho))}^{\intercal}\right\}\rho\right]=  \Tr\!\left[\left\{\mathbf{\hat{R}},\mathbf{\hat{R}}^{\intercal}\right\}\rho\right]-2\textbf{m}(\rho)\textbf{m}(\rho)^\intercal \, .
\ee   
Any covariance matrix $V\!(\rho)$ satisfies the inequality
\bb
V\!(\rho)+i\Omega_n\ge0\,,
\ee
known as \emph{uncertainty relation}. As a consequence, since $\Omega_n$ is skew-symmetric, any covariance matrix $V\!(\rho)$ is positive semi-definite on $\mathbb{R}^{2n}$. Conversely, for any symmetric $W\in\mathbb{R}^{2n,2n}$ such that $W+i\Omega_n\ge0$ there exists an $n$-mode (Gaussian) state $\rho$ with covariance matrix $V\!(\rho)=W$~\cite{BUCCO}.

\begin{Def}[(Gaussian state)]\label{def_gauss_sm}
An $n$-mode state $\rho$ is said to be a Gaussian state if it can be written as a Gibbs state of a quadratic Hamiltonian $\hat{H}$ in the quadrature operators $\{\hat{R}_i\}_{i\in[2n]}$, that is, 
\bb
    \hat{H}\coloneqq \frac{1}{2}(\mathbf{\hat{R}}-\mathbf{m})^{\intercal}H(\mathbf{\hat{R}}-\mathbf{m})
\ee
for some symmetric positive-definite matrix $H\in \mathbb{R}^{2n,2n}$ and some vector $\mathbf{m}\in\mathbb{R}^{2n}$. The Gibbs states associated with the Hamiltonian $\hat{H}$ are given by 
\begin{equation}
\rho= \left(\frac{e^{-\beta \hat{H}}}{\Tr[e^{-\beta \hat{H}}]}\right)_{\beta\in[0,\infty]}\,,
\end{equation}
where the parameter $\beta$ is called the `inverse temperature'.
\end{Def}
\begin{remark}
Definition~\ref{def_gauss_sm} includes also the pathological cases where both $\beta$ and certain terms of $H$ diverge (e.g.,~this is the case for tensor products between pure Gaussian states and mixed Gaussian states). To formalise this mathematically, one can define the set of Gaussian states as the closure, with respect to the trace norm, of the set of Gibbs states of quadratic Hamiltonians~\cite{G-resource-theories}.
\end{remark}

The characteristic function of a Gaussian state $\rho$ is the Fourier transform of a Gaussian probability distribution, evaluated at $\Omega_n \mathbf{r}$, which can be written in terms of $\mathbf{m}(\rho) $ and $V\!(\rho)$ as~\cite{BUCCO}
\bb\label{charact_gaussian}
\chi_{\rho}(\mathbf{r})=\exp\!\left( -\frac{1}{4}(\Omega_n \mathbf{r})^{\intercal}V\!(\rho)\Omega_n \mathbf{r}+i(\Omega_n \mathbf{r})^{\intercal}\mathbf{m}(\rho) \right)\,.
\ee
Consequently, the Wigner function of a Gaussian state can be expressed as the following Gaussian probability distribution: 
\bb
    W_\rho(\textbf{r})&=\frac{e^{- (\textbf{r}-\textbf{m}(\rho))^\intercal [V(\rho) ]^{-1}  (\textbf{r}-\textbf{m}(\rho))}}{\pi^{n}\sqrt{\det [V(\rho) ]}}\,.
\ee
Note that a Gaussian probability distribution with first moment $\textbf{m}$ and covariance matrix $V$ is defined as
\bb
    \NN[\textbf{m},V](\textbf{r})&\coloneqq\frac{e^{-\frac12 (\textbf{r}-\textbf{m})^\intercal V^{-1}  (\textbf{r}-\textbf{m})}}{(2\pi)^{n}\sqrt{\det V }}\,.
\ee
Hence, using such notation, the Wigner function of a Gaussian state $\rho$ is a Gaussian probability distribution with first moment $\textbf{m}(\rho)$ and covariance matrix $\frac{V(\rho)}{2}$, that is,
\bb
    W_\rho(\textbf{r})=\NN\!\left[\textbf{m}(\rho),\frac{V(\rho)}{2}\right]\!(\textbf{r})\,.
\ee
Hence,~\eqref{relation_Husimi_charact} and~\eqref{charact_gaussian} imply that the Husimi function of a Gaussian state $\rho$ is a Gaussian probability distribution with first moment $\textbf{m}(\rho)$ and covariance matrix $\frac{V(\rho)+\mathbb{1}}{2}$, that is,
\bb\label{husimi_gaussian_state}
    Q_\rho(\textbf{r})=\NN\!\left[\textbf{m}(\rho),\frac{V(\rho)+\mathbb{1}}{2}\right]\!(\textbf{r})\,.
\ee

Since any quantum state is uniquely identified by its characteristic function, it follows from~\eqref{charact_gaussian} that any Gaussian state is uniquely identified by its first moment and covariance matrix. An example of Gaussian state is the single-mode thermal state 
\bb
    \tau_{\nu}\coloneqq \frac{1}{\nu+1}\sum_{n=0}^\infty \left(\frac{\nu}{\nu+1}\right)^{n}\ketbra{n}\,,
    \label{eq:termal}
\ee
where $\nu\ge0$. It holds that $\Tr[a^\dagger a\,\tau_{\nu} ]=\nu$, thus $\nu$ is the mean photon number of $\tau_{\nu}$. The first moment and the covariance matrix of $\tau_\nu$ satisfy
\bb\label{moments_thermal}
\mathbf{m}(\tau_{\nu})&=(0,0)^{\intercal}\,,\\
V(\tau_{\nu})&=(2\nu+1)\mathbb{1}_2 \,.
\ee
It is worth noting that the thermal state with $\nu=0$ is the vacuum state vector, $\tau_0=\ketbra{0}$. Thermal states are important since they maximise the von Neumann entropy among all states with a fixed mean photon number, as established by Lemma~\ref{maxthermstate}~\cite{max_entropy_therm}. 
\begin{lemma}[(Extremality of thermal states)]\label{maxthermstate}
	For all $\nu>0$, the maximum von Neumann entropy among all $n$-mode states with a given mean photon number $\nu$ is achieved by a thermal state $\tau_{\nu/n}^{\otimes n}$, i.e. 
	\bb
		&\max\left\{S(\rho):\,\rho\in\pazocal{D}(L^2(\mathbb{R}^n)) \text{, }\Tr\left[\rho\sum_{i=1}^n a_i^\dagger a_i\right]\le \nu\right\} = S\left(\tau_{\nu/n}^{\otimes n}\right)=ng\left(\frac{\nu}{n}\right)\,,
	\ee
	where 
	\begin{equation}\label{bosonicent}
	g(\nu)\coloneqq (\nu+1)\log_2(\nu+1) - \nu\log_2 \nu 
	\end{equation}
	is a monotone increasing function called the bosonic entropy.
\end{lemma}

In the setting of continuous variable systems, symplectic matrices play a crucial role. Recall that a matrix $S\in\mathbb{R}^{2n,2n}$ is \emph{symplectic} matrix if and only if $S\Omega_nS^\intercal=\Omega_n$ and that the group of symplectic matrices is denoted by $\mathrm{Sp}(2n)$. Fixed a symplectic matrix $S\in\mathrm{Sp}(2n)$, one can define a suitable $n$-mode unitary $U_S$ --- dubbed \emph{symplectic unitary} or the \emph{metaplectic representation} of $S$ --- such 
that, for each $k\in[2n]$ it holds 
\begin{equation}
U_S^\dagger \hat{R}_k U_S=\sum^{2n}_{k,l}S_{k,l}\hat{R}_l.
\end{equation}
In the vectorial notation, we can write this as 
\bb\label{def_U_s_quadrature}
    U_S^\dagger \mathbf{\hat{R}}U_S=S\mathbf{\hat{R}}\,.
\ee
More explicitly, the symplecitic unitary $U_S$ is defined in terms of the symplectic matrix $S$ as follows~\cite{BUCCO}. Any symplectic matrix $S$ can be written as $S=S_1S_2$, where $S_1\coloneqq e^{\Omega_n H_1}$ and $S_2\coloneqq e^{\Omega_n H_2}$ for some symmetric matrices $H_1, H_2$~\cite{BUCCO}. Then, $U_S$ is defined as $U_S\coloneqq U_{S_1}U_{S_2}$, where 
\bb
    U_{S_1}&\coloneqq e^{- \frac{i}{2}\mathbf{\hat{R}}^{\intercal}H_1\mathbf{\hat{R}}}\,,\\
    U_{S_2}&\coloneqq e^{- \frac{i}{2}\mathbf{\hat{R}}^{\intercal}H_2\mathbf{\hat{R}}}\,.
\ee
In particular,~\eqref{def_U_s_quadrature} implies that
\bb
    &\mathbf{m}(U_S\rho U_S^\dagger)=S\mathbf{m}(\rho)\,,\\
    &V(U_S\rho U_S^\dagger)=SV\!(\rho)S^\intercal\, .
\ee
Given two symplectic matrices $S_1$ and $S_2$, it holds that $U_{S_1}U_{S_2}=U_{S_1S_2}$. In addition, the displacement operator $\hat{D}_\mathbf{r}$ transforms the quadrature vector as 
\begin{align}
    \label{eq:displ}\hat{D}_\mathbf{r}\mathbf{\hat{R}}\hat{D}_\mathbf{r}^\dagger=\mathbf{\hat{R}}+\mathbf{r}\mathbb{\hat{1}},
\end{align} which implies that
\bb
    &\mathbf{m}\!\left(\hat{D}_\mathbf{r}\rho \hat{D}_\mathbf{r}^\dagger\right)=\mathbf{m}(\rho)+\mathbf{r}\,,\\
    &V\!\left(\hat{D}_\mathbf{r}\rho \hat{D}_\mathbf{r}^\dagger\right)=V\!(\rho)\,.
\label{eq:bille}
\ee
Notably, any covariance matrix $V\!(\rho)$ of an $n$-mode state $\rho$ satisfies the following decomposition, known as the Williamson decomposition~\cite{BUCCO}: there exists a symplectic matrix $S\in\mathrm{Sp}(2n)$ and $n$ real numbers $d_1,d_2,\ldots , d_n\ge1$, called the symplectic eigenvalues of $V\!(\rho)$, such that
\bb
    V\!(\rho)=SDS^\intercal\,,
    \label{eq:will}
\ee
where 
\bb
    D=\bigoplus_{j=1}^n \left(d_j\!\left(\begin{matrix}1&0\\0&1\end{matrix}\right)\right) =\operatorname{diag}\left(d_1,d_1,d_2,d_2,\ldots,d_n,d_n\right)\,.
\ee
In addition, if $\rho$ is Gaussian, it is possible to show that the above Williamson decomposition of the covariance matrix leads to the  decomposition 
\bb
    \rho= \hat{D}_{\mathbf{m}(\rho)}U_S\left(\tau_{\nu_1}\otimes\tau_{\nu_2}\otimes\ldots\tau_{\nu_n}\right)U_S^\dagger \hat{D}_{\mathbf{m}(\rho)}^\dagger\,
    \label{eq:GaussianState}
\ee
for the state $\rho$, 
where $\{\tau_{\nu_k}\}^n_{k=1}$ are thermal states defined in Eq.~\eqref{eq:termal}, and the mean photon numbers $\nu_1,\nu_2,\ldots,\nu_n$ are defined in terms of the symplectic eigenvalues of $V\!(\rho)$ via the relation $\nu_i\coloneqq \frac{d_i-1}{2}$ for all $i\in[n]$. In other words, any Gaussian state is unitarily equivalent --- through displacement and symplectic unitaries --- to a multi-mode thermal state. In particular, any pure Gaussian state vector $\ket{\phi}$ can be written as
\bb
    \ket{\phi}=\hat{D}_{\mathbf{r}}U_S \ket{0}^{\otimes n}\,.
\ee
Moreover, Gaussian unitaries can be defined as follows.
\begin{Def}[(Gaussian unitary)]
    An $n$-mode unitary is said to be Gaussian if it is the composition of unitaries generated by quadratic Hamiltionians. 
    A Hamiltonian $\hat{H}$ is said to be quadratic if it can be written as $\hat{H}\coloneqq \frac{1}{2}(\mathbf{\hat{R}}-\mathbf{m})^{\intercal}H(\mathbf{\hat{R}}-\mathbf{m})$, where $\mathbf{r}\in\mathbb R^{2n}$ and $H\in\mathbb{R}^{2n,2n}$ is a symmetric matrix. 
\end{Def}
In particular, for any $n$-mode Gaussian unitary $U$ there exists a symplectic matrix $S\in\mathrm{Sp}(2n)$ and a vector $\mathbf{r}\in\mathbb R^{2n}$ such that $U=\hat{D}_{\mathbf{r}}U_S$, where $U_S$ is the symplectic unitary associated with $S$. Conversely, for any $S\in\mathrm{Sp}(2n)$ and $\mathbf{r}\in\mathbb R^{2n}$, the unitary $\hat{D}_{\mathbf{r}}U_S$ is a Gaussian unitary. In other words, Gaussian unitaries are the composition of displacement and symplectic unitaries.

Notably, any symplectic matrix $S\in\mathrm{Sp}(2n)$ can be written in the so-called Euler (or Bloch-Messiah) decomposition as follows: there exist symplectic orthogonal matrices $O_1,O_2\in\mathrm{O}(2n)\cap \mathrm{Sp}(2n)$ and real numbers $z_1,z_1,\ldots,z_n\ge 1$ such that 
\bb\label{Euler_dec}
    S=O_1ZO_2
\ee
with 
\bb\label{Euler_dec2}
    Z\coloneqq \bigoplus_{j=1}^n\left(\begin{matrix}z_j&0\\0&z_j^{-1}\end{matrix}\right)\,.
\ee
Symplectic unitaries associated with symplectic matrices of the form of $Z$ are said to be \emph{squeezing} unitaries. Moreover, symplectic unitaries associated with symplectic orthogonal matrices are said to be \emph{passive}. Importantly, passive unitaries preserve the energy operator, meaning that given $O\in\mathrm{O}(2n)\cap \mathrm{Sp}(2n)$ then 
\bb
    U_O^\dagger \hat{E}_n U_O=\hat{E}_n\,.
\ee
Moreover, note that the Euler decomposition in~\eqref{Euler_dec} implies that any Gaussian unitary can be written as a composition of passive unitaries and squeezing unitary, i.e., $U_{S}=U_{O_1}U_ZU_{O_2}$.

Finally, let us briefly mention two measurements that are 
presumably experimentally most feasible to implement in
common quantum optics laboratories: \emph{heterodyne} and \emph{homodyne} measurements~\cite{BUCCO}. Heterodyne measurement is a POVM whose elements are given by all the coherent states, suitably normalised. Specifically, these POVM elements are $\{\frac{1}{(2\pi)^n}\mathrm{d}^{2n}\textbf{r}\ketbra{\textbf{r}}\}_{r\in\mathbb{R}^{2n}}$, which sum up to the identity thanks to the over-complete relation of coherent states~\cite{BUCCO} according to
\bb
    \frac{1}{(2\pi)^n}\int_{\mathbb{R}^{2n}}\mathrm{d}^{2n}\textbf{r}\ketbra{\textbf{r}}=\mathbb{1}\,.
\ee
Moreover, a homodyne measurement entails measuring a quadrature observable, which is a component of either the quadrature operator vector $\hat{\textbf{R}}$ or, more generally, a rotated version $U_S^\dagger \hat{\textbf{R}}U_S=S\hat{\textbf{R}}$, where $U_S$ represents the symplectic unitary associated with the symplectic matrix $S$.

\subsection{Preliminaries on quantum learning theory}
\label{sec:QLTprel}

In this section, we introduce the concept of \emph{quantum state tomography}, which forms the basis of our investigation. We start by formulating the problem of quantum state tomography.
\begin{problem}[(Quantum state tomography)] 
\label{def:PROBtom}
    Let $n \in \N$ be the number of modes/qudits. Let $\pazocal{S} \subseteq \pazocal{D}(\HH^{\otimes n})$ be a subset of the set of quantum states $\pazocal{D}(\HH^{\otimes n})$. Consider $\varepsilon,\delta \in (0,1)$ and $N\in \N$.  
    Let $\rho \in \pazocal{S}$ be an unknown quantum state. Given access to $N$ copies of $\rho$, the goal is to provide a classical description of a quantum state $\tilde{\rho}$ such that
    \bb
        \Pr\left(\frac{1}{2}\|\tilde{\rho} - \rho \|_1 \le \varepsilon\right)\ge 1-\delta\,.
    \ee
    That is, with a probability $\ge 1-\delta$, the trace distance between $\tilde{\rho}$ and $\rho$ is at most $\varepsilon$.    Here, $\varepsilon$ is called the accuracy, while $\delta$ is called the failure probability.
\end{problem}
\begin{Def}[(Sample, time, and memory complexity of a tomography algorithm)]
  A quantum state tomography algorithm is an algorithm that solves Problem~\ref{def:PROBtom}. It is characterised by the sample complexity, i.e., the number of copies of the unknown state needed by the algorithm, represented by $N$ in Problem~\ref{def:PROBtom}. The time complexity of the tomography algorithm is the amount of classical and quantum computation time required to execute the algorithm. The memory complexity quantifies the amount of classical memory required by the algorithm. 
\end{Def}
We say that a tomography algorithm is \emph{efficient} if its sample, time, and memory complexity scale polynomially in the number of modes/qudits $n$. 
It is worth noting that the time complexity of an algorithm always upper bound its memory complexity, as well as its sample complexity. In Fig.~\ref{fig_tomography} we provide a pictorial representation of a quantum state tomography algorithm.

In the literature, several tomography algorithms~\cite{odonnell2015efficient,Haah_2017,kueng2014low} are tailored to finite-dimensional rank-$r$ quantum states $\rho \in \pazocal{D}(\mathbb{C}^{D})$, where $D$ is the dimension of the Hilbert space. In the case of $n$-qudits systems with local dimension $d$, then the dimension of the Hilbert space is $D=d^n$.  If no assumptions on the states are known, then it is known that any tomography algorithm needs at least $\tildeOmega\,(rD/\varepsilon^2)$ copies of the unknown state to solve the tomography problem~\cite{Haah_2017}, where $\varepsilon$ is the accuracy parameter. In this case, the tilde above the Big-O asymptotic functions means that we neglect logarithmic factors. There are algorithms that match these lower bounds, i.e., they require $\tildeO\,(rD/\varepsilon^2)$ copies of the unknown state to solve the tomography problem \cite{Haah_2017,odonnell2015efficient}. Such algorithms that achieve optimal performances use entangled measurements between multiple copies of the unknown state. However, there are also algorithms, which might be more experimentally feasible, that use \emph{unentangled} or \emph{single-copies} measurements between the queried copies of the unknown state; in this case, a lower bound $\tildeOmega\,(r^2D/\varepsilon^2)$ is known to hold~\cite{chen2023does}. Also in this case, there are known algorithms which achieve optimal performances \cite{guta2018fast}. 
In the case of an $n$-qudits system, all these algorithms necessarily scale exponentially with the number of qudits $n$. However, if further restrictions on the class of states are provided, then one might achieve a sample and time complexity scaling which is polynomial in $n$ (or, equivalently, logarithmic in $D$). In the finite-dimensional case, only a few classes of states are currently known to be efficiently learnable, such as matrix product states~\cite{Cramer_2010}, finitely-correlated states~\cite{fanizza2023learning}, states prepared by shallow quantum circuits~\cite{huang2024learning}, quantum phase states~\cite{arunachalam2023optimal}, stabiliser states~\cite{montanaro2017learning}, states prepared by Clifford circuits `doped' with at most $O(\log(n))$ T-gates~\cite{grewal2023efficient,leone2023learning,hangleiter2024bell}, fermionic Gaussian states~\cite{aaronson2023efficient}, and states prepared by fermionic Gaussian circuits `doped' with at most $O(\log(n))$ non-Gaussian gates~\cite{mele2024efficient}. For a detailed literature review of quantum learning theory, see Ref.\ \cite{anshu2023survey}.

\begin{figure*}[t]
    \centering
    \includegraphics[width=0.55\textwidth]{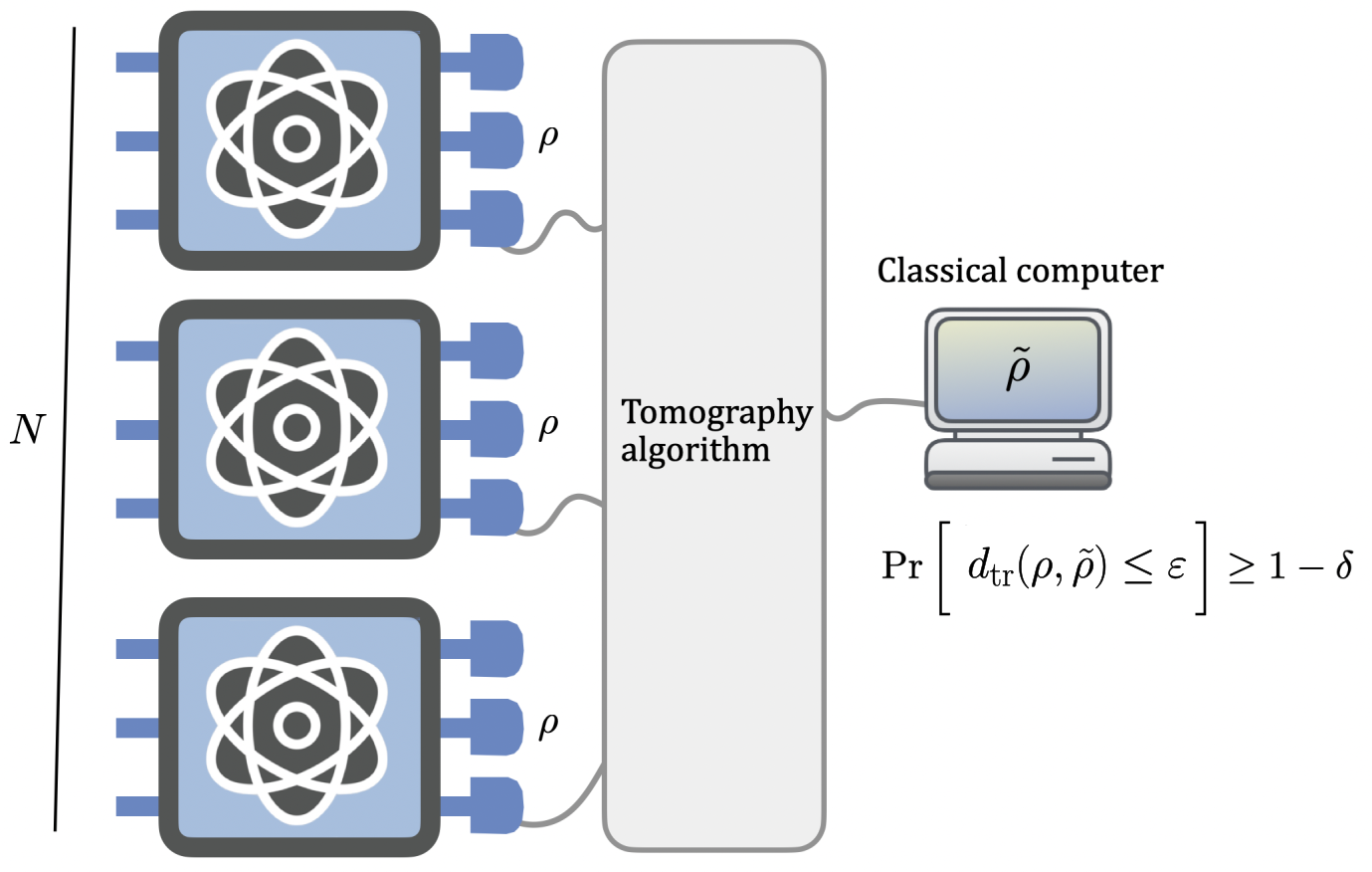}
    \caption{Pictorial representation of a quantum state tomography algorithm. Given access to $N$ copies of an unknown state $\rho$, the goal of a tomography algorithm is to construct a classical description of a state $\tilde{\rho}$ that serves as a `good approximation' of the true unknown state $\rho$. Mathematically, the error incurred in such an approximation is measured by the trace distance between $\rho$ and $\tilde{\rho}$. This is the most meaningful way to measure the error incurred in a tomography algorithm, due to the operational meaning of the trace distance given by the Holevo--Helstrom theorem~\cite{HELSTROM, Holevo1976}. Additionally, since quantum measurements inherently yield probabilistic outcomes, the output $\tilde{\rho}$ is probabilistic rather than deterministic. We thus require that the probability that `the trace distance is small' is high. Mathematically, this translates to $\mathrm{Pr}\!\left[\frac{1}{2}\|\tilde{\rho}-\rho\|_1\le \varepsilon\right]\ge 1-\delta$, where $\varepsilon$ represents the trace distance error and $\delta$ denotes the failure probability. Fixed $\varepsilon$ and $\delta$, the minimum number of copies $N$ required to achieve quantum state tomography with trace distance error $\varepsilon$ and failure probability $\delta$ is called the sample complexity. }
    \label{fig_tomography}
\end{figure*}

Given a tomography algorithm, we define a shorthand for denoting its sample and time complexity that will be useful later. 
\begin{Def}[(Sample and time complexity)]
\label{def:performance}
Let $\rho \in \pazocal{D}(\mathbb{C}^{D})$ be a $D$-dimensional quantum state where $D\in \N$. Let $\varepsilon,\delta\in (0,1)$. Let $\pazocal{A}$ be a tomography algorithm.
We denote as $N_{\mathrm{tom}}(\pazocal{A}, D,\varepsilon,\delta)$ and $T_{\mathrm{tom}}(\pazocal{A},D,\varepsilon,\delta)$, respectively, the sample and time complexity of the algorithm $\pazocal{A}$ to solve the tomography problem of the state $\rho$ with accuracy at most $\varepsilon$ and a failure probability at most $\delta$. That is, given access to $N_{\mathrm{tom}}(\pazocal{A}, D,\varepsilon,\delta)$ copies of $\rho$ and using at most $T_{\mathrm{tom}}(\pazocal{A},D,\varepsilon,\delta)$ computational time, the algorithm $\pazocal{A}$ outputs, with probability $\ge 1-\delta$, a classical description of a state $\tilde{\rho}$ such that
\begin{align}
    \frac{1}{2}\|\tilde{\rho}-\rho\|_1\le \varepsilon.
\end{align}
\end{Def}

We now present a standard proof technique to provide lower bound on the sample complexity of a tomography algorithm (see, e.g.,~Ref.~\cite{haah_optimal_2021}), which relies on the use of Fano's inequality~\cite{Cover2006} and Holevo's bound~\cite{Holevo1973BoundsFT}.

\begin{lemma}[(Lower bound on the sample complexity)]\label{lower_bound_sample_complexity}
Let $M \in \N$ and $\varepsilon > 0$. Consider a distance metric $d(\cdot, \cdot)$ between quantum states (e.g.,~trace distance). Define the discrete set $\pazocal{M}_{\varepsilon} \coloneqq \{\rho_1, \dots, \rho_M\}$ as a set of quantum states such that for every $i \neq j \in [M]$, $d(\rho_i, \rho_j) > 2\varepsilon$. Let $\pazocal{S}$ be a subset of quantum states such that $\pazocal{M}_{\varepsilon} \subseteq \pazocal{S}$. Any tomography algorithm that learns states within the metric $d(\cdot, \cdot)$ from the set $\pazocal{S}$ with accuracy $\le \varepsilon$ and failure probability $\le \delta$ must use a number of state copies $N$ satisfying:
\begin{align}
    \chi \ge (1-\delta)\log_2(M) - H_2(\delta),
\end{align}
where $\chi \coloneqq S\left(\frac{1}{M}\sum^{M}_{j=1} \rho^{\otimes N}_j\right) - \frac{1}{M}\sum^{M}_{j=1} S(\rho_j^{\otimes N})$ is the Holevo information, $ S(\cdot)$ denotes the von Neumann entropy, and $H_2(x) \coloneqq -x \log_2 x - (1-x)\log_2(1-x)$ is the binary entropy.
For instance, if the states in $\pazocal{M}_{\varepsilon}$ live in a $2^n$ dimensional Hilbert space, then this implies that the number of copies $N$ satisfies
\begin{align}
    N \ge \Omega\left(\frac{\log_2(M)}{n}\right)\,.
\end{align}
\end{lemma}

\begin{proof}  
Consider a communication protocol between Alice and Bob. First, define a codebook that Alice will use to encode classical information in quantum states: each number $i \in [n]$ is associated with the quantum state $\rho_i \in \pazocal{M}_{\varepsilon}$. Alice samples uniformly a number $x \in [M]$. Let us denote the random variable associated to $x$ with $X$. Then Alice prepares the state $\rho^{\otimes N}_x$ and sends it to Bob. Now Bob runs the tomography algorithms and learns the classical description of a matrix $\hat{\rho}$ such that $d(\rho_x, \hat{\rho}) \le \varepsilon$ with failure probability $\le \delta$.
In case the tomography protocol succeeds, this would imply that for all $j \in [M]$ such that $j \neq x$, we have $d(\rho_j, \hat{\rho}) > \varepsilon$, because
\begin{align}
    d(\rho_j, \hat{\rho}) \ge d(\rho_j, \rho_x) - d(\rho_x, \hat{\rho}) > 2\varepsilon - \varepsilon = \varepsilon\,.
\end{align}
Thus, in such a case, there exists only one element in the set $\pazocal{M}_{\varepsilon}$ that is $\varepsilon$-close to $\hat{\rho}$ and this has to be $\rho_x$. Hence, by going through all the $M$ states in $\pazocal{M}_{\varepsilon}$ and for each of them computing the trace distance with $\hat{\rho}$, Bob can find the number $x \in [M]$ corresponding to $\rho_x$ (with probability $\ge 1- \delta$). Thus, at the ends of this procedure, Bob selects a number $y \in [M]$ and with probability $\ge 1- \delta$, it will coincide with $x$, i.e., the message sent by Alice. Therefore, the probability that Bob decodes the wrong message is $\le \delta$.
Call $Y$ the random variable associated with $y$. By Fano's inequality (see Ref.~\cite{Cover2006}), we have 
\begin{align}
    I(X:Y) \ge (1-\delta)\log_2(M) - H_2(\delta),
\end{align}
where $I(X:Y)$ is the mutual information (see, e.g. Ref.~\cite{Cover2006} for the standard definition) and $H_2(\cdot)$ is the binary entropy. By the data processing inequality~\cite{Cover2006}, we have
\begin{align}
    I(X:Y) \le I(X:Z),
\end{align}
where $Z$ is the random variable associated with the POVM outcome associated with the tomography algorithm. By Holevo's theorem~\cite{Holevo1973BoundsFT} (see also, e.g.,~Wikipedia), we have
\begin{align}
    I(X:Z) \le S\left(\frac{1}{M}\sum^{M}_{j=1} \rho^{\otimes N}_j\right) - \frac{1}{M}\sum^{M}_{j=1} S(\rho_j^{\otimes N})=\chi.
\end{align}
Hence, we get
\begin{align}
    \chi \ge (1-\delta)\log_2(M) - H_2(\delta).
\end{align}
If the states in $\pazocal{M}_{\varepsilon}$ live in a $2^n$ dimensional Hilbert space, then we have
\begin{align}
    \chi \le S\left(\frac{1}{M}\sum^{M}_{j=1} \rho^{\otimes N}_j\right)\le N \log(2^n).
\end{align}
This implies that
\begin{align}
    N \ge \frac{1}{\log_2(2^n)}\left((1-\delta)\log_2(M) - H_2(\delta)\right) = \Omega\left(\frac{\log_2(M)}{n}\right).
\end{align}

\end{proof}

\subsection{Preliminaries on $\varepsilon$-nets}\label{sec:epsilonnet}
In this section, we introduce key concepts related to $\varepsilon$-nets~\cite{vershynin_2018} before delving into our novel results on the topic.
Let us begin by giving definitions relative to the notion of $\varepsilon$-covering net~\cite{vershynin_2018}.
\begin{definition}[($\varepsilon$-covering net/number)]
Let $(T, \|\cdot\|)$ be a normed space with distance induced by the norm $\|\cdot\|$, $K \subseteq T$ be a subset, and let $\varepsilon > 0$. We define:
\begin{itemize}
    \item \textbf{$\varepsilon$-covering net:} A subset $C \subseteq K$ is an $\varepsilon$-covering net of $K$ if, for every $x \in K$, there exists $x_0 \in C$ such that $\|x-x_0\| \leq \varepsilon$.
    \item \textbf{$\varepsilon$-covering number:} The covering number of $K$, denoted as $\pazocal{C}(K, \|\cdot\|, \varepsilon)$, is the smallest cardinality of the $\varepsilon$-covering nets of $K$.
    \item \textbf{Optimal $\varepsilon$-covering:} $C \subseteq K$ is an optimal $\varepsilon$-covering of $K$ if and only if $|C|=\pazocal{C}(K, \|\cdot\|, \varepsilon)$, where $|C|$ denotes the cardinality of $C$.
\end{itemize}
\end{definition}
We now move to define definitions related to the notion of $\varepsilon$-packing net~\cite{vershynin_2018}.
\begin{definition}[($\varepsilon$-packing net/number)]\label{def_pack}
Let $(T, \|\cdot\|)$ be a normed space with distance induced by the norm $\|\cdot\|$, $K \subseteq T$ be a subset, and let $\varepsilon > 0$. We define:
\begin{itemize}
    \item \textbf{$\varepsilon$-packing net:} A subset $P \subseteq K$ is an \emph{$\varepsilon$-packing net} of $K$ if $\|x-y\| > \varepsilon$ for every $x,y \in P$.
    \item \textbf{$\varepsilon$-packing number:} The packing number of $K$, denoted as $\pazocal{P}(K, \|\cdot\|, \varepsilon)$, is the largest cardinality of the $\varepsilon$-packing nets of $K$.
    \item \textbf{Optimal $\varepsilon$-packing:} $P\subseteq K$ is is an optimal $\varepsilon$-covering of $K$ if and only if $|P|=\pazocal{P}(K, \|\cdot\|, \varepsilon)$.
\end{itemize}
\end{definition}

A result in the theory of $\varepsilon$-net~\cite{vershynin_2018} connects the $\varepsilon$-covering number $\pazocal{C}(K, \|\cdot\|, \varepsilon)$ to the $\varepsilon$-packing number $\pazocal{P}(K, \|\cdot\|, \varepsilon)$ .
\begin{lemma}[(Ref.~\cite{vershynin_2018}, Lemma 4.2.8)
]\label{lemma1111}
Let $(T, \|\cdot\|)$ be a normed space, $K \subseteq T$ be a subset, and let $\varepsilon > 0$. Then, it holds that
\bb
\pazocal{P}(K, \|\cdot\|, \varepsilon)\ge \pazocal{C}(K, \|\cdot\|, \varepsilon)\ge \pazocal{P}(K, \|\cdot\|, 2\varepsilon)\,.
\label{eq2121} 
\ee
\end{lemma}

\subsubsection{Our results on $\varepsilon$-nets}
In this section we will present our results concerning $\varepsilon$-nets. 
We will begin by proving a lemma that relates the covering number of a ($d-1$)-dimensional sphere with unit radius to the covering number of the surface of a $d$-dimensional sphere with unit radius, where $d\in \N_+$. 

\begin{lemma}[(Covering numbers)]\label{lemma3333}
Consider the normed vector space $(\mathbb{R}^{d}, \|\cdot\|_2)$, where $d \in \mathbb{N}_+$. Let $B_1(d)$ be the $d$-dimensional unit ball and $\partial B_1(d)$ its boundary, defined as
\begin{align}
     B_1(d)&\coloneqq \{x\in \mathbb{R}^{d}\,:\, \|x\|_2\le1  \}, \\
    \partial B_1(d)&\coloneqq \{x\in \mathbb{R}^{d}\,:\, \|x\|_2=1  \}
    .
\end{align}
For all $d\in\mathbb{N}_+$ and $\varepsilon>0$, the covering numbers of $B_1(d-1)$ and $\partial B_1(d)$ are related as
\begin{align}
\pazocal{C}(\partial B_1(d), \|\cdot\|_2, \varepsilon) \ge \pazocal{C}( B_1(d-1), \|\cdot\|_2, \varepsilon).
\end{align}
\begin{proof}
Consider an \textit{optimal} $\varepsilon$-covering net of $\partial B_1(d)$, denote with $C^{\varepsilon, \mathrm{opt}}_{\partial B_1(d)}$. Note that for any $y\in B_1(d-1)$, the $d$-dimensional vector $v(y)\coloneqq \left(y,\sqrt{1-\|y\|^2_2}\right)$ satisfies $v(y)\in \partial B_1(d)$.
Therefore, for any $y\in B_1(d-1)$ there exists $x\in C^{\varepsilon, \mathrm{opt}}_{\partial B_1(d)}$ such that
\begin{align}
    \|v(y)-(x_{1},\ldots, x_{d})\|_2\le \varepsilon.
\end{align}
Thus, we also have that 
\be
\|y-(x_{1},\ldots, x_{d-1})\|_2\le \sqrt{\|y-(x_{1},\ldots, x_{d-1})\|^2_2+(\sqrt{1-\|y\|^2_2} - x_d)^2}=\|v(y)-(x_{1},\ldots, x_{d})\|_2\le \varepsilon\,.
\ee
This implies that the set 
\bb
    C^{\varepsilon}_{B_1(d-1)} \coloneqq\{(x_{1},\ldots, x_{d-1})\in \mathbb{R}^{d-1}\,|\, x\in C^{\varepsilon, \mathrm{opt}}_{\partial B_1(d)}\}
    \label{eqdef:cb1}
\ee    
is an $\varepsilon$-covering net for $B_1(d-1)$.  Thus, we have
\be
\pazocal{C}(\partial B_1(d), \|\cdot\|_2, \varepsilon)=|C^{\varepsilon, \mathrm{opt}}_{\partial B_1(d)} | 
 \ge  |C^{\varepsilon}_{B_1(d-1)}|\ge \pazocal{C}( B_1(d-1), \|\cdot\|_2, \varepsilon)
\ee
where the first step follows from the definition of covering number, the second step follows from the fact that $C^{\varepsilon}_{B_1(d-1)}$ contains at most the same number of elements as $C^{\varepsilon, \mathrm{opt}}_{\partial B_1(d)}$, and the last inequality follows from the fact that the $\varepsilon$-covering number is a lower bound on the cardinality of every possible $\varepsilon$-covering net. This concludes the proof.
\end{proof}
\end{lemma}
We now present a lemma which relates the covering number of the set of pure states density matrices with respect to the (matrix) $1$-norm and the covering number of the set of pure states vectors with respect to the (vector) $2$-norm. 

\begin{lemma}[(Covering numbers revisited)]\label{lemma2222}
Let $d\in \N$ and $\varepsilon\in [0,1)$.
We define the sets $\pazocal{S}_{\ket{\ast}\!\bra{\ast}}$ and $\pazocal{S}_{\ket{\ast}}$, as 
\begin{align}
    \pazocal{S}_{\ket{\ast}\!\bra{\ast}}&\coloneqq \{ \ketbra{v} \, : \, \ket{v}\in \mathbb{C}^{d} \text{ such that } |\braket{v|v}|=1 \},\\
    \pazocal{S}_{\ket{\ast}}&\coloneqq \{ \ket{v} \, : \, \ket{v}\in \mathbb{C}^{d} \text{ such that } |\braket{v|v}|=1 \}.
\end{align}
It holds the following relation between their covering numbers 
\be
\pazocal{C}(  \pazocal{S}_{\ket{\ast}\!\bra{\ast}}, \|\cdot\|_1, \varepsilon)\ge \frac{\varepsilon}{8\pi} \pazocal{C}(  \pazocal{S}_{\ket{\ast}}, \|\cdot\|_2, \varepsilon)\,.
\ee
\end{lemma}
\begin{proof}
Let $C^{\varepsilon, \mathrm{opt}}_{\pazocal{S}_{\ket{\ast}\!\bra{\ast}}}=\{\ketbra{\psi_1},\dots, \ketbra{\psi_{M}}\}$ be an optimal $ \varepsilon$-covering net of the set $\pazocal{S}_{\ket{\ast}\!\bra{\ast}}$ with respect to the trace norm $\|\cdot\|_1$, where we defined $M\coloneqq \pazocal{C}(  \pazocal{S}_{\ket{\ast}\!\bra{\ast}}, \|\cdot\|_1, \varepsilon)$. 
For each $j\in [M]$, we take $\ket{\phi_j}$ to be a state such that $\ketbra{\phi_j}=\ketbra{\psi_j}$. 
Let $\delta>0$ be a number that we will fix later in terms of $\varepsilon$. We define the set 
\be
C_{\delta}\coloneqq\left\{e^{i\delta m}\ket{\phi_j}\,:\,   j\in [M] ,\, m\in\Big\{0,1,\dots,\left\lfloor\frac{2\pi}{\delta}\right\rfloor\Big\} \right\}\,.
\label{eq:cdelt}
\ee

Let us show that $C_{\delta}$ is an $(\frac{\varepsilon}{\sqrt{2}}+\delta)$-covering net of $\pazocal{S}_{\ket{\ast}}$ with respect to the $2$-norm $\|\cdot\|_2$. To this end, fixed a state vector $\ket{\psi}\in \mathbb{C}^{d}$, it suffices to construct a state vector $\ket{\tilde{\psi}}\in C_{\delta}$ such that 
\bb\label{eq_epsilon_net_1}
    \|\ket{\psi}-\ket{\tilde{\psi}}\|_2\le \frac{\varepsilon}{\sqrt{2}}+\delta\,.
\ee
By definition of $C^{\varepsilon, \mathrm{opt}}_{\pazocal{S}_{\ket{\ast}\!\bra{\ast}}}$, there exists $j\in [M]$ such that $\|\ketbra{\phi_j}-\ketbra{\psi}\|_1\le \varepsilon$.
Let $\bar{m}\coloneqq \lfloor\frac{2\pi-\arg(\braket{\psi|\phi_j})}{\delta}\rfloor$, where $\arg(\braket{\psi|\phi_j})$ is such that $\braket{\psi|\phi_j}=|\braket{\psi|\phi_j}|e^{i\arg(\braket{\psi|\phi_j})}$, and let us define the state vector $\ket{\tilde{\psi}}\coloneqq e^{i\delta\bar{m}}\ket{\phi_j}\in C_{\delta}$. Let us now show that the so-defined $\ket{\tilde{\psi}}$ 
satisfies Eq.~\eqref{eq_epsilon_net_1}:
\bb
    \|\ket{\psi}-\ket{\tilde{\psi}}\|_2&= \|\ket{\psi}-e^{i\delta\bar{m}}\ket{\phi_j}\|_2\\
    &\le \|\ket{\psi}-e^{-i\arg(\braket{\psi|\phi_j})}\ket{\phi_j}\|_2+\|e^{-i\arg(\braket{\psi|\phi_j})}\ket{\phi_j}-e^{i\delta\bar{m}}\ket{\phi_j}\|_2\\
    &\leqt{(i)} \frac{1}{\sqrt{2}}\|\ketbra{\psi}-\ketbra{\phi_j}\|_1+\|e^{-i\arg(\braket{\psi|\phi_j})}\ket{\phi_j}-e^{i\delta\bar{m}}\ket{\phi_j}\|_2\\
    &\le \frac{\varepsilon}{\sqrt{2}} +|e^{i\left(2\pi-\arg(\braket{\psi|\phi_j})\right)}-e^{i\delta\bar{m}}|\\
    &\leqt{(ii)} \frac{\varepsilon}{\sqrt{2}} +|2\pi -\arg(\braket{\psi|\phi_j})-\delta\bar{m}|\\
    &\le \frac{\varepsilon}{\sqrt{2}} +\delta\,.
\ee
Here, (i) follows from the fact that
\begin{align}
    \left\|\ket{\psi}-e^{-i\arg(\braket{\psi|\phi_j})}\ket{\phi_j}\right\|_2 &= \sqrt{2\left(1-\operatorname{Re}{(e^{-i\arg(\braket{\psi|\phi_j})}\braket{\psi|\phi_j})}\right)} \\
    \nonumber
    &= \sqrt{2\left(1-|\braket{\psi|\phi_j}|\right)} \\
     \nonumber
    &\le \sqrt{2\left(1-|\braket{\psi|\phi_j}|^2\right)} \\
     \nonumber
    &= \frac{1}{\sqrt{2}}\left\|\ketbra{\psi}-\ketbra{\phi_j}\right\|_1\,,
     \nonumber
\end{align}
where in the last step we have used  the well-known formula for the trace distance between two pure states (see, e.g. Ref.~\cite{wilde_quantum_2013}).
Moreover, (ii) is a consequence of the fact that for any $x,y\in\mathbb{R}$ it holds that
\bb
    |e^{i x}-e^{i y}| =\sqrt{2}\sqrt{1-\cos(x-y)}\le |x-y|\,.
\ee
Hence, we have proved that $C_{\delta}$ is a $(\frac{\varepsilon}{\sqrt{2}}+\delta)$-covering net of $\pazocal{S}_{\ket{\ast}}$ with respect to the $2$-norm $\|\cdot\|_2$. Consequently, by setting $\delta\coloneqq \varepsilon-\frac{\varepsilon}{\sqrt{2}}$, we have that $C_{\varepsilon-\frac{\varepsilon}{\sqrt{2}}}$ is an $\varepsilon$-covering net of $\pazocal{S}_{\ket{\ast}}$ with respect to the $2$-norm $\|\cdot\|_2$. In particular, its cardinality has to satisfy
\bb
    |C_{\varepsilon-\frac{\varepsilon}{\sqrt{2}}}|\ge \pazocal{C}(  \pazocal{S}_{\ket{\ast}}, \|\cdot\|_2, \varepsilon)\,.
\ee
Consequently, we can conclude that
\begin{align}
\pazocal{C}(  \pazocal{S}_{\ket{\ast}\!\bra{\ast}}, \|\cdot\|_1, \varepsilon) &= |C^{\varepsilon, \mathrm{opt}}_{\pazocal{S}_{\ket{\ast}\!\bra{\ast}}}| \\
\nonumber
&\ge \left(\left\lfloor \frac{2\pi}{\varepsilon-\frac{\varepsilon}{\sqrt{2}}}\right\rfloor +1\right)^{-1}|C_{\varepsilon-\frac{\varepsilon}{\sqrt{2}}}| \\
\nonumber
&\ge \left(\left\lfloor \frac{2\pi}{\varepsilon-\frac{\varepsilon}{\sqrt{2}}}\right\rfloor +1\right)^{-1}\pazocal{C}(  \pazocal{S}_{\ket{\ast}}, \|\cdot\|_2, \varepsilon) \\
\nonumber
&\ge \frac{\varepsilon}{8\pi}\pazocal{C}(  \pazocal{S}_{\ket{\ast}}, \|\cdot\|_2, \varepsilon),
\nonumber
\end{align}
where the first inequality follows by a simple counting argument, based on Eq.~\eqref{eq:cdelt}, which shows that $|C_{\delta}|\le |C^{\varepsilon, \mathrm{opt}}_{\pazocal{S}_{\ket{\ast}\!\bra{\ast}}}|(\left\lfloor\frac{2\pi}{\delta}\right\rfloor+1)$, and the last inequality follows by inspection. 
\end{proof}
The following lemma establishes a lower bound on the $\varepsilon$-packing number of the set of pure state density matrices with respect to the trace norm.

\begin{lemma}[($\varepsilon$-packings)]\label{lemma6666} 
Let $d\in \N$ and $\varepsilon\in [0,1)$. Define the set of pure state density matrices $\pazocal{S}_{\ket{\ast}\!\bra{\ast}}$ as
\begin{align}
     \pazocal{S}_{\ket{\ast}\!\bra{\ast}}&\coloneqq \{ \ketbra{v} \, : \, \ket{v}\in \mathbb{C}^{d} \text{ such that } |\braket{v|v}|=1 \}.
\end{align}
Then, we have
\be
\pazocal{P}(  \pazocal{S}_{\ket{\ast}\!\bra{\ast}}, \|\cdot\|_1, \varepsilon)\ge\frac{1}{8\pi}\varepsilon^{-2(d-1)}\,.
\ee
\end{lemma}
\begin{proof} 
Consider the set of pure state vectors $\pazocal{S}_{\ket{\ast}\!\bra{\ast}}$ defined as
\begin{align}
    \pazocal{S}_{\ket{\ast}}&\coloneqq \{ \ket{v} \, : \, \ket{v}\in \mathbb{C}^{d} \text{ such that } |\braket{v|v}|=1 \}.
\end{align}
Let $B_1(k)$ be the $k$-dimensional unit ball and $\partial B_1(k)$ its boundary, with $k\in \mathbb{N}$ and $r\in \mathbb{R}_+$, defined as
\begin{align}
     B_r(k)&\coloneqq \{x\in \mathbb{R}^{k}\,:\, \|x\|_2\le r  \},\\
    \partial B_r(k)&\coloneqq \{x\in \mathbb{R}^{k}\,:\, \|x\|_2= r  \}.
\end{align}
We can establish the following chain of inequalities:
\be
\pazocal{P}(  \pazocal{S}_{\ket{\ast}\!\bra{\ast}}, \|\cdot\|_1, \varepsilon) &\geqt{(i)} \pazocal{C}(  \pazocal{S}_{\ket{\ast}\!\bra{\ast}}, \|\cdot\|_1, \varepsilon)\\
&\geqt{(ii)} \frac{\varepsilon}{8\pi} \pazocal{C}(  \pazocal{S}_{\ket{\ast}}, \|\cdot\|_2, \varepsilon)\\
&\eqt{(iii)} \frac{\varepsilon}{8\pi} \pazocal{C}(  \partial B_1(2d), \|\cdot\|_2, \varepsilon)\\
&\geqt{(iv)} \frac{\varepsilon}{8\pi} \pazocal{C}( B_1(2d-1), \|\cdot\|_2, \varepsilon)\\
&\geqt{(v)} \frac{\varepsilon}{8\pi} \frac{\text{Vol}( B_1(2d-1))}{\text{Vol}( B_\varepsilon(2d-1))}\\
&\eqt{(vi)}\frac{\varepsilon^{-2(d-1)}}{8\pi}\,.
\ee
Here, (i) and (ii) follow from Lemma~\ref{lemma1111} and Lemma~\ref{lemma2222}, respectively. Moreover, (iii) stems from viewing the set of pure states $\pazocal{S}_{\ket{\ast}}$ as the boundary of the $2d$-dimensional real ball $B_1(2d)$. In (iv), we apply Lemma~\ref{lemma3333}. In (v), we denote $\text{Vol}(B_1(2d-1))$ as the volume of the $(2d-1)$-dimensional ball with unit radius and $\text{Vol}(B_\varepsilon(2d-1))$ as the volume of the $(2d-1)$-dimensional ball with radius $\varepsilon$. We then exploit a simple counting argument to get $\text{Vol}(B_\varepsilon(2d-1))\pazocal{C}( B_1(2d-1), \|\cdot\|_2, \varepsilon)\ge \text{Vol}(B_1(2d-1))$. Finally, in (vi), we use the known fact that the volume of an $n$-dimensional ball with radius $r$ is $r^n \text{Vol}(B_1(2d-1))$.        
\end{proof}

We now present a lemma which relates the covering number of the set of quantum states on $\mathbb{C}^{d}$ with the covering number of the set of Hermitian matrices with trace norm smaller than one.

\begin{lemma}[(Covering number of the set of Hermitian matrices)]\label{lemma_cov_pos_norm}
Let $\varepsilon>0$ and $d\in \N_+$. Let $\pazocal{D}\!\left(\mathbb{C}^{d}\right)$ and  $\mathrm{Herm}_{\|\ast\|_1\le 1}(d)$ be respectively the set of quantum states on $\mathbb{C}^{d}$ and the set of Hermitian matrices with trace norm smaller than one, 
i.e.
\bb
     \pazocal{D}\!\left(\mathbb{C}^{d}\right) &\coloneqq \left\{\rho\in\mathbb{C}^{d,d}\,:\,\rho\ge0,\,\Tr\rho=1\right\}\,,\\
     \mathrm{Herm}_{\|\ast\|_1\le 1}(d)&\coloneqq \left\{X\in\mathbb{C}^{d,d}\,:\,X^{\dagger}=X,\,\|X\|_1\le1\right\}\,.
\ee
Then, it holds the following relation between their covering numbers 
\bb
\pazocal{C}(  \pazocal{D}\!\left(\mathbb{C}^{d}\right) , \|\cdot\|_1, \varepsilon)  \ge \sqrt{\pazocal{C}(  \mathrm{Herm}_{\|\ast\|_1\le 1}(d-1) , \|\cdot\|_1, 2\varepsilon)}\,.
\ee
\end{lemma}
\begin{proof}
The strategy of the proof is to construct a covering net for $\mathrm{Herm}_{\|\ast\|_1\le 1}(d-1)$ from an optimal $\varepsilon$-covering net of set of quantum states $\pazocal{D}\!\left(\mathbb{C}^{d}\right)$. Let $C^{\varepsilon, \mathrm{opt}}_{\pazocal{D}\left(\mathbb{C}^{d}\right)}$ be an optimal $\varepsilon$-covering net for $\pazocal{D}\!\left(\mathbb{C}^{d}\right)$. Let us fix $X\in  \mathrm{Herm}_{\|\ast\|_1\le 1}(d-1)$, and let us write it in its Jordan–Hahn decomposition~\cite{WATROUS} as follows, $X=X_1-X_2$, where $X_1,X_2\ge0$ and $X_1X_2=X_2X_1=0$. Since 
\bb
    1\ge\|X\|_1=\Tr X_1+\Tr X_2\,,
\ee
it follows that 
\bb
    1-\Tr X_1&\ge0\,,\\
    1-\Tr X_2&\ge0\,.
\ee
Let us consider the $d\times d$ block diagonal matrices
\bb
    \bar{X}_1&\coloneqq\left(\begin{matrix}X_1&\textbf{0}\\\textbf{0}^\intercal &1-\Tr X_1\end{matrix}\right)\,,\\
    \bar{X}_2&\coloneqq\left(\begin{matrix}X_2&\textbf{0}\\\textbf{0}^\intercal &1-\Tr X_2\end{matrix}\right)\,,
\ee
where $\textbf{0}$ is the $(d-1)$-dimensional zero vector. Hence, both $\bar{X}_1$ and $\bar{X}_2$ are density operators, i.e., $\bar{X}_1,\bar{X}_2\in\pazocal{D}\!\left(\mathbb{C}^{d}\right)$. By definition of $C^{\varepsilon, \mathrm{opt}}_{\pazocal{D}\left(\mathbb{C}^{d}\right)}$, there exist $\rho_1,\rho_2\in C^{\varepsilon, \mathrm{opt}}_{\pazocal{D}\left(\mathbb{C}^{d}\right)}$ such that $\|\bar{X}_1-\rho_1\|_1\le\varepsilon$ and $\|\bar{X}_2-\rho_2\|_1\le\varepsilon$. Given a matrix $A$, let us denote as $A|_n$ its $n\times n$ top-left block. Specifically, we denote as $\rho_1|_{d-1}$ and as $\rho_2|_{d-1}$ the $(d-1)\times (d-1)$ top-left blocks of $\rho_1$ and $\rho_2$, respectively. We have that
\bb
    \left\|X-\left(\rho_1|_{d-1}-\rho_2|_{d-1}\right)\right\|_1&=\left\|X_1-X_2-\left(\rho_1|_{d-1}-\rho_2|_{d-1}\right)\right\|_1\\
    &\le \|X_1-\rho_1|_{d-1}\|_1+\|X_2-\rho_2|_{d-1}\|_1\\
    &\leqt{(i)} \|\bar{X}_1-\rho_1 \|_1+\|\bar{X}_2-\rho_2\|_1\\
    &\le 2\varepsilon\,.
\ee
Here, in (i), we exploited the general fact that the trace norm of any matrix $A$ is an upper bound on the trace norm of its $n\times n$ top-left block $A|_n$, i.e., $\|A|_n\|_1\le \|A\|_1$. This latter fact can be easily proved by exploiting the \emph{minimax principle for singular values}~\cite[Problem III.6.1]{BHATIA-MATRIX}. Therefore, we have proved that the set 
\bb
    S\coloneqq \{\rho_1|_{d-1}-\rho_2|_{d-1}\,:\,\rho_1,\rho_2\in C^{\varepsilon, \mathrm{opt}}_{\pazocal{D}\left(\mathbb{C}^{d}\right)}\}
\ee
constitutes a $2\varepsilon$-covering net of $\mathrm{Herm}_{\|\ast\|_1\le 1}(d-1)$, and thus $|S|\ge \pazocal{C}(  \mathrm{Herm}_{\|\ast\|_1\le 1}(d-1) , \|\cdot\|_1, 2\varepsilon)$. Therefore, we conclude that
\bb
    \pazocal{C}(  \mathrm{Herm}_{\|\ast\|_1\le 1}(d-1) , \|\cdot\|_1, 2\varepsilon)\le |S|\le |C^{\varepsilon, \mathrm{opt}}_{\pazocal{D}\left(\mathbb{C}^{d}\right)}|^2=\pazocal{C}(  \pazocal{D}\!\left(\mathbb{C}^{d}\right) , \|\cdot\|_1, \varepsilon)^2\,.
\ee
\end{proof}

We now present a lemma which gives a lower bound on the covering number of the set of Hermitian matrices with trace norm smaller than one.
\begin{lemma}[(Lower bound on the covering number)]\label{lem:Hermitianmatreiceslowerbound}
Let $\varepsilon>0$ and $d\in \N_+$. Let $\mathrm{Herm}_{\|\ast\|_1\le 1}(d)$ be the set of Hermitian matrices with trace norm smaller than one, i.e.
\bb
     \mathrm{Herm}_{\|\ast\|_1\le 1}(d)&\coloneqq \left\{X\in\mathbb{C}^{d,d}\,:\,X^{\dagger}=X,\,\|X\|_1\le1\right\}\,.
\ee
Then, its covering number can be lower bounded by
\bb
    \pazocal{C}(  \mathrm{Herm}_{\|\ast\|_1\le 1}(d) , \|\cdot\|_1, \varepsilon)\ge \varepsilon^{-d^2}\,.
\ee
\end{lemma}
\begin{proof}
    Let us define the ball of $d\times d$ Hermitian matrices with radius $r$ with respect the trace norm as
    \bb
        B_r(d)\coloneqq \left\{X\in\mathbb{C}^{d,d}\,:\,X^{\dagger}=X,\,\|X\|_1\le r\right\}\,,
    \ee
    Note that the ball of radius one $B_1(d)$ is exactly equal to the set $\mathrm{Herm}_{\|\ast\|_1\le 1}(d)$. The volume of the ball with radius $r$ can be expressed in terms of the volume of the ball with unit radius as
    \bb\label{Vol_ball}
        \text{Vol}\!\left( B_r(d) \right)=r^{d^2} \text{Vol}\!\left( B_1(d) \right)\,.
    \ee
    Let us prove this latter fact. First, note that there are $d^2$ independent real parameters $y_1,y_2,\ldots,y_{d^2}$ that define any $d\times d$ Hermitian matrix. Specifically, there are $d^2-d$ independent real parameters for the off diagonal elements, and $d$ independent real parameters for the diagonal ones. Given the $d^2$-dimensional vector $y=(y_1,y_2,\ldots,y_{d^2})\in\R^{d^2}$ of such parameters, let $X(y)$ be the associated Hermitian matrix. 
    Consequently, we have that
\bb
    \text{Vol}\!\left( B_r(d) \right)=\,
    \int\limits_{\substack{y\in\mathbb{R}^{d^2}\\ \|X(y)\|_1\le r}}\!\!\!1 \,\mathrm{d}^{d^2}\!y \eqt{(i)}r^{d^2}\!\!\!\int\limits_{\substack{y\in\mathbb{R}^{d^2}\\ \|X(y)\|_1\le 1}}\!\!\!1 \,\mathrm{d}^{d^2}\!y =\,
    r^{d^2}  \text{Vol}\!\left( B_1(d) \right)\,,
\ee
where in (i) we made the change of variable $y\mapsto r y$ and we exploited that $X(ry)=rX(y)$. Hence, we have proved~\eqref{Vol_ball}. Finally, a simple counting argument shows that
\bb
    \pazocal{C}(  \mathrm{Herm}_{\|\ast\|_1\le 1}(d) , \|\cdot\|_1, \varepsilon)\ge \frac{\text{Vol}\!\left( B_1(d)\right)}{\text{Vol}\!\left( B_\varepsilon(d)\right)}\,,
\ee
which, together with~\eqref{Vol_ball}, concludes the proof.
\end{proof} 
The following corollary establishes a lower bound on the covering number of the set of (mixed) quantum states.
\begin{cor}\label{cor:densitymatrices}
Let $\varepsilon > 0$ and $d\in \N_+$. Let $\pazocal{D}\!\left(\mathbb{C}^{d}\right)$ be the set of quantum states on $\mathbb{C}^{d}$, i.e.
\bb
     \pazocal{D}\!\left(\mathbb{C}^{d}\right) &\coloneqq \left\{\rho\in\mathbb{C}^{d,d}\,:\,\rho\ge0,\,\Tr\rho=1\right\}\,.
\ee
Then, its packing number can be lower bounded as follows
\bb
\pazocal{P}(  \pazocal{D}\!\left(\mathbb{C}^{d}\right) , \|\cdot\|_1, \varepsilon)\ge (2\varepsilon)^{-\frac{1}{2}(d-1)^2}
\ee
\begin{proof}
The result follows from the following chain of inequalities:
\bb
    \pazocal{P}(  \pazocal{D}\!\left(\mathbb{C}^{d}\right) , \|\cdot\|_1, \varepsilon)&\geqt{(i)}\pazocal{C}(  \pazocal{D}\!\left(\mathbb{C}^{d}\right) , \|\cdot\|_1, \varepsilon)\\
    &\geqt{(ii)} \sqrt{\pazocal{C}(  \mathrm{Herm}_{\|\ast\|_1\le 1}(d-1) , \|\cdot\|_1, 2\varepsilon)}\\
    &\geqt{(iii)} (2\varepsilon)^{-\frac{1}{2}(d-1)^{2}}\,.
\ee
Here, (i), (ii), and (iii) follow from Lemma~\ref{lemma1111},  Lemma~\ref{lemma_cov_pos_norm}, and Lemma~\ref{lem:Hermitianmatreiceslowerbound}, respectively.
\end{proof}
\end{cor}

\newpage

\section{Tomography of moment-constrained states}\label{Sec_EC_subset}
In this section we consider the problem of quantum state tomography of continuous-variable systems. Since continuous-variable states are associated with an infinite dimensional Hilbert space, a first trivial observation is that quantum state tomography is impossible if one has no any extra prior information about the unknown state. However, in a practical scenario, experimentalists often possess knowledge about the energy budget available of their light sources --- and hence about the mean energy of the unknown quantum state.
Thus, it is crucial to analyse the problem of tomography of $n$-mode states subject to an energy constraint. Such energy constraint can be expressed as the following upper bound on the mean photon number of the unknown $n$-mode state $\rho$:
\bb\label{def_energy_constraint}
    \Tr[\hat{N}_n \rho]\le nN_{\text{phot}}\,,
\ee
where $\hat{N}_n$ is the $n$-mode photon number operator, and  $N_{\text{phot}}$ is a real number. In other words, one may have the extra prior information that the unknown quantum state is contained in the set 
\bb
    \left\{\rho\in \pazocal{D}\!\left(L^2(\mathbb R^n)\right):\quad\Tr\!\left[\hat{N}_n\rho\right]\le nN_{\text{phot}}\right\}
\ee
of energy constrained states.
Here, the number of modes $n$ and the energy constraint $N_{\text{phot}}$ are \emph{known} parameters, in the sense that the tomography algorithm may explicitly depend on their values. Note that we normalise the right-hand-side of~\eqref{def_energy_constraint} with the factor $n$ because the mean photon number $\Tr[\hat{N}_n \rho]$ is extensive, meaning that $\Tr[\hat{N}_n \sigma^{\otimes n}]=n\Tr[\hat{N}_1 \sigma]$ for all  single-mode states $\sigma$. In this section, we first consider the following question:
\begin{quote} 
 Suppose there exists a tomography algorithm designed to learn, with $\varepsilon$-precision in trace distance and high probability, an unknown $n$-mode pure state $\psi$ satisfying the energy constraint $\Tr[\hat{N}_n \psi] \leq nN_{\text{phot}}$. What is the required sample complexity, represented by the number $N$ of state copies?
\end{quote}
In this section, we answer this question by showing that \emph{any such tomography algorithm needs at least $N=\Omega\!\left(\frac{N_{\text{phot}}}{\varepsilon^2}\right)^{\!n}$ state copies}. This result imposes strong limitations on quantum state tomography of continuous-variable systems, due to the unfavourable scaling not only in the number of modes $n$ but also in the trace-distance error $\varepsilon$. In contrast, in the finite-dimensional setting, quantum state tomography of an unknown $n$-qubit state needs at least $N=\tildeTheta\,(2^{n}/\varepsilon^2)$ copies of the state (indicating a more favorable dependence on $\varepsilon$). Conversely, we also show that there exists a tomography algorithm that achieves the same sample complexity performances. Thus, we have the following main result:

\begin{thm}[(Tomography of energy-constrained pure states - informal version)]
    $N=\Theta\!\left(\frac{N_{\text{phot}}}{\varepsilon^{2}}\right)^{\!n}$ state copies are necessary and sufficient for quantum state tomography of energy-constrained $n$-mode pure states with trace-distance error $\varepsilon$.
\end{thm}
In the general case of $n$-mode energy-constained (possibly mixed) states we show the following:
\begin{thm}[(Tomography of energy-constrained mixed states - informal version)]
    $N=\Omega\! \left(\frac{N_{\text{phot}}}{\varepsilon}\right)^{\!2n}$ state copies are necessary for quantum state tomography of energy-constrained $n$-mode (possibly mixed) states with trace-distance error $\varepsilon$. Conversely, $N=O\!\left(\frac{N_{\text{phot}}}{\varepsilon^{3/2}}\right)^{\!2n}$ state copies are sufficient for this task.
\end{thm}
In practical scenarios, experimentalists often have knowledge not only about the mean energy but also about higher moments (e.g.,~energy variance) of their light sources. Therefore, it is meaningful to analyse tomography of states satisfying constraints on higher moments of energy. Mathematically, the $k$-th moment constraint can be expressed as
\bb\label{def_moment_constraint}
    \left(\Tr[\hat{N}_n^k \rho]\right)^{1/k}\le nN_{\text{phot}}\,.
\ee
Note that we normalise the right-hand-side of~\eqref{def_moment_constraint} with the factor $n$ because the $k$-th moment $\left(\Tr[\hat{N}_n^k \rho]\right)^{1/k}$ is extensive when evaluated on Fock states $\rho=\ketbra{m}^{\otimes n}$, meaning that 
\begin{equation}
\left(\Tr[\hat{N}_n^k \ketbra{m}^{\otimes n}]\right)^{1/k}=n\left(\Tr[\hat{N}_1^k \ketbra{m}]\right)^{1/k}. 
\end{equation}
In summary, it is sometimes meaningful to possess the extra prior information that the unknown quantum state is contained in the following set of $k$-th moment-constrained states:
\bb
    \mathcal{S}(n,N_{\text{phot}},k)\coloneqq \left\{\rho\in \pazocal{D}\!\left(L^2(\mathbb R^n)\right):\quad\Tr\!\left(\left[\hat{N}_n^k\rho\right]\right)^{1/k}\le nN_{\text{phot}}\right\}\,,
\ee
where we stress that $n$, $N_{\text{phot}}$, and $k$ are \emph{known} parameters. Note that the higher the value of $k$, the more prior extra information we have about the unknown quantum state. Indeed, this is established by the following simple result.

\begin{lemma}[(High moment-constraint implies low moment-constraint)]
For any $k, q \in \mathbb{N}_+$ with $k \geq q$, it holds that
\begin{equation}
\mathcal{S}(n, N_{\text{phot}}, k) \subseteq \mathcal{S}(n, N_{\text{phot}}, q).
\end{equation}
\end{lemma}
\begin{proof}
By Eq.~\ref{eq:conc} we have that, for any concave function $f$, any Hermitian matrix $X$, and any quantum state $\rho$, it holds that
\begin{align}
\Tr(\rho f(X)) \leq f(\Tr(\rho X)).
\end{align}
Now, the claim follows from the fact that
\begin{align}
\left(\Tr\!\left[\rho\,\hat{N}_n^q\right]\right)^{1/q} = \left(\Tr\!\left[\rho\left(\hat{N}_n^k\right)^{q/k}\right]\right)^{1/q} \leq \left(\Tr\!\left[\rho\,\hat{N}_n^k\right]\right)^{1/k},
\end{align}
where in the first step we use the semi-definite positivity of $\hat{N}_n$, and the inequality follows from the concavity of the function $x \mapsto x^{q/k}$ for $x \geq 0$.
\end{proof}
In this section, we also analyse tomography of moment-constrained states and establish the following results.
\begin{thm}[(Tomography of moment-constrained states - informal version)]\label{sample_complexity_our_tom}
Let $\rho$ be an unknown $n$-mode state satisfying the $k$-th moment constraint $(\Tr[\hat{N}_n^k\rho])^{1/k} \leq n N_{\text{phot}}$. The following facts hold:
\begin{enumerate}[a.]
    \item The number of copies of $\rho$ required to achieve quantum state tomography with precision $\varepsilon$ in trace distance scales at least as $\Omega\!\left(\frac{N_{\text{phot}}}{\varepsilon^{1/k}}\right)^{2n}$.
    \item There exists a tomography algorithm with sample complexity scaling as $O\!\left(\frac{N_{\text{phot}}}{\varepsilon^{3/(2k)}}\right)^{2n}$.
    \item If we assume the unknown state $\rho$ to be pure, then $\Theta\!\left(\frac{N_{\text{phot}}}{\varepsilon^{2/k}}\right)^{n}$ copies of $\rho$ are necessary and sufficient for tomography.
\end{enumerate}
\end{thm}
We note that the upper bound $O\!\left(\frac{N_{\text{phot}}}{\varepsilon^{3/(2k)}}\right)^{2n}$ on the sample complexity of tomography of moment-constrained (mixed) states provided in Theorem~\ref{sample_complexity_our_tom} outperforms the one attainable through the CV classical shadow algorithm outlined in References~\cite{becker_classical_2023,gandhari_precision_2023}.

This section is organised as follows:
\begin{itemize}
    \item In Subsection~\ref{section_lower_bounds_cv}, we derive lower bounds on the sample complexity of tomography of moment-constrained states.
    \item In Subsection~\ref{sec_upper_bound_sample}, we show upper bounds on the sample complexity by providing explicit tomography algorithms to learn moment-constrained states. Throughout this section, we also establish key properties of moment-constrained states, demonstrating that they can be effectively approximated by states with finite local dimension and rank.    
\end{itemize}

\emph{Note}: Throughout this section, we opted to simplify the notation in the use of parentheses in asymptotic expressions. For example, when we write $\Theta\!\left(\frac{N_{\text{phot}}}{\varepsilon^{2/k}}\right)^{\!n}$ we formally mean $\left[\Theta\!\left(\frac{N_{\text{phot}}}{\varepsilon^{2/k}}\right)\right]^{\!n}$. We stress that this notation omits only factors that are independent of the parameters of the problem, i.e., $n$, $\varepsilon$, and $N_{\text{phot}}$. For instance, with this notation, it is true that $\Theta\!\left(\frac{2N_{\text{phot}}}{\varepsilon^{2/k}}\right)^{\!n}=\Theta\!\left(\frac{N_{\text{phot}}}{\varepsilon^{2/k}}\right)^{\!n}$, but $\Theta\!\left(\frac{N_{\text{phot}}^2}{\varepsilon^{2/k}}\right)^{\!n}\neq \Theta\!\left(\frac{N_{\text{phot}}}{\varepsilon^{2/k}}\right)^{\!n}$. Importantly, to ensure clarity and avoid ambiguity, we express every sample complexity bound using both asymptotic notation and explicit exact expressions.

\subsection{Sample complexity lower bounds}\label{section_lower_bounds_cv}
To establish lower bounds on the sample complexity of tomography in continuous variable systems, we rely on the concept of $\varepsilon$-nets (as discussed in Section~\ref{sec:epsilonnet}). In the subsequent subsubsection~\ref{subsub:lbpure}, we delve into the analysis of sample complexity lower bounds for pure state tomography. Following that, we proceed to demonstrate lower bounds for mixed state tomography in subsubsection~\ref{subsub:lbmixed}.

\subsubsection{Lower bound for pure state tomography}
\label{subsub:lbpure}
We now present and prove our main theorem concerning the number of copies necessary for learning moment-constrained pure states. The proof strategy hinges on the utilisation of Lemma~\ref{lower_bound_sample_complexity}, which relies on Fano's inequality and Holevo's bound, and on identifying a suitable `packing' for the set of moment-constrained pure states.

\begin{thm}[(Lower bound on sample complexity for moment-constrained pure states)]\label{th:lowerboundtomohraphy} 
Let us consider a tomography algorithm that learns, within a trace distance $\le \varepsilon$ and failure probability $\le \delta$, an arbitrary state belonging to the set 
\bb
    \pazocal{S}_{\text{pure}}(n,k,N_{\text{phot}})&\coloneqq\left\{\ket{\psi}:\,n\text{-mode pure state},\,\left(\Tr\!\left[\left(\hat{N}_n\right)^{\!\!k}\ketbra{\psi}\right]\right)^{\!1/k}\le nN_{\text{phot}}\right\} 
\ee
of $n$-mode pure states with bounded $k$-moment. Then, such a tomography algorithm must use a number of state copies $N$ satisfying
\bb
    N&\ge \frac{1}{ng(N_{\text{phot}})}\left[2(1-\delta)\left(\frac{N_{\text{phot}}}{(12\varepsilon)^{2/k}}-\frac{1}{n}\right)^n-(1-\delta)\log_2(32\pi)-H_2(\delta)\right]\\
    &=\Theta\!\left(\frac{N_{\text{phot}}}{\varepsilon^{2/k}}\right)^n\,.
\ee
Here, $g(x)\coloneqq (x+1)\log_2(x+1) - x\log_2 x\,$ is the bosonic entropy, and $H_2(x)\coloneqq -x\log_2x-(1-x)\log_2(1-x)$ is the binary entropy.
\end{thm}
\begin{proof}
Let $M\in\N$ be such that there exist $M$ states 
\bb
    \{\ket{\psi_1}, \ket{\psi_2},\ldots, \ket{\psi_M}\}\subseteq \pazocal{S}_{\text{pure}}(n,k,N_{\text{phot}})
\ee
such that for every $i \neq j \in [M]$ it holds that
\bb
    \frac{1}{2}\|\psi_i-\psi_j\|_1 > 2\varepsilon\,,
\ee
where $\psi_i\coloneqq \ketbra{\psi_i}$ and $\psi_j\coloneqq \ketbra{\psi_j}$ (as depicted in Fig.~\ref{figure_packing}).
\begin{figure*}[t]\label{figure_packing}
    \centering
  \includegraphics[width=0.28\textwidth]{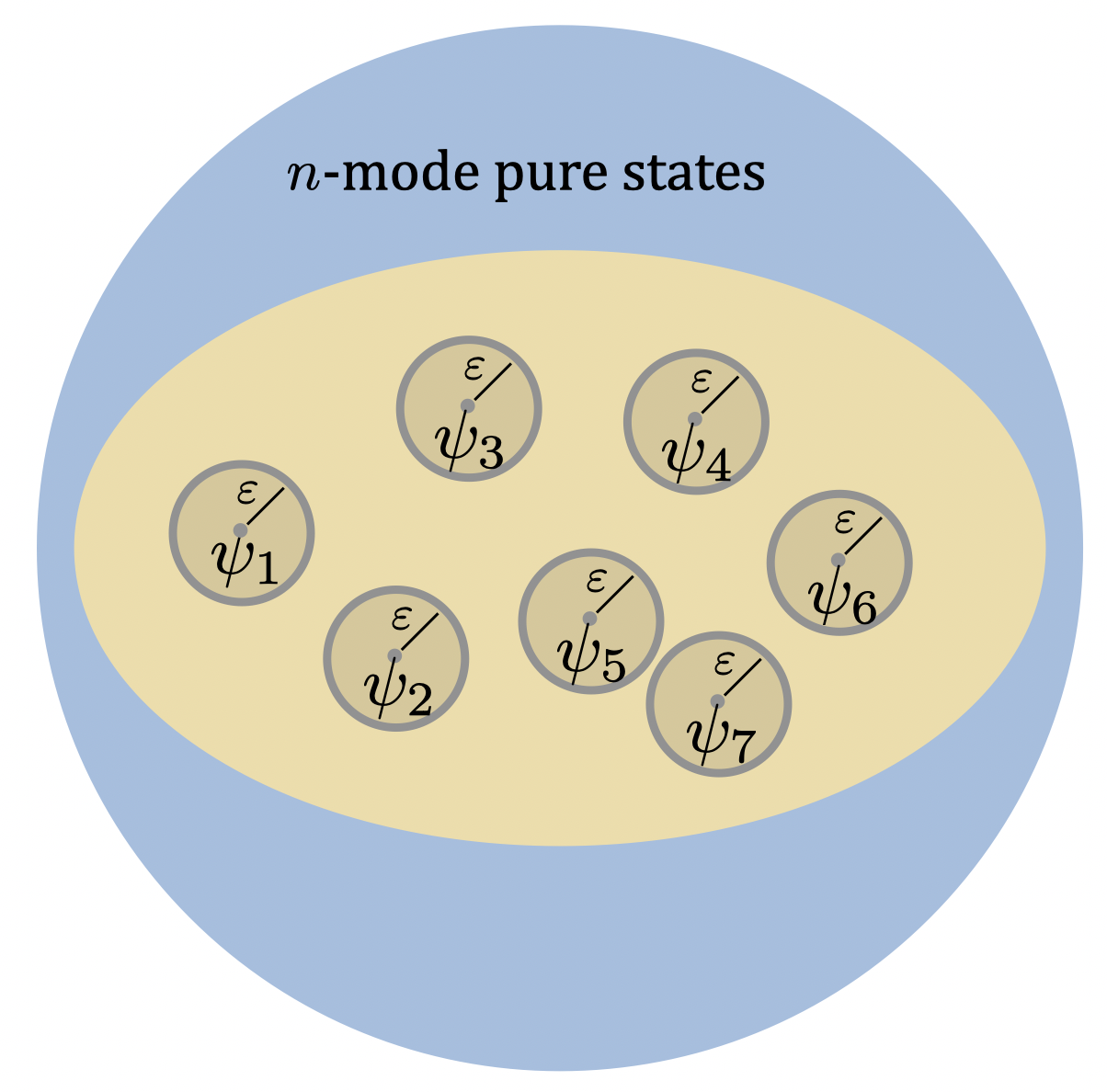}
    \caption{Pictorial representation of the set of all $n$-mode pure states (blue) and the set of of all $n$-mode pure states with bounded $k$-moment (yellow). We consider $M$ moment-constrained pure states $\{\psi_1,\psi_2,\cdots,\psi_M\}$ such that they are at least $2\varepsilon$-far from each other with respect to the trace distance.}
\end{figure*}
Thanks to Lemma~\ref{lower_bound_sample_complexity}, any tomography algorithm that learns states from the set $\pazocal{S}_{\text{pure}}(n,k,N_{\text{phot}})$ with accuracy $\le \varepsilon$ in trace distance and failure probability $\le \delta$ must use a number of state copies $N$ satisfying:
\bb\label{eq_chi1}
    \chi \ge (1-\delta)\log_2(M) - H_2(\delta)\,,
\ee
where $\chi \coloneqq S\left(\frac{1}{M}\sum^{M}_{j=1} \psi_j^{\otimes N}\right) - \frac{1}{M}\sum^{M}_{j=1} S(\psi_j^{\otimes N})$ is the Holevo information, and $S(\cdot)$ is the von Neumann entropy. Note that the latter quantity can be upper bounded as 
\bb\label{eq_chi2}
    \chi &= S\left(\frac{1}{M}\sum^{M}_{j=1} \psi_j^{\otimes N}\right) - \frac{1}{M}\sum^{M}_{j=1} S(\psi_j^{\otimes N}) \le S\left(\frac{1}{M}\sum^{M}_{j=1} \psi_j^{\otimes N}\right) \le nN g(N_{\text{phot}})\,,
\ee
where in the last inequality we exploited the fact that the von Neumann entropy of a state with a fixed mean photon number can be upper bounded in terms of the bosonic entropy function $g(\cdot)$ (see Lemma~\ref{maxthermstate}). Specifically, note that the mean total photon number of the $nN$-mode state $\frac{1}{M}\sum^{M}_{j=1} \psi_j^{\otimes N}$ can be upper bounded as 
\bb
     \Tr\left[\hat{N}_{nN}\left(\frac{1}{M}\sum^{M}_{j=1} \psi_j^{\otimes N}\right)\right]=N\frac{1}{M} \sum_{j=1}^M\Tr[\hat{N}_n\psi_j]\le N\frac{1}{M} \sum_{j=1}^M\left(\Tr[\hat{N}_n^k\psi_j]\right)^{1/k}\le Nn N_{\text{phot}}\,,
\ee
where in the first inequality we have used  that $\Tr[\hat{N}_n\psi_j]=\Tr[(\hat{N}_n^{k})^{\frac{1}{k}}\psi_j]$ and exploited the concavity of $x\mapsto x^{1/k}$ for $x>0,k\in\N$ (see Eq.~\eqref{eq:conc}), while in the last inequality we exploited that $\psi_j\in \pazocal{S}_{\text{pure}}(n,k,N_{\text{phot}})$. Hence, Lemma~\ref{maxthermstate} establishes that
\bb
    S\left(\frac{1}{M}\sum^{M}_{j=1} \psi^{\otimes N}_j\right) \le nN g(N_{\text{phot}})\,.
\ee
By putting together~\eqref{eq_chi1} and~\eqref{eq_chi2}, we deduce that, if such $M$ exists, then the number of state copies $N$ required for tomography of an unknown state belonging to $\pazocal{S}_{\text{pure}}(n,k,N_{\text{phot}})$ has to be at least
\bb\label{lower_bound_N_proof}
    N\ge \frac{(1-\delta)\log_2(M) - H_2(\delta)}{ng(N_{\text{phot}})}\,.
\ee
Hence, in order to get a good lower bound on $N$, we need to find a large value of $M$ such that there exist $M$ states 
\bb
    \{\psi_1, \psi_2,\ldots, \psi_M\}\subseteq \pazocal{S}_{\text{pure}}(n,k,N_{\text{phot}})
\ee
for which it holds that $\frac{1}{2}\|\psi_i-\psi_j\|_1 > 2\varepsilon$ for all $i \neq j \in [M]$. To this end, let us make a suitable construction of such $M$ states.

Let $m\in\N$ such that $m\geq nN_{\text{phot}}$ which will be fixed later. Let us define the finite-dimensional Hilbert space $\HH_m$ spanned by all the $n$-mode Fock states, apart from the $n$-mode vacuum state, with total number of photons less than $m$:
\bb
    \HH_m &\coloneqq \left\{ \ket{\phi}\in\mathrm{span}\!\left(\ket{\textbf{k}}:\,\,\textbf{k}\in\N^n\setminus\{0\}\,\,,\sum_{i=1}^n k_i\leq m\right):\, \braket{\phi|\phi}=1\right\}\,,\label{hmdef}
\ee
where $\ket{\textbf{k}}$ is the $n$-mode Fock state defined by $\ket{\textbf{k}}=\ket{k_1}\otimes\ket{k_2}\otimes\ldots\otimes\ket{k_n}$.

Given a state vector $\ket{\phi}\in \HH_m$, we can define a corresponding state vector $\ket{\psi_{\phi}}\in \pazocal{S}_{\text{pure}}(n,k,N_{\text{phot}})$ as 
\bb 
    \ket{\psi_{\phi}}\coloneqq \sqrt{1-\left(\frac{nN_{\text{phot}}}{m}\right)^{k}}\ket{0}^{\otimes n}+ \sqrt{\left(\frac{nN_{\text{phot}}}{m}\right)^{k}}\ket{\phi}\,.
\ee
Indeed, the $k$-moment of $ \ket{\psi_{\phi}} $ can be upper bounded as
\begin{align}
\Tr[\hat{N}_n^k\ketbra{\psi_\phi}]=\frac{\left(nN_{\text{phot}}\right)^{k}}{m^k}\Tr[\hat{N}_n^k\ketbra{\phi}]\leq\frac{\left(nN_{\text{phot}}\right)^{k}}{m^k}m^k= \left(nN_{\text{phot}}\right)^{k} \,,
\end{align}
and thus $\ket{\psi_\phi}\in\pazocal{S}_{\text{pure}}(n,k,N_{\text{phot}})$. 
Moreover, as we show below in~\eqref{proofofeqtoprove}, the following inequality holds 
\begin{align}
    \frac{1}{2}\left\|\ketbra{\psi_\phi}-\ketbra{\psi_{\tilde{\phi}}}\right\|_1\ge \sqrt{\frac{(nN_{\text{phot}})^k}{8m^k}}  \left\|\ketbra{\phi}-\ketbra{\tilde\phi}\right\|_1\label{eq.toprove}
\end{align}
for all $\ket{\phi},\ket{\tilde{\phi}}\in\HH_m$. Hence, by choosing 
\bb\label{choice_m}
    m\coloneqq \left\lfloor \frac{nN_{\text{phot}}}{(12\varepsilon)^{2/k}}\right\rfloor\,,
\ee
one can check that for all $\ket{\phi},\ket{\tilde{\phi}}\in\HH_m$ such that
\bb
\left\|\ketbra{\phi}-\ketbra{\tilde{\phi}}\right\|_1\geq\frac{1}{2}\,, 
\ee
the inequality in~\eqref{eq.toprove} guarantees that
\bb
    \frac{1}{2}\left\|\ketbra{\psi_\phi}-\ketbra{\psi_{\tilde{\phi}}}\right\|_1\geq 2\varepsilon\,.
\ee
Therefore, let us find $M\in\N$ such that there exist $M$ states $\{\ket{\phi_1},\ket{\phi_2},\ldots,\ket{\phi_M}\}\subseteq \HH_m$ such that $\left\|\ketbra{\phi_i}-\ketbra{\phi_j}\right\|_1\geq \frac{1}{2}$
for all $i\ne j\in[M]$. By definition of packing number (see Definition~\ref{def_pack}), such $M$ can be chosen to be the $\frac{1}{2}$-packing number of $\HH_m$ with respect to the trace norm, i.e., $M\coloneqq \pazocal{P}\!\left(\HH_m, \|\cdot\|_1, \frac{1}{2}\right)\,.$
We can find a simple lower bound on $M$ as 
\bb
    M&= \pazocal{P}\!\left(\HH_m, \|\cdot\|_1, \frac{1}{2}\right)
     \geqt{(i)} \frac{1}{8\pi}2^{2(\dim\HH_{m}-1)}
     \geqt{(ii)} \frac{1}{8\pi}2^{2[\binom{n+m}{n}-2]}
     \geqt{(iii)} \frac{1}{8\pi}2^{2[(\frac{m}{n})^n-1]} 
     = \frac{1}{32\pi}2^{2(\frac{m}{n})^n}
     \,.  
\ee
Here, in (i), we exploit Lemma~\ref{lemma6666}; in (ii), we observed that the dimension of $\HH_m$ is given by 
\bb\label{dim_H_m}
    \dim\HH_m= \left|\left\{\textbf{k}\in\N^n\setminus\{0\}\,,\quad\sum_{i=1}^n k_i\leq m\right\}\right| =\binom{m+n}{n}-1\,,
\ee
where the last equality follows from a standard combinatorial argument; in (iii), we just exploited that $\binom{m+n}{n}\ge (\frac{m+n}{n})^n\ge (\frac{m}{n})^n+1$. Consequently,~\eqref{lower_bound_N_proof} implies that
\bb
    N&\ge \frac{(1-\delta)2(\frac{m}{n})^n-(1-\delta)\log_2(32\pi)-H_2(\delta)}{ng(N_{\text{phot}})}\,.
\ee
Finally, by exploiting the definition of $m$ in~\eqref{choice_m} we conclude that
\bb
    N&\ge \frac{2(1-\delta)\left(\frac{N_{\text{phot}}}{(12\varepsilon)^{2/k}}-\frac{1}{n}\right)^n-(1-\delta)\log_2(32\pi)-H_2(\delta)}{ng(N_{\text{phot}})}\,.
\ee
We are just left to show~\eqref{eq.toprove}:
\newcommand{\coef}{\frac{\left(nN_{\text{phot}}\right)^{k}}{m^k}}
\bb\label{proofofeqtoprove}
\|\ketbra{\psi_\phi}-\ketbra{\psi_{\tilde{\phi}}} \|_1 &= 2\sqrt{1-|\braket{\psi_\phi|\psi_{\tilde{\phi}}}|^2}
\\&=2\sqrt{1-\left|\left(1-\coef\right)+\coef \braket{\phi|\tilde\phi}\right|^2}\\
&\geq 2\sqrt{1-\left(1-\coef+\coef \left|\braket{\phi|\tilde\phi}\right|\right)^2}\\
    &=2\sqrt{2\coef\left(1-\left|\braket{\phi|\tilde\phi}\right|\right)-\left(\coef\left(1-\left|\braket{\phi|\tilde\phi}\right|\right)\right)^2}\\
    &\geq 2\sqrt{\coef\left(1-\left|\braket{\phi|\tilde\phi}\right|\right)}\\
    &\geq \sqrt{2}\sqrt{\coef}\sqrt{\left(1-\left|\braket{\phi|\tilde\phi}\right|\right)\left(1+\left|\braket{\phi|\tilde\phi}\right|\right)}\\
    &=\sqrt{2}\sqrt{\coef}\sqrt{\left(1-\left|\braket{\phi|\tilde\phi}\right|^2\right)}\\
    &=\sqrt{\frac{(nN_{\text{phot}})^k}{2m^k}}  \left\|\ketbra{\phi}-\ketbra{\tilde\phi}\right\|_1\,,
\ee
which concludes the proof.
\end{proof}
In particular, in the single-mode case, i.e., $n=1$, the number of state copies $N$ has to satisfy
\bb
   N&\ge\frac{1}{g(N_{\text{phot}})}\left[2(1-\delta)\left(\frac{N_{\text{phot}}}{(12\varepsilon)^{2/k}}-1\right)-(1-\delta)\log_2(32\pi)-H_2(\delta)\right]\\
     &=\Theta\!\left( \frac{N_{\text{phot}}}{\varepsilon^{2/k}}\right)\,.
\ee

\subsubsection{Lower bound for mixed state tomography}
In this subsection, we establish a theorem concerning the number of copies necessary for learning moment-constrained mixed states. The proof strategy follows a similar logic to the one used previously for the pure state case.
\label{subsub:lbmixed}
\begin{thm}[(Lower bound on sample complexity for moment-constrained mixed states)] 
Let us consider a tomography algorithm that learns, within a trace distance $\le \varepsilon$ and failure probability $\le \delta$, an arbitrary state belonging to the set 
\bb
    \pazocal{S}(n,k,N_{\text{phot}})&\coloneqq\left\{\rho\in\pazocal{D}(L^2(\mathbb{R}^n)),\quad\left(\Tr\!\left[\left(\hat{N}_n\right)^{\!\!k}\rho\right]\right)^{\!1/k}\le nN_{\text{phot}}\right\} 
\ee
of $n$-mode (possibly mixed) states with bounded $k$-moment. Then, such a tomography algorithm must use a number of state copies $N$ satisfying
\bb
    N&\ge \frac{1}{2ng(N_{\text{phot}})}\left[(1-\delta)\left(\frac{N_{\text{phot}}}{(16\varepsilon)^{1/k}}-\frac{1}{n}\right)^{2n} - \frac{1}{2}(1-\delta) -2H_2(\delta)\right] \\
    &=\Theta\!\left(\frac{N_{\text{phot}}}{\varepsilon^{1/k}}\right)^{2n}\,.
\ee
Here, $g(x)\coloneqq (x+1)\log_2(x+1) - x\log_2 x\,$ is the bosonic entropy, and $H_2(x)\coloneqq -x\log_2x-(1-x)\log_2(1-x)$ is the binary entropy. %
\end{thm}
\begin{proof}
Let $M\in\N$ be such that there exist $M$ states 
\bb
    \{\rho_1, \rho_2,\ldots,\rho_M\}\subseteq \pazocal{S}(n,k,N_{\text{phot}})
\ee
such that for every $i \neq j \in [M]$ it holds that
\bb
    \frac{1}{2}\|\rho_i-\rho_j\|_1 > 2\varepsilon\,.
\ee
By applying Lemma~\ref{lower_bound_sample_complexity} in the exact same way as we did in the proof of Theorem~\ref{th:lowerboundtomohraphy}, we deduce that the number of state copies $N$ required for tomography of an unknown state belonging to $\pazocal{S}(n,k,N_{\text{phot}})$ has to be at least
\bb\label{lower_bound_N_proof_mixed}
    N\ge \frac{(1-\delta)\log_2(M) - H_2(\delta)}{ng(N_{\text{phot}})}\,.
\ee
Let $m\in\N$ such that $m\geq nN_{\text{phot}}$. Let $\HH_m$ be the Hilbert space, defined in~\eqref{hmdef}, spanned by all the $n$-mode Fock states, apart from the $n$-mode vacuum state, with total number of photons less than $m$. Moreover, let $\pazocal{D}(\HH_m)$ be the set of density operators on $\HH_m$.
Given a state $\rho\in \pazocal{D}(\HH_m)$, we can define a corresponding state $\sigma_\rho \in \pazocal{S}(n,k,N_{\text{phot}})$ as 
\bb 
    \sigma_\rho\coloneqq \left[1-\left(\frac{nN_{\text{phot}}}{m}\right)^{k}\right]\ketbra{0}^{\otimes n}+ \left(\frac{nN_{\text{phot}}}{m}\right)^{k}\rho\,.
\ee
Indeed, the $k$-moment of $ \ket{\psi_{\phi}} $ can be upper bounded as
\begin{align}
\Tr[\hat{N}_n^k\sigma_\rho]=\left(\frac{nN_{\text{phot}}}{m}\right)^{k}\Tr[\hat{N}_n^k\rho]\leq\frac{\left(nN_{\text{phot}}\right)^{k}}{m^k}m^k= \left(nN_{\text{phot}}\right)^{k} \,,
\end{align}
and thus $\sigma_\rho\in\pazocal{S}(n,k,N_{\text{phot}})$. 
Moreover, it holds that
\bb\label{eq.trace}
    \left\|\sigma_{\rho}-\sigma_{\rho'}\right\|_1=\left(\frac{nN_{\text{phot}}}{m}\right)^{k}\left\|\rho-\rho'\right\|_1
\ee
for all $\rho,\rho'\in\pazocal{D}(\HH_m)$. Hence, by choosing 
\bb\label{choice_m2}
    m\coloneqq \left\lfloor \frac{nN_{\text{phot}}}{(16\varepsilon)^{1/k}}\right\rfloor\,,
\ee
one can check that for all $\rho,\rho'\in\pazocal{D}(\HH_m)$ such that
\bb
\left\|\rho-\rho'\right\|_1\geq\frac{1}{4}\,, 
\ee
the relation in~\eqref{eq.trace} guarantees that
\bb
    \frac{1}{2}\left\|\sigma_\rho-\sigma_{\rho'}\right\|_1\geq 2\varepsilon\,.
\ee
Therefore, let us find $M\in\N$ such that there exist $M$ states $\{\rho_1,\rho_2,\ldots,\rho_M\}\subseteq \pazocal{D}(\HH_m)$ such that $\left\|\rho_i-\rho_j\right\|_1\geq \frac{1}{4}$
for all $i\ne j\in[M]$. By definition of packing number (see Definition~\ref{def_pack}), such $M$ can be chosen to be the $\frac{1}{4}$-packing number of $\pazocal{D}(\HH_m)$ with respect to the trace norm:
\bb
    M\coloneqq \pazocal{P}\!\left(\pazocal{D}(\HH_m), \|\cdot\|_1, \frac{1}{4}\right)  \,.
\ee
We can find a simple lower bound on $M$ as 
\bb
    M&= \pazocal{P}\!\left(\pazocal{D}(\HH_m), \|\cdot\|_1, \frac{1}{4}\right)\,\\
     &\geqt{(i)} 2^{\frac{(\dim\HH_{m}-1)^2}{2}}\\
     &\geqt{(ii)} 2^{\frac{\left(\binom{n+m}{n}-2\right)^2}{2}}\\
     &\geqt{(iii)} 2^{\frac{1}{2}((\frac{m}{n})^{2n}-1)} \\
     &=\frac{1}{\sqrt{2}}2^{\frac{1}{2}(\frac{m}{n})^{2n}}
     \,.  
\ee
Here, in (i), we exploit Corollary~\ref{cor:densitymatrices}; in (ii), we exploited~\eqref{dim_H_m}; in (iii), we have used  that $\binom{m+n}{n}\ge (\frac{m+n}{n})^n\ge (\frac{m}{n})^n+1$. Consequently,~\eqref{lower_bound_N_proof_mixed} implies that
\bb
    N&\ge \frac{\frac{1-\delta}{2}(\frac{m}{n})^{2n}- (1-\delta)\log_2(\sqrt{2}) -H_2(\delta)}{ng(N_{\text{phot}})}\,.
\ee
Finally, by exploiting the definition of $m$ in~\eqref{choice_m} we conclude that
\bb
    N&\ge \frac{1}{2ng(N_{\text{phot}})}\left[(1-\delta)\left(\frac{N_{\text{phot}}}{(16\varepsilon)^{1/k}}-\frac{1}{n}\right)^{2n} - \frac{1}{2}(1-\delta) -2H_2(\delta)\right]\,.
\ee

\end{proof}

In particular, in the single-mode case, i.e., $n=1$, the number of state copies $N$ has to satisfy
\bb
    N&\ge \frac{1}{2g(N_{\text{phot}})}\left[(1-\delta)\left(\frac{N_{\text{phot}}}{(16\varepsilon)^{1/k}}-1\right)^{2}- \frac{1}{2}(1-\delta) -2H_2(\delta)\right] \\
    &=\Theta\!\left( \frac{N_{\text{phot}}^2}{\varepsilon^{2/k}} \right)\,.
\ee

\subsection{Sample complexity upper bounds}\label{sec_upper_bound_sample}

In this subsection, we present a tomography algorithm aimed at learning a classical description of an unknown $n$-mode $k$-th moment-constrained state. We analyse both pure and mixed scenarios.
The tomography algorithm involves two main steps: first, projecting onto a finite-dimensional subspace, and second, applying a known tomography algorithm specifically designed for finite-dimensional systems. 

The subsection is organised as follows. In Subsubsection~\ref{subsub:properties}, we establish useful properties of moment-constrained states. Specifically, we demonstrate how to approximate a moment-constrained state using a finite-dimensional state with bounded effective rank. In Subsubsection~\ref{subsub:comparison}, we compare our approximation with those of other works, highlighting the improvements achieved. Subsubsection~\ref{subsub:performance} outlines the precise performances of an optimal tomography algorithm applicable to finite-dimensional systems. Finally, in Subsubsection~\ref{subsub:algorithm}, we present the tomography algorithms for moment-constrained pure and mixed states and establish their correctness. 
We denote the Euler's Constant as $e$ throughout this subsection.

\subsubsection{Properties of moment-constrained states: approximation with finite-dimensional and finite-rank states}
\label{subsub:properties}
The following lemma establishes that $k$-th moment-constrained states can be effectively approximated by finite-dimensional states. 
\begin{lemma}[(Effective dimension of moment-constrained states)]
\label{le:effectivedim}
Let $\varepsilon \in (0,1)$ and let $\rho$ be an $n$-mode state with $k$-moment bounded as $(\Tr[\hat{N}_n^k\rho])^{1/k} \leq nN_{\text{phot}}$.
Then, $\rho$ is $\varepsilon$-close in trace distance to a state supported on a space of dimension $\Theta\!\left(\frac{eN_{\text{phot}}}{\varepsilon^{2/k}}\right)^n$.

\noindent
More specifically, let $\HH_m$ be the Hilbert space spanned by all the $n$-mode Fock states with total photon number less than $m$:
\bb\label{def_H_m}
     \HH_m\coloneqq\operatorname{Span}\!\left\{\ket{k_1}\otimes\ket{k_2}\otimes\ldots\otimes\ket{k_n}:\,\,
     k_1,k_2,\ldots,k_n\in\N,\,\,\sum_{i=1}^n k_i\le m\right\}\,.
\ee
Let $m \coloneqq \ceil{\frac{n N_{\text{phot}}}{\varepsilon^{2/k}}}$. Define the state $\rho_{\deff}\coloneqq\Pi_m\rho\,\Pi_m/\Tr\!\left[ \Pi_m \rho\, \Pi_m\right]$ as the projection of $\rho$ onto $\HH_m$,
where $\Pi_m$ is the projector onto $\HH_m$, whose dimension satisfies
\begin{align}
\deff\coloneqq \dim\HH_m=\binom{m+n}{n}\le \left(\frac{eN_{\text{phot}}}{\varepsilon^{2/k}}+2e\right)^n= \Theta \!\left(\frac{eN_{\text{phot}}}{\varepsilon^{2/k}}\right)^{\!n}.
\end{align}
Then, $\rho_{\deff}$ is $\varepsilon$-close in trace distance  to $\rho$ in that
\bb
    \frac{1}{2}\left\|\rho-\rho_{\deff}\right\|_1\le\varepsilon\,.
\ee
\end{lemma}
\begin{proof}
First, observe that the dimension $\deff = \dim \HH_m$ is given by    \bb 
        \dim\HH_m&= \left|\left\{k_1,k_2,\ldots,k_n\in\N\,,\quad\sum_{i=1}^n k_i\leq m\right\}\right| \\
        &\eqt{(i)}\binom{m+n}{n} \\
        &\leqt{(ii)} \left(\frac{e(m+n)}{n}\right)^n\\
        &\leqt{(iii)} \left(\frac{eN_{\text{phot}}}{\varepsilon^{2/k}}+2e\right)^n\,.
    \ee
Here, in (i), we exploited a standard combinatorial argument; in (ii), we have used  the fact that $\binom{a}{b}\le (\frac{ea}{b})^b$ for all $a,b\in\N$; in (iii), we have used  the definition of $m \coloneqq \ceil{\frac{n N_{\text{phot}}}{\varepsilon^{2/k}}}$. Moreover, let us note that
\bb\label{ineq_involving_project}
        \Tr[(\mathbb{1}-\Pi_m)\rho]&\leqt{(iv)}\frac{1}{m^k}\Tr\!\left[\hat{N}_{n}^{k}\rho\right]\\
        &\leqt{(v)} \frac{\left(nN_{\text{phot}}\right)^{k}}{m^k}\\
        &\leqt{(vi)} \varepsilon^2\,,
\ee
where: (iv) follows by the operator inequality $ m^k(\mathbb{1}-\Pi_m)\le \hat{N}_{n}^{k}$; (v) follows by the hypothesis on the $k$-th moment of $\rho$; in (vi) we have used  again the definition of $m$. Finally, the trace distance between $\rho$ and the projected state $\rho_{\deff}$ satisfies
        \bb
            \frac{1}{2}\left\|\rho_{\deff}-\rho\right\|_1&=\frac{1}{2}\left\| \frac{\Pi_m\rho\,\Pi_m }{  \Tr\!\left[ \Pi_m \rho\, \Pi_m\right]   } -\rho \right\|_1\\
            &\leqt{(vii)} \sqrt{\Tr\!\left[(\mathbb{1}-\Pi_m)\rho \right]}\\
            &\le \varepsilon\,,
        \ee
where in (vii) we have used  the \emph{gentle measurement lemma} \cite[Lemma 6.15]{Sumeet_book}.
\end{proof}
We now mention a lemma which will be useful consequently.
\begin{lemma}[(Infinite-dimensional Schur-Horn theorem (\cite{kaftal2009infinite}, Proposition 6.4))]
\label{le:infSchur}
Let $\pazocal{H}$ be a Hilbert space. Let $N \in \mathbb{N}$ and $M \geq 0$ be a semi-definite operator on $\pazocal{H}$. For any set of orthonormal vectors $\{ |v_n\rangle \in \pazocal{H} : n \in \{0, 1, \ldots, N-1\} \}$, it holds that
\begin{equation}
\sum_{n=0}^{N-1} \lambda_n \geq \sum_{n=0}^{N-1} \langle v_n | M | v_n \rangle,
\end{equation}
where $\lambda_0 \geq \lambda_1 \geq \cdots \geq \lambda_{N-1}$ are the $N$ largest eigenvalues of $M$.
\end{lemma}

In the following lemma, we establish that any moment-constrained state is close in trace distance to a state with bounded rank.
\begin{lemma}[(Effective rank of moment-constrained states)]
\label{lemma_rank_ec_states}
Let $\varepsilon \in (0,1)$. Consider an $n$-mode state $\rho$ satisfying the $k$-th moment constraint $(\Tr[\rho \hat{N}_n^{k}])^{1/k} \leq n N_{\text{phot}}$. Then, $\rho$ is $\varepsilon$-close in trace distance to a state $\rho_{r_{\text{eff}}}$ with rank
$r_{\text{eff}} = \Theta\!\left(\frac{eN_{\text{phot}}}{\varepsilon^{1/k}}\right)^{\!n}$.

\noindent
More specifically, there exists a state $\rho_{r_{\text{eff}}}$ with rank $r_{\text{eff}}$, where
\begin{equation}
 r_{\text{eff}}\coloneqq \operatorname{rank}(\rho_{r_{\text{eff}}})\le\left(\frac{eN_{\text{phot}}}{\varepsilon^{1/k}}+2e\right)^n= \Theta\!\left(\frac{eN_{\text{phot}}}{\varepsilon^{1/k}}\right)^{\!n}\,,
\end{equation}
such that $\frac{1}{2} \|\rho - \rho_{r_{\text{eff}}}\|_1 \leq \varepsilon.$ 
\end{lemma}

\begin{proof}
    Let $\rho=\sum_{j=1}^\infty \lambda_j^{\downarrow} \psi_j$ be the the spectral decomposition of $\rho$, with $(\lambda^{\downarrow}_j)_{j\in\N_+}$ being its eigenvalues ordered in decreasing order with respect to $j$. For each $r_{\text{eff}}\in\N$, let us define the operator 
    \bb\label{def_theta_r}
        \Theta_{r_{\text{eff}}}\coloneqq \sum_{j=1}^{r_{\text{eff} }}\lambda_j^{\downarrow} \psi_j\,
    \ee
    which has rank equal to $r_{\text{eff}}$. Note that the \emph{infinite-dimensional Schur Horn theorem} (Lemma~\ref{le:infSchur}) implies that for any $r_{\text{eff}}$-rank projector $\Pi^{(r_{\text{eff}})}$  it holds that
    \bb
        \| \rho-\Theta_{r_\text{eff}}  \|_1=\sum_{j=r_{\text{eff}}+1}^{\infty}\lambda_j^{\downarrow}\le \Tr[(\mathbb{1}-\Pi^{(r_{\text{eff}})})\rho]\,.
    \ee
    Now, let us choose $\Pi^{(r_{\text{eff}})}$ to be the projector $\Pi_{m'}$ onto the subspace $\HH_{m'}$, defined in~\eqref{def_H_m}, spanned by the $n$-mode Fock states with total photon number at most $m'$, with $m'\coloneqq \ceil{\frac{nN_{\text{phot}}}{\varepsilon^{1/k}}}$. Hence, we have that
    \bb
        r_{\text{eff}}\coloneqq \Tr\Pi_{m'}=\binom{m'+n}{n}\le \left(\frac{e(m'+n)}{n}\right)^n\le\left(\frac{eN_{\text{phot}}}{\varepsilon^{1/k}}+2e\right)^n\,.
    \ee
    Thanks to~\eqref{ineq_involving_project}, we have that
    \bb
        \Tr[(\mathbb{1}-\Pi_{m'})\rho]\le \left(\frac{nN_{\text{phot}}}{m'}\right)^k\le \varepsilon\,,
    \ee
    and thus it holds that $\Tr[(\mathbb{1}-\Pi^{(r_{\text{eff}})})\rho]\le \varepsilon$. Consequently, we deduce that $\|\rho-\Theta_{r_\text{eff}}\|_1\le \varepsilon$. Therefore, by setting 
    \bb\label{def_rho_reff}
        \rho_{r_{\text{eff}}}\coloneqq \frac{\Theta_{r_\text{eff}}}{\Tr\Theta_{r_\text{eff}}}\,
    \ee
    we have that 
    \bb\label{chain_inequalities_theta}
        \|\rho-\rho_{r_{\text{eff}}}\|_1&=
        \left\|\rho-\frac{\Theta_{r_\text{eff}}}{\Tr\Theta_{r_\text{eff}}}\right\|_1\\
        &\le\|\rho-\Theta_{r_\text{eff}}\|_1+\left\|\Theta_{r_\text{eff}}-\frac{\Theta_{r_\text{eff}}}{\Tr\Theta_{r_\text{eff}}}\right\|_1\\
        &=\|\rho-\Theta_{r_\text{eff}}\|_1+\left|1-\frac{1}{\Tr\Theta_{r_\text{eff}}}\right|\Tr\Theta_{r_\text{eff}}\\
        &\le \varepsilon+\left(1- \Tr\Theta_{r_\text{eff}}\right)\\
        &\le \varepsilon+ \sum_{j=r_{\text{eff}}+1}^{\infty}\lambda_j^{\downarrow}\\
        &\le 2\varepsilon\,.
    \ee 
\end{proof}

We have demonstrated that an $n$-mode $k$-th moment-constrained state is $\varepsilon$-close in trace distance to a qudit state residing in a known Hilbert space of dimension $\deff \leq \left(\frac{eN_{\text{phot}}}{\varepsilon^{2/k}} + 2e\right)^n$. Additionally, we have shown that this moment-constrained state is also $\varepsilon$-close in trace distance to a state with rank $r_{\text{eff}} \leq \left(\frac{eN_{\text{phot}}}{\varepsilon^{1/k}} + 2e\right)^n$.

\noindent
Now, we combine the previous two lemmas to show that any moment-constrained state is close to a finite-dimensional state which is close to a state with smaller rank.
\begin{lemma}[(Low rank and finite-dimensional approximation of moment-constrained states)]
\label{le:lowrankapprox}
Let $\varepsilon \in (0,1)$.
Let $\rho$ be an $n$-mode state with $k$-moment bounded as $\left(\Tr[\hat{N}_n^k\rho]\right)^{1/k}\le nN_{\text{phot}}$. Let $\HH_m$ be the Hilbert space, defined in~\eqref{def_H_m}, spanned by all the $n$-mode Fock states with total photon number less than $m$ with $m \coloneqq \ceil{\frac{n N_{\text{phot}}}{\varepsilon^{2/k}}}$, with dimension 
\bb
    \dim\HH_m\le \left(\frac{eN_{\text{phot}}}{\varepsilon^{2/k}}+2e\right)^n= \Theta\!\left(\frac{N_{\text{phot}}}{\varepsilon^{2/k}}\right)^{\!n}\,.
\ee
Let $\rho_{\deff}$ be the projected state of $\rho$ onto $\HH_m$. 
Then, $\rho_{\deff}$ is $\eta(\varepsilon)$-close in trace distance to a state on $\HH_m$ with rank $r_{\text{eff}}$ satisfying
\bb
        r_{\text{eff}}\le \left(\frac{N_{\text{phot}}}{\varepsilon^{1/k}}+2e\right)^n= \Theta\!\left(\frac{N_{\text{phot}}}{\varepsilon^{1/k}}\right)^{\!n}\,,
\ee
where we have defined $\eta(\varepsilon)\coloneqq \left(2+\frac{1}{\sqrt{1-\varepsilon}}\right)\varepsilon$. 
\end{lemma}
\begin{proof}   
    Thanks to Lemma~\ref{lemma_rank_ec_states}, there exists an $n$-mode state $\rho_{r_{\text{eff}}}$, with rank satisfying
 \bb
    \operatorname{rank}(\rho_{r_{\text{eff}}})\le\left(\frac{eN_{\text{phot}}}{\varepsilon^{1/k}}+2e\right)^n= \Theta\!\left(\frac{eN_{\text{phot}}}{\varepsilon^{1/k}}\right)^{\!n}\,,
 \ee
which is $\varepsilon$-close in trace distance to $\rho$. Let $\rho'$ be the projected state of $\rho_{r_{\text{eff}}}$ onto $\HH_m$, that is,
\bb
    \rho'\coloneqq\frac{\Pi_m\rho_{r_{\text{eff}}}\Pi_m}{\Tr[\Pi_m\rho_{r_{\text{eff}}}\Pi_m]},
\ee
where $\Pi_m$ is the projector onto $\HH_m$. Clearly, $\rho'$ is a state supported on $\HH_m$. The following chain of inequality shows that $\rho'$ is $\eta(\varepsilon)$-close in trace distance to $\rho_{\deff}$:
\bb
    \frac{1}{2}\|\rho_{\deff}-\rho'\|_1&\leqt{(i)} \frac{1}{2}\|\rho_{\deff}-\rho\|_1+ \frac{1}{2}\|\rho-\rho_{r_{\text{eff}}}\|_1+\frac{1}{2}\|\rho_{r_{\text{eff}}}-\rho'\|_1\\
    &\leqt{(ii)} 2\varepsilon+ \frac{1}{2}\|\rho_{r_{\text{eff}}}-\rho'\|_1\\
    &\leqt{(iii)} 2\varepsilon+\sqrt{\Tr[(\mathbb{1}-\Pi_m)\rho_{r_{\text{eff}}}]}\\
    &\leqt{(iv)} 2\varepsilon+\frac{\varepsilon}{\sqrt{1-\varepsilon}}\\
    &=\eta(\varepsilon)\,.
\ee
Here, in (i), we have used  triangle inequality. In (ii), we exploited Lemma~\ref{le:effectivedim} and Lemma~\ref{lemma_rank_ec_states} to ensure that $\frac{1}{2}\|\rho_{\deff}-\rho\|_1\le \varepsilon$ and $\frac{1}{2}\|\rho-\rho_{r_{\text{eff}}}\|_1$, respectively. In (iii), we have used the \emph{gentle measurement lemma} \cite[Lemma 6.15]{Sumeet_book}. In (iv), we have used  that
\bb
    \Tr[(\mathbb{1}-\Pi_m)\rho_{r_{\text{eff}}}]&\eqt{(v)}\frac{1}{\Tr \Theta_{r_{\text{eff}}}}\Tr\left[(\mathbb{1}-\Pi_m)  \Theta_{r_{\text{eff}}}   \right]\\
    &\leqt{(vi)} \frac{1}{1-\varepsilon}\Tr\left[(\mathbb{1}-\Pi_m)  \Theta_{r_{\text{eff}}}   \right]\\
    &\leqt{(vii)} \frac{1}{1-\varepsilon}\Tr\left[(\mathbb{1}-\Pi_m)  \rho  \right]\\
    &\leqt{(viii)} \frac{\varepsilon^2}{1-\varepsilon}\,.
\ee
Here, in (v), we have used  the definition $\rho_{r_{\text{eff}}}\coloneqq\frac{\Theta_{r_{\text{eff}}} }{\Tr \Theta_{r_{\text{eff}}} }$ in~\eqref{def_rho_reff}, where we recall that $\Theta_{r_{\text{eff}}}$ is defined in~\eqref{def_theta_r} as $\Theta_{r_{\text{eff}}}\coloneqq \sum_{j=1}^{r_{\text{eff} }}\lambda_j^{\downarrow} \psi_j$. In (vi), we have used  that $\Tr \Theta_{r_{\text{eff}}}\ge 1-\varepsilon$, as it follows from the third and fifth line of~\eqref{chain_inequalities_theta}. In (vii), we exploited that 
\bb
    \Theta_{r_{\text{eff}}}= \sum_{j=1}^{r_{\text{eff} }}\lambda_j^{\downarrow} \psi_j\le \sum_{j=1}^{\infty}\lambda_j^{\downarrow} \psi_j=\rho\,.
\ee
Finally, in (viii) we applied that~\eqref{ineq_involving_project}.
\end{proof}

\subsubsection{Related works on approximating moment-constrained states}
\label{subsub:comparison}

We demonstrated that an $n$-mode state $\rho$ with $k$-moment bounded by $N_{\text{phot}}$ can be approximated within trace distance precision $\varepsilon$ by a projected state living in a subspace of \textit{effective dimension} $\deff = \lceil (e N_{\text{phot}} / \varepsilon^{2/k} + 2e)^n \rceil$ (Lemma~\ref{le:effectivedim}). Naturally, the lower the effective dimension, the better the efficiency of the subsequent tomography algorithm for learning the projected state.

Previous works such as~\cite{PLOB,becker_classical_2023} have established similar results for approximating $n$-mode moment-constrained states with a projected finite-dimensional state. However, our Lemma~\ref{le:effectivedim} improves upon these findings. Specifically, our effective dimension scales as $\deff = \Theta (e N_{\text{phot}} / \varepsilon^{2/k})^n $, is in contrast to the scaling $\deff = \Theta (nN_{\text{phot}} / \varepsilon^{2/k})^n$ obtained in~\cite{PLOB,becker_classical_2023}.
Unlike \cite[Proposition 8]{becker_classical_2023}, which involves approximating an $n$-mode state with an $n$-qudit state, our projection approximates an $n$-mode state with a single-qudit state. Furthermore, Ref.\ \cite[Supplementary Note 1]{PLOB} exclusively focuses on the case $k=1$ (energy-constrained states). While they use a projection similar to ours, they end up with an unfavourable scaling $\deff = \Theta(nN_{\text{phot}} / \varepsilon)^n$ due to an overly crude upper bound on the effective dimension.
In conclusion, our bound on the effective dimension of moment-constrained states represents to our knowledge the first achieved bound that is non-super-exponential in the number of modes.

\subsubsection{Known optimal tomography algorithm for qudits}
\label{subsub:performance}
In this section, we review the guarantees of the optimal tomography algorithm known as the \emph{truncated version of Keyl’s algorithm}, which is detailed in Wright's PhD thesis~\cite[Theorem 1.4.13]{wrightHowLearnQuantum}.
\begin{lemma}[(Truncated version of Keyl’s algorithm~\cite{wrightHowLearnQuantum})]\label{Lemma_Wright}
    Let $\rho$ be an unknown qudit state of dimension $d$, and let $\lambda_1^{\downarrow} \ge \lambda_2^{\downarrow} \ge \ldots \ge \lambda_d^{\downarrow}$ be its eigenvalues. For any $r \in [d]$, $N$ copies of $\rho$ are sufficient to build an estimator $\tilde{\rho}$ with rank $r$ such that 
    \begin{equation}
        \mathbb{E} \|\rho - \tilde{\rho}\|_1 \le \sum_{j=r+1}^d \lambda_j^{\downarrow} + 6\sqrt{\frac{rd}{N}}.
    \end{equation}
\end{lemma}
We now introduce a lemma that later will be useful for expressing the sample complexity of the tomography algorithm.
\begin{lemma}\label{lemma_coda_eigenvalues}
    Let $\rho$ be a qudit state of dimension $d$ (possibly $d = \infty$), and let $\lambda_1^{\downarrow} \ge \lambda_2^{\downarrow} \ge \ldots \ge \lambda_d^{\downarrow}$ be its eigenvalues ordered in decreasing order. For any state $\rho_r$ with rank $r$, we have
    \begin{equation}
        \sum_{j=r+1}^d \lambda_j^{\downarrow} \le \frac{1}{2} \|\rho - \rho_r\|_1.
    \end{equation}
\end{lemma}

\begin{proof}
    Let $\Pi_r$ be the projector onto the support of $\rho_r$. Note that
    \bb
        \sum_{j=r+1}^{d}\lambda_j^{\downarrow}&\leqt{(i)} \Tr[(\mathbb{1}-\Pi_r)\rho]\\
        &\eqt{(ii)} \Tr[(\mathbb{1}-\Pi_r)(\rho-\rho_r)]\\
        &\le \max_{0\le E\le\mathbb{1}}\Tr[E(\rho-\rho_r)]\\
        &\eqt{(iii)} \frac{1}{2}\|\rho-\rho_r\|_1\,.
    \ee
    Here, in (i), we have used the \emph{(infinite) dimensional Schur Horn theorem} (Lemma~\ref{le:infSchur}); in (ii), we just exploited the fact that $\Pi_r\rho_r=\rho_r\Pi_r=\rho_r$, which holds true since $\Pi_r$ is the projector onto the support of $\rho_r$; in (iii), we just used the well-known variational characterisation of the trace distance in terms of POVM~\cite{NC}.
\end{proof}
We are now ready to provide the precise performance guarantees of the algorithm~\cite{wrightHowLearnQuantum}.
\begin{lemma}[(Sample complexity of the optimal tomography algorithms for qudit systems)]\label{lemma_tomography_rank_r}
    Let $\varepsilon,\delta\in(0,1)$. Let $\rho$ be an (unknown) state of dimension $d$ such that it is $\frac{\varepsilon}{3}$-close in trace distance to a state with rank $r$. Then, there exists a tomography algorithm such that given
    \bb
        N\ge 2^{18}\frac{rd}{\varepsilon^2} \log\left(\frac{2}{\delta}\right)
    \ee
    copies of $\rho$ can build (a classical description of) an $r$-rank state estimator $\tilde{\rho}$ satisfying
    \bb
        \Pr\left[\frac{1}{2} \|\rho-\tilde{\rho}\|_1\le\varepsilon \right]\ge 1-\delta\,.
    \ee
\end{lemma}
\begin{proof}
    Thanks to Lemma~\ref{Lemma_Wright}, $N'$ copies of $\rho$ are sufficient in order to build an estimator $\tilde{\rho}$ (with rank $r$) such that 
    \bb
        \E \|\rho-\tilde{\rho}\|_1 \le \sum_{j=r+1}^d\lambda_j^{\downarrow} +6\sqrt{\frac{rd}{N'}}\,,
    \ee
    where $\lambda_1^{\downarrow}\ge\lambda_2^{\downarrow}\ge\ldots\ge \lambda_d^{\downarrow}$ are the eigenvalues of $\rho$.
    By assumption, there exists a state $\rho_r$ with rank $r$ such that $\frac{1}{2}\|\rho-\rho_r\|_1\le\frac{\varepsilon}{3}$. 
    Consequently, Lemma~\ref{lemma_coda_eigenvalues} implies that $\sum_{j=r+1}^d\lambda_j^{\downarrow}\le\frac{\varepsilon}{3}$.
    Hence, by choosing $N'\coloneqq\left\lceil\frac{(18)^2rd}{\varepsilon^2}\right\rceil$ and by applying Markov's inequality, we have that
    \bb
        \Pr\left[ \frac{1}{2}\|\rho-\tilde{\rho}\|_1\ge \varepsilon \right]& \le \frac{\E\|\rho-\tilde{\rho}\|_1}{2\varepsilon}\\
        &\le \frac{\sum_{j=r+1}^d\lambda_j^{\downarrow} +6\sqrt{\frac{rd}{N'}}}{2\varepsilon}\\
        &\le\frac{1}{6}+\frac{3}{\varepsilon}\sqrt{\frac{rd}{N'}}\\
        &\le \frac{1}{3}\,.
    \ee
    Consequently, we have proved that $N'= \left\lceil 324 \frac{rd}{\varepsilon^2}\right\rceil$ copies of $\rho$ are sufficient in order to build an estimator $\tilde{\rho}$ such that $\Pr\left[ \frac{1}{2}\|\rho-\tilde{\rho}\|_1\le \varepsilon \right]\ge \frac{2}{3}$. Now, we can apply the standard argument presented in Lemma~\ref{le:enhance-success_states} with $P_{\text{succ}}=\frac{2}{3}$ in order to enhance the probability of success from $\frac{2}{3}$ to $1-\frac{\delta}{2}$. From such a lemma, by defining
    \bb
        m &\coloneqq \left\lceil \frac{2}{\left(1-\frac{1}{2p_{\mathrm{succ}}}\right)^2p_{\mathrm{succ}}}\log\!\left(\frac{2}{\delta}\right)  \right\rceil =\left\lceil48\log\left(\frac{2}{\delta}\right)\right\rceil\,,
    \ee
    it follows that a total of 
    \bb
        mN'=\left\lceil48\log\left(\frac{2}{\delta}\right)\right\rceil\left\lceil 324 \frac{rd}{\varepsilon^2}\right\rceil
    \ee
    copies of $\rho$ suffices in order to build an estimator $\tilde{\rho}$ such that $\Pr\left[ \frac{1}{2}\|\rho-\tilde{\rho}\|_1\le 3\varepsilon \right]\ge 1-\frac{\delta}{2}\ge1-\delta$. Finally, by redefining $\varepsilon\mapsto\frac{\varepsilon}{3}$, we have that
    \bb
    \left\lceil48\log\left(\frac{2}{\delta}\right)\right\rceil\left\lceil 324 \frac{rd}{(\varepsilon/3)^2}\right\rceil\le 2^{18}\frac{rd}{\varepsilon^2}\log\left(\frac{2}{\delta}\right)\,,
    \ee
    and thus we conclude that $2^{18}\frac{rd}{\varepsilon^2}\log\left(\frac{2}{\delta}\right)$ copies of $\rho$ are sufficient in order to build an estimator $\tilde{\rho}$ such that $\Pr\left[ \frac{1}{2}\|\rho-\tilde{\rho}\|_1\le \varepsilon \right]\ge1-\delta$.
\end{proof}
One can prove that $\tildeTheta\left(\frac{rd}{\varepsilon^2}\right)$ copies are not only sufficient (as established by the lemma above) but also necessary to perform quantum state tomography of a finite-dimensional quantum state~\cite{haah_optimal_2021, anshu2023survey}.


\subsubsection{Tomography algorithm for moment-constrained mixed and pure states}
\label{subsub:algorithm}
In this subsection, we present a tomography algorithm to learn a classical description of an unknown $n$-mode moment-constrained state (pure and mixed), together with its sample complexity analysis.
The tomography algorithm involves two main steps: first, projecting onto a finite-dimensional subspace (specifically, the subspace defined in Lemma~\ref{le:effectivedim}), and second, applying a known tomography algorithm designed for finite-dimensional states (specifically, the algorithm described in Lemma~\ref{lemma_tomography_rank_r}). Table~\ref{Table_tomography_ec} presents the steps of our tomography algorithm. We begin with Theorem~\ref{correctness_algorithm_ECmixed}, which analyses the sample complexity of our tomography algorithm for moment-constrained mixed states. The basic idea of Theorem~\ref{correctness_algorithm_ECmixed} is the following: the unknown moment-constrained mixed state is effectively a qudit state of dimension $\deff=O\!\left(\frac{eN_{\text{phot}}}{\varepsilon^{2/k}}\right)^n$ and rank $r_{\text{eff}}=O\!\left(\frac{eN_{\text{phot}}}{\varepsilon^{1/k}}\right)^n$ (as established by Lemma~\ref{le:lowrankapprox}), and thus the sample complexity of tomography is $O\!\left(r_{\text{eff}}\,\deff\right)=O\!\left(\frac{N_{\text{phot}}}{\varepsilon^{3/(2k)}}\right)^{2n}$ (thanks to Lemma~\ref{lemma_tomography_rank_r}).

\begin{table}[t]
  \caption{Tomography algorithm for $n$-mode $k$-th moment constrained states (pure and mixed). The algorithm for pure states and the one for mixed states are the same, apart from two distinctions: the number $N$ of copies of the unknown state, and the specific input parameters provided to the subroutine called at Line 10 of the algorithm. For the mixed case the details of the algorithm are provided in the proof of Theorem~\ref{correctness_algorithm_ECmixed}; while for the pure case the details are reported in the proof of Theorem~\ref{correctness_algorithm_ECpure}.}
  \label{Table_tomography_ec}
  \begin{mdframed}[linewidth=2pt, roundcorner=10pt, backgroundcolor=white!10, innerbottommargin=10pt, innertopmargin=10pt]
    \textbf{Input:} Accuracy $\varepsilon$, failure probability $\delta$, $N$ copies of the unknown $n$-mode moment-constrained state $\rho$ (as defined in Theorem~\ref{correctness_algorithm_ECmixed} and Theorem~\ref{correctness_algorithm_ECpure} for mixed and pure states, respectively).\\
    \textbf{Output:} A classical description of $\tilde{\rho}$, such that $\frac{1}{2}\|\tilde{\rho}- \rho\|_1\le \varepsilon$ with probability at least $1-\delta$.
    \begin{algorithmic}[1]
      \For{$i \leftarrow 1$ \textbf{to} $N$}
         \State Query a copy of $\rho$.
          \State Perform the POVM $\{\Pi_{m},\mathbb{1}-\Pi_{m}\}$, where $\Pi_m$ is the projector onto the subspace $\HH_{m}$, defined in Eq.~\eqref{def_H_m}.
          \If{the POVM outcome corresponds to $\mathbb{1}-\Pi_{m}$ }
              \State Discard the post-measurement state.
          \Else
              \State Keep the post-measurement state.
          \EndIf
      \EndFor
     \State Perform the full state tomography algorithm described in Lemma~\ref{lemma_tomography_rank_r} on the kept copies of the post-measurement states, obtaining as output the classical description of a state $\tilde{\rho}$. 
      \State \Return $\tilde{\rho}$.
    \end{algorithmic}
  \end{mdframed}
\end{table}

\begin{thm}[(Learning moment-constrained mixed states)]
\label{correctness_algorithm_ECmixed}
Let $\varepsilon, \delta \in (0,1)$. Let $\rho$ be an $n$-mode state with $k$-moment bounded as $\left(\Tr[\hat{N}_n^k\rho]\right)^{1/k} \le nN_{\text{phot}}$. There exists a quantum algorithm that, utilising 
\begin{align}
N = \left \lceil 2^{21}\frac{r_{\text{eff}}\,\deff}{\varepsilon^2} \log\left(\frac{4}{\delta}\right) \right \rceil = O\!\left(\frac{N_{\text{phot}}}{\varepsilon^{3/(2k)}}\right)^{2n}
\end{align}
copies of $\rho$, generates a classical representation of a $\deff$-dimensional state $\tilde{\rho}$ with rank $r_{\text{eff}}$ such that
\begin{align}
\Pr\left[\frac{1}{2} \|\rho - \tilde{\rho}\|_1 \le \varepsilon \right] \ge 1-\delta\,,
\end{align}
where $\deff \le \left(\frac{eN_{\text{phot}}}{\varepsilon^{2/k}} + 2e\right)^n$ and $r_{\text{eff}} \le \left(\frac{eN_{\text{phot}}}{(\frac{\varepsilon}{20})^{1/k}} + 2e\right)^n$. 
\end{thm}

 \begin{proof}
We aim to establish the correctness of the algorithm presented in Table~\ref{Table_tomography_ec}. The algorithm queries $\ceil{2N'+24\log\!\left(\frac{2}{\delta}\right)}$ copies of $\rho$, where we will fix $N'$ later. Let $m \coloneqq \ceil{\frac{n N_{\text{phot}}}{(\varepsilon/2)^{2/k}}}$. On each copy, it executes the POVM $\{\Pi_{m},\mathbb{1}-\Pi_{m}\}$, where $\Pi_m \coloneqq \sum_{\textbf{m}:\, \sum_{i=1}^n m_i \leq m}\ketbra{\textbf{m}}$ is the projector onto the subspace $\HH_{m}$, defined in Eq.~\eqref{def_H_m}, spanned by the $n$-mode Fock states with total photon number at most $m$, with dimension (see Lemma~\ref{le:effectivedim})
\begin{align}
\deff \coloneqq \dim\HH_m \le \left(\frac{eN_{\text{phot}}}{\varepsilon^{2/k}}+2e\right)^n = \Theta\!\left(\frac{eN_{\text{phot}}}{\varepsilon^{2/k}}\right)^n.
\end{align}
The probability of obtaining the first POVM outcome is 
\bb
\Tr\left[\Pi_m \rho \right] &= 1 - \Tr[(\mathbb{1}-\Pi_m)\rho] \\
&\ge 1 - \frac{1}{m^k}\Tr\left[\hat{N}_{n}^{k}\rho\right] \\
&\ge 1 - \frac{(nN_{\text{phot}})^k}{m^k} \\
&\ge 1 - \frac{\varepsilon^2}{4} \\
&\ge \frac{3}{4}\,,
\ee
where the second step follows from the operator inequality $m^k(\mathbb{1}-\Pi_m) \le \hat{N}_{n}^{k}$ and the third step follows from the hypothesis on the $k$-th moment of $\rho$.

\noindent
Therefore, applying Lemma~\ref{le:enhance-success}, we can assert that the algorithm, which uses a number of copies $\ceil{2N'+24\log\!\left(\frac{2}{\delta}\right)}$, with a probability $\geq 1-\frac{\delta}{2}$ obtains at least $N'$ copies of the post-measurement state $\rho_{\deff} \coloneqq \Pi_m\rho\,\Pi_m/\Tr\left[ \Pi_m \rho\right]$. Due to Lemma~\ref{le:effectivedim}, the post-measurement state satisfies
\begin{align}
\frac{1}{2}\left\| \rho_{\deff} -\rho \right\|_1 \leq \frac{\varepsilon}{2}.
\end{align}
Moreover, Lemma~\ref{le:lowrankapprox} implies that $\rho_{\deff}$ is $\eta(\varepsilon/20)$-close in trace distance to a state supported on $\HH_m$ with rank $r_{\text{eff}}$ satisfying
\begin{align}
r_{\text{eff}} \le \left(\frac{eN_{\text{phot}}}{(\frac{\varepsilon}{20})^{1/k}} + 2e\right)^n = \Theta\left(\frac{N_{\text{phot}}}{\varepsilon^{1/k}}\right)^n\,,
\end{align}
where $\eta(\frac{\varepsilon}{20}) \coloneqq \left(2 + \frac{1}{\sqrt{1 - \frac{\varepsilon}{20}}}\right)\frac{\varepsilon}{20} \le \frac{\varepsilon}{6}$.

\noindent
Thus, from Lemma~\ref{lemma_tomography_rank_r}, it follows that there exists a quantum algorithm that, utilising only
\begin{align}
N' &\ge 2^{18}\frac{r_{\text{eff}}\,\deff}{(\varepsilon/2)^2} \log\left(\frac{4}{\delta}\right)
\end{align}
copies of the post-measurement state $\rho_{\deff}$, builds a classical description of an $r_{\text{eff}}$-rank state estimator $\tilde{\rho}$ such that, with a probability $\geq 1-\frac{\delta}{2}$, it holds that
\begin{align}
\frac{1}{2} \|\rho_{\deff}-\tilde{\rho}\|_1 \le \frac{\varepsilon}{2}.
\end{align}
By the triangle inequality, we then conclude that $\frac{1}{2}\left\| \rho -\tilde{\rho} \right\|_1 \leq \varepsilon$. The total failure probability of the algorithm is $\leq \delta$ by applying a union bound.
\end{proof}

We now proceed with Theorem~\ref{correctness_algorithm_ECpure}, which shows an upper bound on the sample complexity of tomography of moment-constrained pure states. The basic idea of Theorem~\ref{correctness_algorithm_ECpure} is the following: the unknown moment-constrained pure state is effectively a pure qudit state of dimension $\deff=O\!\left(\frac{eN_{\text{phot}}}{\varepsilon^{2/k}}\right)^n$ (as established by Lemma~\ref{le:effectivedim}), and thus the sample complexity of tomography is $O\!\left(\deff\right)= O\!\left(\frac{N_{\text{phot}}}{\varepsilon^{2/k}}\right)^{n}$ (thanks to Lemma~\ref{lemma_tomography_rank_r}).
\begin{thm}[(Learning moment-constrained pure states)]
\label{correctness_algorithm_ECpure}
Let $\varepsilon, \delta \in (0,1)$. Consider an $n$-mode pure state $\psi$ with its $k$-th moment bounded by $\left(\Tr[\hat{N}_n^k\psi ]\right)^{1/k} \leq nN_{\text{phot}}$. There exists a quantum algorithm that, utilising 
\begin{align}
N = \left\lceil 2^{21}\frac{\deff}{\varepsilon^2} \log\left(\frac{4}{\delta}\right)\right\rceil= O\!\left(\frac{N_{\text{phot}}}{\varepsilon^{2/k}}\right)^{n}
\end{align}
copies of $\psi$, generates a classical representation of a pure state $\tilde{\psi}$ such that
\bb
\Pr\left[\frac{1}{2} \|\psi-\tilde{\psi}\|_1 \le \varepsilon \right] \ge 1-\delta,
\ee
where $\deff \le \left(\frac{eN_{\text{phot}}}{\varepsilon^{2/k}}+2e\right)^n$. 
\end{thm}

\begin{proof}
   We adopt the same notation used in Theorem~\ref{correctness_algorithm_ECmixed}.
   The algorithm is presented in Table~\ref{Table_tomography_ec}, and the proof of its correctness follows similar steps to the previous mixed state case (Theorem~\ref{correctness_algorithm_ECmixed}). Specifically, the algorithm queries $\ceil{2N'+24\log\!\left(\frac{2}{\delta}\right)}$ copies of $\psi$ and on each of them performs the POVM $\{\Pi_{m},\mathbb{1}-\Pi_{m}\}$, where $\Pi_m$ is the projector onto the subspace $\HH_{m}$ (defined in Eq.~\eqref{def_H_m}). As in the proof of Theorem~\ref{correctness_algorithm_ECmixed}, with probability $\geq 1-\frac{\delta}{2}$, on at least $N'$ copies the post-measurement state will be $\psi_{\deff} \coloneqq \Pi_m\psi \,\Pi_m/\Tr\left[ \Pi_m \psi\right]$, satisfying $\frac{1}{2}\left\| \psi_{\deff} -\psi \right\|_1 \leq \frac{\varepsilon}{2}$.

\noindent
   By applying Lemma~\ref{lemma_tomography_rank_r} with $r=1$, there exists a quantum algorithm that, utilising
   \begin{align}
   N' &\ge 2^{18}\frac{\deff}{(\varepsilon/2)^2} \log\left(\frac{4}{\delta}\right),
   \end{align}
   copies of the post-measurement state $\psi_{\deff}$, builds a classical description of a pure state $\tilde{\psi}$ such that, with a probability $\geq 1-\frac{\delta}{2}$, it holds that
   \bb
   \frac{1}{2} \|\psi_{\deff}-\tilde{\psi}\|_1 \le \frac{\varepsilon}{2}.
   \ee
   By the triangle inequality, we conclude that $\frac{1}{2}\left\| \psi -\tilde{\psi} \right\|_1 \leq \varepsilon$. The total failure probability of the algorithm is $\leq \delta$ by applying a union bound.
\end{proof}

Remarkably, as a consequence of the sample-complexity lower bound proved in Theorem~\ref{th:lowerboundtomohraphy} and the sample-complexity upper bound proved in Theorem~\ref{correctness_algorithm_ECpure}, we have the following.

\begin{thm}[(Optimal sample complexity of tomography of moment-constrained pure states)]
\label{thm_optimal_sample_pure}
Let $\psi$ be an unknown pure $n$-mode state vector satisfying the $k$-th moment constraint $(\Tr[\hat{N}_n^k\rho])^{1/k}\le n N_{\text{phot}}$. Then, $\Theta \left(\frac{N_{\text{phot}}}{\varepsilon^{2/k}}\right)^{n}$ copies of $\psi$ are necessary and sufficient to perform quantum state tomography with precision $\varepsilon$ in trace distance. 
\end{thm}

\newpage

\section{Tomography of bosonic Gaussian states}
\label{Sec_Gaussian}
Gaussian states play a crucial role in applications of quantum optics, such as quantum sensing, quantum communication, and optical quantum computing~\cite{Wang2007,weedbrook_gaussian_2012}, and they form a small subset of the entire infinite-dimensional Hilbert space of continuous variable quantum states.
In contrast to an arbitrary continuous variable quantum state, which is defined in terms of an infinite number of parameters, a Gaussian state is uniquely characterised by only a few parameters --- specifically those present in its first moment and its covariance matrix. Indeed, it is well established that: `In order to \emph{know} a Gaussian state it is sufficient to \emph{know} its first moment and its covariance matrix.' However, in practice, we never know the first moment and the covariance matrix \emph{exactly}, but we can only have estimates of them, meaning that we can only approximately know the Gaussian state. Crucially, the \emph{trace distance} between the exact quantum state and its approximation is the most meaningful figure of merit to measure of the error incurred in the approximation, due to the operational meaning of the trace distance given by the Holevo--Helstrom theorem~\cite{HELSTROM, Holevo1976}. It is thus a fundamental problem --- yet never tackled before --- of Gaussian quantum information to determine what is the error incurred in \emph{trace distance} when estimating the first moment and covariance matrix of an unknown Gaussian state up to a precision $\varepsilon$. In this section, we address this fundamental problem, by finding upper and lower bounds on the trace distance between two arbitrary Gaussian states, determined by the norm distance of their covariance matrices and first moments. We present such bounds in the forthcoming Theorem~\ref{spoiler_bounds}.

One might be inclined to believe that there is a simple approach to solving this problem, involving first bounding the trace distance in terms of the fidelity and then employing the known formula for fidelity between Gaussian states \cite{Banchi_2015}. Although this approach may seem promising at first glance, it is actually highly non-trivial because the expressions involved in the fidelity formula appear to be too complicated to allow the derivation of a bound in terms of the norm distance between the first moments and the covariance matrices.

\begin{thm}[(Bounds on the trace distance between Gaussian states)]\label{spoiler_bounds}
    Let $\rho_1$ and $\rho_2$ be $n$-mode Gaussian states satisfying the energy constraint $\Tr[\rho_1\,\hat{N}_n]\le N$ and $\Tr[\rho_2\,\hat{N}_n]\le N$. Let $\textbf{m}_1$ and $\textbf{m}_2$ be the first moments and let $V_1$ and $V_2$ be the covariance matrices of $\rho_1$ and $\rho_2$, respectively. The trace distance between $\rho_1$ and $\rho_2$ can be upper bounded as
    \bb 
     \frac12\|\rho_1-\rho_2\|_1\le f(N)\left(\|\textbf{m}_1-\textbf{m}_2\|_2+\sqrt{2}\sqrt{\|V_1-V_2\|_1}\right)\,,
    \ee
    where
    $f(N)\coloneqq \frac{1}{\sqrt{2}} \left(\sqrt{N} + \sqrt{N+1}\right)$. Moreover, it can be lower bounded as
\bb 
    \frac{1}{2}\|\rho_1-\rho_2\|_1\ge \frac{1}{200}\max\left(\min\!\left(1,\frac{\| \textbf{m}_1-\textbf{m}_2 \|_2}{\sqrt{h(N,n)}}\right), \min\!\left(1, \frac{\|V_1-V_2\|_2}{ h(N,n) }\right)\right)\,,
\ee
where $ h(N,n)\coloneqq 4N+2n+1$.
\end{thm}
Theorem~\ref{spoiler_bounds} is not only a technical result of independent interest for the field of Gaussian quantum information, but also it answers the fundamental question: \emph{if we approximate the first moment and covariance matrix of an unknown Gaussian state with precision $\varepsilon$, what is the resulting trace distance error on the state?} Theorem~\ref{spoiler_bounds} establishes that if we make an error $O(\varepsilon)$ in approximating the first moment and the covariance matrix, then the trace-distance error that we make in approximating the unknown Gaussian state is \emph{at most} $O(\sqrt{\varepsilon})$ and \emph{at least} $O(\varepsilon)$. 

The upper bound presented in Theorem~\ref{spoiler_bounds} allows us to analyse the sample complexity of tomography of Gaussian states. 
In particular, we prove that tomography of (energy-constrained) Gaussian states is efficient, as there exists a tomography algorithm whose sample and time complexity scales polynomially in the number of modes. Notably, our result demonstrates that Gaussian states can be \emph{efficiently} learned by estimating the first moment and the covariance matrix, a result that has been previously assumed but never rigorously proved in the literature.

\begin{thm}[(Upper bound on the sample complexity of tomography of Gaussian states - informal version)]\label{thm_tomography_gaussian}
A number $O\!\left( \frac{n^7N_{\text{phot}}^4}{\varepsilon^4} \right)$ of state copies are sufficient for quantum state tomography of $n$-mode Gaussian states $\rho$ satisfying the energy constraint $\Tr[\rho \hat{N}_n]\le n N_{\text{phot}}$, where $\varepsilon$ is the accuracy in trace distance.
\end{thm}

Notably, we can improve the trace distance upper bound in Theorem~\ref{spoiler_bounds} if we assume one of the states, say $\rho_2$, to be a pure Gaussian state (interestingly, such an improved bound holds even if $\rho_1$ is not Gaussian). 
\begin{thm}[(Upper bound on the trace distance between a pure Gaussian state and a possibly non-Gaussian mixed state)]\label{spoiler_bounds_pure}
    Let $\psi\coloneqq \ketbra{\psi}$ be a pure $n$-mode Gaussian state with first moment $\textbf{m}(\psi)$ and second moment $V\!(\psi)$.  Let $\rho$ be an $n$-mode (possibly non-Gaussian) state with first moment $\textbf{m}(\rho)$ and second moment $V\!(\rho)$.  Assume that $\rho$ and $\psi$ satisfy the energy constraint $\Tr[\rho \hat{N}_n]\le N$. Then
    \bb
        \frac{1}{2}\left\|\, \rho-\ketbra{\psi}\, \right\|_1\le  \sqrt{N+\frac{n}{2}}\sqrt{ 2\|\textbf{m}(\rho)-\textbf{m}(\psi)\|_2^2+\|V\!(\rho)-V\!(\psi)\|_\infty}\,.
    \ee
\end{thm}
The improved bound presented in Theorem~\ref{spoiler_bounds_pure} allows us to analyse the sample complexity of tomography of pure Gaussian states.

\begin{thm}[(Upper bound on the sample complexity of tomography of pure Gaussian states - informal version)]\label{thm_tomography_gaussian_pure}
A number $O\!\left( \frac{n^5N_{\text{phot}}^4}{\varepsilon^4} \right)$ of state copies are sufficient for quantum state tomography of $n$-mode pure Gaussian states $\psi$ satisfying the energy constraint $\Tr[\psi \hat{N}_n]\le n N_{\text{phot}}$, where $\varepsilon$ is the accuracy in trace distance.
\end{thm}

We have seen that learning unknown Gaussian states is efficient, with sample complexity scaling polynomially in the number of modes. But what if the unknown state is not exactly Gaussian? Is our tomography procedure robust against slight perturbations of the set of Gaussian states? These questions are crucial conceptually, given that experimental imperfections may transform Gaussian states into slightly-perturbed Gaussian states. In this section, we will also prove that \emph{quantum state tomography of slightly-perturbed Gaussian states is efficient}. Technically, by `slightly-perturbed Gaussian state' we mean that the minimum quantum relative entropy between the state and any Gaussian state is sufficiently small. The latter is a meaningful measure of `non-Gaussianity', thanks to results from~\cite{Hiai1991,Ogawa2000,Genoni2008,Marian2013}.

Our tomography algorithm for Gaussian states is feasible to realise in practice, as it only requires the estimation of the first moment and the covariance matrix of the unknown Gaussian state, tasks routinely accomplished in quantum optics laboratories through homodyne detection~\cite{BUCCO, Aolita_2015}. We stress that the non-trivial aspect of our tomography algorithm is that it comes with rigorous performance guarantees. That is, if the number of state copies is larger than a critical value (which scales polynomially in the number of modes), then the algorithm outputs a classical description of a Gaussian state which is guaranteed to be a `good' approximation of the true state in trace distance with high probability.
 
This section is organised as follows:
\begin{itemize}
    \item In Subsection~\ref{subsec_bounds_trace distance}, we prove upper and lower bounds on the trace distance between Gaussian states in terms of the norm distance between their first moments and their covariance matrices. 
    \item In Subsection~\ref{Subsection_useful_relations}, we analyse the sample complexity of learning the first moment and the covariance matrix of unknown (possibly non-Gaussian) states. 
    
    \item In Subsection~\ref{subs_sec_tom_gauss}, we put together all the above-mentioned results to analyse the sample complexity of tomography of (energy-constrained) Gaussian states both in the mixed-state and in the pure-state setting. We also how the second moment of the energy of a Gaussian state can be expressed in terms of its first moment and its covariance matrix, a result that turns out to be useful to obtain rigorous performance guarantees on tomography of energy-constrained Gaussian states. Finally, we show the our tomography procedure is robust under little perturbations from the set of Gaussian states, meaning that learning slightly-perturbed Gaussian states is efficient.
\end{itemize}
\emph{Note}: Throughout this section, we sometimes represent energy constraints using the energy operator $\hat{E}_n$, such as $\Tr[\rho\hat{E}_n]\leq nE$, rather than the total photon number operator $\hat{N}_n$, like $\Tr[\rho\hat{N}_n]\leq nN_{\text{phot}}$. This simplifies the analysis without altering the sample complexity scaling, because of the identity $\hat{E}_n=\hat{N}_n+\frac{n}{2}\mathbb{1}$ that allows us to identify $E=N_{\text{phot}}+\frac{1}{2}$. For example, the upper bound on the sample-complexity of Gaussian-state tomography in Theorem~\ref{thm_tomography_gaussian}, given by $O\!\left( \frac{n^7N_{\text{phot}}^4}{\varepsilon^4} \right)$, can equivalently be expressed as $O\!\left( \frac{n^7E^4}{\varepsilon^4} \right)$.

\subsection{Bounds on the trace distance between Gaussian states}\label{subsec_bounds_trace distance}
Possessing knowledge of the first moment and covariance matrix of a Gaussian state is sufficient to determine the state itself. However, when dealing with a finite number of copies of an unknown Gaussian state $\rho$, it is possible to obtain only estimates $\tilde{\textbf{m}}$ and $\tilde{V}$ of its first moment $\textbf{m}(\rho)$ and of its covariance matrix $V\!(\rho)$. Consequently, the resulting Gaussian state $\tilde{\rho}$ with first moment $\tilde{\textbf{m}}$ and covariance matrix $\tilde{V}$ constitutes only an approximation of the true, unknown Gaussian state $\rho$. A natural question arises: What is the error incurred in the approximation $\tilde{\rho}\approx\rho$ in terms of the errors incurred in the approximations $\tilde{\textbf{m}}\approx\textbf{m}(\rho)$ and $V\!(\rho)\approx \tilde{V}$? The existing literature lacked such an error estimate, despite it being a natural question in Gaussian quantum information theory. In this section, we address this gap. To quantify the error incurred in the approximation $\tilde{\rho}\approx\rho$, we employ the trace distance $\frac{1}{2}\|\tilde{\rho}-\rho\|_1$ as it is the most meaningful notion of distance between quantum states, given its operational meaning~\cite{HELSTROM, Holevo1976}. 
Specifically:
\begin{itemize}
    \item In Subsubsection~\ref{subsubsection_upper_bound_mixed_state}, we derive an upper bound on the trace distance between two arbitrary (possibly mixed) Gaussian state in terms of the norm distance between their first moment and their covariance matrices. More explicitly, by using the notation introduced above, we find an upper bound on $\frac{1}{2}\|\tilde{\rho}-\rho\|_1$ in terms of $\|\tilde{\textbf{m}}-\textbf{m}(\rho)\|_2$ and $\|\tilde{V}-V\!(\rho)\|_1$. 
    \item In Subsubsection~\ref{subsubsection_upper_bound_pure_state}, we find an improved upper bound on the trace distance between a pure Gaussian state and a possibly-mixed possibly-non-Gaussian state. \item  In Subsubsection~\ref{subsubsection_lower_bound}, we find a lower bound on the trace distance between two arbitrary (possibly mixed) Gaussian state. One may apply this lower bound, together with $\varepsilon$-net techniques, to obtain a lower bound on the sample complexity on tomography of Gaussian states. However, we do not conduct such an analysis in this work since it leads to a too weak lower bound on the sample complexity. 
\end{itemize}

\subsubsection{Upper bound in the general (possibly mixed-state) setting}\label{subsubsection_upper_bound_mixed_state}
This subsubsection is devoted to the proof of the following theorem, which is one of our main technical result.
\begin{thm}\label{upper_bound_trace_distance_Gaussian}
Let $\rho$, $\sigma$ be $n$-mode Gaussian states with mean photon number bounded by $N$, that is,
\bb
\Tr\left[\rho\,\hat{N}_n\right] &\le N\,,\\
\Tr\left[\sigma\,\hat{N}_n \right] &\le N\,.
\ee
Then, the trace distance between $\rho$ and $\sigma$ is upper bounded by
\bb\label{ludo_upp_bound}
\frac12 \left\|\rho - \sigma\right\|_1 \leq f(N) \left( \|\textbf{m}(\rho)-\textbf{m}(\sigma)\|_2 + \sqrt2\, \sqrt{\|V\!(\rho)-V\!(\sigma)\|_1}\right) ,
\ee
where $f(x) \coloneqq \frac{1}{\sqrt{2}} \left(\sqrt{x} + \sqrt{x+1}\right)$. Here, $\textbf{m}(\rho), V\!(\rho)$ denote the first moment and covariance matrix of $\rho$, while $\textbf{m}(\sigma), V\!(\sigma)$ denote the first moment and covariance matrix of $\sigma$.
\end{thm}
\begin{remark}
     One might be tempted to believe that a straightforward approach exists for establishing an upper bound akin to the one presented in Theorem~\ref{upper_bound_trace_distance_Gaussian}. Specifically, the Fuchs-van de Graaf inequality~\cite{Fuchs1999} implies that the trace distance between the Gaussian states $\rho$ and $\sigma$ can be upper bounded in terms of the fidelity $F(\rho,\sigma)$ as
     \bb
        \frac 12 \|\rho-\sigma\|_1\le \sqrt{1-F(\rho,\sigma)^2}. 
    \ee
    Additionally, a closed formula for the fidelity between Gaussian states is available~\cite{Banchi_2015}. However, it is important to note that this closed formula is exceedingly intricate, involving determinants, square roots, and inverses of functions of the covariance matrices $V\!(\rho),V\!(\sigma)$. Consequently, attempting to derive such an upper bound through this method appears too hard, and, indeed, our attempts in this direction were unsuccessful.
 \end{remark}

Before proving Theorem~\ref{upper_bound_trace_distance_Gaussian}, let us first provide some useful preliminary results. 

For any positive semi-definite matrix $K\in\mathbb{R}^{2n,2n}$, the \emph{Gaussian noise channel} $\NN_K$ is the convex combination of displacement transformations $(\cdot)\mapsto \hat{D}_u(\cdot)\hat{D}_u^\dagger$, where $u\in\mathbb{R}^{2n}$ is distributed according to a Gaussian probability distribution with vanishing mean value and covariance matrix equal to $K$~\cite[Chapter 5]{BUCCO}. Mathematically, $\NN_K$ can be written as
\bb\label{Gaussian_noise}
\NN_K(\Theta) \coloneqq \int_{\mathbb{R}^{2n}} \mathrm{d}^{2n}u\ P_K(u)\, \hat{D}_u \Theta \hat{D}_{u}^\dagger\,,
\ee
where $P_K(u)$ is a Gaussian probability distribution with vanishing mean and covariance matrix equal to $K$, and $\hat{D}_{u}\coloneqq e^{i u^{\intercal}\Omega_n \mathbf{\hat{R}}}$ is the displacement operator (see the preliminaries subsection~\ref{sub:prelCV}). Of course, in the case when $K$ is strictly positive, $P_K(u)$ can be written as
\bb\label{def_pku}
    P_K(u)\coloneqq \frac{e^{-\frac12 u^\intercal K^{-1} u}}{(2\pi)^{n}\sqrt{\det K}}\,.
\ee
Instead, when $K$ has some zero eigenvalues, the definition of $P_K(u)$ involves Dirac deltas as follows. Let $K=\sum_{i=1}^{2n}\lambda_iv_iv_i^\intercal$ be a spectral decomposition of $K$, where $(v_i)_{i=1,\ldots,2n}$ are orthonormal eigenvectors, while $\lambda_1,\ldots,\lambda_r>0$ and $\lambda_{r+1},\ldots,\lambda_{2n}=0$ are its eigenvalues, with $r$ being the rank of $K$. Then, $P_K(u)$ can be written as
\bb
    P_K(u)\coloneqq \frac{e^{-\frac12 u^\intercal\left(\sum_{i=1}^{r}\lambda_i^{-1} v_i v_i^\intercal \right)u
 }}{(2\pi)^{(2n-r)/2}\sqrt{\lambda_1\ldots\lambda_r}}\prod_{i=r+1}^{2n}\delta(v_i^\intercal u)\,,
\ee
where $\delta(\cdot)$ denotes the Dirac delta distribution.

It turns out that $\NN_K$ leaves the first moments unchanged, and it acts on the covariance matrix of the input state by adding $K$~\cite[Chapter 5]{BUCCO}. For the sake of completeness, we prove the latter fact in the following lemma.
\begin{lemma}\emph{(\cite[Chapter 5]{BUCCO})}\label{Lemma_ec_diamond}
    Let $K\in\mathbb{R}^{2n,2n}$ be a positive semi-definite matrix. The Gaussian noise channel $\NN_K$, defined in~\eqref{Gaussian_noise}, acts on the first moments and covariance matrices as 
    \begin{eqnarray}
        \textbf{m}\!\left(\NN_K(\rho)\right)&=&\textbf{m}\!\left(\rho\right)\,,\\
        V\!\left(\NN_K(\rho)\right)&=& V\!(\rho)+K\,\,,
    \end{eqnarray}
    for all input states  $\rho\in\pazocal{D}\!\left(L^2(\mathbb{R}^{n})\right)$.
\end{lemma}
\begin{proof}
    It holds that
    \bb
        \textbf{m}\!\left(\NN_K(\rho)\right)&=\int_{\mathbb{R}^{2n}} \mathrm{d}^{2n}u\ P_K(u)\, \textbf{m}\!\left(\hat{D}_u \rho \hat{D}_{u}^\dagger\right)\\
        &=\int_{\mathbb{R}^{2n}} \mathrm{d}^{2n}u\ P_K(u)\, \left(\textbf{m}(\rho)+u\right)\\
        &=\textbf{m}(\rho)+\int_{\mathbb{R}^{2n}} \mathrm{d}^{2n}u\ P_K(u)u\\
        &=\textbf{m}(\rho)\,,
        \label{eq:momGN}
    \ee
    where in the second step we have used  Eq.~\eqref{eq:bille}.
    Moreover, we have that
    \bb
        V\!\left(  \NN_K(\rho)  \right)&=\Tr[\{\hat{\textbf{R}},\hat{\textbf{R}}^\intercal\}\NN_K(\rho)]-2\textbf{m}(\NN_K(\rho))\textbf{m}(\NN_K(\rho))^\intercal\\
        &=\int_{\mathbb{R}^{2n}} \mathrm{d}^{2n}u\ P_K(u)\Tr\left[\rho \{\hat{\textbf{R}}+u\mathbb{1},\hat{\textbf{R}}^\intercal+u^\intercal\mathbb{1}\}  \right]-2\textbf{m}(\rho)\textbf{m}(\rho)^\intercal\\
        \\&=V\!(\rho)+\int_{\mathbb{R}^{2n}} \mathrm{d}^{2n}u\ P_K(u) \,u\mathbf{m}(\rho)^\intercal+\int_{\mathbb{R}^{2n}} \mathrm{d}^{2n}u\ P_K(u) \,\mathbf{m}(\rho)u^\intercal+ \int_{\mathbb{R}^{2n}} \mathrm{d}^{2n}u\ P_K(u) \,uu^\intercal
        \\&= V\!(\rho)+K\,,
    \ee
where in the second step we have used  Eq.~\eqref{eq:displ} and Eq.~\eqref{eq:momGN}, and in the last step we solved a Gaussian integral.
\end{proof}

\begin{lemma}\label{N_k_is_gaussian}
    \emph{(\cite[Chapter 5]{BUCCO})} For any $K\ge0$, the Gaussian noise channel $\NN_K$ is actually a Gaussian channel. That is, $\NN_K(\rho)$ is a Gaussian state for any Gaussian state $\rho$.
\end{lemma}

Let us recall the definition of \emph{diamond norm}, as it will be useful for the following. The diamond norm of a superoperator $\Delta:\pazocal{D}(\HH)\mapsto\pazocal{D}(\HH)$ is defined as
\begin{equation}\label{dnorm}
	\left\|\Delta\right\|_{\diamond}\coloneqq\sup_{\rho\in\pazocal{D}(\HH\otimes\HH_C)} \left\|\Delta\otimes \Id_C(\rho)\right\|_1\text{ ,}
\end{equation}
where the $\sup$ is taken also over the choice of the ancilla system $\HH_C$. In the infinite-dimensional scenario the topology induced by the diamond norm is often too strong (e.g.,~see Ref.~\cite[Proposition 1]{VVdiamond}) and thus it is customary to define the \emph{energy-constrained diamond norm}. Given $N>0$, given an $n$-mode system $\HH_S=L^2(\mathbb{R}^n)$, the energy-constrained diamond norm of a superoperator $\Delta:\pazocal{D}(\HH_S)\mapsto\pazocal{D}(\HH_S)$ is defined as~\cite{PLOB, Shirokov2018, VVdiamond}
\begin{equation}\label{ednorm}
	\left\|\Delta\right\|_{\diamond N}\coloneqq\sup_{\rho\in\pazocal{D}(\HH_S\otimes\HH_C)\text{ : }\Tr\left[\rho\, \hat{N}_n\otimes\mathbb{1}_C\right]\le N} \left\|(\Delta\otimes I_C)\rho\right\|_1\text{ ,}
\end{equation}
where the $\sup$ is taken also over the choice of the ancilla system $\HH_C$, and $\hat{N}_n$ is the $n$-mode photon number operator on $\HH_S$. The supremum in the definition of diamond norm in~\eqref{dnorm}, as well as the one in~\eqref{ednorm}, is achieved by taking the ancilla system $\HH_C$ to be equal to the input system $\HH$. 

\begin{lemma}\label{upp_bound_ludo} \emph{(\cite[Eq. (3)]{EC-diamond})}
    Let $u\in\mathbb{R}^{2n}$ and let $\D_u(\cdot)\coloneqq \hat{D}_u(\cdot)\hat{D}_u^\dagger$ be the displacement channel. Then, for all $N>0$ the energy-constrained diamond norm of the difference between $\D_u$ and the identity map $\Id$ can be upper bounded as
    \bb
        \frac{1}{2}\|\D_u-\Id\|_{\diamond N}\le\sin\left( \min\left\{ \|u\|_2 f(N),\, \frac{\pi}{2} \right\} \right)\,,
    \ee
    where $f(N)\coloneqq \frac{1}{\sqrt{2}} \left(\sqrt{N} + \sqrt{N+1}\right)$.
\end{lemma}

\begin{lemma}[(Diamond norm bound)]\label{lemma_upp_diamond_norm33}
Let $K\in\mathbb{R}^{2n,2n}$ be a positive semi-definite matrix. Then, for all $N\geq 0$ the energy-constrained diamond norm of the difference between the Gaussian noise channel $\NN_K$ and the identity map $\Id$ can be upper bounded as
\bb
\frac 12 \left\| \NN_K - \Id \right\|_{\diamond N} \leq f(N) \sqrt{\Tr K}\, ,
\ee
where $f(N)\coloneqq \frac{1}{\sqrt{2}} \left(\sqrt{N} + \sqrt{N+1}\right)$.
\end{lemma}

\begin{proof}
It holds that
\bb
\frac12 \left\| \NN_K -\Id\right\|_{\diamond N} &\leqt{(i)} \int \mathrm{d}^{2n}u\ P_K(u)\ \frac12 \left\| \D_u -\Id\right\|_{\diamond N} \\
&\leqt{(ii)} \int \mathrm{d}^{2n}u\ P_K(u)\ \sin\left( \min\left\{ \|u\|_2 f(N),\, \frac{\pi}{2} \right\} \right) \\
&\leq f(N) \int \mathrm{d}^{2n}u\ P_K(u)\ \sqrt{u^\intercal u} \\
&\leqt{(iii)} f(N) \left(\int \mathrm{d}^{2n}u\ P_K(u)\ u^\intercal u \right)^{1/2} \\
&= f(N) \left(\Tr K\right)^{1/2} .
\ee
Here, in (i), we have used the triangle inequality of the energy-constrained diamond norm and we have introduced the displacement channel  $\D_u(\cdot)\coloneqq \hat{D}_u(\cdot)\hat{D}_u^\dagger$. In (ii), we have exploited Lemma~\ref{upp_bound_ludo}. In (iii), we have used the concavity of the square root and Jensen inequality.
\end{proof}
\begin{lemma}[(Expectation value of energy operator)]\label{lemma_energy_traceV}
Let $\rho$ be an $n$-mode state. The expectation value of the energy operator $\hat{E}_n\coloneqq \frac{\hat{\mathbf{R}}^\intercal\hat{\mathbf{R}}}{2}$ can be expressed in terms of the covariance matrix and the first moment as
    \bb
        \Tr\!\left[\rho \hat{E}_n\right]=\frac{\Tr\! V\!(\rho)}{4}+\frac{\|\textbf{m}(\rho)\|_2^2}{2}\,.
    \ee
In particular, in terms of the photon number operator $\hat{N}_n\coloneqq \sum_{i=1}^n a^\dagger_i a_i$, it holds that
    \bb\label{eq_mean_photon_number}
    \Tr\!\left[\rho\, \hat{N}_n\right]=\frac{\Tr[V\!(\rho)-\mathbb{1}]}{4}+\frac{\|\textbf{m}(\rho)\|_2^2}{2}\,.
    \ee
\end{lemma}
\begin{proof}
By definition of covariance matrix, it holds that
\bb
    \Tr\! V\!(\rho)&= \sum_{i=1}^{2n} V_{i,i}(\rho) 
    \\&= 2 \sum_{i=1}^{2n} \left[ \Tr[\hat{R}_i^2\rho]-m_i(\rho)^2\right]
    \\&= 4 \Tr[\rho \hat{E}_n]- 2\|\textbf{m}(\rho)\|_2^2\,.
\ee
Moreover,~\eqref{eq_mean_photon_number} follows by exploiting $\hat{E}_n=\hat{N}_n+\frac{n}{2}\hat{\mathbb{1}}$.
\end{proof}
We are now ready to prove Theorem~\ref{upper_bound_trace_distance_Gaussian}.
\begin{proof}[Proof of Theorem~\ref{upper_bound_trace_distance_Gaussian}]
Let us first prove the upper bound in~\eqref{ludo_upp_bound} under the assumption that the Gaussian states have the same first moments. Without loss of generality, we can clearly assume that their first moments vanish, as the trace distance is invariant under (displacement) unitary transformations. Given two Gaussian states $\rho$ and $\sigma$ with zero first moment, let $V$ be the covariance matrix of $\rho$, and let $W$ be the covariance matrix of $\sigma$. In addition, set
\bb
T\coloneqq \frac12 \left(V+W + |V-W|\right),
\ee
where we recall that $|V-W|\coloneqq\sqrt{(V-W)^{\dagger}(V-W)}$. Since $T-V\ge0$ and $T-W\ge0$, we can consider the Gaussian noise channels $\NN_{T-V}$ and $\NN_{T-W}$, as defined in~\eqref{Gaussian_noise}. Since $\rho$ and $\sigma$ are Gaussian states and since $\NN_{T-V}$ and $\NN_{T-W}$ are Gaussian channels (thanks to Lemma~\ref{N_k_is_gaussian}), it follows that $\NN_{T-V}(\rho)$ and $\NN_{T-W}(\sigma)$ are Gaussian states. Moreover, by exploiting Lemma~\ref{Lemma_ec_diamond}, both the covariance matrix of $\NN_{T-V}(\rho)$ and that of $\NN_{T-W}(\sigma)$ are equal to $T$. Additionally, both the first moment of $\NN_{T-V}(\rho)$ and that of $\NN_{T-W}(\sigma)$ are equal to the zero vector. Hence, since the Gaussian states $\NN_{T-V}(\rho)$ and $\NN_{T-W}(\sigma)$ have identical first moments and covariance, they are actually the same state: $\NN_{T-V}(\rho)=\NN_{T-W}(\sigma)$. Then, we can write
\bb
\frac12 \left\|\rho - \sigma\right\|_1 &\le\frac12 \left\|\rho - \NN_{T-V}(\rho)+\NN_{T-W}(\sigma)-\sigma\right\|_1\\
&\leq \frac12 \left\|\rho - \NN_{T-V}(\rho)\right\|_1 + \frac12 \left\|\NN_{T-W}(\sigma) - \sigma\right\|_1 \\
&= \frac12 \left\|( \NN_{T-V} - \Id)(\rho)\right\|_1 + \frac12 \left\|(\NN_{T-W} - \Id)(\sigma)\right\|_1 \\
&\leq \frac12 \left\| \NN_{T-V} - \Id\right\|_{\diamond N} + \frac12 \left\|\NN_{T-W} - \Id\right\|_{\diamond N} \\
&\leqt{(i)} f(N) \left( \sqrt{\Tr [T-V]} + \sqrt{\Tr [T-W]} \right) \\
&\leqt{(ii)} 2\,f(N) \sqrt{\frac12 \Tr [T-V] + \frac12 \Tr [T-W]} \\
&= \sqrt{2}\,f(N) \sqrt{\Tr|V-W|}\\ 
&= \sqrt{2}\,f(N)\, \sqrt{\|V-W\|_1}\, ,
\label{eq:mom0proof}
\ee
where in (i) we have used  Lemma~\ref{lemma_upp_diamond_norm33}, and in (ii) we have used  the concavity of the square root. This proves the claim under the assumption that the Gaussian states $\rho$ and $\sigma$ have zero first moment.

For the general case, we can denote by $\delta\coloneqq\textbf{m}(\rho)-\textbf{m}(\sigma)$ the difference between the first moments of $\rho$ and $\sigma$. Then
\bb
\frac12 \left\|\rho - \sigma\right\|_1& \eqt{(iii)} \frac12 \left\|\hat{D}_{-\textbf{m}(\sigma) }\rho \hat{D}_{-\textbf{m}(\sigma) }^\dagger - \hat{D}_{-\textbf{m}(\sigma) }\sigma \hat{D}_{-\textbf{m}(\sigma) }^\dagger\right\|_1 \\ 
&= \frac12 \left\|\hat{D}_\delta \,\hat{D}_{-\textbf{m}(\rho) }\rho \hat{D}_{-\textbf{m}(\rho) }^\dagger\, \hat{D}_\delta^\dagger - \hat{D}_{-\textbf{m}(\sigma) }\sigma \hat{D}_{-\textbf{m}(\sigma) }^\dagger\right\|_1 \\
&\leq \frac12 \left\|\hat{D}_\delta \,\hat{D}_{-\textbf{m}(\rho) }\rho \hat{D}_{-\textbf{m}(\rho) }^\dagger\, \hat{D}_\delta^\dagger - \hat{D}_{-\textbf{m}(\rho) }\rho \hat{D}_{-\textbf{m}(\rho) }^\dagger \right\|_1 + \frac12 \left\| \hat{D}_{-\textbf{m}(\rho) }\rho \hat{D}_{-\textbf{m}(\rho) }^\dagger - \hat{D}_{-\textbf{m}(\sigma) }\sigma \hat{D}_{-\textbf{m}(\sigma) }^\dagger\right\|_1 \\
&\leqt{(iv)} \frac12 \left\|\D_\delta - \Id\right\|_{\diamond N} + \frac12 \left\| \hat{D}_{-\textbf{m}(\rho) }\rho \hat{D}_{-\textbf{m}(\rho) }^\dagger - \hat{D}_{-\textbf{m}(\sigma) }\sigma \hat{D}_{-\textbf{m}(\sigma) }^\dagger\right\|_1 \\
&\leqt{(v)} \frac12 \left\|\D_\delta - \Id\right\|_{\diamond N} + \sqrt{2}\,f(N)\, \sqrt{\|V\!(\rho)-V\!(\sigma)\|_1} \\
&\leqt{(vi)} f(N) \left( \|\delta\|_2 + \sqrt2\, \sqrt{\|V\!(\rho)-V\!(\sigma)\|_1}\right)\,.
\ee
Here, in (iii), we have used  the unitary invariance of the trace norm. In (iv), we exploited the definition of energy-constrained diamond norm in Eq.~\eqref{ednorm}, together with the fact that 
\bb
    \Tr\left[\hat{N}_n\hat{D}_{-\textbf{m}(\rho)}\rho\hat{D}_{-\textbf{m}(\rho)}^\dagger\right]= \frac{\Tr[V\!(\rho)-\mathbb{1}]}{4}\le \frac{\Tr[V\!(\rho)-\mathbb{1}]}{4}+\frac{\|\textbf{m}(\rho)\|_2^2}{2}= \Tr\left[\hat{N}_n\rho\right]\le N\,,
\ee
where we exploited Lemma~\ref{lemma_energy_traceV}. In (v), we leveraged~\eqref{eq:mom0proof}, as the two Gaussian states $\hat{D}_{-\textbf{m}(\rho) }\rho \hat{D}_{-\textbf{m}(\rho) }^\dagger$ and $\hat{D}_{-\textbf{m}(\sigma) }\sigma \hat{D}_{-\textbf{m}(\sigma) }^\dagger$ have zero first moment.
Finally, in (vi) we exploited Lemma~\ref{upp_bound_ludo}, together with the inequality $\sin (x)\le x$ for all $x\ge0$.
\end{proof}
\subsubsection{Upper bound in the pure-state setting}\label{subsubsection_upper_bound_pure_state}
Theorem~\ref{upper_bound_trace_distance_Gaussian} provides an upper bound on the trace distance between two Gaussian states $\rho$ and $\sigma$, which are possibly mixed. We can obtain a tighter bound if we assume that one of the two states is pure, as proved in the following theorem. 
\begin{thm}\label{inequality_dist_gauss}
    Let $\psi\coloneqq \ketbra{\psi}$ be a pure $n$-mode Gaussian state with first moment $\textbf{m}(\psi)$ and second moment $V\!(\psi)$.  Let $\rho$ be an $n$-mode possibly non-Gaussian state with first moment $\textbf{m}(\rho)$ and second moment $V\!(\rho)$.  Then
    \bb
        \left\|\, \rho-\ketbra{\psi}\, \right\|_1\le  \sqrt{\Tr V\!(\psi)}\sqrt{ \|V\!(\rho)-V\!(\psi)\|_\infty+2\|\textbf{m}(\rho)-\textbf{m}(\psi)\|_2^2}
    \ee
    In particular, if the mean energy of $\psi$ is upper bounded by a constant $E$, that is,
    \bb
    \Tr\!\left[\rho \hat{E}_n\right] \le E
    \ee
    where $\hat{E}_n\coloneqq \frac{\hat{\mathbf{R}}^\intercal\hat{\mathbf{R}}}{2}$, then
    \bb
        \frac{1}{2}\left\|\, \rho-\ketbra{\psi}\, \right\|_1\le  \sqrt{E}\sqrt{ \|V\!(\rho)-V\!(\psi)\|_\infty+2\|\textbf{m}(\rho)-\textbf{m}(\psi)\|_2^2}
    \ee
\end{thm}
Before proving Theorem~\ref{inequality_dist_gauss}, let us show some useful lemmas.
\begin{lemma}\label{lemma_inverse_cov_pure}
    Let $\psi$ be an $n$-mode pure Gaussian state. The inverse of the covariance matrix reads
    \bb\label{eq_pure_state}
        V\!(\psi)^{-1}=\Omega_n V\!(\psi)\Omega_n^\intercal\,,
    \ee
    where $\Omega_n$ is the symplectic form defined in~\eqref{symplectic_form}.
\end{lemma}
\begin{proof}
        The covariance matrix $V\!(\psi)$ of the pure Gaussian state $\psi$ can be written as $V\!(\psi)=SS^\intercal$ with $S\in \mathrm{Sp}(2n)$ being a symplectic matrix (due to~\eqref{eq:GaussianState}). Consequently, it holds that
    \bb
        \Omega_n^\intercal V\!(\psi)^{-1}\Omega_n=-\Omega_n V\!(\psi)^{-1}\Omega_n=-\Omega_n(S^\intercal)^{-1}S^{-1}\Omega_n=-S\Omega_n^2 S^{\intercal}=SS^\intercal=V\!(\psi)\,,
    \ee
    where we have used that $S^{-1}\in \mathrm{Sp}(2n)$.
\end{proof}
\begin{lemma}\label{lemma_S_inv}
    For any symplectic $S\in \mathrm{Sp}(2n)$ and for any unitarily invariant norm $\| \cdot\|$, it holds that 
    \bb
    \|S\|=\|S^{-1}\|\,.
    \ee
\end{lemma}
\begin{proof}
    By exploiting that $S^{-1}\in \mathrm{Sp}(2n)$, we have that 
    \bb
        S^{-1}=-S^{-1}\Omega_n\Omega_n=-\Omega_nS^\intercal \Omega_n=\Omega_n S^\intercal \Omega_n^\intercal\,.
    \ee
    In particular, since $\Omega_n$ is orthogonal, it holds that $\|S^{-1}\|=\|S^\intercal\|=\|S\|$.
\end{proof}
We are now ready to prove Theorem~\ref{inequality_dist_gauss}.
\begin{proof}[Proof of Theorem~\ref{inequality_dist_gauss}]
    Since $\psi$ is a pure Gaussian state, the symplectic eigenvalues of its covariance matrix $V\!(\psi)$ are all equal to $1$. Consequently, there exist a symplectic matrix $S$ such that
    \bb
        V\!(\psi)&=S S^\intercal\,,\\
        \ket{\psi}&=\hat{D}_{\textbf{m}(\psi)}U_{S}\ket{0}\,,
    \ee
    where we have denoted as $\ket{0}$ the $n$-mode vacuum state vector.
    Hence, it holds that
    \bb
        \bra{\psi}\rho\ket{\psi}=\bra{0}U_{S}^\dagger \hat{D}_{\textbf{m}(\psi)}^\dagger \rho \hat{D}_{\textbf{m}(\psi)}U_{S}\ket{0}=\bra{0}\omega\ket{0}\,,
    \ee
    where we have introduced the state
    \bb
        \omega\coloneqq U_{S}^\dagger \hat{D}_{\textbf{m}(\psi)}^\dagger \rho \hat{D}_{\textbf{m}(\psi)}U_{S}\,.
    \ee
    By exploiting the operator inequality
    \bb
        \ketbra{0}\ge \mathbb{1}-\sum_{i=1}^n a^\dagger_ia_i\,,
    \ee
    we find that
    \bb\label{step_proof_ineq}
        \bra{\psi}\rho\ket{\psi}&=
        \bra{0}\omega\ket{0}
        \\&=\Tr[\omega \ketbra{0}]\\
        &\ge 1-\Tr\!\left[\omega\sum_{i=1}^n a^\dagger_i a_i\right]\\
        &=1-\frac{\Tr[V\!(\omega)-\mathbb{1}]}{4}-\frac{\|\textbf{m}(\omega)\|_2^2}{2}\,,
    \ee
    where we also used Lemma~\ref{lemma_energy_traceV}.
    Note that the first moment and covariance matrix of $\omega$ are given by
    \bb
        \textbf{m}(\omega)&= S^{-1}\left(\textbf{m}(\rho)-\textbf{m}(\psi)\right)\,,\\
        V\!(\omega)&= S^{-1}V\!(\rho) (S^{-1})^\intercal\,.
    \ee
    Consequently, it holds that
    \bb
        \left|\Tr[V\!(\omega)-\mathbb{1}]\right|&= \left|\Tr[S^{-1} V\!(\rho) (S^{-1})^\intercal-\mathbb{1}]\right|\\
        &= \left|\Tr[S^{-1}(V\!(\rho)-V\!(\psi))(S^{-1})^\intercal]\right|\\
        &\leqt{(i)} \|(S^{-1})^\intercal S^{-1}\|_1\,\|V\!(\rho)-V\!(\psi)\|_\infty\\
        &\texteq{(ii)}\Tr\!\left[(S^{-1})^\intercal S^{-1}\right] \|V\!(\rho)-V\!(\psi)\|_\infty\\
        &=\Tr[V\!(\psi)^{-1}]\,\|V\!(\rho)-V\!(\psi)\|_\infty\\
        &\texteq{(iii)}\Tr V\!(\psi)\,\|V\!(\rho)-V\!(\psi)\|_\infty\,.
    \ee
     Here, in (i), we have used H\"older inequality $|\Tr[AB]|\le\|A\|_1\|B\|_\infty$. In (ii), we have exploited that $(S^{-1})^\intercal S^{-1}$ is positive semi-definite. Moreover, in (iii) we utilised Lemma~\ref{lemma_inverse_cov_pure}, which asserts that the inverse of the covariance matrix $V\!(\psi)$ of the pure Gaussian state $\psi$ can be written as $V\!(\psi)^{-1}=-\Omega_n V\!(\psi)\Omega_n$, and, in particular, this implies $\Tr V\!(\psi)=\Tr[V\!(\psi)^{-1}]$. In addition, it holds that
    \bb
      \|\textbf{m}(\omega)\|_2&=\|S^{-1}\left(\textbf{m}(\rho)-\textbf{m}(\psi)\right)\|_2
      \\&\le \|S^{-1}\|_\infty\,\|\textbf{m}(\rho)-\textbf{m}(\psi)\|_2
      \\&\texteq{(iv)} \|S \|_\infty\,\|\textbf{m}(\rho)-\textbf{m}(\psi)\|_2
      \\&\le \sqrt{\Tr[SS^\intercal]}\,\|\textbf{m}(\rho)-\textbf{m}(\psi)\|_2
      \\&=\sqrt{\Tr V\!(\psi)}\,\|\textbf{m}(\rho)-\textbf{m}(\psi)\|_2\,,
    \ee
    where in (iv) we have exploited Lemma~\ref{lemma_S_inv}. Consequently,~\eqref{step_proof_ineq} implies that
    \bb
        \bra{\psi}\rho\ket{\psi}\ge 1-\frac{\Tr V\!(\psi)}{4}\left(  \|V\!(\rho)-V\!(\psi)\|_\infty+2\|\textbf{m}(\rho)-\textbf{m}(\psi)\|_2^2    \right)
    \ee
    Hence, Fuchs-van de Graaf inequality implies that
    \bb
        \| \rho-\psi  \|_1&\le2\sqrt{1-\bra{\psi}\rho\ket{\psi}}\\
        &\le \sqrt{\Tr V\!(\psi)}\sqrt{ \|V\!(\rho)-V\!(\psi)\|_\infty+2\|\textbf{m}(\rho)-\textbf{m}(\psi)\|_2^2}\,. 
    \ee
    In particular, by exploiting Lemma~\ref{lemma_energy_traceV} and the assumption that the mean energy of $\psi$ is upper bounded by $E$, we conclude that
    \bb
        \Tr V\!(\psi)&= 4\Tr[ \hat{E}_n \psi]-2\|\textbf{m}(\psi)\|_2^2\le 4E\,.
    \ee
\end{proof}

\subsubsection{Lower bound}\label{subsubsection_lower_bound}
In this subsection we obtain a lower bound on the trace distance between Gaussian states in terms of the norm distance between their covariance matrices and their first moments.
\begin{thm}\label{thm_trace_distance_lower_bound}
    Let $\rho_1$ and $\rho_2$ be $n$-mode Gaussian states with mean energy upper bounded by $E$, i.e., 
\bb\label{EC_assumption}
\Tr\left[\rho_1\hat{E}_{n}\right] &\le E\,,\\
\Tr\left[\rho_2\hat{E}_{n} \right] &\le E\,.
\ee
Then, the trace distance between $\rho_1$ and $\rho_2$ can be lower bounded in terms of the norm distance between their first moments as
\bb
    \frac{1}{2}\|\rho_1-\rho_2\|_1\ge \frac{1}{200}\min\!\left\{1,\frac{\| \textbf{m}(\rho_1)-\textbf{m}(\rho_2) \|_2}{\sqrt{4E+1}}\right\} \,,
\ee
and in terms of the norm distance between their covariance matrices as
\bb
    \frac{1}{2}\|\rho_1-\rho_2\|_1\ge \frac{1}{200}\min\!\left\{1, \frac{\|V(\rho_2)-V(\rho_1)\|_2}{ 4E+1}\right\}\,.
\ee
\end{thm}

\begin{proof}
    Recall that the Husimi function $Q_\rho:\mathbb{R}^{2n}\longmapsto\mathbf{R}$ of an $n$-mode quantum state $\rho$ is defined as~\cite{BUCCO}
    \bb
        Q_\rho(\textbf{r})\coloneqq \frac{1}{(2\pi)^n}\bra{\textbf{r}}\rho\ket{\textbf{r}}\,,
    \ee
    where $\ket{\textbf{r}}\coloneqq \hat{D}_\textbf{r}\ket{0}$ is a coherent state. The Husimi function $Q_\rho$ is the probability distribution of the outcome $\textbf{r}$ of the heterodyne measurement performed on $\rho$~\cite{BUCCO}. Hence, by exploiting the monotonicity of the trace norm under quantum channels, the trace distance between $\rho_1$ and $\rho_2$ can be lower bounded in terms of the TV distance between their Husimi functions as
    \bb\label{ineq_trace_distance_husimi}
        \frac{1}{2}\|\rho_1-\rho_2\|_1&\ge \frac{1}{2}\int_{\textbf{r}\in\mathbb{R}^{2n}}\mathrm{d}^{2n}\textbf{r}\,|Q_{\rho_1}(\textbf{r})-Q_{\rho_2}(\textbf{r})|=d_{\text{TV}}\left(Q_{\rho_1},Q_{\rho_2}\right)\,,
    \ee
    where in the last equality we have used  the definition of TV distance. The Husimi function of a Gaussian state $\rho$ is a Gaussian probability distribution with first moment $\textbf{m}(\rho)$ and covariance matrix $\frac{V(\rho)+\mathbb{1}}{2}$ (see~\eqref{husimi_gaussian_state} in the preliminaries section):
    \bb\label{Husimi_rho}
        Q_\rho(\textbf{r})=\NN\!\left[\textbf{m}(\rho),\frac{V(\rho)+\mathbb{1}}{2}\right]\!(\textbf{r})\,,
    \ee
    where $\NN$ denotes a Gaussian probability distribution defined as follows. Given a real vector $\textbf{m}$ and a positive matrix $V$, the Gaussian probability distribution $\NN[\textbf{m},V]$ with first moment $\textbf{m}$ and covariance matrix $V$ is defined as
\bb
    \NN[\textbf{m},V](\textbf{r})&\coloneqq\frac{e^{-\frac12 (\textbf{r}-\textbf{m})^\intercal V^{-1}  (\textbf{r}-\textbf{m})}}{(2\pi)^{n}\sqrt{\det V }}\,.
\ee
Consequently,~\eqref{ineq_trace_distance_husimi} implies that
\bb\label{lower_1_trace_dist}
    \frac{1}{2}\|\rho_1-\rho_2\|_1\ge d_{\text{TV}}\left(\NN\!\left[\textbf{m}(\rho_1),\frac{V(\rho_1)+\mathbb{1}}{2}\right]\,,\,\NN\!\left[\textbf{m}(\rho_2),\frac{V(\rho_2)+\mathbb{1}}{2}\right]\right)\,.
\ee
Now, let us leverage recent results about tight bounds on the TV distance between two arbitrary Gaussian probability distributions~\cite{devroye2023total}. By exploiting \cite[Theorem~1.1]{devroye2023total}, together with~\cite[Lemma~B.3]{arbas2023polynomial}, the TV distance between two Gaussian probability distributions $\NN[\textbf{m}_1,V_1]$ and $\NN[\textbf{m}_2,V_2]$ can be lower bounded in terms of the norm distance between covariance matrices as
\bb\label{lower_bound_normal_prob}
    d_{\text{TV}}\left(\NN[\textbf{m}_1,V_1],\NN[\textbf{m}_2,V_2]\right)\ge \frac{\min\!\left\{1,\left\|V_1^{-1/2}V_2V_1^{-1/2}-\mathbb{1}\right\|_2\right\}}{200}\,.
\ee
Note that we can lower bound the term $\left\|V_1^{-1/2}V_2V_1^{-1/2}-\mathbb{1}\right\|_2$ in terms of $\|V_2-V_1\|_2$ as
\bb
    \left\|V_1^{-1/2}V_2V_1^{-1/2}-\mathbb{1}\right\|_2&=\left\|V_1^{-1/2}\left(V_2-V_1\right)V_1^{-1/2}\right\|_2\\
    &=\sqrt{\Tr\left[V_1^{-1}(V_2-V_1)V_1^{-1}(V_2-V_1)\right]}\\
    &\ge \sqrt{\Tr\left[V_1^{-1}(V_2-V_1)\frac{\mathbb{1}}{\|V_1\|_\infty}(V_2-V_1)\right]}\\
    &\ge\frac{\sqrt{\Tr[(V_2-V_1)^2]}}{\|V_1\|_\infty}\\
    &=\frac{\|V_2-V_1\|_2}{\|V_1\|_\infty}\,.
\ee
Consequently, we obtain the following lower bound on the TV distance between two Gaussian probability distributions 
\bb\label{lower_bound_normal_prob2}                 
    d_{\text{TV}}\left(\NN[\textbf{m}_1,V_1],\NN[\textbf{m}_2,V_2]\right)\ge \frac{1}{200}\min\!\left\{1,\frac{\|V_2-V_1\|_2}{\|V_1\|_\infty}\right\}\,.
\ee
In addition, thanks to Ref.~\cite[Theorem~1.1]{devroye2023total}, the TV distance between two Gaussian probability distributions can also be lower bounded as
\bb\label{lower_bound_normal_prob3}                 
    d_{\text{TV}}\left(\NN[\textbf{m}_1,V_1],\NN[\textbf{m}_2,V_2]\right)\ge \frac{1}{200}\min\!\left\{1,\frac{\| \textbf{m}_1-\textbf{m}_2 \|_2^2}{\sqrt{\left(\textbf{m}_1-\textbf{m}_2\right)^\intercal V_1\left(\textbf{m}_1-\textbf{m}_2\right)}}\right\}\,,
\ee
which implies that
\bb\label{lower_bound_normal_prob4}                 
    d_{\text{TV}}\left(\NN[\textbf{m}_1,V_1],\NN[\textbf{m}_2,V_2]\right)\ge \frac{1}{200}\min\!\left\{1,\frac{\| \textbf{m}_1-\textbf{m}_2 \|_2}{\sqrt{ \|V_1\|_\infty}}\right\}\,.
\ee
Consequently, by using~\eqref{lower_1_trace_dist}, together with~\eqref{lower_bound_normal_prob2} and~\eqref{lower_bound_normal_prob4}, we obtain that
\bb
    \frac{1}{2}\|\rho_1-\rho_2\|_1&\ge \frac{1}{200}\min\!\left\{1,\frac{\| \textbf{m}(\rho_1)-\textbf{m}(\rho_2) \|_2}{\sqrt{ \|V(\rho_1)+\mathbb{1}\|_\infty}}\right\}\,,\\
    \frac{1}{2}\|\rho_1-\rho_2\|_1&\ge \frac{1}{200}\min\!\left\{1,\frac{\|V(\rho_2)-V(\rho_1)\|_2}{\|V(\rho_1)+\mathbb{1}\|_\infty}\right\}\,.
\ee
Hence, since any $n$-mode state $\rho$ satisfies $\| V(\rho )\|_\infty\le 4\Tr[\rho\hat{E}_n]$, the assumption on the energy constraints in~\eqref{EC_assumption} implies that
\bb
    \frac{1}{2}\|\rho_1-\rho_2\|_1&\ge \frac{1}{200}\min\!\left\{1,\frac{\| \textbf{m}(\rho_1)-\textbf{m}(\rho_2) \|_2}{\sqrt{4E+1}}\right\}\,,\\
    \frac{1}{2}\|\rho_1-\rho_2\|_1&\ge \frac{1}{200}\min\!\left\{1,\frac{\|V(\rho_2)-V(\rho_1)\|_2}{4E+1}\right\}\,.
\ee
\end{proof}


\subsection{Learning first moments and covariance matrices}\label{Subsection_useful_relations}
In this subsection, we conduct the sample complexity analysis of an algorithm for estimating the first moment and covariance matrix of a (possibly non-Gaussian) quantum state. Recall that the covariance matrix $V(\rho)$ is defined by
\begin{equation}
V(\rho) = \Tr[\{\hat{\textbf{R}}, \hat{\textbf{R}}^\intercal\} \rho] - 2\textbf{m}(\rho)\textbf{m}(\rho)^\intercal,
\end{equation}
where $\textbf{m}(\rho) \coloneqq \Tr[\rho \hat{\textbf{R}}]$ represents the first moment vector, and $\hat{\textbf{R}} = (\hat{x}_1, \hat{p}_1, \hat{x}_2, \hat{p}_2, \ldots, \hat{x}_n, \hat{p}_n)$ denotes the quadrature vector. Additionally, the covariance matrix of any state must satisfy the uncertainty relation:
\begin{equation}
V(\rho) + i \Omega_n \geq 0.
\end{equation}
A direct method for estimating the covariance matrix involves estimating the first moment and the expectation values of all observables $(\{\hat{R}_i, \hat{R}_j\})_{i,j \in [2n]}$. However, there is a more efficient algorithm, as detailed by Aolita et al.~\cite{Aolita_2015}, which requires fewer copies of $\rho$. We adopt this approach and provide a sample complexity analysis that extends beyond what was outlined in Ref.~\cite{Aolita_2015}.

The algorithm's strategy involves grouping commuting observables whose expectation values form entries of the covariance matrix and performing joint measurements on these groups of commuting observables.
This estimation procedure yields a symmetric matrix $\tilde{V} \in \mathbb{R}^{2n \times 2n}$ as an approximation of $V(\rho)$. However, the estimated matrix $\tilde{V}$ may not satisfy the uncertainty relation $\tilde{V} + i \Omega_n \geq 0$, implying that there might not exist a quantum state $\tilde{\rho}$ with covariance matrix $\tilde{V}$. We address this issue by showing that slightly perturbing $\tilde{V}$ suffices to obtain a valid covariance matrix. Ensuring that the estimated matrix is a valid covariance matrix is crucial, for example, when we need to consider a Gaussian state with covariance matrix $\tilde{V}$ (as we will do in Subsection~\ref{subs_sec_tom_gauss}), or when applying the Williamson decomposition to $\tilde{V}$ (e.g.,~we will need this in Section~\ref{Sec_t_doped}).

The algorithm's detailed steps are provided in Table~\ref{Table_covariance}. Notably, this algorithm exclusively relies on homodyne measurements, which are experimentally feasible in photonic platforms.
In the forthcoming Theorem~\ref{correctness_algorithm_cov}, we establish the algorithm's correctness and demonstrate that its sample complexity is $O\!\left(\log\left(\frac{n^2}{\delta}\right) \frac{n^3 E^2}{\varepsilon^2}\right)$.

\begin{table}[t]
  \caption{Estimation algorithm for the first moment $\textbf{m}(\rho)$ and covariance matrix $V\!(\rho)$ of an arbitrary $n$-mode quantum state $\rho$ satisfying a second moment constraint.}
  \begin{mdframed}[linewidth=2pt, roundcorner=10pt, backgroundcolor=white!10, innerbottommargin=10pt, innertopmargin=10pt]
    \textbf{Input:} Accuracy $\varepsilon$, failure probability $\delta$, second moment constraint $E$, $N$ copies of the unknown $n$-mode quantum state $\rho$ satisfying the moment constraint $\sqrt{\Tr\!\left[\hat{E}_n^2 \rho \right]}\le n E$, where
    \bb
        N&\coloneqq\ (n+3)\ceil{  68 \log\!\left(\frac{2 (2n^2+3n)  }{\delta}\right) \frac{200(8n^2 E^2+3n)}{\varepsilon^2}}=O\!\left( \log\!\left(\frac{n^2}{\delta}\right)\frac{n^3E^2}{\varepsilon^2}   \right)\,.
    \ee
    \textbf{Output:} A vector $\tilde{\textbf{m}}\in\mathbb{R}^{2n}$ and a symmetric matrix $\tilde{V}'\in\mathbb{R}^{2n,2n}$ such that 
    \bb
        \Pr\left(  \|\tilde{V}'-V\!(\rho)\|_\infty\le \varepsilon\quad \text{ and }\quad \|\tilde{\textbf{m}}-\textbf{m}(\rho)\|_2\le \frac{\varepsilon}{10\sqrt{8En}}\right)\ge 1-\delta\,.
    \ee
    \begin{algorithmic}[1]
    \State Set $N'\coloneqq N/(n+3)$.
    \State Query $N'$ copies of $\rho$ and, for each, perform a joint measurement of the position observables $\hat{x}_1,\hat{x}_2,\ldots,\hat{x}_n$. Then, construct median-of-means estimators (Lemma~\ref{median_of_means}) of the expectation values of $(\{\hat{x}_i,\hat{x}_j\})_{i\le j}$ and of $(\hat{x}_i)_{i\in[n]}$.\label{step1}
    \State Query $N'$ copies of $\rho$ and, for each, perform a joint measurement of the momentum observables $\hat{p}_1,\hat{p}_2,\ldots,\hat{p}_n$. Then, construct median-of-means estimators of the expectation values of $(\{\hat{p}_i,\hat{p}_j\})_{i\le j}$ and of $(\hat{p}_i)_{i\in[n]}$.\label{step2}
    \State Query $N'$ copies of $\rho$ and, for each, perform a joint measurement of $\{\hat{x}_1,\hat{p}_1\}, \{\hat{x}_2,\hat{p}_2\},\ldots, \{\hat{x}_n,\hat{p}_n\}$. Then, construct median-of-means estimators of the expectation value of $\left( \{\hat{x}_i,\hat{p}_i\} \right)_{i\in[n]}$.\label{step4}

      \For{$k \leftarrow 1$ \textbf{to} $n$}
         \State Query $N'$ copies of $\rho$ and, for each, perform a joint measurement of $(\hat{x}_i)_{i\ne k}$ and $\hat{p}_k$. Then, construct median-of-means estimators of the expectation values of $(\{\hat{p}_k,\hat{x}_i\})_{i\ne k}$.\label{step7}
      \EndFor
     \State Combine all the aforementioned estimates to form the estimator $\tilde{\textbf{m}}\in\mathbb{R}^{2n}$ for the first moment $\textbf{m}(\rho)$, and the estimator $\tilde{W}\in\mathbb{R}^{2n,2n}$ for the matrix $\Tr[\rho\{\hat{\mathbf{R}},\hat{\mathbf{R}}^\intercal\}]$.
     \State Set $\tilde{V}\coloneqq \tilde{W}-2\tilde{\textbf{m}}\tilde{\textbf{m}}^\intercal$.
    \State Set $\tilde{V}'\coloneqq \tilde{V}+\frac{\varepsilon}{2}\mathbb{1}$.
    \If{ $\tilde{V}'+i\Omega_n$ is positive semi-definite }
              \State \Return the estimator $\tilde{\textbf{m}}$ for the first moment and the estimator $\tilde{V}'$ for the covariance matrix.
          \Else
              \State Declare \emph{failure} and abort.
          \EndIf
    \end{algorithmic}
  \end{mdframed}
  \label{Table_covariance}
\end{table}
\begin{thm}\label{correctness_algorithm_cov}
Let $\varepsilon,\delta\in(0,1)$ and $E\geq 0$. Let $\rho$ be an $n$-mode quantum state satisfying the second moment constraint $\sqrt{\Tr\!\left[\hat{E}_n^2 \rho \right]}\le n E$. Then, $N$ copies of $\rho$, satisfying
\bb\label{number_of_copies_cov}
N&\coloneqq (n+3)\ceil{  68 \log\!\left(\frac{2 (2n^2+3n)  }{\delta}\right) \frac{200(8n^2 E^2+3n)}{\varepsilon^2}}\\
&=O\!\left( \log\!\left(\frac{n^2}{\delta}\right)\frac{n^3E^2}{\varepsilon^2}   \right)\,,
\ee
are sufficient to build a vector $\tilde{\textbf{m}}\in\mathbb{R}^{2n}$ and a symmetric matrix $\tilde{V}'\in\mathbb{R}^{2n,2n}$, such that 
    \bb
        \Pr\left(  \|\tilde{V}'-V\!(\rho)\|_\infty\le \varepsilon\quad \text{ and }\quad \tilde{V}'+i\Omega_n \ge 0 \quad \text{ and }\quad \|\tilde{\textbf{m}}-\textbf{m}(\rho)\|_2\le \frac{\varepsilon}{10\sqrt{8En}}\right)\ge 1-\delta\,
    \ee
\end{thm}
Before proving Theorem~\ref{correctness_algorithm_cov}, let us show the following lemma.
\begin{lemma}\label{lemma_formula_rr2}
For all $n\in\N$, it holds that
\bb
    \sum_{j,k=1}^{2n} \{\hat{R}_j,\hat{R}_k    \} ^2=16\hat{E}_n^2+6n\hat{\mathbb{1}}\,,
\ee
where $\hat{E}_n\coloneqq \frac{\hat{\mathbf{R}}^\intercal\hat{\mathbf{R}}}{2}$ is the energy operator.
\end{lemma}
\begin{proof}
Note that 
\bb
\sum_{j,k=1}^{2n} \{\hat{R}_j,\hat{R}_k    \} ^2&=\sum_{j,k=1}^{2n}\left(\hat{R}_j\hat{R}_k+\hat{R}_k\hat{R}_j\right)^2 \\
    &=2\sum_{j,k=1}^{2n}\left( \hat{R}_j\hat{R}_k\hat{R}_j\hat{R}_k+\hat{R}_j\hat{R}_k^2\hat{R}_j\right)\\
    &\texteq{(i)} 2\sum_{j,k=1}^{2n}\left( 2\hat{R}_j^2\hat{R}_k^2-3i\Omega_{jk}\hat{R}_j\hat{R}_k\right)\\
    &\texteq{(ii)}4\sum_{j,k=1}^{2n}\hat{R}_j^2\hat{R}_k^2+6n\hat{\mathbb{1}}\\
    &=4\left(\mathbf{\hat{R}}^\intercal\mathbf{\hat{R}}\right)^2+6n\hat{\mathbb{1}}\\
    &=16\hat{E}_n^2+6n\hat{\mathbb{1}}\,.
\ee
Here, in $(i)$, we have used three times that $\hat{R}_k\hat{R}_j=\hat{R}_j\hat{R}_k-i\Omega_{jk}\hat{\mathbb{1}}$ and, in (ii), we have used that $\sum_{j,k=1}^{2n}\Omega_{jk}\hat{R}_j\hat{R}_k=in\hat{\mathbb{1}}$.
\end{proof}
\begin{proof}[Proof of Theorem~\ref{correctness_algorithm_cov}]
We aim to establish the correctness of the algorithm presented in Table~\ref{Table_covariance}. In Line~\ref{step1}, we estimate the expectation value of $\{\hat{x}_i,\hat{x}_j\}$ for each $ i\le j \in [n]$ and the expectation value of $\hat{x}_i$ for each $i\in[n]$, amounting to $\frac{n(n+1)}{2}+n$ quantities. In Line~\ref{step2}, we estimate the expectation value of $\{\hat{p}_i,\hat{p}_j\}$ for each $i\le j \in [n]$ and the expectation value of $\hat{p}_i$ for each $i\in[n]$, which are also $\frac{n(n+1)}{2}+n$ quantities. In Line~\ref{step4}, we estimate the expectation value of $\{ \hat{x}_i, \hat{p}_i \}$ for each $i\in[n]$, constituting $n$ quantities. Finally, in all iterations of Line~\ref{step7}, we estimate the expectation value of $\{\hat{p}_k,\hat{x}_i\} $ for each $i\ne k \in [n]$, resulting in $n(n-1)$ quantities. Consequently, across Lines~\ref{step1}--\ref{step7}, we estimate a total of $2n^2+3n$ quantities. By combining all these estimates, we can construct an estimator $\tilde{\textbf{m}}\in\mathbb{R}^{2n}$ for the first moment $\textbf{m}(\rho)$, and an estimator $\tilde{W}\in\mathbb{R}^{2n,2n}$ for the matrix $\Tr[\rho\{\hat{\mathbf{R}},\hat{\mathbf{R}}^\intercal\}]$. Therefore, by making use of Lemma~\ref{median_of_means} together with an union bound, for all $\varepsilon'\in(0,1)$ we can choose
    \bb
        N' \ge 68 \log\!\left(\frac{2 (2n^2+3n)  }{\delta}\right) \frac{16 E^2+6/n}{\varepsilon'^2}
    \ee
    to ensure that with probability at least $1-\delta$ it holds that
    \bb\label{condition_W_ij}
        \left|\tilde{W}_{i,j}-\Tr[\rho\{\hat{R}_i,\hat{R}_j\}]\right|\le \varepsilon'\sqrt{\frac{\Tr[\rho\{\hat{R}_i,\hat{R}_j\}^2]}{16 E^2+6/n}}\qquad\forall\,i,j\in[2n]\,,
    \ee
    and
    \bb\label{condition_m_i}
        \left|\tilde{m}_i-m_i(\rho)\right|\le \varepsilon'\sqrt{\frac{ \Tr[\rho \hat{R}_i^2] }{16 E^2+6/n}}\qquad\forall\,i\in[2n]\,.
    \ee
    Here, we have used  the fact that the variance of the random variable considered to construct the estimator $\tilde{W}_{i,j}$ can be upper bounded by $\Tr[\rho\{\hat{R}_i,\hat{R}_j\}^2]$ for each $i,j\in [2n]$, and the variance of the random variable considered to construct the estimator $\tilde{m}_i$ can be upper bounded by $\Tr[\rho\hat{R}_i^2]$ for each $i\in [2n]$.
    Moreover, if~\eqref{condition_W_ij} holds, then we have that
    \bb
        \left\|\tilde{W}-\Tr[\rho\{\hat{\mathbf{R}},\hat{\mathbf{R}}^\intercal\}]\right\|_\infty&\le\left\|\tilde{W}-\Tr[\rho\{\hat{\mathbf{R}},\hat{\mathbf{R}}^\intercal\}]\right\|_2\\
        &=\sqrt{\sum_{i,j=1}^{2n} \left|\tilde{W}_{i,j}-\Tr[\rho\{\hat{R}_i,\hat{R}_j\}]\right|^2 }\\
        &\le \varepsilon'\sqrt{\frac{\sum_{i,j=1}^{2n}\Tr[\rho\{\hat{R}_i,\hat{R}_j\}^2]}{16 E^2+6/n}  }\\
        &= n\varepsilon' \,,
    \ee
    where the last equality is a consequence of Lemma~\ref{lemma_formula_rr2}. In addition, if~\eqref{condition_m_i} holds, we have that
    \bb\label{upper_bound_error_m}
        \|\tilde{\textbf{m}}-\textbf{m}(\rho)\|_2
        &\le \varepsilon'\sqrt{\frac{\sum_{i=1}^{2n}\Tr[\rho\hat{R}_i^2]}{16 E^2+6/n}}\\
        &=\varepsilon'\sqrt{\frac{2\Tr[\rho\hat{E}_n]}{16 E^2+6/n}}\\
        &\leqt{(i)} \varepsilon'\sqrt{\frac{2\sqrt{\Tr[\rho\hat{E}_n^2]}}{16 E^2+6/n}}\\
        &\le \varepsilon'\sqrt{\frac{2n\sqrt{E^2}}{16 E^2+6/n}}\\
        &\le \frac{\sqrt{n}\varepsilon'}{\sqrt{8}E^{1/2}}
    \ee
    where (i) is a consequence of the fact that $\Tr[\hat{E}_n\rho]=\Tr[(\hat{E}_n^{2})^{\frac{1}{2}}\rho]$ and exploited the concavity of $x\mapsto \sqrt{x}$ for $x>0$ (see Eq.~\eqref{eq:conc}). Furthermore, if~\eqref{condition_m_i} holds, we have also that
    \bb
        \|\tilde{\textbf{m}}\tilde{\textbf{m}}^\intercal-\textbf{m}(\rho)\textbf{m}(\rho)^\intercal\|_\infty&\le \|(\tilde{\textbf{m}}-\textbf{m}(\rho))\tilde{\textbf{m}}^\intercal\|_\infty+\|\textbf{m}(\rho)(\tilde{\textbf{m}}-\textbf{m}(\rho))^\intercal\|_\infty\\
        &= \|\tilde{\textbf{m}}-\textbf{m}(\rho)\|_2\|\tilde{\textbf{m}}\|_2+\|\textbf{m}(\rho)\|_2\|\tilde{\textbf{m}}-\textbf{m}(\rho)\|_2\\
        &\le \|\tilde{\textbf{m}}-\textbf{m}(\rho)\|_2\left(\|\tilde{\textbf{m}}-\textbf{m}(\rho)\|_2+2\|\textbf{m}(\rho)\|_2\right)\\
        &\leqt{(ii)} \frac{\sqrt{n}\varepsilon'}{\sqrt{8}E^{1/2}}\left(\frac{\sqrt{n}\varepsilon'}{\sqrt{8}E^{1/2}}+2\|\textbf{m}(\rho)\|_2\right)\\
        &\leqt{(iii)} \frac{\sqrt{n}\varepsilon'}{\sqrt{8}E^{1/2}}\left(\frac{\sqrt{n}\varepsilon'}{\sqrt{8}E^{1/2}}+\sqrt{8n}E^{1/2}\right)\\
        &= \left(\frac{\varepsilon'}{{8}E}+1\right)n\varepsilon'\\
        &\leqt{(iv)}  2n\varepsilon'
    \ee
    Here, in (ii) we have used~\eqref{upper_bound_error_m}, while in (iii) we have exploited that
    \bb\label{upper_bound_norma_m}
        \|\textbf{m}(\rho)\|_2=\sqrt{\sum_{i=1}^{2n}\Tr[\rho\hat{R}_i]^2}\le\sqrt{\sum_{i=1}^{2n}\Tr[\rho\hat{R}_i^2]}= \sqrt{2\Tr[\rho\hat{E}_n]}\le \sqrt{2\sqrt{\Tr[\rho\hat{E}_n^2]}}\le\sqrt{2n}E^{1/2}\,,
    \ee
    where the first inequality and second inequality are a consequence of Eq.~\eqref{eq:conc}. In (iv) we have used that the second moment satisfies $E\ge \frac{1}{2} $ (due to~\eqref{energy_operator}). Therefore, if~\eqref{condition_W_ij} and~\eqref{condition_m_i} hold, then it holds that
    \bb
        \|\tilde{\textbf{m}}-\textbf{m}(\rho)\|_2\le \frac{\sqrt{n}\varepsilon'}{\sqrt{8}E^{1/2}}
    \ee     
    and
    \bb
        \|\tilde{V}-V\!(\rho)\|_\infty\le 5n\varepsilon'\,,
    \ee
    where we defined $\tilde{V}\coloneqq \tilde{W}-2\tilde{\textbf{m}}\tilde{\textbf{m}}^\intercal$ and we have used triangle inequality. Consequently, by setting $\varepsilon'\coloneqq \frac{\varepsilon}{10n}$, we have that the choice
    \bb
        N' \ge 68 \log\!\left(\frac{2 (2n^2+3n)  }{\delta}\right) \frac{200(8n^2 E^2+3n)}{\varepsilon^2}\,,
    \ee
    allows us to guarantee that
    \bb
        \Pr\left(  \|\tilde{V}-V\!(\rho)\|_\infty\le \frac{\varepsilon}{2}\quad \text{ and }\quad \|\tilde{\textbf{m}}-\textbf{m}(\rho)\|_2\le \frac{\varepsilon}{10\sqrt{8}E^{1/2}\sqrt{n}}\right)\ge 1-\delta\,.
    \ee
    The total number of copies of $\rho$ used in the algorithm in Table~\ref{Table_covariance} is $N=(n+3)N'$.
    Then, the algorithm defines the matrix $\tilde{V}'\coloneqq \tilde{V}+\frac{\varepsilon}{2}\,\mathbb{1}$. Consequently, in order to conclude the proof it suffices to show that if the condition $\|\tilde{V}-V\!(\rho)\|_\infty\le \frac{\varepsilon}{2}$ is satisfied, then $\tilde{V}'$ is a proper covariance matrix satisfying $\| \tilde{V}'-V\!(\rho) \|_\infty\le \varepsilon$. The latter condition follows by the triangle inequality. In order to show that $\tilde{V}'$ is a proper covariance matrix, let us observe that $\tilde{V}'+i\Omega_n\ge0$:
\bb             \tilde{V}'+i\Omega_n&=\tilde{V}+i\Omega_n+\frac{\varepsilon}{2}\mathbb{1}\\
        &\geqt{(i)} \tilde{V}-V\!(\rho)+\frac{\varepsilon}{2}\mathbb{1}\\
        &\geqt{(ii)} -\|\tilde{V}-V\!(\rho)\|_\infty \mathbb{1}+\frac{\varepsilon}{2}\mathbb{1}\\
        &\ge 0\,,
    \ee
where in (i) we have used that $V\!(\rho)+i\Omega_n\ge0$, while in (ii) we have used that for any operator $\Theta$ it holds that $\Theta\ge -\|\Theta\|_\infty \mathbb{1}$.   
\end{proof}

The sample complexity for estimating the covariance matrix of the above algorithm is $O\!\left(\frac{n^3E^2}{\varepsilon^2}\log\!\left(\frac{n^2}{\delta}\right)\right)$, outperforming the straightforward method of estimating indipendently the expectation values of all the observables $(\{\hat{R}_i,\hat{R}_j\})_{i,j\in[2n]}$, which leads to a worse sample complexity of $O\!\left(\frac{n^4E^2}{\varepsilon^2}\log\!\left(\frac{n^2}{\delta}\right)\right)$.

\subsection{Learning Gaussian states}\label{subs_sec_tom_gauss}
In this subsection, we analyse the algorithm for learning Gaussian states. In Subsubsection~\ref{subsection_second_moment_gaussian}, we discuss the relation between the second moment energy constraint of a Gaussian state and its first moment and covariance matrix. Subsequently, in Subsubsection~\ref{subsub:mixedG}, we detail the algorithm for learning mixed Gaussian states and provide the associated recovery guarantees. Moreover, in Subsubsection~\ref{subsub:mixedGnois}, we rigorously show that this algorithm is noise-robust, if the state deviates slightly from the set of Gaussian states, the same tomography algorithm can still be applied. Finally, in Subsubsection~\ref{subsub:pureG}, we outline the algorithm for learning pure Gaussian states and prove its recovery guarantees.

\subsubsection{Second moment of the energy of a Gaussian state}\label{subsection_second_moment_gaussian}
A Gaussian state $\rho$ is uniquely identified by its first moment $\textbf{m}(\rho)$ and its covariance matrix $V\!(\rho)$. In the forthcoming lemma we will see how to express the second moment of the energy in terms of these quantities.
\begin{lemma}\label{lemma_second_moment_energy}
    Let $\rho$ be an $n$-mode Gaussian state with first moment $\textbf{m}$ and covariance matrix $V$. The first and the second moment of the energy can be expressed in terms of the covariance matrix and the first moment as
    \begin{align}
        \Tr\!\left[\rho \hat{E}_n\right]&= \frac{\Tr\! V}{4}+\frac{\|\textbf{m}\|_2^2}{2}\,,\label{eq_1_en}\\
        \Tr\!\left[\rho \hat{E}_n^2\right]&= \left(\frac{1}{4}\Tr V+\frac12\|\textbf{m}\|_2^2\right)^2 + \frac{1}{8}\Tr[V^2] +\frac12\textbf{m}^\intercal V\textbf{m} - \frac{n}{4}
        \,,\label{eq_2_en}
    \end{align}
    where $\hat{E}_n\coloneqq \frac{\hat{\mathbf{R}}^\intercal\hat{\mathbf{R}}}{2}$ is the energy operator. In particular, this implies that the second moment of the energy can be upper bounded in terms of the mean energy as
    \bb
        \Tr[\rho \hat{E}_n^2]\le 3(\Tr[\rho \hat{E}_n])^2\,.
    \ee
\end{lemma}
\begin{proof}
   ~\eqref{eq_1_en} holds because of Lemma~\ref{lemma_energy_traceV}. In order to prove~\eqref{eq_2_en}, let us employ~\cite[Lemma~7]{bittel2024optimalestimatestracedistance}, which states that:  for any  $2n\times2n$ real symmetric matrix $X$, it holds that
	\begin{equation}
	\Tr\!\left[ (\hat{\textbf{R}}^\intercal X \hat{\textbf{R}})^2\, \rho(V,\textbf{m}) \right]  =  \left(\frac{1}{2}\Tr[VX]+\textbf{m}^\intercal X\textbf{m}\right)^2 + \frac{1}{2}\Tr[XVXV] + \frac{1}{2}\Tr[\Omega X\Omega X]+2\textbf{m}^\intercal XVX\textbf{m} \,.
\end{equation}
By setting $X=\frac{\mathbb{1}}{2}$, this proves~\eqref{eq_2_en}. Moreover, note that
    \bb
        \Tr\!\left[\rho \hat{E}_n^2\right] &= (\Tr[\rho\hat{E}_n])^2+\frac{1}{8}\Tr[V^2]+\frac{1}{2}\textbf{m}^\intercal V\textbf{m}-\frac{n}{4}\\
        &\le  (\Tr[\rho\hat{E}_n])^2+\frac{1}{8}\left[\Tr\! V  \right]^2+\frac{1}{2}\|\textbf{m}\|_2^2\Tr\! V\\
        &= (\Tr[\rho\hat{E}_n])^2+\frac{1}{2}\Tr[\rho\hat{E}_n]\Tr\! V+\frac{1}{4}\|\textbf{m}\|_2^2\Tr\! V\\
        & = 3(\Tr[\rho\hat{E}_n])^2-\Tr[\rho\hat{E}_n]\|\textbf{m}\|_2^2-\frac{1}{2}\|\textbf{m}\|_2^4\\
        & \le 3(\Tr[\rho\hat{E}_n])^2\,,
    \ee
    where we have used that $\Tr\!\left[\rho \hat{E}_n\right]= \frac{\Tr\! V}{4}+\frac{\|\textbf{m}\|_2^2}{2}$.
\end{proof}

\subsubsection{Tomography algorithm for mixed Gaussian states}
\label{subsub:mixedG}
In Table~\ref{Table_covariance3} we present a tomography algorithm to learn a classical description of an unknown $n$-mode Gaussian state $\rho$, by possessing the prior knowledge that $\rho$ has mean energy per mode bounded by some known constant $E>0$. The correctness of such a tomography algorithm is proved in Theorem~\ref{correctness_algorithm_Gaussian}.
The idea of such a tomography algorithm is trivial: estimate first moment and covariance matrix, and then outputs the Gaussian state with these first moment and covariance matrix. The non-trivial aspect of the algorithm concerns the underlying sample complexity analysis, which crucially relies on upper bounding the error in trace distance in terms of the errors of first moment and covariance matrix. We have proved such an upper bound in Theorem~\ref{upper_bound_trace_distance_Gaussian} above. 
\begin{table}[t]
  \caption{Estimation algorithm of an unknown $n$-mode Gaussian state $\rho$ satisfying the energy constraint $\Tr[\rho \hat{E}_n]\le nE$.}
  \begin{mdframed}[linewidth=2pt, roundcorner=10pt, backgroundcolor=white!10, innerbottommargin=10pt, innertopmargin=10pt]
    \textbf{Input:} Accuracy $\varepsilon$, failure probability $\delta$, mean energy per mode upper bound $E$, $N$ of copies of the unknown $n$-mode Gaussian state $\rho$ satisfying the energy constraint $\Tr[\rho \hat{E}_n]\le nE$, where
    \bb 
        N&\coloneqq\ (n+3)\ceil{  68 \log\!\left(\frac{2 (2n^2+3n)  }{\delta}\right) \frac{200(24n^2 E^2+3n)}{\varepsilon^4}2^{14}E^2n^4}=O\!\left( \log\!\left(\frac{n^2}{\delta}\right)\frac{n^7E^4}{\varepsilon^4}   \right)\,.
    \ee
    \textbf{Output:} With probability less than $\delta$, the output is \emph{failure}. Otherwise, with probability $\ge 1-\delta$, the output is a classical description of a Gaussian state $\tilde{\rho}$, such that
    \bb
        \frac{1}{2}\|\tilde{\rho}- \rho\|_1\le \varepsilon\,.
    \ee
    The classical description consists of the first moment $\textbf{m}(\tilde{\rho})\in\mathbb{R}^{2n}$ and in the covariance matrix $ V(\tilde{\rho})\in\mathbb{R}^{2n,2n}$, which uniquely characterise the Gaussian state $\tilde{\rho}$ and satisfy
    \bb
        \Pr\left(  \|\tilde{V}-V\!(\rho)\|_\infty\le \frac{\varepsilon^2}{2^7En^2}   \quad \text{ and }\quad \|\tilde{\textbf{m}}-\textbf{m}(\rho)\|_2\le \frac{\varepsilon^2}{2^7E^{3/2}n^{5/2}}  \right)\ge 1-\delta\,.
    \ee
    \begin{algorithmic}[1]
    \State Query $N$ copies of $\rho$ and apply the algorithm in Table~\ref{Table_covariance}, which outputs, with probability $\ge 1-\delta$, a vector $\tilde{\textbf{m}}\in\mathbb{R}^{2n}$ and a covariance matrix $\tilde{V}\in\mathbb{R}^{2n,2n}$. In case of failure of this algorithm, which happens with probability less than $\delta$, declare \emph{failure} and abort.
    \State Define $\tilde{\rho}$ as the Gaussian state with first moment $\textbf{m}(\tilde{\rho})=\tilde{\textbf{m}}$ and covariance matrix $V(\tilde{\rho})=\tilde{V}$.\\
    \Return the vector $\tilde{\textbf{m}}$ and the matrix $\tilde{V}$, which form a classical description of the Gaussian state $\tilde{\rho}$.
    \end{algorithmic}
  \end{mdframed}
  \label{Table_covariance3}
\end{table}

\begin{thm}\label{correctness_algorithm_Gaussian}
    Let $\varepsilon,\delta\in(0,1)$ and $E\geq 0$. Let $\rho$ be an $n$-mode Gaussian state satisfying the energy constraint $\Tr[\rho \hat{E}_n]\le nE$. A number of copies $N$ of $\rho$, such that
\bb 
N&\coloneqq\ (n+3)\ceil{  68 \log\!\left(\frac{2 (2n^2+3n)  }{\delta}\right) \frac{200(24n^2 E^2+3n)}{\varepsilon^4}2^{14}E^2n^4}\\
&=O\!\left( \log\!\left(\frac{n^2}{\delta}\right)\frac{n^7E^4}{\varepsilon^4}   \right)\,,
\ee
are sufficient to build a classical description of a Gaussian state $\tilde{\rho}$ such that
    \bb
        \Pr\left(\frac{1}{2}\|\tilde{\rho}- \rho\|_1\le \varepsilon\right)\ge 1-\delta\,.
    \ee
The classical description consists of the first moment $\textbf{m}(\tilde{\rho})\in\mathbb{R}^{2n}$ and in the covariance matrix $ V(\tilde{\rho})\in\mathbb{R}^{2n,2n}$, which uniquely characterise the Gaussian state $\tilde{\rho}$ and satisfy
    \bb
        \Pr\left(  \|\tilde{V}-V\!(\rho)\|_\infty\le \frac{\varepsilon^2}{2^7En^2}   \quad \text{ and }\quad \|\tilde{\textbf{m}}-\textbf{m}(\rho)\|_2\le \frac{\varepsilon^2}{2^{7}E^{3/2}n^{5/2}}  \right)\ge 1-\delta\,.
    \ee
\end{thm}
\begin{proof}
    We aim to establish the correctness of the algorithm presented in Table~\ref{Table_covariance3}. First, let us observe that the second moment of the energy of the Gaussian state $\rho$, which satisfies the energy constraint $\Tr[\rho\hat{E}_n] \le nE$, can be upper bounded as
    \bb
            \frac{1}{n}\sqrt{\Tr[\rho\hat{E}_n^2]}\le \sqrt{3}\frac{\Tr[\rho\hat{E}_n]}{n} \le \sqrt{3}E\,,
    \ee
    where we have used Lemma~\ref{lemma_second_moment_energy}. Consequently, for any $\varepsilon'\in(0,\frac{1}{2})$ Theorem~\ref{correctness_algorithm_cov} establishes that a number 
    \bb
        N&\coloneqq\ (n+3)\ceil{  68 \log\!\left(\frac{2 (2n^2+3n)  }{\delta}\right) \frac{200(24n^2 E^2+3n)}{\varepsilon'^2}}
    \ee
    of copies of $\rho$ suffices to construct a vector $\tilde{\textbf{m}}\in\mathbb{R}^{2n}$ and a covariance matrix $\tilde{V}\in\mathbb{R}^{2n,2n}$ such that
    \bb
        \Pr\left(  \|\tilde{V}-V\!(\rho)\|_\infty\le \varepsilon'\quad \text{ and }\quad \|\tilde{\textbf{m}}-\textbf{m}(\rho)\|_2\le \frac{\varepsilon'}{10(3)^{1/4}\sqrt{8En}}  \right)\ge 1-\delta\,.
    \ee
    Let $\tilde{\rho}$ be the Gaussian state with first moment $\textbf{m}(\tilde{\rho})=\tilde{\textbf{m}}$ and covariance matrix $V(\tilde{\rho})=\tilde{V}$. Let us show that, with probability at least $1-\delta$, the energy of the estimator $\tilde{\rho}$ is at most twice the energy of the true state $\rho$:
    \bb\label{eq_en_proof_tomo_gauss}
        \Tr[\tilde{\rho}\hat{E}_n]&\eqt{(i)}\frac{\Tr \tilde{V}}{4}+\frac{\|\tilde{\textbf{m}}\|_2^2}{2} \\
        &\leqt{(ii)} \frac{\Tr V(\rho)}{4}+\frac{2n\|\tilde{V}-V(\rho)\|_\infty}{4}+\frac{(\|\tilde{\textbf{m}}-\textbf{m}(\rho)\|_2+\|\textbf{m}(\rho)\|_2)^2}{2}\\
        &\leqt{(iii)}\Tr[\rho\hat{E}_n]+\frac{n\varepsilon'}{2}+\frac{\left(\frac{\varepsilon'}{10(3)^{1/4}\sqrt{8En}}\right)^2+2\|\textbf{m}(\rho)\|_2\left(\frac{\varepsilon'}{10(3)^{1/4}\sqrt{8En}}\right)}{2}\\
        &\le\Tr[\rho\hat{E}_n]+\frac{n\varepsilon'}{2}+\frac{1 +2\|\textbf{m}(\rho)\|_2}{2}\left(\frac{\varepsilon'}{10(3)^{1/4}\sqrt{8En}}\right)\\
        &\leqt{(iv)}\Tr[\rho\hat{E}_n]+\frac{n\varepsilon'}{2}+\frac{1 +2\sqrt{2nE}}{2}\left(\frac{\varepsilon'}{10(3)^{1/4}\sqrt{8En}}\right)\\
        &\le\Tr[\rho\hat{E}_n]+\frac{n\varepsilon'}{2}+ \frac{\varepsilon'}{10(3)^{1/4}} \\
        &\le \Tr[\rho\hat{E}_n]+n\varepsilon'\\
        &\le \Tr[\rho\hat{E}_n]+\frac{n}{2}\\
        &\le 2\Tr[\rho\hat{E}_n]\,.
    \ee
    Here, in (i), we exploited Lemma~\ref{lemma_energy_traceV}. In (ii), we have used  H\"older inequality, triangle inequality, and the fact that $\|\tilde{V}-V(\rho)\|_\infty\le\varepsilon'$. In (iii), we exploited again Lemma~\ref{lemma_energy_traceV} and we have used  that $\|\tilde{\textbf{m}}-\textbf{m}(\rho)\|_2\le \frac{\varepsilon'}{10(3)^{1/4}\sqrt{8En}}$. The inequality in (iv) follows from~\eqref{upper_bound_norma_m}. Then, building on Theorem~\ref{upper_bound_trace_distance_Gaussian}, we obtain that, with probability at least $1-\delta$, it holds that
    \bb
        \frac12 \left\|\rho - \tilde{\rho}\right\|_1 &\leq 2\sqrt{\max\left(\Tr[\rho\hat{E}_n],\Tr[\tilde{\rho}\hat{E}_n]\right)}\left( \|\textbf{m}(\rho)-\textbf{m}(\tilde{\rho})\|_2 + \sqrt2\, \sqrt{\|V\!(\rho)-V\!(\tilde{\rho})\|_1}\right)\\
        &\le 2\sqrt{2nE}\left( \|\textbf{m}(\rho)-\textbf{m}(\tilde{\rho})\|_2 + \sqrt2\, \sqrt{\|V\!(\rho)-V\!(\tilde{\rho})\|_1}\right)\\
        &\le \sqrt{8nE}\left(  \frac{\varepsilon'}{10(3)^{1/4}\sqrt{8En}} + 2\, \sqrt{n\|V\!(\rho)-V\!(\tilde{\rho})\|_\infty}\right)\\
        &\le \sqrt{8nE}\left(  \frac{\varepsilon'}{10(3)^{1/4}\sqrt{8En}} + 2\, \sqrt{n\varepsilon'}\right)\\
        &\le \sqrt{8nE}\left(  4\, \sqrt{n\varepsilon'}\right)\\
        &=\sqrt{2^7En^2\varepsilon'}\,.
    \ee
    Consequently, by setting $\varepsilon'\coloneqq \frac{\varepsilon^2}{2^7En^2}\in(0,\frac{1}{2})$, we have that the choice
    \bb
        N&\coloneqq\ (n+3)\ceil{  68 \log\!\left(\frac{2 (2n^2+3n)  }{\delta}\right) \frac{200(24n^2 E^2+3n)}{\varepsilon^4}2^{14} E^2 n^4}
    \ee
    allows us to guarantee that 
    \bb
        \Pr\left(\frac{1}{2}\|\tilde{\rho}- \rho\|_1\le \varepsilon\right)\ge 1-\delta\,.
    \ee
    With this choice of $\varepsilon'$, we have also that
    \bb
        \Pr\left(  \|\tilde{V}-V\!(\rho)\|_\infty\le \frac{\varepsilon^2}{2^7En^2}   \quad \text{ and }\quad \|\tilde{\textbf{m}}-\textbf{m}(\rho)\|_2\le \frac{\varepsilon^2}{10(3)^{1/4}2^7\sqrt{8}E^{3/2}n^{5/2}}  \right)\ge 1-\delta\,,
    \ee
    which, in particular, implies that
    \bb
        \Pr\left(  \|\tilde{V}-V\!(\rho)\|_\infty\le \frac{\varepsilon^2}{2^7En^2}   \quad \text{ and }\quad \|\tilde{\textbf{m}}-\textbf{m}(\rho)\|_2\le \frac{\varepsilon^2}{2^7E^{3/2}n^{5/2}}  \right)\ge 1-\delta\,,
    \ee
\end{proof}

\subsubsection{Noise robustness of the algorithm}
\label{subsub:mixedGnois}
Above we have seen that learning an unknown Gaussian states is efficient, as the sample complexity scales polynomially in the number of modes. Here, we address the questions: \emph{What happens if the state is not exactly Gaussian?} \emph{Is our tomography procedure stable under little perturbations of the set of Gaussian states?} Answering these questions is conceptually relevant, as experimental imperfections may transform Gaussian states into slightly-perturbed Gaussian states.

Here, in order to quantify the non-Gaussian character of a quantum state, we use the \emph{relative entropy of non-Gaussianity}~\cite{Genoni2008,Marian2013}. For any state $\rho$, the relative entropy of non-Gaussianity  $d_{\mathcal{G}}(\rho)$ is defined as the minimum relative entropy between $\rho$ and any Gaussian state:
\bb
    d_{\mathcal{G}}(\rho)\coloneqq \min_{\sigma\in \mathcal{G}}S(\rho\|\sigma)\,,
\ee
where $\mathcal{G}$ denotes the set of Gaussian states and
\bb
S(\rho\|\sigma)\coloneqq \Tr[\rho\log_2\rho]-\Tr[\rho\log_2\sigma]
\ee
denotes the quantum relative entropy between $\rho$ and $\sigma$. The quantum relative entropy $S(\rho\|\sigma)$ has a strong operation meaning in terms of the optimal rate in asymmetric quantum hypothesis testing between $\rho$ and $\sigma$, due to the quantum Stein’s lemma~\cite{Hiai1991,Ogawa2000}.

As proved by \cite{Genoni2008}, the relative entropy of non-Gaussianity satisfies desirable properties that make it a meaningful measure of non-Gaussianity. In particular, note that $d_{G}(\rho)$ is faithful, in the sense that $d_{G}(\rho)\ge0$ and it vanishes if and only if $\rho$ is Gaussian. Notably, the minimum in the definition of $d_{\mathcal{G}}(\rho)$ is achieved by the \emph{Gaussianification} of $\rho$~\cite{Marian2013}. The Gaussianification $G(\rho)$ of a state $\rho$ is the Gaussian state with the same first moment and covariance matrix of $\rho$. In formula, the relative entropy of non-Gaussianity is given by~\cite{Marian2013}:
\bb\label{charact_rel_entropy_non_gauss}
    d_{\mathcal{G}}(\rho)=S\!\left(\rho\|G(\rho)\right)\,.
\ee
Let us now analyse the stability of our tomography algorithm for Gaussian states. We make use the of the following simple observation.
\begin{lemma}\label{remark_gauss}
    If the unknown state $\rho$ is not Gaussian, the algorithm designed for learning Gaussian states in Table~\ref{Table_covariance3} effectively learns the Gaussianification $G(\rho)$. Mathematically, $O\!\left( \log\!\left(\frac{n^2}{\delta}\right)\frac{n^7E^4}{\varepsilon^4}   \right)$ copies of $\rho$ suffices in order to build a classical description of a Gaussian state $\tilde{\rho}$ such that $\Pr\left(\frac{1}{2}\|\tilde{\rho}- G(\rho)\|_1\le \varepsilon\right)\ge 1-\delta$.
\end{lemma}
\begin{proof}
    The claim follows by two main observations: (i) the algorithm in Table~\ref{Table_covariance3} only involves the first moment and the covariance matrix of the unknown state, and (ii) the mean photon number of a state $\rho$ and the one of its Gaussianification $G(\rho)$ coincide (because of Lemma~\ref{lemma_energy_traceV}).
\end{proof}
Now we show that \emph{quantum state tomography of slightly-perturbed Gaussian states is efficient}. Here, by `slightly-perturbed Gaussian state' we mean that the relative entropy of non-Gaussianity is sufficiently small.

\begin{thm}[(Quantum state tomography of slightly-perturbed Gaussian states is efficient)]\label{thm_robustness_gaussian_states}
    Let $\varepsilon,\delta\in(0,1)$ and $E\ge0$. Let $\rho$ be an unknown $n$-mode (possibly non-Gaussian) state such that its relative entropy of non-Gaussianity satisfies $d_{\mathcal{G}}(\rho)\le \varepsilon^2$. Assume that $\rho$ satisfies the energy constraint $\Tr[\rho \hat{E}_n]\le nE$. Then, $O\!\left( \log\!\left(\frac{n^2}{\delta}\right)\frac{n^7E^4}{\varepsilon^4}   \right)$ copies of $\rho$ suffices to build a classical description of a Gaussian state $\tilde{\rho}$ such that $\Pr\left(\frac{1}{2}\|\tilde{\rho}- \rho\|_1\le \varepsilon\right)\ge 1-\delta$.
\end{thm}
\begin{proof}
    Thanks to Remark~\ref{remark_gauss}, $O\!\left( \log\!\left(\frac{n^2}{\delta}\right)\frac{n^7E^4}{\varepsilon^4}   \right)$ copies of $\rho$ suffices to build a classical description of a Gaussian state $\tilde{\rho}$ such that $\frac{1}{2}\|\tilde{\rho}- G(\rho)\|_1\le \left(1-\sqrt{\frac{\ln2}{2}}\right)\varepsilon$ with probability at least $1-\delta$, where $G(\rho)$ denotes the Gaussianification of $\rho$. If this event happens, then 
\bb
    \frac{1}{2}\|\rho-\tilde{\rho}\|_1&\le\frac{1}{2}\|G(\rho)-\tilde{\rho}\|_1+\frac{1}{2}\|G(\rho)-\rho\|_1\\
    &\le \left(1-\sqrt{\frac{\ln2}{2}}\right)\varepsilon+\frac{1}{2}\|G(\rho)-\rho\|_1\\
    &\leqt{(i)} \left(1-\sqrt{\frac{\ln2}{2}}\right)\varepsilon+\sqrt{\frac{\ln2}{2}}\sqrt{S(\rho\|G(\rho))}\\
    &\eqt{(ii)}\left(1-\sqrt{\frac{\ln2}{2}}\right)\varepsilon+\sqrt{\frac{\ln2}{2}} \sqrt{d_{\mathcal{G}(\rho)}}\\
    &\le \varepsilon\,.
\ee
Here, in (i), we employed the quantum Pinsker inequality~\cite[Theorem 11.9.1]{MARK}, which states that for any $\tau$ and $\sigma$ the trace distance can be upper bounded in terms of the relative entropy as $\frac{1}{2}\|\tau-\sigma\|_1\le\sqrt{\frac{\ln2}{2}S(\tau\|\sigma)}$. Finally, in (ii), we just employed the characterisation of the relative entropy of non-Gaussianity in~\eqref{charact_rel_entropy_non_gauss}.
\end{proof}

\subsubsection{Improved tomography algorithm for pure Gaussian states}
\label{subsub:pureG}
In this subsubsection we present an improved upper bound on the sample complexity of tomography of pure Gaussian states. We can show that $O\!\left( \log\!\left(\frac{n^2}{\delta}\right)\frac{n^5E^4}{\varepsilon^4}   \right)$ state copies suffices to achieve tomography of pure Gaussian states, thus obtaining an improvement with respect the scaling $O\!\left( \log\!\left(\frac{n^2}{\delta}\right)\frac{n^7E^4}{\varepsilon^4}   \right)$ of the mixed setting analysed in subsubsection~\ref{subsub:mixedG}. The main technical tool employed here is the improved upper bound on the trace distance between a pure Gaussian state and an arbitrary (possibly-mixed) state presented in Theorem~\ref{inequality_dist_gauss}. 
\begin{thm}[(Tomography of pure Gaussian states)]\label{tomography_pure_gaussian}
Let $\varepsilon,\delta\in(0,1)$ and $E\geq 0$. Let $\psi$ be an $n$-mode pure Gaussian state satisfying the energy constraint $\Tr[\psi \hat{E}_n]\le nE$. A number of copies $N$ of $\psi$, such that
\bb 
N&\coloneqq\ (n+3)\ceil{  68 \log\!\left(\frac{2 (2n^2+3n)  }{\delta}\right) \frac{200(24n^2 E^2+3n)}{\varepsilon^4}16E^2n^2}\\
&=O\!\left( \log\!\left(\frac{n^2}{\delta}\right)\frac{n^5E^4}{\varepsilon^4}   \right)\,,
\ee
are sufficient to build a classical description of a Gaussian state $\tilde{\rho}$ such that
    \bb
        \Pr\left(\frac{1}{2}\|\tilde{\rho}- \psi\|_1\le \varepsilon\right)\ge 1-\delta\,.
    \ee
\end{thm}
\begin{proof}
   We proceed as in Lemma~\ref{correctness_algorithm_Gaussian}. First, let us observe that Lemma~\ref{lemma_second_moment_energy} establishes that the second moment of the energy of the Gaussian state $\psi$, which satisfies the energy constraint $\Tr[\psi\hat{E}_n] \le nE$, can be upper bounded as
    \bb
            \frac{1}{n}\sqrt{\Tr[\psi\hat{E}_n^2]}\le \sqrt{3}\frac{\Tr[\psi\hat{E}_n]}{n} \le \sqrt{3}E\,.
    \ee
     Consequently, for any $\varepsilon'\in(0,\frac{1}{2})$ Theorem~\ref{correctness_algorithm_cov} establishes that a number 
    \bb
        N&\coloneqq\ (n+3)\ceil{  68 \log\!\left(\frac{2 (2n^2+3n)  }{\delta}\right) \frac{200(24n^2 E^2+3n)}{\varepsilon'^2}}
    \ee
    of copies of $\psi$ suffices to construct a vector $\tilde{\textbf{m}}\in\mathbb{R}^{2n}$ and a covariance matrix $\tilde{V}\in\mathbb{R}^{2n,2n}$ such that
    \bb
        \Pr\left(  \|\tilde{V}-V\!(\psi)\|_\infty\le \varepsilon'\quad \text{ and }\quad \|\tilde{\textbf{m}}-\textbf{m}(\psi)\|_2\le \frac{\varepsilon'}{10 (3)^{1/4}\sqrt{8En}}  \right)\ge 1-\delta\,.
    \ee
    Let $\tilde{\rho}$ be the Gaussian state with first moment $\textbf{m}(\tilde{\rho})=\tilde{\textbf{m}}$ and covariance matrix $V(\tilde{\rho})=\tilde{V}$. By exploiting the same reasoning used in~\eqref{eq_en_proof_tomo_gauss}, we have that the energy of the estimator $\tilde{\rho}$ is at most twice the energy of $\psi$, i.e.
    \bb
        \Tr[\tilde{\rho}\hat{E}_n]\le 2\Tr[\psi\hat{E}_n]\,,
    \ee
    with probability at least $1-\delta$. Then, by exploiting Theorem~\ref{inequality_dist_gauss}, we have that, with probability at least $1-\delta$, it holds that
    \bb
        \frac{1}{2}\left\|\, \tilde{\rho}-\psi\, \right\|_1&\le \sqrt{\max\left(\Tr[\tilde{\rho}\hat{E}_n],\Tr[\psi\hat{E}_n]\right)}\sqrt{ \|V\!(\tilde{\rho})-V\!(\psi)\|_\infty+2\|\textbf{m}(\tilde{\rho})-\textbf{m}(\psi)\|_2^2}\\
        &\le\sqrt{2nE}\sqrt{ \|V\!(\tilde{\rho})-V\!(\psi)\|_\infty+2\|\textbf{m}(\tilde{\rho})-\textbf{m}(\psi)\|_2^2}\\
        &\le \sqrt{2nE}\sqrt{ \varepsilon'+\frac{\varepsilon'^2}{400\sqrt{3}\,En}   }\\
        &\le \sqrt{4nE\varepsilon'}\,.
    \ee
    Consequently, by setting $\varepsilon'\coloneqq \frac{\varepsilon^2}{4nE}\in(0,\frac{1}{2})$, we have that the choice
    \bb
        N&\coloneqq\ (n+3)\ceil{  68 \log\!\left(\frac{2 (2n^2+3n)  }{\delta}\right) \frac{200(24n^2 E^2+3n)}{\varepsilon^4}16E^2n^2}
    \ee
    allows us to guarantee that 
    \bb
        \Pr\left(\frac{1}{2}\|\tilde{\rho}- \psi\|_1\le \varepsilon\right)\ge 1-\delta\,.
    \ee
\end{proof}

\newpage

\section{Tomography of $t$-doped bosonic Gaussian states}\label{Sec_t_doped}
In the preceding section, we rigorously proved that quantum state tomography of bosonic Gaussian states is efficient. In this section, we delve into the analysis of \emph{$t$-doped bosonic Gaussian states}, which are states prepared by Gaussian unitaries and at most $t$ non-Gaussian local unitaries. They encompass a much broader class of efficiently learnable states than exact bosonic Gaussian states. 

The results we will show in this section can be seen as an extension to the bosonic setting of what was previously shown for $t$-doped stabiliser states~\cite{grewal2023efficient,leone2023learning,hangleiter2024bell} (states prepared by Clifford gates and at most $t$ $\operatorname{T}$-gates) and $t$-doped fermionic Gaussian states of Ref.~\cite{mele2024efficient} (states prepared by fermionic Gaussian unitaries and at most $t$ fermionic non-Gaussian local unitaries). 
However, extending these results is far from trivial, as in the bosonic setting one must deal not only with different commutation relations than in the fermionic setting, but also with subtleties arising from the energy constraints and from the infinite-dimensional Hilbert space.
Just as T-gates are considered \emph{magic} gates for Clifford circuits, which are classically simulable~\cite{gottesman1998heisenberg}, local non-Gaussian gates can similarly be viewed as \emph{magic} gates for Gaussian circuits, which are also classically simulable~\cite{Mari-Eisert}.
Recent works have also been focusing on the classical simulability of states prepared by a Gaussian evolution applied to an input state that is a superposition of a bounded number of Gaussian states~\cite{Chabaud_2021,dias2024classical,hahn2024classical}. 

An $n$-mode unitary $U$ is said to be a \emph{$(t,\kappa)$-doped Gaussian unitary} if it is a composition of Gaussian unitaries and at most $t$ non-Gaussian $\kappa$-local unitaries, where $\kappa$-local means that each non-Gaussian gate involves at most $\kappa$ quadrature operators. In other words, $U$ is of the form 
\bb
    U=G_{t}W_t\cdots G_1 W_1 G_0\,,
\ee
where each $G_{i}$ is an $n$-mode Gaussian unitary and $W_i$ is a $\kappa$-local non-Gaussian unitary. Strictly speaking, we assume each $W_i$ to be a unitary generated by a Hamiltonian which is a (non-quadratic) polynomial in at most $\kappa$ quadrature operators. 
An $n$-mode state vector $\ket{\psi}$ is said to be a \emph{$(t,\kappa)$-doped Gaussian state} if it can be prepared by applying a $(t,\kappa)$-doped Gaussian unitary to the vacuum, that is,
\bb
    \ket{\psi}=U\ket{0}^{\otimes n}\,.
\ee
Informally, we sometimes say that a state is a $t$-doped Gaussian state (by omitting the dependence on $\kappa$) if it is a $(t,\kappa)$-doped Gaussian state for $\kappa=O(1)$ in the number of modes. Fig.~\ref{figure_t_doped} shows a depiction of a $t$-doped Gaussian state.
\begin{figure*}[t]\label{figure_t_doped}
    \centering
    \includegraphics[width=0.55\textwidth]{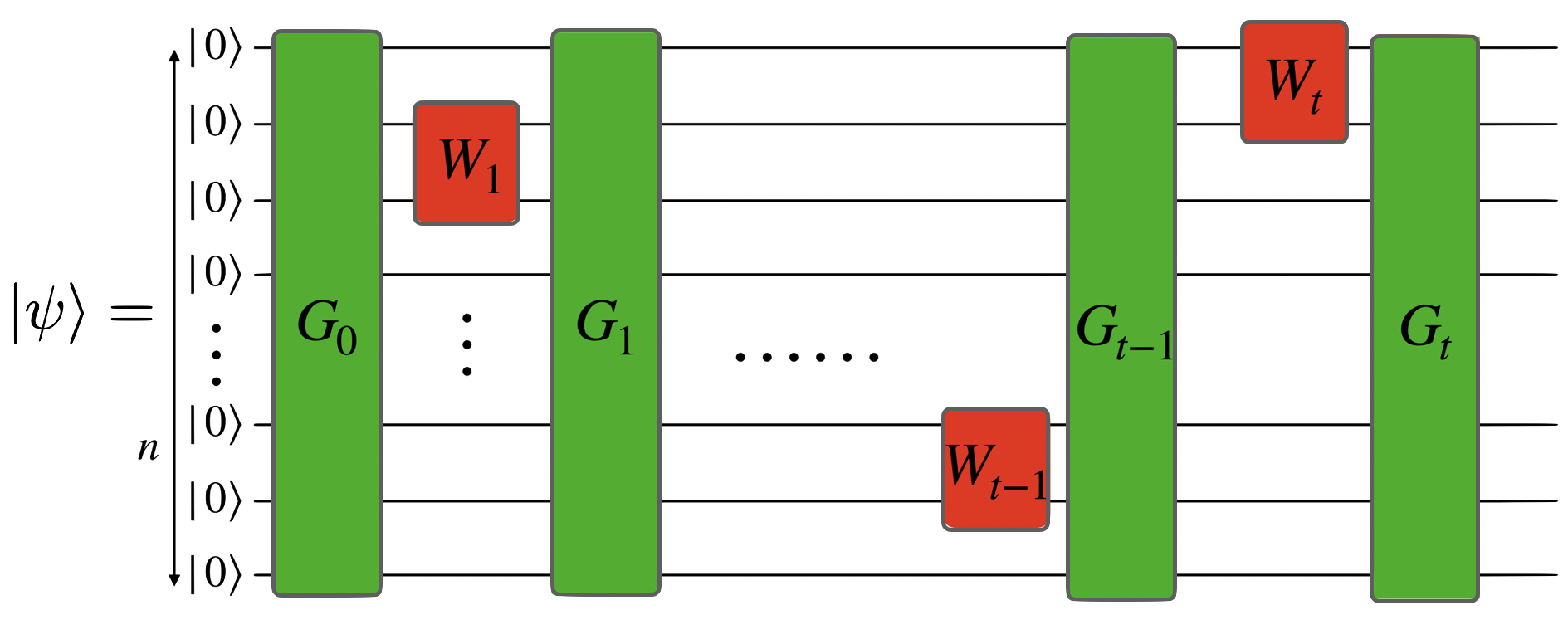}
    \caption{Depiction of a $t$-doped Gaussian state. The state vector $\ket{\psi}$ is prepared by applying Gaussian unitaries (green gates) and at most $t$ local non-Gaussian unitaries (red gates). In the specific example depicted in the figure, the locality is $\kappa=4$ as each of the non-Gaussian unitaries acts non-trivially on two modes (thus $4$ quadratures). }
\end{figure*}

The forthcoming Theorem~\ref{thm_compression} provides a remarkable decomposition of $(t,\kappa)$-doped unitaries and states, which shows that, if $\kappa t\le n$, one can turn any $t$-doped state into a tensor product between a $\kappa t$-mode non-Gaussian state and the $(n-\kappa t)$-mode vacuum state via a suitable Gaussian unitary. We provide a depiction of such a decomposition in Fig.~\ref{fig_compression_thm}
\begin{figure*}[t]\label{fig_compression_thm}
    \centering
    \includegraphics[width=0.75\textwidth]{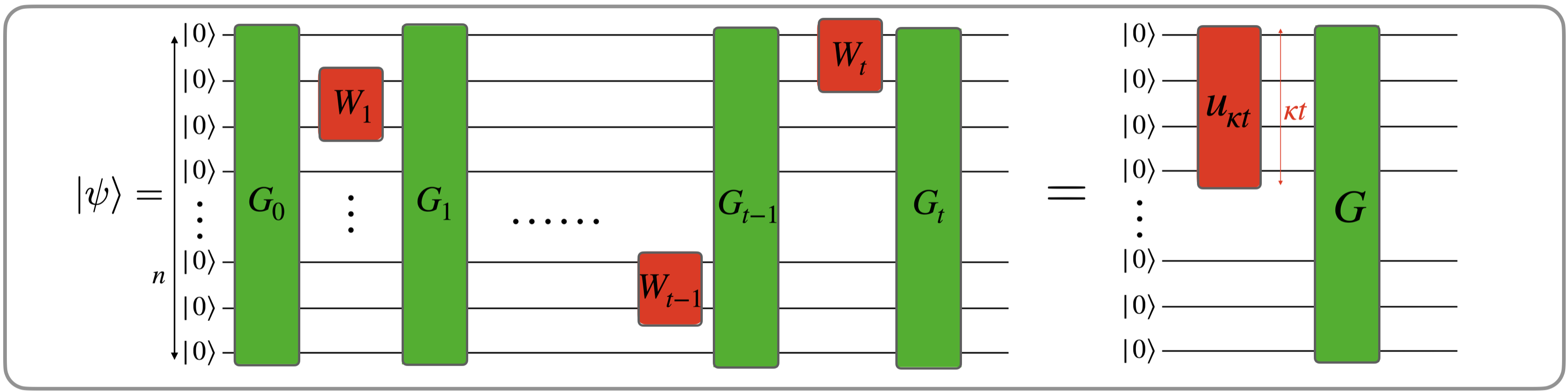}
    \caption{Depiction of the decomposition of a $t$-doped Gaussian state vector $\ket{\psi}$ proved in Theorem~\ref{thm_compression}. All the `non-Gaussianity' of $\ket{\psi}$ can be compressed on the first $\kappa t$ modes by applying to $\ket{\psi}$ a suitable Gaussian unitary $G^\dagger$.  }
\end{figure*}

\begin{thm}[(Non-Gaussianity compression in $t$-doped Gaussian unitaries and states)]\label{thm_compression}
If $n\ge \kappa t$, any $n$-mode $(t,\kappa)$-doped Gaussian unitary $U$ can be decomposed as 
\bb
    U = G(u_{\kappa t}\otimes \mathbb{1}_{n-\kappa t}) G_{\text{passive}}\,
\ee
for some suitable Gaussian unitary $G$, passive Gaussian unitary $G_{\text{passive}}$, and $\kappa t$-mode (non-Gaussian) unitary $u_{\kappa t}$. In particular, any $n$-mode $(t,\kappa)$-doped Gaussian state can be decomposed as 
\bb\label{dec_t_doped}
    \ket{\psi}=G\left(\ket{\phi_{\kappa t}}\otimes\ket{0}^{\otimes(n-\kappa t)}\right)\,
\ee
for some suitable Gaussian unitary $G$ and  $\kappa t$-mode (non-Gaussian) state vector $\ket{\phi_{\kappa t}}$.
\end{thm}
The decomposition in~\eqref{dec_t_doped} shows that it is possible to compress all the non-Gaussianity of a $t$-doped Gaussian state via a suitable Gaussian unitary. By leveraging such a property, we devise a tomography algorithm for $t$-doped Gaussian states which has a sample and time complexity that scales polynomially in $n$ as long as $\kappa t = O(1)$, \emph{thereby establishing that tomography of (energy-constrained) $t$-doped states is efficient} in this regime. This establishes the robustness of tomography of Gaussian states, in the sense that, even if few non-Gaussian unitaries are applied to a Gaussian state, the resulting state remains efficiently learnable.
 
The rough idea behind our tomography algorithm for unknown $t$-doped Gaussian states $\ket{\psi}$ involves first estimating the Gaussian unitary $G$ (see~\eqref{dec_t_doped}), then applying its inverse to the state in order to compress the non-Gaussianity, and finally performing tomography of the first $\kappa t$ modes. Remarkably, our tomography algorithm is experimentally feasible, as it uses only Gaussian unitaries and easily implementable Gaussian measurements, such as homodyne and heterodyne detection~\cite{BUCCO,Aolita_2015}.

In the following theorem, we analyse the performance guarantees of our tomography algorithm.
\begin{thm}[(Informal version)] 
Let $\ket{\psi}$ be an unknown $n$-mode $(t,\kappa)$-doped Gaussian state vector, which satisfies the second-moment constraint ${\Tr[\psi \hat{N}_n^2]}\le n^2 N_{\text{phot}}^2$. Then, $
O\!\left( \left(\frac{n^2  N_{\text{phot}}^2  }{\varepsilon^2}\right)^{\!\kappa t}\right)$ state copies suffices to construct a succint classical description of an estimator 
$\ket{\tilde{\psi}}$ which is $\varepsilon$-close in trace distance to $\ket{\psi}$ with high probability. Such a classical description consists of a triplet $(  \tilde{\textbf{m}}, \tilde{S},\ket{\tilde{\phi}_1} )$, which defines the estimator via the relation 
\bb
    \ket{\tilde{\psi}}\coloneqq \hat{D}_{\tilde{\textbf{m}}}U_{\tilde{S}} \ket{\tilde{\phi}_1}\otimes\ket{0}^{\otimes(n-\kappa t)}\,,
\ee
where $\hat{D}_{\tilde{\textbf{m}}}$ is the displacement operator with $\tilde{\textbf{m}}\in\mathbb{R}^{2n}$, $U_{\tilde{S}}$ is a Gaussian unitary associated with the $2n\times 2n$ symplectic matrix $\tilde{S}$, and $\ket{\tilde{\phi}_1}$ is a $\kappa t$-mode state contained in the $d_{\text{eff}}$-dimensional subspace spanned by all the $n$-mode Fock
states with total photon number less than $O\!\left( \left(\frac{n^2  N_{\text{phot}}^2   }{\varepsilon^2}\right)\right)$, where 
\bb
    d_{\text{eff}}&\coloneqq O\!\left( \left(\frac{n^2  N_{\text{phot}}^2   }{\varepsilon^2}\right)^{\!\kappa t}\right)\,.
\ee
In particular, tomography of $t$-doped Gaussian states is efficient in the regime $\kappa t = O(1)$, as its sample, time, and memory complexity scales polynomially in $n$. 
\end{thm}
We conclude by showing that any tomography algorithm that aims to learn quantum states which can be written as $G(\ket{\phi_{ t}}\otimes\ket{0}^{\otimes(n-t)})$, where $G$ is a Gaussian unitary and $\ket{\phi_{ t}}$ a $t$-mode quantum state vector, must be inefficient if $t$ scales \emph{slightly} more than a constant in the number of modes.

\subsection{$t$-doped Gaussian states and unitaries}

We now present the formal definitions crucial for our analysis.
\begin{Def}[(Non-Gaussian $k$-local unitary)]
    Let $n\in\N$ and $\kappa\in[2n]$. A \emph{$\kappa$-local non-Gaussian unitary} is an $n$-mode unitary generated by a Hamiltonian which is a polynomial in at most $k$ quadrature operators $\{\hat{R}_{\mu(r)}\}^{\kappa}_{r=1}$, where $\mu(1),\dots,\mu(\kappa) \in [2n]$.
\end{Def}
For instance, unitaries like $\exp\!\left[i \hat{x}_1^4 \right]$, $\exp\!\left[i (\hat{x}_1^4+\hat{p}_1^4) \right]$, $\exp\!\left[i (\hat{x}_1^4+\hat{p}_1^4 + \hat{x}_3^6)\right]$, and $\exp\!\left[i 0.7 \hat{a}^{\dagger}_1\hat{a}^{\dagger}_2\hat{a}_3\hat{a}_4 \right]$ are examples of $1$-, $2$-, $3$-local, and $4$-local non-Gaussian unitaries, respectively.

\begin{Def}[($t$-doped Gaussian unitary)]
\label{def:tdoped}
A unitary $U_t$ is a $(t,\kappa)$-doped Gaussian unitary if it can be prepared by Gaussian unitaries $\{G_i\}^{t}_{i=0}$ and $t$ non-Gaussian $\kappa$-local unitaries $\{W_i\}^{t}_{i=1}$, specifically given by
\begin{align}
    U_t = G_{t}W_t\cdots G_1 W_1 G_0 \,.
\end{align}
\end{Def}

\begin{Def}[($t$-doped Gaussian state)]
\label{def_t_doped}
    A $(t,\kappa)$-doped Gaussian state vector $\ket{\psi}$ is obtained by the action of a $(t,\kappa)$-doped Gaussian unitary $U_t$ on the vacuum, i.e., $\ket{\psi}\coloneqq U_t\ket{0}^{\otimes n}$.
\end{Def}
\subsubsection{Compression of the non-Gaussianity}
In this section, we introduce a finding that sheds light on the structure of $t$-doped Gaussian unitaries: the 'non-Gaussianity' can be compressed into the first modes through a Gaussian operation. This result is an extension to the bosonic setting of a result proved in Ref.~\cite{mele2024efficient}, formulated for $t$-doped \emph{fermionic} unitaries. The proof for bosons can be easily generalised from the one for fermions of Ref.~\cite{mele2024efficient}, by making use of the following parallelism between the theory of bosonic Gaussian and fermionic Gaussian unitaries. Specifically, in the bosonic setting, Gaussian unitaries correspond to symplectic matrices, and Gaussian passive unitaries are associated with symplectic-orthogonal matrices. In the fermionic context, the role of symplectic and orthogonal matrices is exchanged: Gaussian unitaries correspond to orthogonal matrices, and Gaussian passive unitaries are associated with symplectic-orthogonal matrices. Consequently, when handling Gaussian passive transformations, the theory of Gaussian bosons aligns with the theory of Gaussian fermions. In the following proof, we follow the construction presented in Ref.~\cite{mele2024efficient}.
\begin{thm}[(Compression of non-Gaussianity in $t$-doped unitaries)]
\label{th:comprUni}
Let $n\in\N$ be the number of modes. Let $\kappa,t\in N$ such that $\kappa t\le n$. For any $(t,\kappa)$-doped bosonic Gaussian unitary $U_t$ there exist a Gaussian unitary $G$, a passive Gaussian unitary $G_{\text{passive}}$, and a (possibly non-Gaussian) unitary $u_{\kappa t}$ supported exclusively on the first $\kappa t$ modes such that
\bb
    U_t = G(u_{\kappa t}\otimes \mathbb{1}_{n-\kappa t}) G_{\text{passive}},
\ee
where $\mathbb{1}_{n-\kappa t}$ denotes the identity on the last $n-\kappa t$ modes.
\end{thm}
\begin{proof}
By definition, the $t$-doped unitary can be expressed as $U_t=(\prod^t_{t'=1}G_{t'}W_{t'})G_0$. Rearranging it, we express it as
\bb
    U_t=\tilde{G}_t\prod^t_{t'=1}\tilde{W}_{t'},
\ee
where $\tilde{W}_{t'} \coloneqq  \tilde{G}^\dag_{t'-1} W_{t'}\tilde{G}_{t'-1}$ and $\tilde{G}_{t'} \coloneqq  G_{t'}..G_0$. Let $G_{\mathrm{aux}}$ be a passive Gaussian unitary, which we will fix later in order to compress all the non-Gaussianity to the first $\kappa t$ modes. We can rewrite $U_t$ as 
\bb
    U_t=\tilde{G}_tG_{\mathrm{aux}}\prod^t_{t'=1}(G^{\dag}_{\mathrm{aux}}\tilde{W}_{t'}G_{\mathrm{aux}})G^{\dag}_{\mathrm{aux}}\,.
\ee
Moreover, by defining
\bb
     G& \coloneqq  \tilde{G}_tG_{\mathrm{aux}}\,,\\
     u_{\kappa t}& \coloneqq  \prod^t_{t'=1}(G^{\dag}_{\mathrm{aux}}\tilde{W}_{t'}G_{\mathrm{aux}})\,,\\
     G_{\text{passive}}& \coloneqq  G^{\dag}_{\mathrm{aux}}\,,
\ee
we have 
\bb
    U_t = G u_{\kappa t} G_{\text{passive}}\,.
\ee
Observe that $G$ is a Gaussian unitary, as the product of Gaussian unitaries remains Gaussian. Similarly, $G_{\text{passive}}$ is a passive Gaussian unitary, given that the adjoint of a passive Gaussian unitary remains passive and Gaussian. In order to conclude the proof, we have to show that we can choose $G_{\mathrm{aux}}$ such that $u_{\kappa t}$ is supported only on the first $\kappa t$ modes. It suffices to show that the Heisenberg evolution, under the Gaussian unitary $\tilde{G}_{t'-1}G_{\mathrm{aux}}$, of the Hamiltonian generating $W_{t'}$ involves only the first $\kappa t$ quadratures $(\hat{R}_{i})_{i\in[\kappa t]}$. In other words, we need to show that the Heisenberg evolution of each quadrature operator involved in the Hamiltonian generating $W_{t'}$ involves only the first $\kappa t$ quadratures. Let $O_{\mathrm{aux}}$ be the symplectic orthogonal matrix associated with $G_{\mathrm{aux}}$. For each $t'\in [t]$ let $\mu(t',1),\mu(t',2),\ldots,\mu(t',\kappa)$, with $\mu(t',1)\le\mu(t',2)\le\ldots\le\mu(t',\kappa)$, be the quadratures involved in the Hamiltonian generating $W_{t'}$.  For each $t'\in [t]$ and $r\in[\kappa]$ it holds that
\bb
        G^{\dag}_{\mathrm{aux}} \tilde{G}^{\dag}_{t'-1} \hat{R}_{\mu(t',r)} \tilde{G}_{t'-1}G_{\mathrm{aux}}=\sum^{2n}_{m=1} (S_{t'-1} O_{\mathrm{aux}})_{\mu(t',r),m}\hat{R}_m\,,
\ee
where $S_{t'-1}$ is the symplectic matrix associated with $\tilde{G}_{t'-1}$. Consequently, we have to choose $O_{\mathrm{aux}}$ so that
\bb
    (S_{t'-1} O_{\mathrm{aux}})_{\mu(t',r),m}=(O^\intercal_{\mathrm{aux}} S^\intercal_{t'-1} )_{m,\mu(t',r)}=0
    \label{eq:condi}
\ee
for all $t'\in [t]$, $r\in[\kappa]$, and all $m\in\{2\kappa t+1,\ldots,2n\}$. Let $(\mathbf{e}_{i})_{i\in[2n]}$ be the canonical basis vectors of $\mathbb{R}^{2n}$. For simplicity, let us denote the vectors $(S^\intercal_{t'-1}\mathbf{e}_{{\mu(t',r)}})_{t'\in [t], r\in [\kappa]}$ with $(\mathbf{v}_j)_{j\in[\kappa t]}$. It suffices to show that there exists an orthogonal symplectic matrix $O_{\mathrm{aux}}$ such that for all $j\in[\kappa t]$ and all $m\in\{2\kappa t+1,\ldots, 2n\}$ it holds that
\bb
    \mathbf{e}^T_{m} O_{\mathrm{aux}}^\intercal \mathbf{v}_{j}=0\,.
\ee
For any arbitrary vectors $\{\mathbf{v}_j\}_{j\in[\kappa t]}$, there exist such an orthogonal matrix $O_{\mathrm{aux}}$. This is a consequence of the forthcoming Lemma~\ref{le:exist}, pointed out in \cite{mele2024efficient}, which is a consequence of the well-known isomorphism between $2n\times 2n$ symplectic orthogonal real matrices and $n\times n$ unitaries~\cite{BUCCO}.
\end{proof}
\begin{lemma}
\label{le:exist}
    Let $\{\mathbf{e}_{i}\}^{2n}_{i=1}$ be the canonical basis of $\mathbb{R}^{2n}$.
    Let $\mathbf{v}_1, \dots, \mathbf{v}_M \in \mathbb{R}^{2n}$ be a set of real vectors, where $M \leq 2n$. There exists an orthogonal symplectic matrix $O \in \mathrm{O}(2n)\cap \mathrm{Sp}(2n,\mathbb{R})$ such that 
    \begin{align}
        \mathbf{e}^T_{i}O\mathbf{v}_j = 0,
    \end{align} for all $i \in \{2M+1, \dots, 2n\}$ and $j \in [M]$, meaning that all $\{O\mathbf{v}_j\}^{M}_{j=1}$ are exclusively supported on the span of the first $2M$ canonical basis vectors.
\end{lemma}
The proof of Lemma~\ref{le:exist} can be found in \cite{mele2024efficient}. The decomposition for $t$-doped Gaussian unitaries proved in Theorem~\ref{th:comprUni} implies a similar for $t$-doped states, as established by the following theorem.
\begin{thm}\label{thm_compression_state}
    Let $n\in\N$ be the number of modes. Let $\kappa,t\in \N$ such that $\kappa t\le n$. For any $(t,\kappa)$-doped Gaussian state vector $\ket{\psi}$ there exists a Gaussian unitary $G$ and a (possibly non-Gaussian) $\kappa t$-mode state vector $\ket{\phi}$ such that
    \bb\label{eq_decomposition_pure}
        \ket{\psi}=G \left(\ket{\phi}\otimes\ket{0}^{\otimes(n-\kappa t)}\right)\,.
    \ee
    Moreover, $G$ can be expressed in terms of the first moment and the covariance matrix of $\ket{\psi}$ as follows. The symplectic diagonalisation of the covariance matrix is of the form
    \bb
        V(\ketbra{\psi})=S(D\oplus  I_{2(n-\kappa t)})S^\intercal\,,
    \ee
    with $D$ being a $2\kappa t\times 2\kappa t$ diagonal matrix and $S$ being a symplectic matrix. Then, $\ket{\psi}$ can be expressed as
    \bb\label{eq_dec_pure_lorenzo}
        \ket{\psi}= \hat{D}_{\textbf{m}(\psi)}U_S \ket{\varphi_{\kappa t}}\otimes\ket{0}^{\otimes(n-\kappa t)}\,,
    \ee
    for some $\kappa t$-mode (possibly non-Gaussian) state vector $\ket{\varphi_{\kappa t}}$, where $\hat{D}_{\textbf{m}(\psi)}$ is the displacement operator and $U_S$ is the Gaussian unitary associated with the symplectic $S$.
\end{thm}
\begin{proof}
    By definition, any $(t,\kappa)$-doped Gaussian state vector $\ket{\psi}$ can be written in terms of a $(t,\kappa)$-doped Gaussian unitary $U_t$ as $\ket{\psi}\coloneqq U_{t}\ket{0}^{\otimes n}$. Consequently, Theorem~\ref{th:comprUni} implies that 
    \bb
        \ket{\psi}&=G(u_{\kappa t}\otimes\mathbb{1}_{n-\kappa t})G_{\text{passive}}\ket{0}^{\otimes n}\,\\
        &=G(u_{\kappa t}\otimes\mathbb{1}_{n-\kappa t}) \ket{0}^{\otimes n},
    \ee
    where we have used that $G_{\text{passive}}$ is passive and thus it preserves the vacuum state, i.e., $G_{\text{passive}}\ket{0}^{\otimes n}=\ket{0}^{\otimes n}$. Hence,~\eqref{eq_decomposition_pure} follows by defining $\ket{\phi}\coloneqq u_{\kappa t}\ket{0}^{\otimes \kappa t}$. Now, let us prove~\eqref{eq_dec_pure_lorenzo}. Note that the first moment and the covariance matrix of the state vector $U_S^\dagger \hat{D}_{\textbf{m}(\psi)}^\dagger\ket{\psi}$ are given by
    \bb
        \textbf{m}\!\left(U_S^\dagger \hat{D}_{\textbf{m}(\psi)}^\dagger\ketbra{\psi}\hat{D}_{\textbf{m}(\psi)}U_S\right)&=\textbf{0}_n\,,
        \\V\!\left(U_S^\dagger \hat{D}_{\textbf{m}(\psi)}^\dagger\ketbra{\psi}\hat{D}_{\textbf{m}(\psi)}U_S\right)&= D\otimes  I_{n-\kappa t}\,.
    \ee
    and thus
    \bb
        \textbf{m}\!\left(\Tr_{[\kappa t]}\!\left[U_S^\dagger \hat{D}_{\textbf{m}(\psi)}^\dagger\ketbra{\psi}\hat{D}_{\textbf{m}(\psi)}U_S\right]\right)&=\textbf{0}_{n-\kappa t}\,,
        \\V\!\left(\Tr_{[\kappa t]}\!\left[U_S^\dagger \hat{D}_{\textbf{m}(\psi)}^\dagger\ketbra{\psi}\hat{D}_{\textbf{m}(\psi)}U_S\right]\right)&=  I_{n-\kappa t}\,.
    \ee
    Hence, the forthcoming Lemma~\ref{lemma_vacuum} establishes that
    \bb
        \Tr_{[\kappa t]}\!\left[U_S^\dagger \hat{D}_{\textbf{m}(\psi)}^\dagger\ketbra{\psi}\hat{D}_{\textbf{m}(\psi)}U_S\right]=\ketbra{0}^{\otimes(n-\kappa t)}\,.
    \ee
    Consequently, there exists a $\kappa t$-mode state vector $\ket{\phi}$ such that 
    \bb
        U_S^\dagger \hat{D}_{\textbf{m}(\psi)}^\dagger\ket{\psi}=\ket{\phi}\otimes\ket{0}^{\otimes(n-\kappa t)}\,,
    \ee
    which concludes the proof.
\end{proof}
\begin{lemma}\label{lemma_vacuum}
    Let $\rho$ be an $n$-mode state with covariance matrix equal to the identity and first moment equal to $\textbf{r}\in\mathbb{R}^{2n}$,  
    \bb
        V\!(\rho)&=I_{2n}\,,\\
        \textbf{m}(\rho)&=\textbf{r}\,.
    \ee
    Then, $\rho$ is the coherent state $\rho=\ketbra{\mathbf{r}}$, where $\ket{\mathbf{r}}\coloneqq\hat{D}_{\mathbf{r}}\ket{0}^{\otimes n}$.   
\end{lemma}
\begin{proof}
    It suffices to show that $\Tr\!\left[ \ketbra{\mathbf{r}}  \rho  \right]=1$, which is valid since it holds that 
    \bb
        1&\ge \Tr\!\left[ \ketbra{\mathbf{r}}  \rho  \right]\\
        &= \Tr\!\left[ \ketbra{0}^{\otimes n}  \hat{D}_{\mathbf{r}}^\dagger \rho \hat{D}_{\mathbf{r}} \right]
        \\&\ge \Tr\!\left[\left(\mathbb{1}-  \sum_{i=1}^n a_i^\dagger a_i  \right)\hat{D}_{\mathbf{r}}^\dagger \rho \hat{D}_{\mathbf{r}}\right]
        \\&\eqt{(i)} 1- \frac{\Tr\!\left[V\!\left(\hat{D}_{\mathbf{r}}^\dagger \rho \hat{D}_{\mathbf{r}}\right)-I_{2n}\right]}{4}-\frac{\|\textbf{m}(\hat{D}_{\mathbf{r}}^\dagger \rho \hat{D}_{\mathbf{r}})\|_2^2}{2}
        \\&= 1- \frac{\Tr\!\left[V\!\left( \rho \right)-I_{2n}\right]}{4}-\frac{\|\textbf{m}(\rho )-\textbf{r}\|_2^2}{2}
        \\&=1\,,
    \ee
    where in (i) we have used Lemma~\ref{lemma_energy_traceV}.
\end{proof}

\subsection{$t$-compressible Gaussian states}
In this section, we present the notion of $t$-compressible states.
\begin{Def}[($t$-compressible states)]
\label{def:tcompr}
Let $t,n\in\N$ with $t\le n$. An $n$-mode state vector $\ket{\psi}$ is said to be $t$-compressible when there exists a Gaussian unitary $G$ and a $t$-mode state vector $\ket{\phi}$ such that
\bb
    \ket{\psi}=G\ket{\phi}\otimes \ket{0}^{\otimes(n-t)}\,.
\ee
\end{Def}
The above Gaussian unitary $G$ can be chosen as follows.
\begin{lemma}[(Choice of states)]\label{lemma_t_compr}
    Let $\ket{\psi}$ be an $n$-mode $t$-compressible state vector. Let 
    \bb
        V(\ketbra{\psi})=S \left( D_t \oplus I_{2(n-t)}\right) S^\intercal\,.
    \ee 
    be the symplectic diagonalisation of the covariance matrix of $V(\ketbra{\psi})$. Then, $\ket{\psi}$ can be written as
    \bb
        \ket{\psi}=\hat{D}_{\textbf{m}(\ketbra{\psi})}U_S\ket{\phi}\otimes \ket{0}^{\otimes (n-t)}\,,
    \ee
    where $\ket{\phi}$ is a $t$-mode state vector with first moment and covariance matrix given by
    \bb
        V(\ketbra{\phi})&=D_t\,,\\
        \textbf{m}(\ketbra{\phi})&=0\,.
    \ee
\end{lemma}

\begin{proof}
    The proof is the same as the one of Theorem~\ref{thm_compression_state}.
\end{proof}
As a consequence of Theorem~\ref{thm_compression_state}, we have the following.
\begin{cor}[(Compressibility)]
\label{tdopedcompr}
    A $(t,\kappa)$-doped Gaussian state is $\kappa t$-compressible.
\end{cor}
Note that the covariance matrix of a $t$-compressible state has at least $n-t$ symplectic eigenvalues equal to $1$. This motivates the following definition, in analogy to the stabiliser dimension~\cite{grewal2023efficient} and fermionic Gaussian dimension~\cite{mele2024efficient}.
\begin{Def}[(Gaussian dimension of a state)]
\label{def:tcompr2}
The Gaussian dimension of a state vector $\ket{\psi}$ is the number of symplectic eigenvalues of the covariance matrix $V(\ketbra{\psi})$ which are equal to one. 
\end{Def}
We now observe the following fact.
\begin{lemma}[(Characterization of $t$-compressible pure states)]
    An $n$-mode pure state is $t$-compressible if and only if its Gaussian dimension is at least $n-t$. 
\end{lemma}
\begin{proof}
    First, let us assume that the Gaussian dimension of the pure state vector $\ket{\psi}$ is $n-t$. Hence, the symplectic diagonalisation of its covariance matrix is of the form
    \bb
        V(\ketbra{\psi})=S \left( D_t \oplus I_{2(n-t)}\right) S^\intercal\,.
    \ee
    Consequently, it holds that
    \bb
        V\!\left(\Tr_{[t]}[U_S^\dagger \ketbra{\psi} U_S]\right) = I_{2(n-t)}\,.
    \ee
    Hence, Lemma~\ref{lemma_vacuum} establishes that $\Tr_{[t]}[U_S^\dagger \ketbra{\psi} U_S]=\ketbra{\mathbf{r}}$ for some $\mathbf{r}\in\mathbb{R}^{2(n-t)}$. Consequently, there exists a $t$-mode state vector $\ketbra{\phi}$ such that
    \bb
        U_S^\dagger \ket{\psi}   = \ket{\phi}\otimes \ket{\mathbf{r}}\,.
    \ee
    This implies that
    \bb
        \ket{\psi}= U_S (\mathbb{1}_{t}\otimes \hat{D}_{\mathbf{r}}) \ket{\phi}\otimes \ket{0}^{\otimes (n-t)}\,,
    \ee
    and in particular that $\ket{\psi}$ is $t$-compressible. On the other hand, the fact that the Gaussian dimension of a $t$-compressible state is at least $n-t$ is trivial. 
\end{proof}

\subsection{Tomography algorithm for $t$-compressible Gaussian states}
In this subsection, we present a tomography algorithm for learning $t$-compressible $n$-mode bosonic states. Thus, because of Theorem~\ref{tdopedcompr}, this algorithm can also be applied to doped quantum states. Our algorithm is outlined in Table~\ref{Table_t_compressible}, and its correctness is proved in the forthcoming Theorem~\ref{thm_corr_compression_state}.
The algorithm turns out to be efficient for $t = O(1)$ (i.e., it runs in polynomial time in the number of modes), while super-poly (yet sub-exponential) for $t = O(\log n)$. Conversely, in Theorem~\ref{thm_lower_bound_t_compressible} we show that, in the regime $t = O(\log n)$ \emph{any} tomography algorithm designed to learn $t$-compressible states must be inefficient. 

Our tomography algorithm is motivated by Theorem~\ref{thm_compression_state} and consists of two main parts, summarised as follows: 
In the first part, it estimates the covariance matrix of the unknown state, expresses it in its Williamson decomposition~\eqref{eq:will}, and uses this decomposition to find a Gaussian unitary $\tilde{G}^{\dag}$ that can \emph{approximately} transform the $n$-mode state into a tensor product between an arbitrary state on the first $t$ modes and the vacuum state in the last $n-t$ modes.
In the second part, the algorithm applies $\tilde{G}^{\dag}$ to the unknown state. It then measures the occupation number of the last $n-t$ modes, and if the vacuum state is obtained, it performs full state tomography on only the first $t$ modes, resulting in a state vector $\ket{\tilde{\phi}}$.
The output of the algorithm will then be a classical representation of the state $\tilde{G}(\ket{\tilde{\phi}} \otimes \ket{0}^{\otimes (n-t)})$. Refer to Table~\ref{Table_t_compressible} for more precise details of the algorithm.

Remarkably, our tomography algorithm is experimentally feasible, as it requires only Gaussian evolutions, avalanche photodiodes (devices capable of reliably distinguishing between zero and one or more photons, commonly referred to as 'on/off detectors'~\cite{BUCCO}), and easily implementable Gaussian measurements, such as homodyne and heterodyne detection~\cite{BUCCO}. Specifically, the estimation of the first moment and covariance matrix in Line 1 relies solely on homodyne measurements~\cite{BUCCO, Aolita_2015}. Line 4 involves only Gaussian unitaries, while Line 5 uses a photon counting that post-selects on the vacuum of the last $n-t$ modes when no photons are detected. This latter measurement can be practically achieved using photodetectors, like avalanche photodiodes, capable of discriminating between zero and one or more photons. Lastly, in Line 8, quantum state tomography of the first $t$ modes is performed. To accomplish this in an experimentally feasible manner, we may utilise the continuous-variable classical shadow algorithm proposed in \cite{becker_classical_2023}, which relies solely on randomised Gaussian unitaries and homodyne and heterodyne measurements. (Alternatively, in order to improve the performance, we may also consider utilising our tomography algorithm for moment-constrained pure states detailed in Section~\ref{sec_upper_bound_sample}. However, despite its better performance, the latter algorithm is not feasible for experimental implementation with our current technology.)

Before delving into the correctness proof of the algorithm, let us establish the notation. 
Let $N_{\mathrm{cov}}(n,\varepsilon,\delta,E_2)$ be the sample complexity of the algorithm outlined in Table~\ref{Table_covariance} and Theorem~\ref{correctness_algorithm_cov} to estimate the first moment and the covariance matrix of an unknown $n$-mode state $\rho$, subject to the second-moment energy constraint ${\Tr[\rho\hat{E}_n^2]}\le n^2 E_2^2$, with precision $\varepsilon$ and failure probability $\delta$. 
Specifically, thanks to Theorem~\ref{correctness_algorithm_cov}, a number
 \bb\label{N_cov}
    N_{\mathrm{cov}}\!\left(n,\varepsilon,\delta,E_2\right)=O\!\left( \log\!\left(\frac{n^2}{\delta}\right)\frac{n^3E_2^2}{\varepsilon^2}   \right)\,
\ee
of copies of $\rho$ are sufficient to build a covariance matrix $\tilde{V}$ and a vector $\tilde{m}$ such that
\bb
    \Pr\left(  \|\tilde{V}'-V\!(\rho)\|_\infty\le \varepsilon\quad \text{ and }\quad \|\tilde{\textbf{m}}-\textbf{m}(\rho)\|_2\le \varepsilon \right)\ge 1-\delta\,.
\ee

Furthermore, we denote as $N_{\mathrm{tom,CV}}(n,\varepsilon,\delta,E_1)$ the sample complexity of the full state tomography algorithm outlined in Table~\ref{Table_tomography_ec} and Theorem~\ref{correctness_algorithm_ECpure} to learn an unknown $n$-mode pure state $\psi$, subject to the energy constraint $\Tr[\psi\hat{E}_n]\le n E_1$, with accuracy $\varepsilon$ in trace distance and failure probability $\delta$.
Specifically, Theorem~\ref{correctness_algorithm_ECpure} implies that
\begin{align}
\label{N_cv_tom}
N_{\mathrm{tom,CV}}(n,\varepsilon,\delta,E_1) = \left\lceil 2^{21}\frac{\deff(n,\varepsilon,E_1)}{\varepsilon^2} \log\left(\frac{4}{\delta}\right)\right\rceil= O\left(\frac{E_1}{\varepsilon^{2}}\right)^{n}
\end{align}
copies of $\psi$ are sufficient to generate a classical representation (i.e., $\deff(n,\varepsilon,E_1)$-sized vector) of a pure state $\tilde{\psi}$ such that
\bb
\Pr\left[\frac{1}{2} \|\psi-\tilde{\psi}\|_1 \le \varepsilon \right] \ge 1-\delta,
\ee
where $\deff(n,\varepsilon,E_1) \le \left(\frac{eE_1}{\varepsilon^{2}}+2e\right)^n$.

\begin{table}[t]
  \caption{Tomography algorithm for bosonic $t$-compressible pure states}
  \begin{mdframed}[linewidth=2pt, roundcorner=10pt, backgroundcolor=white!10, innerbottommargin=10pt, innertopmargin=10pt]
    \textbf{Input:} Accuracy $\varepsilon$, failure probability $\delta$, second moment upper bound $E$, $N$ copies of the unknown $n$-mode $t$-compressible state vector $\ket{\psi}$ satisfying the 2-moment constraint ${\Tr\!\left[\hat{E}_n^2 \ketbra{\psi}\right]}\le n^2 E^2$, where
    \bb
        N&\coloneqq\ N_1+ \left\lceil 2N_2 + 24\log\!\left(\frac{3}{\delta}\right)\right\rceil\\
         &=O\!\left(\frac{n^9E^6}{\varepsilon^4}\right)+O\!\left(\left(\frac{n^2E^2}{\varepsilon^2}\right)^{\! \!t}\right)\,,
    \ee
    where $N_1\coloneqq N_{\mathrm{cov}}\!\left(n,\frac{\varepsilon^2}{2(n+1)(1+4nE)^2},\frac{\delta}{3}, E\right)$ is reported in~\eqref{N_cov} and $N_2\coloneqq N_{\mathrm{tom,CV}}\!\left(t,\frac{\varepsilon}{2},\frac{\delta}{3},80n^2E^2\right)$ is reported in~\eqref{N_cv_tom}. \\
    \textbf{Output:} A classical description of a pure state vector $\ket{\hat{\psi}}$ such that
    \bb
        \Pr\left(  \frac{1}{2}\left\|\ketbra{\hat{\psi}}-\ketbra{\psi}\right\|_1\le \varepsilon \right)\ge 1-\delta\,.
    \ee
    Such a classical description consists of the triplet $(  \tilde{\textbf{m}}, \tilde{S},\ket{\tilde{\phi}_1} )$ which defines $\ket{\hat{\psi}}$ via the relation
    \bb
        \ket{\hat{\psi}}\coloneqq \hat{D}_{\tilde{\textbf{m}}}U_{\tilde{S}} \ket{\tilde{\phi}_1}\otimes\ket{0}^{\otimes(n-t)} \,,
    \ee 
    where $\tilde{\textbf{m}}\in\mathbb{R}^{2n}$, $\tilde{S}\in\mathrm{Sp}(\mathbb{R}^{2n})$, and $\ket{\tilde{\phi}_1}$ is a $t$-mode pure state contained in a $\left\lceil\left(\frac{e(80n^2E^2-\frac{1}{2})}{\varepsilon^2}\right)^{\!t}\right\rceil$-dimensional subspace.
    \begin{algorithmic}[1]
    \State Query $N_1$ copies of $\ket{\psi}$ and apply the algorithm in Table~\ref{Table_covariance} to construct a vector $\tilde{\textbf{m}}$ and a covariance matrix $\tilde{V}$ which are estimates of the first moment $\textbf{m}(\ketbra{\psi})$ and of the covariance matrix $V(\ketbra{\psi})$, respectively.\label{steptd:1}
    \State Compute the symplectic diagonalisation $\tilde{V}=\tilde{S}\tilde{D}\tilde{S}^\intercal$, where $\tilde{S}\in\mathrm{Sp}(2n)$ and $\tilde{D}\coloneqq\diag(\tilde{d}_1,\tilde{d}_1,\tilde{d}_2,\tilde{d}_2,\ldots, \tilde{d}_n,\tilde{d}_n)$, with $\tilde{d}_1\ge \tilde{d}_2\ge \ldots\ge \tilde{d}_n\ge 1$. \label{steptd:2}
    \For{$k \leftarrow 1$ \textbf{to} $\left\lceil 2N_2 + 24\log\!\left(\frac{3}{\delta}\right)\right\rceil$}
    \State Query one copy of $\ket{\psi}$ and apply $\hat{D}^\dagger_{\tilde{\textbf{m}}}U_{\tilde{S}}^\dagger$, obtaining the state vector $\ket{\tilde{\psi}_t}\coloneqq \hat{D}^\dagger_{\tilde{\textbf{m}}}U_{\tilde{S}}^\dagger\ket{\psi}$.\label{steptd:4}
    \State Measure $\ket{\tilde{\psi}_t}$ with respect the POVM $\{M_0\coloneqq \mathbb{1}_{t}\otimes\ketbra{0}^{\otimes (n-t)}, M_1\coloneqq \mathbb{1}-M_0 \}$, and discard if the outcome corresponds to the POVM element $M_1$. The post-measurement state is thus of the form $\ket{\tilde{\phi}}\otimes \ket{0}^{\otimes (n-t)}$, where $\ket{\tilde{\phi}}$ is the pure state proportional to $\bra{0}^{\otimes(n-t)}\ket{\tilde{\psi}_t}$.\label{steptd:5}
    \State Do a step of the algorithm in Table~\ref{Table_tomography_ec} to perform pure-state tomography of the $t$-mode state vector $\ket{\tilde{\phi}}$, which has mean energy upper bounded by $80n^2E^2$. \label{steptd:6}
    \EndFor
    \State The tomography algorithm of Line~6 returns a classical description of a $t$-mode pure state vector $\ket{\tilde{\phi}_1}$, which is supported on a $\left\lceil\left(\frac{e(80n^2E^2-\frac{1}{2})}{\varepsilon^2}\right)^{\!t}\right\rceil$-dimensional subspace.\label{steptd:8}
    \State\Return the triplet $\left(  \tilde{\textbf{m}}, \tilde{S},\ket{\tilde{\phi}_1} \right)$.\label{steptd:9}
    \end{algorithmic}
  \end{mdframed}
  \label{Table_t_compressible}
\end{table}

\begin{thm}[(Number of copies required)]\label{thm_corr_compression_state}
Let $n, t\in\N$ with $t\le n$, and $E\geq 0$. Let $\ket{\psi}$ be an $n$-mode $t$-compressible state satisfying the second-moment constraint ${\Tr[\ketbra{\psi} \hat{E}_n^2]}\le n^2 E^2$. Let $\varepsilon,\delta\in(0,1)$, and $N$ such that
    \bb
        N&\coloneqq\ N_1+ \left\lceil 2N_2 + 24\log\!\left(\frac{3}{\delta}\right)\right\rceil=O\!\left(\frac{n^9E^6}{\varepsilon^4}\right)+O\!\left(\left(\frac{nE^2}{\varepsilon^2}\right)^{\! \!t}\right)\,,
    \ee
where 
\begin{equation}
N_1\coloneqq N_{\mathrm{cov}}\!\left(n,\frac{\varepsilon^2}{2(n+1)(1+4nE)^2},\frac{\delta}{3}, E\right)
\end{equation}
is reported in~\eqref{N_cov} and $N_2\coloneqq N_{\mathrm{tom,CV}}\!\left(t,\frac{\varepsilon}{2},\frac{\delta}{3},80n^2E^2\right)$ is reported in~\eqref{N_cv_tom}.
Then, $N$ copies of $\ket{\psi}$ are sufficient to build a classical description of a $t$-compressible pure state vector $\ket{\hat{\psi}}$ such that
    \bb
        \Pr\left(  \frac{1}{2}\left\|\ketbra{\hat{\psi}}-\ketbra{\psi}\right\|_1\le \varepsilon \right)\ge 1-\delta\,.
    \ee
Such a classical description consists of the triplet $\left(  \tilde{\textbf{m}}, \tilde{S},\ket{\tilde{\phi}_1} \right)$, which defines $\ket{\hat{\psi}}$ as
\bb
    \ket{\hat{\psi}}\coloneqq \hat{D}_{\tilde{\textbf{m}}}U_{\tilde{S}} \ket{\tilde{\phi}_1}\otimes\ket{0}^{\otimes(n-t)} \,,
\ee 
where $\tilde{\textbf{m}}\in\mathbb{R}^{2n}$, $\tilde{S}\in\mathrm{Sp}(\mathbb{R}^{2n})$, and $\ket{\tilde{\phi}_1}$ is a $t$-mode pure state contained in a $\left\lceil\left(\frac{e(80n^2E^2-\frac{1}{2})}{\varepsilon^2}\right)^{\!t}\right\rceil$-dimensional subspace.
\end{thm}
Before proving Theorem~\ref{thm_corr_compression_state} let us state some useful lemmas.
\begin{lemma}[(Bounds to symplectic matrices)]\label{lem:SlowerboundV}
Let $V\!(\rho)$ be a covariance matrix associated with an $n$-mode state $\rho$ and let $V\!(\rho)=SDS^{\intercal}$ its symplectic diagonalisation. It holds that 
\bb
\|S\|_{\infty}\le \sqrt{\|V(\rho)\|_{\infty}}
\ee
In particular, it holds that
\bb
\|V(\rho)\|_{\infty}\le 4\Tr[\rho \hat{E}_n]
\ee
and consequently $\|S\|_{\infty}\le \sqrt{4\Tr[\rho \hat{E}_n]}$.
\end{lemma}
\begin{proof}
    Since $D\ge \mathbb{1}$, we have that $V(\rho)\ge SS^{\intercal}$. Consequently, it holds that 
    \bb
        \|S\|_\infty^2&=\|SS^{\intercal}\|_\infty
        \\&\le \|V\!(\rho)\|_\infty
        \\&\le \Tr V\!(\rho) 
        \\&= 4\Tr[\rho \hat{E}_n]-2\|\textbf{m}(\rho)\|_2^2
        \\&\le 4\Tr[\rho \hat{E}_n]\,,
    \ee
    where we have exploited Lemma~\ref{lemma_energy_traceV}.
    This implies $\|S\|_{\infty}\le \sqrt{4\Tr[\rho \hat{E}_n]}$, where $\hat{E}_n\coloneqq \frac{1}{2}\hat{\textbf{R}}^\intercal\hat{\textbf{R}}$ is the energy operator. 
\end{proof}

\begin{lemma}[(Perturbation on symplectic diagonalisation~\cite{idel2017})]\label{lem:wolf}
    Let $V_1,V_2\in\mathbb{R}^{2n\times 2n}$ be two covariance matrices with symplectic diagonalisations $V_1=S_{1}D_1S_{1}^{\intercal}$ and $V_2=S_{2}D_2S_{2}^{\intercal}$, where the elements on the diagonal of $D_1$ and $D_2$ are arranged in descending order. Then 
\bb
    \|D_1-D_2\|_\infty\le \sqrt{K\!(V_1)K\!(V_2)}\|V_1-V_2\|_\infty\,,
\ee
where $K\!(V)$ is the condition number of the covariance matrix $V$, defined as $K\!(V)\coloneqq \|V\|_\infty \|V^{-1}\|_\infty$.
\end{lemma}

\begin{lemma}[Inverse of a covariance matrix]\label{lem:lboundsVm1V}
Let $V$ be a covariance matrix. The inverse of the covariance matrix $V^{-1}$ satisfies
\bb
    \|V^{-1}\|\le \|V\|\,
\ee
where $\|\cdot\|$ is an orthogonal invariant matrix norm. As a consequence, the condition number of $V$, which is defined as $K(V)\coloneqq \|V\|_\infty\|V^{-1}\|_\infty$, can be upper bounded as 
\bb
    K(V)\le \|V\|^2_{\infty}\,.
\ee
\end{lemma}
\begin{proof}
    Let $V=SDS^\intercal$ be the symplectic diagonalisation of $V$ with $S$ symplectic and $D$ diagonal of the form $D=\diag(d_1,d_1,d_2,d_2,\dots,d_n,d_n)$. In particular, we recall that $S\Omega_n S^{\intercal}=\Omega_n$ and that $D$ commutes with $\Omega_n$. Then it holds that 
    \bb
        V^{-1}&= (S^{-1})^\intercal D^{-1}S^{-1}
        \\&=\Omega_n \Omega_n(S^{-1})^\intercal D^{-1}S^{-1} \Omega_n \Omega_n
        \\&= \Omega_n S \Omega_n D^{-1} \Omega_n S^\intercal\Omega_n
        \\&= - \Omega_n S  D^{-1}  S^\intercal\Omega_n
        \\&=  \Omega_n S  D^{-1}  (\Omega_n S)^\intercal\,.
    \ee
    Since $V$ is a covariance matrix, then $D\ge \mathbb{1}$ and thus $D^{-1}\le D$. Consequently, the matrix inequality 
    \bb
        V^{-1}&\le \Omega_n S  D  (\Omega_n S)^\intercal
        \\&=  \Omega_n V \Omega_n^\intercal\,
    \ee
    holds.    Since $\Omega_n$ is orthogonal, then for any orthogonal invariant norm $\|\cdot\|$ it holds that $ \|V^{-1}\|\le \|V\| $\cite{BHATIA-MATRIX}, that concludes the proof.
\end{proof}

\begin{lemma}[(Upper bound to mean energy)]\label{lemma:energyboundsympl}
    Let $\rho$ be an $n$-mode state. Let $S\in \mathrm{Sp}(2n)$ be a symplectic matrix. Let $U_S$ be the Gaussian symplectic unitary associated with $S$. Then, the mean energy of $U_S \rho U_S^\dagger$ can be upper bounded as
    \bb
        E\!\left(U_S \rho U_S^\dagger\right)\le \|S\|^2_\infty E(\rho)\,,
    \ee
    where $E(\rho)\coloneqq \Tr\!\left[ \frac{\hat{\textbf{R}}^\intercal\hat{\textbf{R}}}{2}\,\rho\right]$. 
\end{lemma}
\begin{proof}
    It holds that
    \bb
        E\!\left(U_S \rho U_S^\dagger\right)&=\frac{1}{2}\Tr\!\left[ U_S \rho U_S^\dagger \hat{\mathbf{R}}^\intercal\hat{\mathbf{R}} \right]
        \\&\eqt{(i)}\frac{1}{2}\Tr\!\left[ \rho  \hat{\mathbf{R}}^\intercal S^\intercal S\hat{\mathbf{R}} \right]
        \\&\eqt{(ii)} \frac{1}{2}\Tr\!\left[ \rho  \hat{\mathbf{R}}^\intercal O_1^\intercal Z^2 O_1 \hat{\mathbf{R}} \right]
        \\&\eqt{(iii)} \frac{1}{2}\Tr\!\left[ U_{O_1}\rho U_{O_1}^\dagger  \hat{\mathbf{R}}^\intercal  Z^2  \hat{\mathbf{R}} \right]
        \\&\leqt{(iv)} \|S\|_\infty^2 \Tr\!\left[ U_{O_1}\rho U_{O_1}^\dagger \hat{E} \right]
        \\&\eqt{(v)} \|S\|_\infty^2 \Tr\!\left[ \rho \hat{E} \right].
    \ee
Here, in (i), we have used that $U_S^\dagger \hat{\mathbf{R}}U_S= S\mathbf{R}$. In (ii), we have exploited the Euler decomposition $S=O_1ZO_2$ given in~\eqref{Euler_dec}. In (iii), we have used $U_{O_1}^\dagger \hat{\mathbf{R}}U_{O_1}= O_1\mathbf{R}$. In (iv), we first used that $\hat{\mathbf{R}}^{\intercal}Z^{2}\hat{\mathbf{R}}=\sum_{i=1}^{2n}z_i^2 \hat{R}_{i}^{2}\le \|Z\|_{\infty}^{2}\hat{\mathbf{R}}^{\intercal}\hat{\mathbf{R}}$, and, second, we have exploited that $\|S\|_{\infty}^{2}=\|S^{\intercal}S\|_{\infty}=\|O_{2}^{\intercal}Z^{2}O_{2}\|_{\infty}=\|Z\|_{\infty}^{2}$ to obtain that $\hat{\mathbf{R}}^{\intercal}Z^{2}\hat{\mathbf{R}}\le 2\|S\|_\infty^2\hat{E}$. Finally, in (v) we have used again $U_{O_1}^\dagger \hat{\mathbf{R}}U_{O_1}= O_1\mathbf{R}$ and the fact that $O_1$ is orthogonal.
\end{proof}

\begin{lemma}[(Refined mean energy bound)]\label{lemma:energydispl}
    Let $\rho$ be an $n$-mode state. Let $\mathbf{r}\in\mathbb{R}^{2n}$ and let $\hat{D}_{\mathbf{r}}$ be the associated displacement operator. Then, the mean energy of $\hat{D}_{\mathbf{r}} \rho \hat{D}_{\mathbf{r}}^\dagger$ reads
    \bb
        E\!\left(\hat{D}_{\mathbf{r}} \rho \hat{D}_{\mathbf{r}}^\dagger\right)= E(\rho)+\mathbf{r}^\intercal \mathbf{m}(\rho)+\frac{1}{2}\|\mathbf{r}\|_2^2\,,
    \ee
    where $E(\rho)\coloneqq \Tr\!\left[ \frac{\hat{\textbf{R}}^\intercal\hat{\textbf{R}}}{2}\,\rho\right]$. In particular, it can be upper bounded as
    \bb
        E\!\left(\hat{D}_{\mathbf{r}} \rho \hat{D}_{\mathbf{r}}^\dagger\right)\le E(\rho)+ \sqrt{2E(\rho)}\|\textbf{r}\|_2 +\frac{1}{2}\|\textbf{r}\|_2^2\,.
    \ee
    
\end{lemma}
\begin{proof}
    It holds that
    \bb
    E\!\left(\hat{D}_{\mathbf{r}} \rho \hat{D}_{\mathbf{r}}^\dagger\right)&=\frac{1}{2}\Tr[\rho (\hat{\mathbf{R}}+\mathbf{r}\mathbb{1})^\intercal(\hat{\mathbf{R}}+\mathbf{r}\mathbb{1})]
    \\&= E(\rho)+\mathbf{r}^\intercal\textbf{m}(\rho)+ \frac{1}{2}\|\mathbf{r}\|_2^2\,.
    \ee
    Moreover, note that 
    \bb
        m_i(\rho)=\Tr[\rho R_i]\le \sqrt{\Tr[\rho R_i^2]}
    \ee
    for each $i\in[2n]$, where we have exploited~\eqref{eq:conc}. Consequently, it holds that 
    \bb\label{eq:bounddisplene}
        \|\textbf{m}(\rho)\|_2\le \sqrt{\sum_{i=1}^n\Tr[\rho R_{i}^2]}= \sqrt{2E(\rho)}\,.
    \ee
    Thus, it follows that
    \bb
        E\!\left(\hat{D}_{\mathbf{r}} \rho \hat{D}_{\mathbf{r}}^\dagger\right)&\le E(\rho)+\|\textbf{r}\|_2\|\textbf{m}(\rho)\|_2+\frac{1}{2}\|\textbf{r}\|_2^2
        \\&\le E(\rho)+ \sqrt{2E(\rho)}\|\textbf{r}\|_2+ \frac{1}{2}\|\textbf{r}\|_2^2\,.
    \ee
\end{proof}

Now, we are ready to prove Theorem~\ref{thm_corr_compression_state}.
\begin{proof}[Proof of Theorem~\ref{thm_corr_compression_state}]
    We aim to establish the correctness of the algorithm presented in Table~\ref{Table_t_compressible}.

    In Line~\ref{steptd:1}, we use $N_{\mathrm{cov}}\!\left(n,\varepsilon_{\mathrm{cov}},\frac{\delta}{3}, E\right)$ copies of $\ket{\psi}$ in order to build a vector $\tilde{\textbf{m}}\in\mathbb{R}^{2n}$ and a covariance matrix $\tilde{V}\in\mathbb{R}^{2n,2n}$ such that
    \bb\label{prob_state_perturb2}
        \Pr\!\left(  \|\tilde{V}-V\!(\ketbra{\psi})\|_\infty\le \varepsilon_{\mathrm{cov}}\quad \text{ and }\quad \|\tilde{\textbf{m}}-\textbf{m}(\ketbra{\psi})\|_2\le \varepsilon_{\mathrm{cov}}  \right)\ge 1-\frac{\delta}{3}\,,
    \ee 
    where $N_{\mathrm{cov}}(\cdot)$ is defined in~\eqref{N_cov} and the accuracy $\varepsilon_{\mathrm{cov}}$ will be fixed later.
    From now on, let us assume that we are in the probability event in which 
    \bb
        \|\tilde{V}-V\!(\ketbra{\psi})\|_\infty\le \varepsilon_{\mathrm{cov}}\quad \text{ and }\quad \|\tilde{\textbf{m}}-\textbf{m}(\ketbra{\psi})\|_2\le \varepsilon_{\mathrm{cov}}\,.
    \ee
    In Line~\ref{steptd:2}, we compute the symplectic diagonalisation $\tilde{V}=\tilde{S}\tilde{D}\tilde{S}^\intercal$, where $\tilde{S}\in\mathrm{Sp}(2n)$ and  
    \bb
        \tilde{D}\coloneqq\diag\left(\tilde{d}_1,\tilde{d}_1,\tilde{d}_2,\tilde{d}_2,\ldots, \tilde{d}_n,\tilde{d}_n\right)
    \ee
    with $\tilde{d}_1\ge \tilde{d}_2\ge \ldots\ge \tilde{d}_n\ge 1$.

    The sequence of steps in Lines~\ref{steptd:4},~\ref{steptd:5}, and~\ref{steptd:6} have to be repeated a number
    \bb
        \left\lceil 2N_{\mathrm{tom,CV}}\!\left(t,\varepsilon_{\text{tom}},\frac{\delta}{3},E\right) + 24\log\!\left(\frac{3}{\delta}\right)\right\rceil
    \ee
    of times, where the quantity $N_{\mathrm{tom,CV}}\!\left(t,\varepsilon_{\text{tom}},\frac{\delta}{3},E\right)$ is the number of copies sufficient for tomography of a $t$-mode pure state with mean energy per mode at most $E$ with accuracy $\varepsilon_{\text{tom}}$ and failure probability $\frac{\delta}{3}$ (see~\eqref{N_cv_tom}). The accuracy $\varepsilon_{\text{tom}}$ and the energy $E$ will be fixed later.

    In Line~\ref{steptd:4}, we query one copy of $\ket{\psi}$ and apply $\hat{D}^\dagger_{\tilde{\textbf{m}}}U_{\tilde{S}}^\dagger$, obtaining the state vector $\ket{\tilde{\psi}_t}\coloneqq \hat{D}^\dagger_{\tilde{\textbf{m}}}U_{\tilde{S}}^\dagger\ket{\psi}$. Since $\ket{\psi}$ is $t$-compressible and since it holds that $\tilde{\textbf{m}}\simeq \textbf{m}(\ketbra{\psi})$ and $\tilde{V}=\tilde{S}\tilde{D}\tilde{S}^\intercal\simeq V(\ketbra{\psi})$, we intuitively expect that the reduced state of $\ket{\tilde{\psi}_t}$ onto the last $n-t$ modes is very close to the vacuum $\ket{0}^{\otimes(n-t)}$ (as we are going to rigorously show).
    
    In Line~\ref{steptd:5}, we measure $\ket{\tilde{\psi}_t}$ with respect the POVM $\{ \mathbb{1}_{t}\otimes\ketbra{0}^{\otimes (n-t)}, \mathbb{1}-\mathbb{1}_{t} \otimes\ketbra{0}^{\otimes (n-t)} \}$, and discard the copies associated with the outcome corresponding to the POVM element $\mathbb{1}-\mathbb{1}_{t} \otimes\ketbra{0}^{\otimes (n-t)}$. The post-measurement state is not discarded with probability 
    \bb
        P_{\text{succ}}\coloneqq \Tr\!\left[ \mathbb{1}_{t} \otimes\ketbra{0}^{\otimes (n-t)}\, \ketbra{\tilde{\psi}_t}\right]\label{eq:psucc}\,,
    \ee
    and the remaining copies are thus in a state of the form $ \ket{\tilde{\phi}}\otimes \ket{0}^{\otimes (n-t)}$, where $\ket{\tilde{\phi}}$ is given by
    \bb
        \ket{\tilde{\phi}}\coloneqq \frac{\bra{0}^{\otimes(n-t)}\ket{\tilde{\psi}_t}}{\sqrt{P_{\text{succ}}}}\,.
    \ee
    Let us now find a lower bound on the probability of success $P_{\text{succ}}$. For simplicity in the following we will use the notation $\psi\coloneqq \ketbra{\psi}$, $\tilde{\psi}_t\coloneqq \ketbra{\psi}$ and $\tilde{\phi}\coloneqq \ketbra{\tilde{\phi}}$.
    Note that
    \bb\label{eq:3}
        P_{\text{succ}}&= \Tr[ \mathbb{1}_{t}\otimes \ketbra{0}^{\otimes (n-t)}\, \tilde{\psi}_t]
        \\&= \Tr[ \ketbra{0}^{\otimes (n-t)}\, \Tr_{[t]}\tilde{\psi}_t]
        \\&\geqt{(i)} \Tr\!\left[ \left(\mathbb{1}-\sum_{i=t+1}^{n}a_i^\dagger a_i\right)\, \Tr_{[t]}\tilde{\psi}_t\right]
        \\&\eqt{(ii)}1-\frac{\Tr[V(\Tr_{[t]}\tilde{\psi}_t)-I_{n-t}]}{4}-\frac{\|\textbf{m}(\Tr_{[t]}\tilde{\psi}_t)\|_2^2}{2}
        \\&\ge 1-\frac{2(n-t)\|V(\Tr_{[t]}\tilde{\psi}_t)-I_{n-t}\|_\infty}{4}-\frac{\|\textbf{m}(\tilde{\psi}_t)\|_2^2}{2}\,,
    \ee
    Here, in (i), we have used the following operator inequality
    \bb
       \mathbb{1}_{t}\otimes \ketbra{0}^{\otimes (n-t)}\ge \mathbb{1}-\sum_{i=t+1}^{n}a_i^\dagger a_i\,.
    \ee
    which can be readily verified. In (ii), we have exploited Lemma~\ref{lemma_energy_traceV}. Moreover, note that
    \bb
        V(\tilde{\psi}_t)&=V\!\left(  U_{\tilde{S}}^\dagger \hat{D}^\dagger_{\tilde{\mathbf{m}}} \psi \hat{D}_{\tilde{\mathbf{m}}} U_{\tilde{S}}  \right)
        \\&=\tilde{S}^{-1} V(\psi) (\tilde{S}^{-1})^\intercal
        \\&= \tilde{S}^{-1} \left(V(\psi)- \tilde{V}\right)(\tilde{S}^{-1})^\intercal+  \tilde{D}\,.
    \ee
    Since $\psi$ is $t$-compressible, Lemma~\ref{lemma_t_compr} establishes that $\ket{\psi}$ can be written as
    \bb
        \ket{\psi}=\hat{D}_{\textbf{m}(\psi)}U_S \ket{\phi_t}\otimes\ket{0}^{\otimes (n-t)}\,,\label{eq:2}
    \ee
    where $\ket{\phi_t}$ is some $t$-mode state, $S$ is a symplectic matrix such that $V(\psi)=S\left(D_t\oplus \mathbb{1}_{n-t}\right)S^\intercal$, and $D_t$ is $2t\times 2t$ diagonal matrix of the form
    \bb
        D_t\coloneqq\diag\left({d}_1,{d}_1,{d}_2,{d}_2,\ldots, {d}_n,{d}_n\right)
    \ee
    with ${d}_1\ge {d}_2\ge \ldots\ge {d}_n\ge 1$. Let us analyse the first term in the r.h.s.~of~\eqref{eq:3}. It holds that
    \bb
        \|V(\Tr_{[t]}\tilde{\psi}_t)-I_{n-t}\|_\infty &\eqt{(iii)} \left\| \left[\tilde{S}^{-1} \left(V(\psi)- \tilde{V}\right)(\tilde{S}^{-1})^\intercal\right]_{n-t}+  [\tilde{D}]_{n-t}-I_{n-t}\right\|_\infty\\
        &\le \left\| \left[\tilde{S}^{-1} \left(V(\psi)- \tilde{V}\right)(\tilde{S}^{-1})^\intercal\right]_{n-t}\right\|_\infty + \left\|[\tilde{D}]_{n-t}-I_{n-t}\right\|_\infty
        \\&\leqt{(iv)} \left\| \tilde{S}^{-1} \left(V(\psi)- \tilde{V}\right)(\tilde{S}^{-1})^\intercal\right\|_\infty+ \left\| \tilde{D} -D \right\|_\infty     
        \\&\leqt{(v)} \|\tilde{S}\|_\infty^2\|V(\psi)- \tilde{V}\|_\infty+\left\| \tilde{D} -D \right\|_\infty  
        \\&\leqt{(vi)} \|\tilde{V}\|_\infty\varepsilon_{\mathrm{cov}}+\left\| \tilde{D} -D \right\|_\infty  
        \\&\leqt{(vii)} (\varepsilon_{\mathrm{cov}}+ \|V(\psi)\|_\infty)\varepsilon_{\mathrm{cov}}+ \sqrt{K\!(V(\psi))K\!(\tilde{V})}\| \tilde{V}-V(\psi) \|_\infty
        \\&\leqt{(viii)} \left(\varepsilon_{\mathrm{cov}}+ 4\Tr[\psi \hat{E}_n]\right)\varepsilon_{\mathrm{cov}}+ \|V(\psi)\|_\infty\|\tilde{V}\|_\infty    \varepsilon_{\mathrm{cov}}
        \\&\leqt{(ix)} \left(\varepsilon_{\mathrm{cov}}+ 4\Tr[\psi \hat{E}_n]\right)\varepsilon_{\mathrm{cov}}+ 4\Tr[\psi \hat{E}_n]\left(4\Tr[\psi \hat{E}_n]+\varepsilon_{\mathrm{cov}}\right)\varepsilon_{\mathrm{cov}}
        \\&\leqt{(x)}\left(\varepsilon_{\mathrm{cov}}+ 4nE\right)\left(1+ 4nE\right)\varepsilon_{\mathrm{cov}}
        \\&\le \left(1+ 4nE\right)^2\varepsilon_{\mathrm{cov}}\,.\label{eq:bigbound}
    \ee 
    Here, in (iii), we have exploited~\eqref{eq:bigbound} and we have introduced the notation $[A]_{k}$ to denote the $k\times k$ lower-right block of the matrix $A$. In (iv), first, we have used that for any $n\times n$ matrix $A$ and any $k\le n$ it holds that $\|[A]_k\|_\infty\le \|A\|_\infty$, and, second, we have exploited that $[D]_{n-t}=I_{n-t}$. In (v), we first used the submultiplicativity of the operator norm, and second, we applied Lemma~\ref{lemma_S_inv}, which implies that $\|\tilde{S}^{-1}\|_{\infty}=\|\tilde S\|_{\infty}$. In (vi), we have used that $\|\tilde{S}\|_{\infty}\le \sqrt{\|\tilde{V}\|_{\infty}}$, as established by Lemma~\ref{lem:SlowerboundV}, and that $\|V(\psi)-\tilde{V}\|_{\infty}\le \varepsilon_{V}$. In (vii), first we have used that $\|\tilde{V}\|_{\infty}\le \|V(\psi)\|_{\infty}+\varepsilon_V$, and second, we have exploited Lemma~\ref{lem:wolf} that establishes that the maximum difference between the ordered symplectic eigenvalues $\|\tilde{D}-D\|_{\infty}$ can be upper bounded in terms of the \emph{condition numbers} $K\!(V(\psi)),K\!(\tilde{V})$ (see Lemma~\ref{lem:wolf}) and in terms of $\|V(\psi)-\tilde{V}\|_{\infty}$. In (viii), we have used that $K\!(V)\le \|V\|_{\infty}^{2}$, as stated by Lemma~\ref{lem:lboundsVm1V}. In (ix), we have exploited that $\|V(\psi)\|_{\infty}\le 4\Tr[\psi \hat{E}_n]$, as established by Lemma~\ref{lem:SlowerboundV}. Finally, in (x), we have used Cauchy Schwarz inequality to deduce that $\Tr[\psi \hat{E}_n]\le \sqrt{\Tr[\psi \hat{E}_n^2]}$ and then we have used the assumption that $ \sqrt{\Tr[\psi \hat{E}_n^2]}\le n E$.

    Now, let us analyse the second term in the r.h.s.~of~\eqref{eq:3}. With similar techniques of~\eqref{eq:bigbound}, we can upper bound $\|\textbf{m}(\tilde{\psi}_t)\|_2$ as 
    \bb
        \|\textbf{m}(\tilde{\psi}_t)\|_2&=\left\|\textbf{m}\!\left(  U_{\tilde{S}}^\dagger \hat{D}^\dagger_{\tilde{\mathbf{m}}} \psi \hat{D}_{\tilde{\mathbf{m}}} U_{\tilde{S}}  \right)\right\|_2
        \\&=\left\|\tilde{S}^{-1}\left(  \mathbf{m}(\psi)   -\tilde{\mathbf{m}} \right)\right\|_2
        \\&\le \|\tilde{S}^{-1}\|_\infty \|   \mathbf{m}(\psi)   -\tilde{\mathbf{m}}  \|_2
        \\&\le \|\tilde{S}\|_\infty\, \varepsilon_{\mathrm{cov}}
        \\&\le \sqrt{\|\tilde{V}\|_\infty}\varepsilon_{\mathrm{cov}}
        \\&\le \sqrt{4n E+\varepsilon_{\mathrm{cov}}}\, \varepsilon_{\mathrm{cov}} 
        \\&\le \sqrt{4n E+1}\, \varepsilon_{\mathrm{cov}}\,. \label{eq:smallbound}
    \ee
    By putting~\eqref{eq:3},~\eqref{eq:bigbound}, and~\eqref{eq:smallbound} together, we deduce that $P_{\text{succ}}$ is lower bounded by
    \bb\label{P_succ_lower_bound}
        P_{\text{succ}}&\ge 1-\frac{2(n-t)\|V(\Tr_{[t]}\tilde{\psi}_t)-\mathbb{1}\|_\infty}{4}-\frac{\|\textbf{m}(\tilde{\psi}_t)\|_2^2}{2}
        \\&\ge 1-\frac{n-t}{2} (1+4nE)^2\varepsilon_{\mathrm{cov}}  - \frac{1}{2}(1+4nE)\varepsilon_{\mathrm{cov}}^2
        \\&\ge 1-\frac{(1+4nE)^2}{2}(n\,\varepsilon_{\mathrm{cov}}+ \varepsilon_{\mathrm{cov}}^2)
        \\&\ge 1-\frac{(1+4nE)^2(n+1)}{2}\varepsilon_{\mathrm{cov}}\,.
    \ee
    The choice of $\varepsilon_{\mathrm{cov}}$, which we will make later in~\eqref{choice_epsilon}, implies that $P_{\text{succ}}\ge\frac{3}{4}$. Consequently, thanks Lemma~\ref{le:enhance-success}, since we conduct in total $\left\lceil 2N_{\mathrm{tom,CV}}\!\left(t,\varepsilon_{\text{tom}},\frac{\delta}{3},E\right) + 24\log\!\left(\frac{3}{\delta}\right)\right\rceil$ measurements, each of which with a success probability of $P_{\text{succ}}\ge\frac{3}{4}$, the probability of obtaining at least $N_{\mathrm{tom,CV}}\!\left(t,\varepsilon_{\text{tom}},\frac{\delta}{3},E\right)$ successful outcomes is $\ge 1-\frac{\delta}{3}$.  In other words, the probability of getting at least $N_{\mathrm{tom,CV}}\!\left(t,\varepsilon_{\text{tom}},\frac{\delta}{3},E\right)$ post-measurement states $\ket{\tilde{\phi}}\otimes \ket{0}^{\otimes (n-t)}$ is $\ge 1-\frac{\delta}{3}$. From now on, let us assume that we are in such a probability event.
    In Line~\ref{steptd:6}, we do a step of the algorithm in Table~\ref{Table_tomography_ec} to perform pure-state tomography of the energy-constrained $t$-mode state vector $\ket{\tilde{\phi}}$, by considering the first $t$ modes of the post-measurement state vector $\ket{\tilde{\phi}}\otimes \ket{0}^{\otimes (n-t)}$. In order to apply such an algorithm we need to find an upper bound on the mean energy of $\ket{\tilde{\phi}}$. By denoting the mean energy of a state $\psi$ as $E(\psi)$, let us note that 
    \bb
        E(\tilde{\phi})&\leqt{(i)}\frac{1}{P_{\text{succ}}}E(\tilde{\psi}_t
        )
        \\&\leqt{(ii)} \frac{\|\tilde S\|_{\infty}^{2}}{P_{\text{succ}}}E\!\left(\hat{D}^{\dagger}_{\tilde{\mathbf{m}}}\ketbra{\psi} \hat{D}_{\tilde{\mathbf{m}}}\right)\\
        &\leqt{(iii)} \frac{\|\tilde S\|_{\infty}^{2} }{P_{\text{succ}}}\left[E(\psi)+\sqrt{2E(\psi)}\|\tilde{\mathbf{m}}\|_2+\frac{1}{2}\|\tilde{\mathbf{m}}\|_2^2\right]\\
        &\leqt{(iv)} \frac{\|\tilde S\|_{\infty}^{2} }{P_{\text{succ}}} \left[E(\psi)+\sqrt{2E(\psi)}(\|{\mathbf{m}(\psi)}\|_2+\varepsilon_{m})+\frac{1}{2}(\|{\mathbf{m}(\psi)}\|_2+\varepsilon_{m})^2\right]\\
        &\leqt{(v)} \frac{4E(\psi)+\varepsilon_{V}}{P_{\text{succ}}} \left[E(\psi)+\sqrt{2E(\psi)}(\sqrt{2E(\psi)}+\varepsilon_{m})+\frac{1}{2}(\sqrt{2E(\psi)}+\varepsilon_{m})^2\right]\\
        &\leqt{(vi)} \frac{4}{3}[4E(\psi)+\varepsilon_{V}]\left[4E(\psi)+2\sqrt{2E(\psi)}\varepsilon_{m}+\frac{\varepsilon_{m}^2}{2}\right]\\
        &\leqt{(vii)} 80\,[E(\psi)]^2\\
        &\le 80\, \Tr[\psi \hat{E}_n^2]\\
        &\le 80\,E^2 n^2 .
    \ee 
    Here, in (i), we have used that 
    \bb
        E(\tilde{\phi})&=\frac{1}{P_{\text{succ}}}E\!\left(\bra{0}^{\otimes (n-t)}\tilde{\psi}_t\ket{0}^{\otimes (n-t)}\right)
        \\& = \frac{1}{2P_{\text{succ}}}\Tr\!\left[ \sum_{i=1}^t R_i^2  \bra{0}^{\otimes (n-t)}\tilde{\psi}_t\ket{0}^{\otimes (n-t)}   \right]
        \\&=  \frac{1}{2P_{\text{succ}}}\Tr\!\left[ \left(\sum_{i=1}^t R_i^2\right) \otimes \ketbra{0}^{\otimes (n-t)}\tilde{\psi}_t    \right]
        \\&\le  \frac{1}{2P_{\text{succ}}}\Tr\!\left[ \left(\sum_{i=1}^t R_i^2\right) \otimes \mathbb{1}_{n-t}\,\tilde{\psi}_t    \right]
        \\&= \frac{1}{2P_{\text{succ}}}\Tr\!\left[ \left(\sum_{i=1}^t R_i^2\right) \,\tilde{\psi}_t    \right]
        \\&\le \frac{1}{2P_{\text{succ}}}\Tr\!\left[ \left(\sum_{i=1}^n R_i^2\right) \,\tilde{\psi}_t    \right]
        \\&=\frac{E(\tilde{\psi}_t)}{P_{\text{succ}}}.
    \ee
    In (ii), first we have used that $\tilde{\psi}=U_{\tilde S}^{\dag}\hat{D}^{\dag}_{\mathbf{\tilde{m}}}\psi \hat{D}_{\mathbf{\tilde{m}}}U_{\tilde S}$, and then we have used Lemma~\ref{lemma:energyboundsympl}. In (iii), we have used Lemma~\ref{lemma:energydispl}. Then, in (iv), we have used triangular inequality and the bound $\|\mathbf{m}(\psi)-\tilde{\mathbf{m}}\|_2\le \varepsilon_{\mathrm{cov}}$. In (v), we have used Lemma~\ref{lem:SlowerboundV}, together with $\|\tilde V-V\|_{\infty}\le \varepsilon_{\mathrm{cov}}$ and~\eqref{eq:bounddisplene}. In (vi), we have exploited that $P_{\text{succ}}\ge \frac{3}{4}$. In (vii), we have exploited that $\varepsilon_{\mathrm{cov}}\le 1$ and $E(\psi)\ge \frac{1}{2}$. Consequently, the mean energy per mode of the state $\tilde{\phi}$ is upper bounded by 
    \bb
        \frac{E(\tilde{\phi})}{t}\le \frac{80\,E^2 n^2 }{t}\le 80\,E^2 n^2\,,
    \ee
    and hence we can set $E\coloneqq 80\,E^2 n^2$.

    In Line~\ref{steptd:8}, the tomography algorithm called in the iterations of Line~6 outputs a $t$-mode state vector $\ket{\tilde{\phi}_1}$, which is supported on a $\left\lceil\left(\frac{e(80n^2E^2-\frac{1}{2})}{\varepsilon^2}\right)^{\!t}\right\rceil$-dimensional subspace of $L^2(\mathbb{R}^t)$. Since the total number of copies of $\ket{\tilde{\phi}}$ employed in such a tomography algorithm is $\ge N_{\mathrm{tom,CV}}\!\left(t,\varepsilon_{\text{tom}},\frac{\delta}{3},E\right)$, the state vector $\ket{\tilde{\phi}_1}$ satisfies
    \bb
        \Pr\left( \frac{1}{2} \left\| \ketbra{\tilde{\phi}} - \ketbra{\tilde{\phi}_1}  \right\|_1< \varepsilon_{\text{tom}} \right)\ge 1-\frac{\delta}{3}\,.
    \ee
    From now on, let us assume that we are in the probability event in which 
    \bb\label{assumption_tom_cv}
        \frac{1}{2}\left\| \ketbra{\tilde{\phi}} - \ketbra{\tilde{\phi}_1}  \right\|_1< \varepsilon_{\text{tom}}\,.
    \ee
In Line~\ref{steptd:9}, we output a classical description of the $n$-mode state vector $\ket{\hat{\psi}}$, which is defined as
    \bb
        \ket{\hat{\psi}}\coloneqq \hat{D}_{\tilde{\textbf{m}}}U_{\tilde{S}}\left(\ket{\tilde{\phi}_1}\otimes\ket{0}^{\otimes(n-t)}\right)  \,.
    \ee
Note that
\bb
    \frac{1}{2}\left\|\ketbra{\hat{\psi}}-\ketbra{\psi}\right\|_{1}&\eqt{(i)} \frac{1}{2} \left\| \ketbra{\tilde{\phi}_1}\otimes\ketbra{0}^{\otimes(n-t)}    - \ketbra{\tilde{\psi}_t}
    \right\|_{1}
    \\&\leqt{(ii)} \frac12\left\| \ketbra{\tilde{\phi}_1} - \ketbra{\tilde{\phi}}  
    \right\|_1   +   \frac12\left\|  \ketbra{\tilde{\phi}}\otimes \ketbra{0}^{\otimes(n-t)}   - \ketbra{\tilde{\psi}_t} \right\|_{1}\\
    &\leqt{(iii)} \varepsilon_{\text{tom}}   +  \sqrt{1-P_{\text{succ}}}\\
    &\leqt{(iv)} \varepsilon_{\text{tom}}+(1+4nE)\sqrt{\frac{n+1}{2}\varepsilon_{\mathrm{cov}}}\,.
\ee
Here, in (i), we have used unitarily invariance of the trace norm. In (ii), we have used the triangular inequality of the trace norm. In (iii), we have used~\eqref{assumption_tom_cv} and the \emph{gentle measurement lemma}~\cite{MARK}. Finally, in (iv) we have used~\eqref{P_succ_lower_bound}. Consequently, by choosing
\bb\label{choice_epsilon}
    \varepsilon_{\mathrm{cov}}&\coloneqq \frac{\varepsilon^2}{2(n+1)(1+4nE)^2}\,,\\
    \varepsilon_{\text{tom}}&\coloneqq \frac{\varepsilon}{2}\,,
\ee
we have that
\bb
    \frac{1}{2}\left\|\ketbra{\hat{\psi}}-\ketbra{\psi}\right\|_{1}\le \varepsilon\,.
\ee
By using an union bound, we can conclude that a total number of
\bb
    N_{\mathrm{cov}}\!\left(n,\varepsilon_{\mathrm{cov}},\frac{\delta}{3}, E\right)+\left\lceil 2N_{\mathrm{tom,CV}}\!\left(t,\varepsilon_{\text{tom}},\frac{\delta}{3},E\right) + 24\log\!\left(\frac{3}{\delta}\right)\right\rceil
\ee
copies of $\ket{\psi}$ suffices to construct a classical representation of a state vector $\ket{\hat{\psi}}$ such that
\bb     
    \Pr\left(\frac{1}{2}\left\|\ketbra{\hat{\psi}}-\ketbra{\psi}\right\|_{1}\le \varepsilon\right)\ge 1-\delta\,.
\ee
\end{proof}
We now show that any algorithm to learn $t$-compressible states should be inefficient if $t$ scales faster than a constant in the number of modes.
\begin{thm}[(Lower bound on sample complexity of tomography of $t$-compressible states)]\label{thm_lower_bound_t_compressible}
Let us consider a tomography algorithm that learns, within a trace distance $\le \varepsilon$ and failure probability $\le \delta$, an unknwon $n$-mode $t$-compressible state $\psi$ satisfying the second-moment constraint $\sqrt{\Tr[\psi\hat{E}_n^2]}\le nE$. Then, such a tomography algorithm must use a number of state copies $N$ satisfying
\bb\label{condition_sample_t_doped}
    N&\ge \frac{1}{t\,g\!\left(  \frac{n}{t}\!\left(E-\frac{1}{2}\right) \right)}\left[2(1-\delta)\left(\frac{ \frac{n}{t}\!\left(E-\frac{1}{2}\right)  }{12\varepsilon  }-\frac{1}{t}\right)^t-(1-\delta)\log_2(32\pi)-H_2(\delta)\right]\\
    &=\Theta\!\left(\frac{nE}{ t\varepsilon }\right)^{\!\!t}\,.
\ee
Here, $g(x)\coloneqq (x+1)\log_2(x+1) - x\log_2 x\,$ is the bosonic entropy, and $H_2(x)\coloneqq -x\log_2x-(1-x)\log_2(1-x)$ is the binary entropy.
\end{thm}
\begin{proof}
    Since the tomography algorithm can learn arbitrary $n$-mode $t$-compressible states $\psi$ satisfying the second-moment constraint $\sqrt{\Tr[\psi\hat{E}_n^2]}\le nE$, then it can also learn \emph{arbitrary} $t$-mode states $\ket{\phi_t}$ satisfying 
    \bb
        \sqrt{\Tr[\ketbra{\phi_t}\hat{N}_t^2]}\le n\!\left(E-\frac{1}{2}\right)\,.
    \ee
    Indeed, $\ket{\psi}\coloneqq\ket{\phi_t}\otimes\ket{0}^{\otimes(n-t)}$ is a $t$-compressible state with second-moment upper bounded by
    \bb
        \sqrt{\Tr\!\left[\ketbra{\psi}\hat{E}_n^2\right]}&=\sqrt{\Tr\!\left[\ketbra{\phi_t}\otimes\ketbra{0}^{\otimes(n-t)}\hat{E}_n^2\right]}\\
        &=\sqrt{\Tr\!\left[\ketbra{\phi_t}\otimes\ketbra{0}^{\otimes(n-t)}\left(\hat{N}_n+\frac{n}{2}\mathbb{1}\right)^2\right]}\\
        &=\sqrt{\Tr[\ketbra{\phi_t}\hat{N}_t^2]+n\Tr[\ketbra{\phi_t}\hat{N}_t]+\frac{n^2}{4}}\\
        &\le\sqrt{\Tr[\ketbra{\phi_t}\hat{N}_t^2]+n\sqrt{\Tr[\ketbra{\phi_t}\hat{N}_t^2]}+\frac{n^2}{4}}\\
        &=\sqrt{\Tr[\ketbra{\phi_t}\hat{N}_t^2]}+\frac{n}{2}\\
        &\le nE\,.
    \ee
Finally, Theorem~\ref{th:lowerboundtomohraphy} ensures that any algorithm that learns unknown $t$-mode states $\ket{\phi_t}$, satisfying the second-moment constraint
    \bb
        \frac{1}{t}\sqrt{\Tr[\ketbra{\phi_t}\hat{N}_t^2]}\le \frac{n}{t}\!\left(E-\frac{1}{2}\right),
    \ee
must use a number of state copies $N$ that satisfies the condition in~\eqref{condition_sample_t_doped}.
\end{proof} 
\end{document}